\documentclass[electronic]{kthesis}
\pdfoutput=1
\usepackage[pdftex]{graphicx}
   \DeclareGraphicsExtensions{.pdf,.jpeg,.png}
\usepackage{amsmath}
\usepackage{amsthm}
\usepackage{amsfonts}
\usepackage{amssymb}
\usepackage{amscd}
\usepackage[hypertexnames=false]{hyperref}
\usepackage{url}

\usepackage{longtable}
\usepackage{semantic}
\usepackage{listings}
\usepackage{alltt}
\usepackage[pdftex,dvipsnames]{xcolor}
\usepackage{array}
\usepackage{color,colortbl}
\usepackage{multirow}
\usepackage[swedish,english]{babel}
\usepackage[T1]{fontenc}
\usepackage{textcomp}
\usepackage{lmodern}
\usepackage[utf8]{inputenc}

\usepackage{xargs} 
\usepackage[textwidth=1.5cm,textsize=tiny]{todonotes}

\usepackage[pages=some,contents={},opacity=1.0,scale=1,angle=90]{background}
\usepackage{tikz}
\usetikzlibrary{decorations.pathmorphing}

\usepackage{makeidx}\makeindex
\newcommand*{\formatfirst}[1]{\MakeUppercase{#1}}
\makeatletter
\newcommand*{\mymacro}[1]{%
  \expandafter\formatfirst\expandafter{\@car #1\@empty\@nil}%
  \@cdr #1\@empty\@nil}
\newcommand*\capletter[1]{%
  \def\@myuppercasewords{\myuppercase@i#1 \@nil}%
    {\@myuppercasewords}\index{#1@\@myuppercasewords}}
\def\myuppercase@i#1 #2\@nil{%
  \mymacro{#1}%
  \ifx\\#2\\%
  \else
    \@ReturnAfterFi{%
      \space
      \myuppercase@i#2\@nil
    }%
  \fi} 
\long\def\@ReturnAfterFi#1\fi{\fi#1} 
\makeatother

\newcommand*\VerBar[2]{%
	\begin{tikzpicture}
	\fill[#1] (0,0)--(2,0)--(2,2)--(0,2)--cycle;
	\draw[white] (1,1)node{\rotatebox{-90}{\HUGE \textbf \thechapter}};
	\end{tikzpicture}
}

\usepackage{makeidx}  
\usepackage{epigraph}
\usepackage{ragged2e}
\usepackage[all]{xy}
\usepackage[acronym,nohypertypes={acronym},indexonlyfirst,nomain,toc]{glossaries}
\usepackage{latexsym}

\newcolumntype{P}[1]{>{\centering\let\newline\\\arraybackslash\hspace{0pt}\scriptsize}p{#1}}
\newcolumntype{R}[0]{>{\let\newline\\\arraybackslash\hspace{0pt}\scriptsize}r}
\newcolumntype{g}{>{\columncolor{gray!70}}p{\textwidth}}

\usepackage[framemethod=default]{mdframed}
\mdfdefinestyle{extensionstyle}{
	skipabove=0ex,
	skipbelow=1ex,
	innerlinewidth=0pt,
	middlelinewidth=0pt,
	outerlinewidth=0pt,
	backgroundcolor=lightgray!35,
	innertopmargin=1ex,
	innerleftmargin=0.65em,
	innerrightmargin=0.65em,
	innerbottommargin=1ex,
	rightline=false,
	leftline=false,
	bottomline = false,
	topline = false,
    frametitlerule=true,
    frametitlerulecolor=white,
    frametitlebackgroundcolor=gray,
    frametitlerulewidth=2pt
}

\usepackage{tcolorbox}

\newenvironment{definition}[1][Definition]{\begin{trivlist}
\item[\hskip \labelsep {\bfseries #1}]}{\end{trivlist}}
  
\definecolor{shadecolor}{RGB}{210,210,210}

 \usepackage{dcolumn}
 \usepackage{listings}
 \usepackage{caption}
 \usepackage{subcaption}

\usepackage{boxedminipage}
\usepackage{multirow}

\usepackage{paralist}
\usepackage{fancyvrb}

\DeclareFontFamily{U}{mathx}{\hyphenchar\font45}
\DeclareFontShape{U}{mathx}{m}{n}{<-> mathx10}{}
\DeclareSymbolFont{mathx}{U}{mathx}{m}{n}
\DeclareMathAccent{\widebar}{0}{mathx}{"73}

\usepackage{afterpage}

\makeatletter
\newcommand*{\declarecommand}{%
  \@star@or@long\declare@command
}
\newcommand*{\declare@command}[1]{%
  \provide@command{#1}{}%
  \renew@command{#1}%
}

\newcommand*\cleartoleftpage{%
  \clearpage
  \ifodd\value{page}\hbox{}\newpage\fi
}
\makeatother

\makeglossaries

\setsecnumdepth{subsection}

\newtheorem{preposition}{Proposition}
\newtheorem{theorem}{Theorem}
\numberwithin{theorem}{section}
\newtheorem*{theorem*}{Theorem}
\newtheorem{lemma}{Lemma}
\numberwithin{lemma}{section}

\newcommand{\SCTLR}{SCTLR}

\newcommand{\priv}{\mathit{mode}}
\newcommand{\regs}{\mathit{regs}}
\newcommand{\psrs}{\mathit{psrs}}
\newcommand{\coregs}{\mathit{coregs}}
\newcommand{\mem}{\mathit{mem}}
\newcommand{\pgtype}{\mathit{pgtype}}
\newcommand{\pgrefs}{\mathit{pgrefs}}
\newcommand{\mode}{\mathit{mode}}
\newcommand{\refs}{\mathit{refs}}
\newcommand{\Data}{\mathit{D}}
\newcommand{\sound}{\mathit{sound}}
\newcommand{\countt}{\mathit{count}}
\newcommand{\block}{\mathit{block}}
\newcommand{\fnsize}{\footnotesize \mathit{size}}
\newcommand{\fnref}{\footnotesize \mathit{ref}}

\newcommand{\negat}{$\neg$}

\newcommand{\sigmas}{$\sigma$}
\newcommand{\botm}{$\bot$}
\newcommand{\wrt}{\ensuremath{\texttt{write}_{32}}}
\newcommand{\rd}{\ensuremath{\texttt{read}_{32}}}
\newcommand{\mmu}[1]{\ensuremath{\texttt{mmu}_{#1}}}
\newcommand{\mathunderline}[2]{\ensuremath{\texttt{#1}_{#2}}}

\newcommand{\armv}{ARMv7}
\newcommand{\tuple}[1]{\langle #1 \rangle}
\newcommand{\RealStateSpace}{\Sigma}
\newcommand{\IdealStateSpace}{\mathcal{Q}}

\title{Secure System Virtualization: \\
     {End-to-End Verification of Memory Isolation}}
\author{Hamed Nemati}
\date{October 2017}
\thesistype{Doctoral Thesis}
\imprint{Stockholm, Sweden 2017}
\examen{teknologie doktorsexamen i datalogi}
\disputationsdatum{fredagen den 20 oktober 2017 klockan 14.00}
\disputationslokal{Kollegiesalen, Kungl Tekniska h\"{o}gskolan, Brinellv\"{a}gen 8, Stockholm}
\isbn{ISBN 978-91-7729-478-8}
\issn{ISSN 1653-5723}
\isrn{ISRN-KTH/CSC/A-{}-17/18-SE}
\trita{TRITA-CSC-A-2017:18}
\publisher{Universitetsservice US AB}
\address{KTH Royal Institute of Technology \\ School of Computer Science and Communication\\ 
SE-100 44 Stockholm\\
SWEDEN}
\kthlogo{kthlogo/kth_cmyk_comp_science_comm}

\makeglossaries

\begin{document}
\frontmatter
\maketitle

\newcommand{\sk}{security kernel}
\newcommand{\Sk}{Security kernel}
\begin{abstract}
Over the last years, {\sk}s have played a promising role in reshaping the landscape of platform security on today's ubiquitous embedded devices. {\Sk}s, such as separation kernels, 
enable constructing high-assurance mixed-criticality execution platforms. They reduce the software portion of the system's trusted computing base to a thin layer, which enforces isolation
between low- and high-criticality components. The reduced trusted computing base  minimizes the system attack surface and facilitates the use of formal methods to ensure functional correctness and security of the kernel.

In this thesis, we explore various aspects of building a provably secure separation kernel using virtualization technology. In particular, we examine techniques related to the appropriate management
of the memory subsystem. Once these techniques were implemented and functionally verified, they provide reliable a foundation for application scenarios that require strong guarantees of isolation and facilitate 
formal reasoning about the system's overall security. 

We show how the memory management subsystem can be virtualized to enforce isolation of system components. Virtualization is done by using direct paging that enables a guest software under 
some circumstances to manage its own memory configuration. We demonstrate the soundness of our approach by verifying that the high-%
level model of the system fulfills the desired security properties. Through refinement, we then propagate these properties (semi-) automatically to the machine-code of the virtualization mechanism.


An application of isolating platforms is to provide external protection to an untrusted guest operating system to restrict its attack surface.
We show how a runtime monitor can be securely deployed alongside a Linux guest on a hypervisor to prevent code injection attacks targeting Linux.
The monitor takes advantage of the provided separation to protect itself and to retain a complete view of the guest software.

Separating components using a low-level software, while important, is not by itself enough to guarantee security. Indeed, current processors architecture involves features, such as caches, 
that can be utilized to violate the isolation of components.  In this  thesis, we present  a new low noise attack vector constructed by measuring cache effects. The vector is capable of breaching isolation between components
of different criticality levels, and it invalidates the verification of software that has been verified on a memory coherent (cacheless) model. To restore isolation, we  provide a number of countermeasures. Further, 
we propose a new methodology to repair the verification of the software by including data-caches in the statement of the top-level security properties of the system.

\end{abstract}
\newpage

\begin{otherlanguage}{swedish}
  \begin{abstract}
Inbyggda system finns överallt idag. Under senare år har utvecklingen av platformssäkerhet för inbyggda system spelat en allt större roll. Säkerhetskärnor, likt isoleringskärnor, möjliggör konstruktion av tillförlitliga 
exekveringsplatformar ämnade för både kritiska och icke-kritiska tillämpningar. Säkerhetskärnor reducerar systemets kritiska kodmängd i ett litet tunt mjukvarulager. Detta mjukvarulager upprättar en tillförlitlig 
isolering mellan olika mjukvarukomponenter, där vissa mjukvarukomponenter har kritiska roller och andra inte. Det tunna mjukvarulagret minimerar systemets attackyta och underlättar användningen av formella metoder för
att försäkra mjukvarulagrets funktionella korrekthet och säkerhet.

I denna uppsats utforskas olika aspekter för att bygga en isoleringskärna med virtualiseringsteknik sådant att det går att bevisa att isoleringskärnan är säker. I huvudsak undersöks tekniker för säker hantering av
minnessystemet. Implementering och funktionell verifiering av dessa minneshanteringstekniker ger en tillförlitlig grund för användningsområden med höga krav på isolering och underlättar formell argumentation om att
systemet är säkert.

Det visas hur minneshanteringssystemet kan virtualiseras i syfte om att isolera systemkomponenter. Virtualiseringstekniken bakom minnessystemet är direkt paging och möjliggör gästmjukvara, under vissa begränsningar,
att konfigurera sitt eget minne. Denna metods sundhet demonstreras genom verifiering av att högnivåmodellen av systemet satisfierar de önskade säkerhetsegenskaperna. Genom förfining överförs dessa egenskaper 
(semi-) automatiskt till maskinkoden som utgör virtualiseringsmekanismen.

En isolerad gästapplikation ger externt skydd för ett potentiellt angripet gästoperativsystem för att begränsa den senares attackyta. 
Det visas hur en körningstidsövervakare kan köras säkert i systemet ovanpå en hypervisor bredvid en Linuxgäst för att förhindra kodinjektionsattacker mot Linux.
Övervakaren utnyttjar systemets isoleringsegenskaper för att skydda sig själv och för att ha en korrekt uppfattning av gästmjukvarans status.

Isolering av mjukvarukomponenter genom lågnivåmjukvara är inte i sig tillräckligt för att garantera säkerhet. Dagens processorarkitekturer har funktioner, som exempelvis cacheminnen, som kan utnyttjas för att kringgå 
isoleringen mellan mjukvarukomponenter. I denna uppsats presenteras en ny attack som inte är störningskänslig och som utförs genom att analysera cacheoperationer. Denna attack kan kringgå isoleringen mellan olika 
mjukvarukomponenter, där vissa komponenter har kritiska roller och andra inte. Attacken ogiltigförklarar också verifiering av mjukvara om verifieringen antagit en minneskoherent (cachefri) modell. För att motverka
denna attack presenteras ett antal motmedel. Utöver detta föreslås också en ny metodik för att återvalidera mjukvaruverifieringen genom att inkludera data-cacheminnen i formuleringen av systemets säkerhetsegenskaper.

  \end{abstract}
\end{otherlanguage}

\thispagestyle{empty}

\section*{Acknowledgements}

I would like to express my sincere gratitude and appreciation to my main supervisor, Prof. Mads Dam for his encouragement and support, 
which brought to the completion of this thesis. Mads' high standards have significantly enriched my research.

I am deeply grateful to my colleagues Roberto Guanciale, Oliver Schwarz (you are one of ``the best'' ;)), Christoph Baumann, Narges Khakpour, Arash Vahidi, and Viktor Do for their supports, 
valuable comments, and advices during the entire progress of this thesis.


I would also like to thank all people I met at TCS, especially
Adam Schill Collberg, Andreas Lindner, Benjamin Greschbach, Cenny Wenner, Cyrille Artho (thanks for reading the thesis draft and providing good suggestions), Daniel Bosk (trevligt att tr\"{a}ffas Daniel), Emma Enstr{\"o}m,
Hojat Khosrowjerdi, Ilario Bonacina, Jan Elffers, Jes{\'u}s Gir{\'a}ldez Cr{\'u},
Jonas Haglund (thanks for helping with the abstract), Joseph Swernofsky, Linda Kann, Marc Vinyals, Mateus de Oliveira Oliveira, Mika Cohen,
Mladen Mik{\v s}a, Muddassar Azam Sindhu, Musard Balliu, Pedro de Carvalho Gomes, Roelof Pieters (I like you too :)), Sangxia Huang,
Susanna \linebreak Figueiredo de Rezende, Thatchaphol Saranurak, Thomas Tuerk, and Xin Zhao for lunch company and making TCS a great place.

I am truly thankful to my dear friends Vahid Mosavat, Mohsen Khosravi, Guillermo Rodr{\'i}guez Cano, and Ali Jafari, who made my stay in Sweden more pleasant.

Last but not the least, huge thanks to my parents, my sister and Razieh for their endless love, moral support and encouragement, without 
which I would have never been able to complete this important step of my life. 

\newpage

\tableofcontents
\mainmatter

\part{Introduction and Summary}

\newacronym{mmu}{MMU}{Memory Management Unit}
\newacronym{os}{OS}{Operating System}
\newacronym{tcb}{TCB}{Trusted Computing Base}
\newacronym{cots}{COTS}{Commercial Off The Shelf}
\newacronym{risc}{RISC}{Reduced Instruction Set Computing}
\newacronym{mpu}{MPU}{Memory Protection Unit}
\newacronym{iommu}{IOMMU}{Input/Output Memory Management Unit}
\newacronym{mils}{MILS}{Multiple Independent Levels of Security}
\newacronym{vm}{VM}{Virtual Machine}
\newacronym{lsm}{LSM}{Linux Security Module}
\newacronym{vmi}{VMI}{Virtual Machine Introspection}
\newacronym{cfg}{CFG}{Control Flow Graph}
\newacronym{bap}{BAP}{Binary Analysis Platform}
\newacronym{bil}{BIL}{BAP Intermediate Language}
\newacronym{cnf}{CNF}{Conjunctive Normal Form}
\newacronym{smt}{SMT}{Satisfiability Modulo Theories}
\newacronym{lts}{LTS}{Labeled Transition System}
\newacronym{tls}{TLS}{Top Level Specification}
\newacronym{psos}{PSOS}{Provably Secure Operating System}
\newacronym{hdm}{HDM}{Hierarchical Development Methodology}
\newacronym{kit}{KIT}{Kernel for Isolated Tasks}
\newacronym{dia}{dia}{Direct Interaction Allowed}
\newacronym{ed}{ED}{Embedded Devices}
\newacronym{sva}{SVA}{Secure Virtual Architecture}
\newacronym{sgx}{SGX}{Software Guard Extensions}
\newacronym{dma}{DMA}{Direct Memory Access}
\newacronym{tocttou}{TOCTTOU}{Time-Of-Check-To-Time-Of-Use, asdasdasd}
\newacronym{isa}{ISA}{Instruction Set Architecture}
\newacronym{ast}{AST}{Abstract Syntax Tree}
\newacronym{dacr}{DACR}{Domain Access Control Register}
\newacronym{iot}{IoT}{Internet of Things}
\newacronym{gdb}{GDB}{GNU Debugger}

\newglossaryentry{angelsperarea}{
  name = $a$ ,
  description = The number of angels per unit area,
}
\newglossaryentry{numofangels}{
  name = $N$ ,
  description = The number of angels per needle point
}
\newglossaryentry{areaofneedle}{
  name = $A$ ,
  description = The area of the needle point
}

\renewcommand{\glossarypreamble}{This list contains the acronyms used in the first part of the thesis. The page numbers indicate primary occurrences.\\}
 
\printglossary[type=\acronymtype,title=Abbreviations]

\chapter{Introduction}\label{kappa:ch:intro}
\label{problemDomain}

Nowadays, for better or worse, the demand for new (embedded) devices and applications to manage all aspects of daily life from traffic control to public safety continues unabated.
However, as reliance on automated sensors and software has improved productivity and people's lives, it also increases risks of information theft, data security breaches, and vulnerability
to privacy attacks. This also raises concerns about illicit activities, sometimes supported by governments, to exploit bugs in systems to take over the control of security-critical
infrastructures or gain access to confidential information. Series of attacks that hit Ukrainian targets in the commercial and government sectors~\cite{symantec}, including power grid and 
the railway system, are examples reinforcing how vulnerable a society can be to attacks directed at its core infrastructure.

Ever-evolving cyber-attacks are expanding their targets, and techniques used to mount these attacks are becoming increasingly complex, large-scale, and multifaceted. The advent of new 
attack techniques in rapid succession and the number of daily uncovered zero-day vulnerabilities~\cite{zerodayvulnerability} are convincing evidences that if the approach  to 
security does not become more formal and systematic, building  trustworthy and reliable computing platforms seems as distant as ever.

\section{Overview}
In order to secure a system against attacks, including those that are currently unknown, one needs to consider the security of all layers from hardware to applications available to the end users. Nevertheless,  among all 
constituents, the correctness of the most privileged software layer within the system architecture, henceforth the \textit{kernel}, is an important factor for the security of the system.
This privileged layer can be the kernel of an \gls{os}  or any other software that directly runs on  bare hardware,  manages resources, and is allowed to execute processor privileged instructions. Therefore, any bug at 
this layer can significantly undermine the overall security of the  system.

It is widely acknowledged that the monolithic design and huge codebase of mainstream execution platforms such as \gls{cots} operating systems make them vulnerable to security
flaws~\cite{KargerS02, Kemerlis:2014:RRK:2671225.2671286, Xu:2015:CEU:2810103.2813637}. This monolithic design is also what makes integrating efficient and fine-grained security mechanisms into current OSs, if not 
impossible, extremely hard. The vulnerability of modern desktop platforms to malwares can be attributed to these issues, and it  reveals the inability of the underlying kernel to protect system components properly
from one another. In order to increase the security and reliability, existing operating systems such as Linux  use access control and supervised execution techniques, such as LSM~\cite{Wright:2002:LSM:647253.720287} or 
SELinux~\cite{Loscocco:2001:IFS:647054.715771},  together with their built-in process mediation mechanisms. However, while these techniques represent a significant step towards improving security, they are insufficient
for application scenarios which demand higher levels of trustworthiness. This is essentially because the kernel of the hosting OS 
as the \gls{tcb}\footnote{The TCB of a system is the set of all components that are critical to establishing and maintaining the system's security} of the system is itself susceptible to attacks which can compromise
these security mechanisms.

A practical solution to increase the system security and to mitigate the impact of any probable malfunctioning is reducing the size and  complexity of the kernel. Minimizing the kernel can be done by splitting it into smaller 
components with restricted\footnote{The least authority necessary for a component to function correctly.} privileges that are isolated from each other and can interact through controlled communication channels.
This concept is an old one. Systems based on such a design principle are usually said to implement the \gls{mils} philosophy~\cite{Alves-foss06themils}. \gls{mils} is a design
approach to build high-assurance systems and is the main idea behind the development of  \textit{{\sk}s}.

A \sk, in the context of this thesis, is an executive that partitions system components into different security classes, each with certain privileges, manages system resources,
allocates resources to \textit{partitions}, and controls the interaction of partitions with each other and the outside world. Among others, \textit{microkernels}~\cite{liedtke1995micro,heiser2010okl4, DBLP:conf/sosp/KleinEHACDEEKNSTW09}, 
\textit{separation kernels}~\cite{DBLP:conf/sosp/Rushby81,Heitmeyer:2008:AFM:1340674.1340715}, and security \textit{hypervisors}~\cite{rutkowska2010qubes,sHype,McDermott:2012:SVM:2420950.2421011}
are important representatives of {\sk}s.

\begin{figure}
\centering
  \includegraphics[width=1\linewidth]{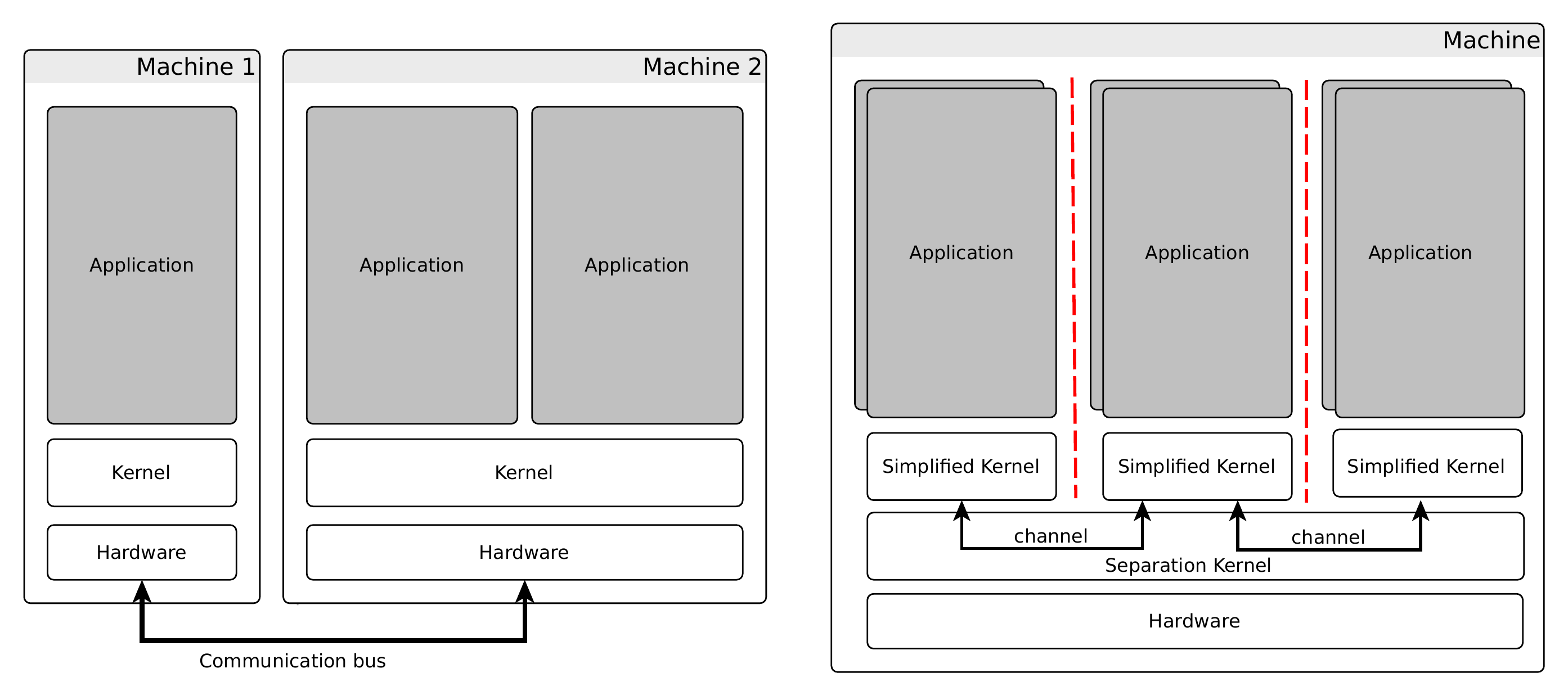}
  \caption{Traditional design (left) isolating components using discrete hardware platforms. A separation kernel (right) provides isolation on single processor.}
  \label{fig:separationKernel}
\end{figure}

Rushby was the first to introduce the notion of separation kernels~\cite{DBLP:conf/sosp/Rushby81}. A separation kernel resembles a physically distributed system
to provide a more robust execution environment for \textit{processes} running on the kernel~\cite{DBLP:conf/sosp/Rushby81} (cf. Figure~\ref{fig:separationKernel}).
The primary security property enforced by a separation kernel is 
\textit{isolation}; i.e., separation of processes in a system to prevent a running process from unauthorized reading  or writing 
into memory space and register locations allocated to other ones. 
Separation kernels reduce the software portion of the system's TCB to a thin layer which provides the isolation between partitions and implements communication channels.

The reduced TCB  minimizes the attack surface of the system and enables the use of rigorous inspection techniques, in the form of a formal proof, to ensure bug-freeness of the partitioning layer.
Formal verification lends a greater degree of trustworthiness to the system and ensures that --- within an explicitly given set of assumptions --- the system design behaves as expected.

Separating components using a low-level system software, assisted with some basic hardware features, reduces overhead in terms of the hardware complexity and increases modularity. 
However, it is not by itself enough to guarantee security. In fact, current processor architectures involves a wealth of features that, while invisible to
the software executing on the platform, can affect the system behavior in many aspects. For example, the \gls{mmu} relies on a caching mechanism 
to speed up accesses to \textit{page-tables} stored in the memory. A cache is a shared resource between all partitions, and it thus potentially  affects and is affected by activities of 
each partition. Consequently, enabling caches may cause unintended interaction between software running on the same processor, which can lead to cross-partition information leakage
and violation of the principle of isolation. Therefore, it is essential to consider the impact of such hardware pieces  while designing and verifying isolation solutions. 

\section{Thesis Statement}
In this thesis, we investigate various aspects of building a provably secure separation kernel using \textit{virtualization technology} and \textit{formal methods}. In particular, we examine
techniques related to the appropriate  management of the memory subsystem. Once these techniques were implemented and verified, they provide reliable mechanisms for application scenarios that
require strong guarantees of isolation and facilitate formal reasoning about the system functional correctness and its security. We are concerned with understanding how formal verification 
influences the design of a separation kernel and its trustworthiness, and what the impact of resource sharing  on the verification of the system is (cf. Figure~\ref{fig:contribution}).

\begin{figure}[]
\centering
  \includegraphics[width=0.5\linewidth]{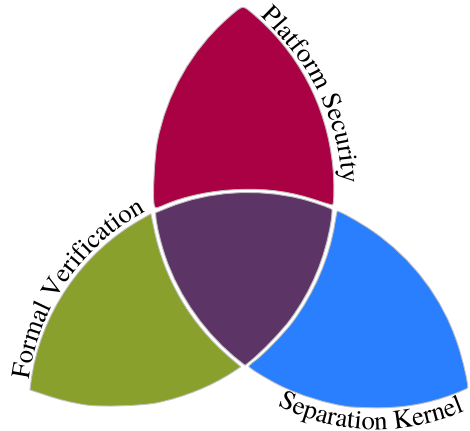}
  \caption{Contribution areas.}
  \label{fig:contribution}
\end{figure}

We limit our discussion to single-core processors to provide a detailed analysis, rather than dealing with complexities which arise due to the use of multiple cores.
More specifically, in this  work we tried to address the following issues:

\begin{itemize}
 \item How can we build a light-weight and performant virtualized memory management subsystem that allows dynamic management of the memory and is small enough to make formal 
 verification possible? This is important since  supporting dynamic memory management is essential to run commodity operating systems on separation kernels.

 \item How can we formally ensure that machine code executing on hardware implements the desired functionalities and complies with the security properties verified at the source-code level? 
 This obviates the necessity of trusting the compilers and provides a higher degree of security assurance.
 
 \item How can we take advantage of isolation of system components on a separation kernel to guarantee that only certified programs can execute on the system?
 
 \item What are the effects of enabling low-level hardware features such as caches on the security of the system? How can one neutralize these effects to ascertain proper isolation of the
 system components?
 \end{itemize}


\section{Thesis Outline}
This thesis is divided into two parts. The first part provides the required background to understand the results, and gives a summary of the papers. This part is structured as 
follows: Chapter~\ref{kappa:ch:background}  discusses the general context and different approaches to formal verification and process isolation, describes the tools we have used and provides 
the formalization of our main security
properties; Chapter~\ref{kappa:ch:relatedWorks} gives an overview of related work, including some historical background and recent developments; Chapter~\ref{kappa:ch:contributions}
describes briefly each of the included papers together with the author's individual contribution; Chapter~\ref{kappa:ch:conclusions} presents concluding remarks. 
The second part includes four of the author's papers.
\chapter{Background}\label{kappa:ch:background}
Embedded computing systems face two conflicting phenomena. First, security is becoming an increasingly important aspect of new devices such as smartphones and \gls{iot} systems. Secondly,
the design of current operating systems and hardware platforms puts significant limitations on the security guarantees that these systems can provide. The problem arises mainly because:

\begin{itemize}
 \item Commodity operating systems are mostly designed with focus on usability rather than security and inadequately protect 
 applications from one another, in the case of mission-critical scenarios. Therefore, bugs in one application can potentially lead to the compromise of the whole system. This
 limits the capability of the system to securely execute in parallel applications with different criticality levels.


 \item Mainstream operating systems  are very complex programs with large software stacks. This complexity undermines the system trustworthiness~\cite{KargerS02, Kemerlis:2014:RRK:2671225.2671286, Xu:2015:CEU:2810103.2813637}
 and  imposes significant hurdles  on building applications with high-criticality on these platforms. 
 
 \item Current platforms do not provide the end users with any reliable method to verify the identity of applications.
 For example, a user has no way of verifying if they are interacting with a trusted banking application or with a malicious or compromised program which impersonates the original one.
 
 \item Finally, modern processor architectures comprises a wealth of features that, while important for performance, could potentially open up
 for  unintended and invisible paths for sensitive information to escape.
\end{itemize}

A possible mitigation of these problems is employing special-purpose closed platforms~\cite{Garfinkel:2003:TVM:945445.945464}, customized for particular  application scenarios, e.g., missile controlling systems and game
consoles. The design of such closed systems, both at the hardware level and for the software stack,  can be carefully tailored to meet required security guarantees, thus significantly reducing vulnerability to attacks 
and the cost of formal verification. 

Despite multiple security benefits of using closed platforms, in most cases, flexibility and rich functionalities offered by general-purpose open systems make the choice of these systems
preferable. In this thesis, we look for a compromise between these two rather conflicting approaches. We try to keep a well-defined structure (i.e., minimal and suitable for a formal analysis)
and security guarantees of the closed systems while taking the advantages of using a general-purpose open platform. Our solution is based on a separation kernel that
brings together the capabilities of these two systems. Further, we apply mathematical reasoning to show correctness  and security of our solution.
Even though the focus of the work in this thesis is on design and verification of a system software for embedded devices, namely ARM processors, our results can be adapted to other
architectures. 


In the following, we give some preliminaries that will be used in this thesis. In particular, 
Section~\ref{sec:platform:security} briefly surveys techniques used in practice to provide a trusted execution environment for high-criticality applications on commodity platforms,
introduces the concept of security kernel and compare it to other approaches to clarify the advantage of using such kernels.
This section also shows how a security kernel can be used to help a commodity OS to protect itself against external attacks and explains why enabling low-level hardware features can potentially affect the security of
a system. Section~\ref{sec:fromal:verification} introduces the notion of formal analysis and explains techniques that can be used to verify a software.  Additionally, in this section, we 
formalize properties which we have used to analyze the security of a separation kernel and explain our verification methodology.

\section{Platform Security}\label{sec:platform:security}

\subsection{Process Protection}\label{sec:platform:security:informal:separation}
Anderson~\cite{Anderson72computersecurity} argued that resource sharing is the key cause of many security issues in operating systems, and programs together with their assigned resources must be securely compartmentalized 
to minimize undesired interaction between applications.
The classical method of preventing applications concurrently executing on the same processor from interfering with one another is \textit{process\footnote{A process is an instance of an executable program in 
the system.} isolation}~\cite{Anderson72computersecurity}. Isolation is a fundamental concept in platform security, and it aims to segregate different applications' processes to prevent the private data of a process
from being written or read by  other ones; the exception is when an explicit communication channel exists.
Process isolation preserves the \textit{integrity} of processes and guarantees their \textit{confidentiality}.

\paragraph{Integrity} refers to the ability of the system to prevent corruption of information, as both program and data. Corruption can happen either intentionally with the
aim of benefiting from information rewriting or erasure, or unintentionally, for example, due to hardware malfunction. The integrity property restricts who can create and modify  
trusted information.

\paragraph{Confidentiality} is the property of preventing information disclosure; i.e., making information unreachable by unauthorized actors.
This property covers both the existence of information and its flow in the system. 

Isolation of processes can be done by hardware or software based solutions; however both these approaches try to limit access to system resources and to keep programs isolated to their assigned resources.
In what follows we briefly describe some of the widely applied techniques to achieve isolation.

\subsubsection{Memory Management Unit}
The notion of isolation in current operating systems is provided at a minimum by abstraction of virtual memory, constructed through a combination of core hardware functionalities and kernel level
mechanisms. For all but the most rudimentary processors, the primary device used to enable isolation is the memory management unit, which provides in one go both virtual addressing and memory
protection. Operating systems use the MMU to isolate processes by assigning to each  a separated virtual space. This prevents errors in one user mode program from being propagated within the system.
It is the role of the kernel to configure the MMU to confine memory accesses of less privileged applications.
Configurations of the MMU, which are also called page-tables, determine the binding of physical memory locations to virtual addresses and hold  restrictions
to access  resources. Hence page-tables are critical for security and must be protected. Otherwise, malicious processes can bypass the MMU and gain illicit accesses.

\begin{figure}
\centering
  \includegraphics[width=0.4\linewidth]{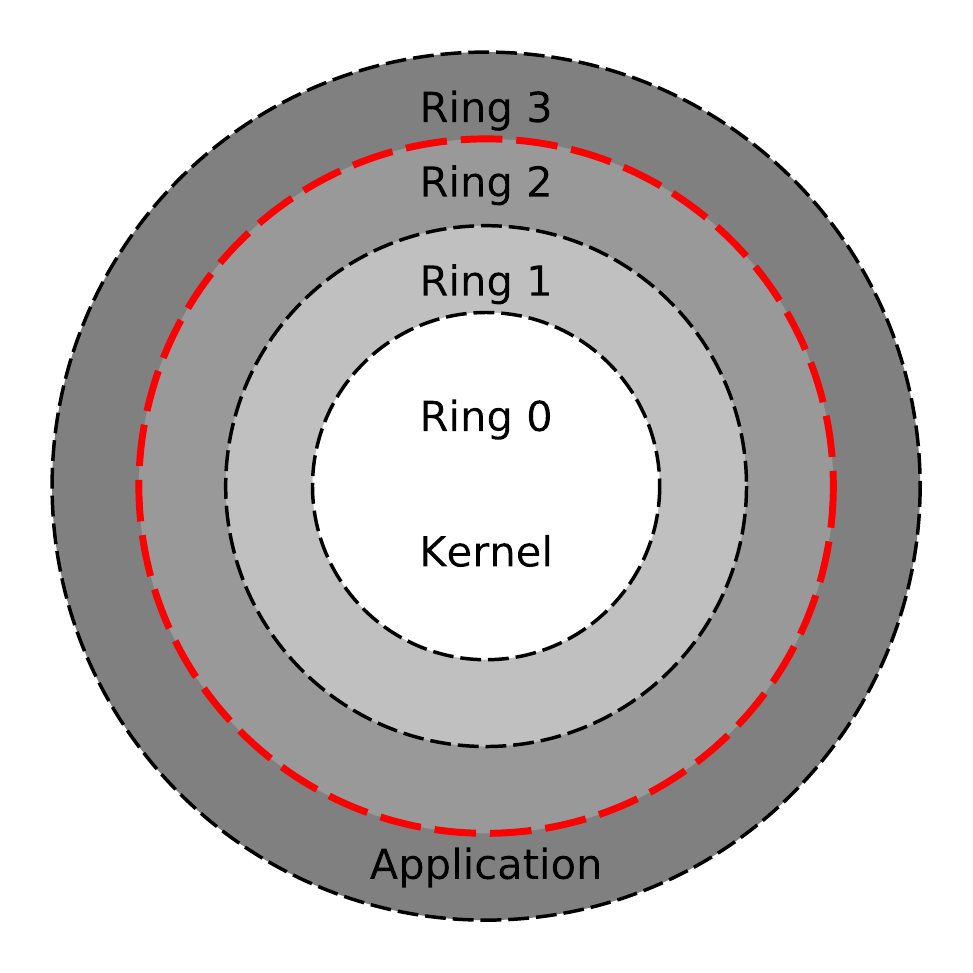}
  \caption{Hierarchical protection domains. Ring $0$ denotes the most privileged domain.}
  \label{fig:ring}
\end{figure}

 To protect security-critical components  such as  configurations of the MMU, most processors allow execution of 
processes with different criticality in separate \textit{processor modes} or \textit{protection rings} (c.f. Figure~\ref{fig:ring}).
On these processors, a privileged program such as the kernel of an OS runs in a special kernel domain, also called kernel space/mode, that is protected from the user specific domains, also referred to as user space/mode,
in which less privileged and potentially malicious applications execute.

Some processors such as ARM family CPUs --- instead of providing a memory management unit --- protect system resources through a \gls{mpu}. The MPU is a light weight specialization of the MMU which
offers fewer features and is mainly used  when the software executing on the platform is far simpler than off-the-shelf operating systems.
In contrast to the MMU which uses hardware protection together with the virtual memory capability, the MPU provides only hardware protection over software-designated memory regions.

Process isolation solely based on these approaches has been proven inadequate to guarantee  isolation in most cases. Due to the security concerns related to weak isolation, research 
on platform security has been focused on complementary techniques to strengthen isolation of processes, e.g., \textit{access control} mechanisms, \textit{sandboxing}, and \textit{hardware based} solutions.

\subsubsection{Access Control}
To improve the basic protection provided by the virtual memory abstraction, traditionally operating systems also use access enforcement methods to validate processes' request
(e.g., read, write) to access resources  (e.g., files, sockets). Each access control mechanism consists of two main parts: 
an \textit{access policy store} and an \textit{authorization module}. The access policy describes the set of allowed operations that processes can perform on resources and is
specific to each system. An example of such a policy  is Lampson's access matrix~\cite{Lampson:1974:PRO:775265.775268}.
At the heart of this protection system is the authorization module which is commonly referred to as \textit{reference monitor}. 
For a given input request, the reference monitor returns a binary response showing if the request is authorized by the monitor's policy.
AppArmor~\cite{Bauer:2006:PPI:1149826.1149839} is an example of such an access control mechanism, which is used 
in some Unix-based operating systems.

%
%

\subsubsection{Sandbox Based Isolation}
Sandboxing, as defined in~\cite{Wahbe:1993:ESF:168619.168635}, is the method of encapsulating an unknown and potentially malicious code in a region of memory to  restrict its impact on the system state.
Accesses of a sandboxed program are limited to only memory locations that are inside the assigned address space, and the program cannot execute binaries not placed in its code segment.
In this approach effects of a program are hidden from the outside world. However, for practical reasons sometimes this isolation has to be violated to allow data transmission.

\textit{Instruction Set Architecture} based sandboxing is the technique of controlling activities of a program at the instruction level, through adding instructions to  the binary of the program to check its 
accesses. Software Fault Isolation~\cite{Wahbe:1993:ESF:168619.168635} and Inline Reference Monitors~\cite{Erlingsson:2000:IEJ:867132}  are two examples of sandboxes that are implemented using this technique.
Application sandboxing can also be achieved  by taking control of the interface between the application and its libraries or the underlying operating system~\cite{Wagner:1999:JAC:894360}. Further, one can 
restrict permissions of a program to access system resources using access control mechanisms~\cite{chroot,Cheswick92anevening}.

\subsubsection{Hardware-Extension Based Isolation}
Process isolation through features embedded in hardware provides a strong form of separation. These features are either part of the processor
or hardware extensions that augment basic protection supplied with the CPU. Examples of these features include {IOMMU}~\cite{Ben-yehuda06utilizingiommus}, {ARM TrustZone}~\cite{TrustZone}, and 
{Intel \gls{sgx}}~\cite{mckeen2013innovative, anati2013innovative}.
The \gls{iommu} is a hardware extension to control memory accesses of I/O devices.
The IOMMU  prevents malicious devices from  performing arbitrary \gls{dma} operations  and can be used to isolate device drivers.
\begin{figure}
\centering
  \includegraphics[width=0.6\linewidth]{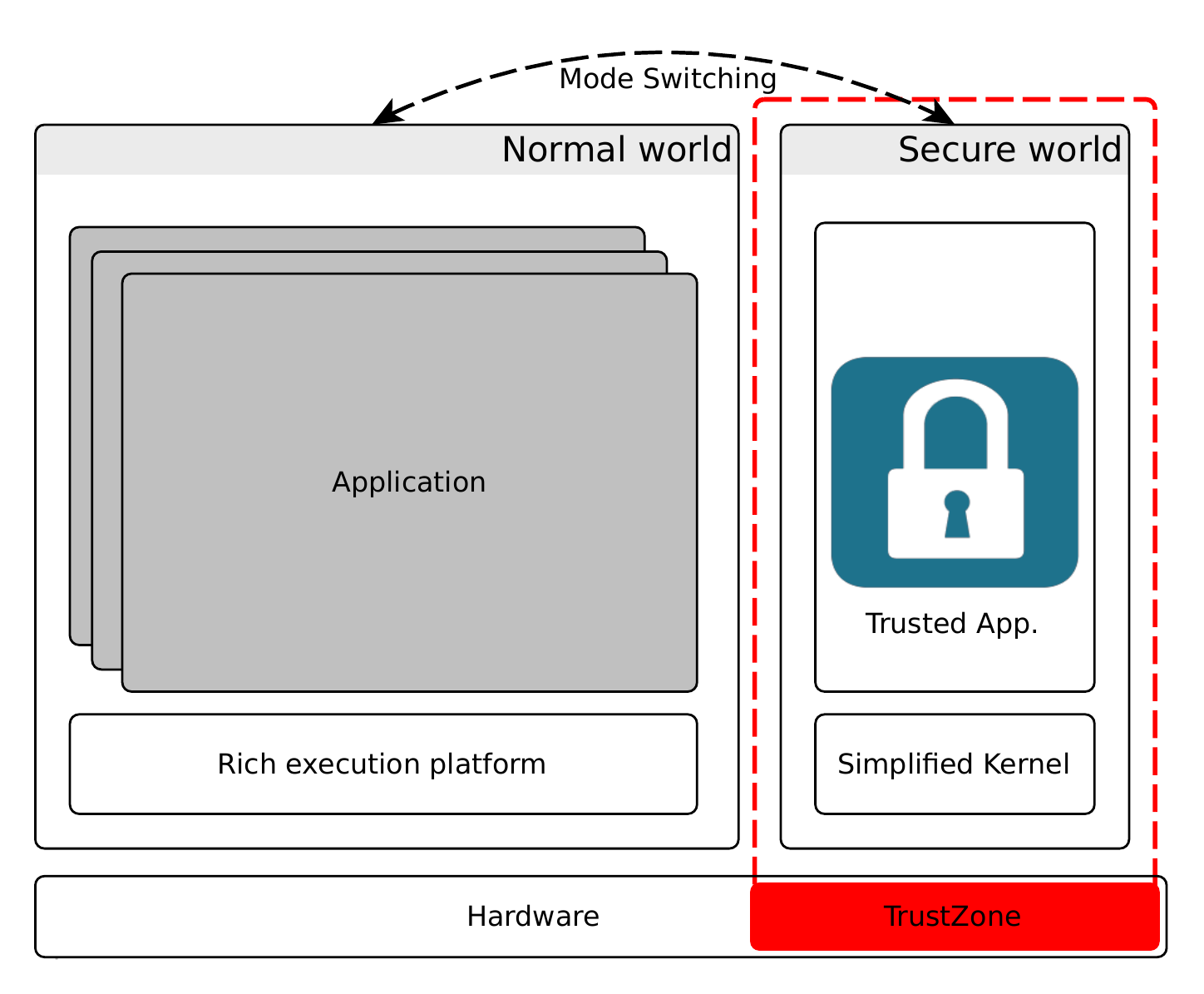}
  \caption{ARM TrustZone}
  \label{fig:trustzone}
\end{figure}

TrustZone (c.f. Figure~\ref{fig:trustzone}) is a set of security extensions added to some ARM architecture CPUs to improve the system security. TrustZone creates a secure execution environment, also called \textit{secure world},
for software that must be isolated from less critical components. These extensions allow the processor to switch between two security domains that are orthogonal to the standard 
capabilities of the CPU. TrustZone can be used to execute an unmodified commodity OS in the less secure domain, and to run security-critical subsystems (e.g., cryptographic algorithms)
inside the secure domain. This hardware feature guarantees that the critical subsystems will remain safe regardless of  malicious activities influencing the system outside the 
\textit{secure world}.

Intel's SGX are extensions to x86 processors that aim to guarantee integrity and confidentiality of the code executing inside SGX's secure containers, called 
\textit{enclaves}. The main application of the SGX is in \textit{secure remote  computation}. In this scenario the user uploads his
secret data and computation method into an enclave and SGX guarantees the confidentiality and integrity of  the secret user data  while the computation is being carried out.


\subsection{Strong Isolation and Minimal TCB}
Process isolation is essential to platform security. Nevertheless, in today's increasingly sophisticated malware climate, methods commonly used to achieve isolation, such as the
ones we have discussed above, are gradually proving insufficient.
The main typical problems related to the kernel-level solutions like access control and sandboxing mechanisms is that they enlarge the trusted computing base of the system
and their trustworthiness depends heavily on the characteristics of the underlying OS.
Language-based techniques~\cite{Schneider:2001:LAS:647348.724331} are mostly experimental, and tools available to enforce them are research prototypes and not applicable to large-scale
real-world systems. Using type systems~\cite{Schneider:2001:LAS:647348.724331} to enforce isolation is not feasible since most system software 
mix C (which is not generally a type-safe programming language~\cite{Necula:2002:CTR:503272.503286}) with assembly (which does not support fine-grained types).
Hardware-based solutions are helpful, but they do not solve the problem entirely either. They increase the bill-of-materials costs, 
and (potential) bugs in their implementation can be exploited by attackers to violate isolation enforced using these features~\cite{shen2015exploiting,costan2016intel,intelsgxvulnerability,sparks2007security,Gotzfried:2017:CAI:3065913.3065915}.

Isolation can be most reliably achieved by deploying high- and low-criticality components onto different CPUs. This, however, leads to higher
design complexity and costs. These make such an approach less appealing 
and emphasize the significance of software based solutions again. {\Sk} (e.g., microkernels~\cite{liedtke1995micro,heiser2010okl4, DBLP:conf/sosp/KleinEHACDEEKNSTW09}, separation
kernels~\cite{INTEGRITY}, and hypervisors~\cite{xen,Seshadri:2007:STH:1294261.1294294}) are recognized as
practical solutions to mitigate the problems of the aforementioned techniques that bring together the isolation of dedicated hardware and the enjoyments of having a small TCB.

\subsubsection{Microkernels}
The primary objective of microkernels is to minimize the trusted computing base of the system while consolidating both high- and low-criticality components on a single processor. 
This is usually done by retaining inside the most privileged layer of the system only those kernel services that are security-critical such as memory and thread management subsystems 
and inter-process communication. Other kernel level functionalities can then be deployed as user space processes with limited access rights.
Using a microkernel the user level services are permitted to perform only accesses that are deemed secure based on some predefined policies. 

A fully-fledged operating system can be executed on a microkernel by delegating  the process management of the hosted OS completely to the microkernel (e.g., L$^4$Linux) through  mapping the guest's 
threads directly to the microkernel threads. However, this generally involves an invasive and error-prone OS adaptation
process. Alternatively, the microkernel can be extended to virtualize the memory subsystem of the guest OS (e.g., using shadow-paging~\cite{Alkassar:FMCAD2010-b} or
nested-paging~\cite{Bhargava:2008:ATP:1346281.1346286}).

\subsubsection{Separation Kernels}
Separation kernels are software that enforce separation among system components and are able  
to control the flow of information between partitions existing on the kernel~\cite{DBLP:conf/sosp/Rushby81}. Programs running on separation kernels should behave equivalently as they were executing on distributed
hardware. Communication between partitions  is only allowed to flow as authorized along well-defined channels. Similar to microkernels, separation kernels implement the \gls{mils} philosophy~\cite{Alves-foss06themils}.

The idea of separation kernels is primarily developed to enforce (security-) isolation and many such kernels do not support essential functionalities to host a complete operating 
system including  device drivers and file systems. The commonly used technique to compensate this shortcoming is virtualization.
The advent of virtualization technologies provides significant improvements in efficiency and capabilities of the {\sk}. Virtualization allows building high-assurance systems
having the same functionalities as the commodity execution platforms. In this thesis, we use virtualization as an enabler for isolation and to implement a separation kernel capable of hosting
a complete operating system.

\subsubsection{Virtualization}
Virtualization, as it is used in this thesis, is the act of abstracting the underlying hardware to multiplex resources among multiple guests and security-check 
accesses to system resources before executing them.
The virtualization layer, which is also called \textit{hypervisor}, executes at the most privileged mode of the processor and can interpose between the guests and hardware.
This enables the hypervisor to intercept guests' sensitive instructions before being executed on hardware.
The complete mediation of events allows the creation of isolated partitions (sometimes referred to as \gls{vm}) in which applications with an unknown degree of trustworthiness
can execute safely. 
Platform virtualization can be done in many ways, two predominant approaches to virtualization are: \textit{full virtualization} and \textit{paravirtualization}.

\begin{figure} \footnotesize
\begin{subfigure}[t]{0.5\textwidth}
 \includegraphics[width=0.9\linewidth]{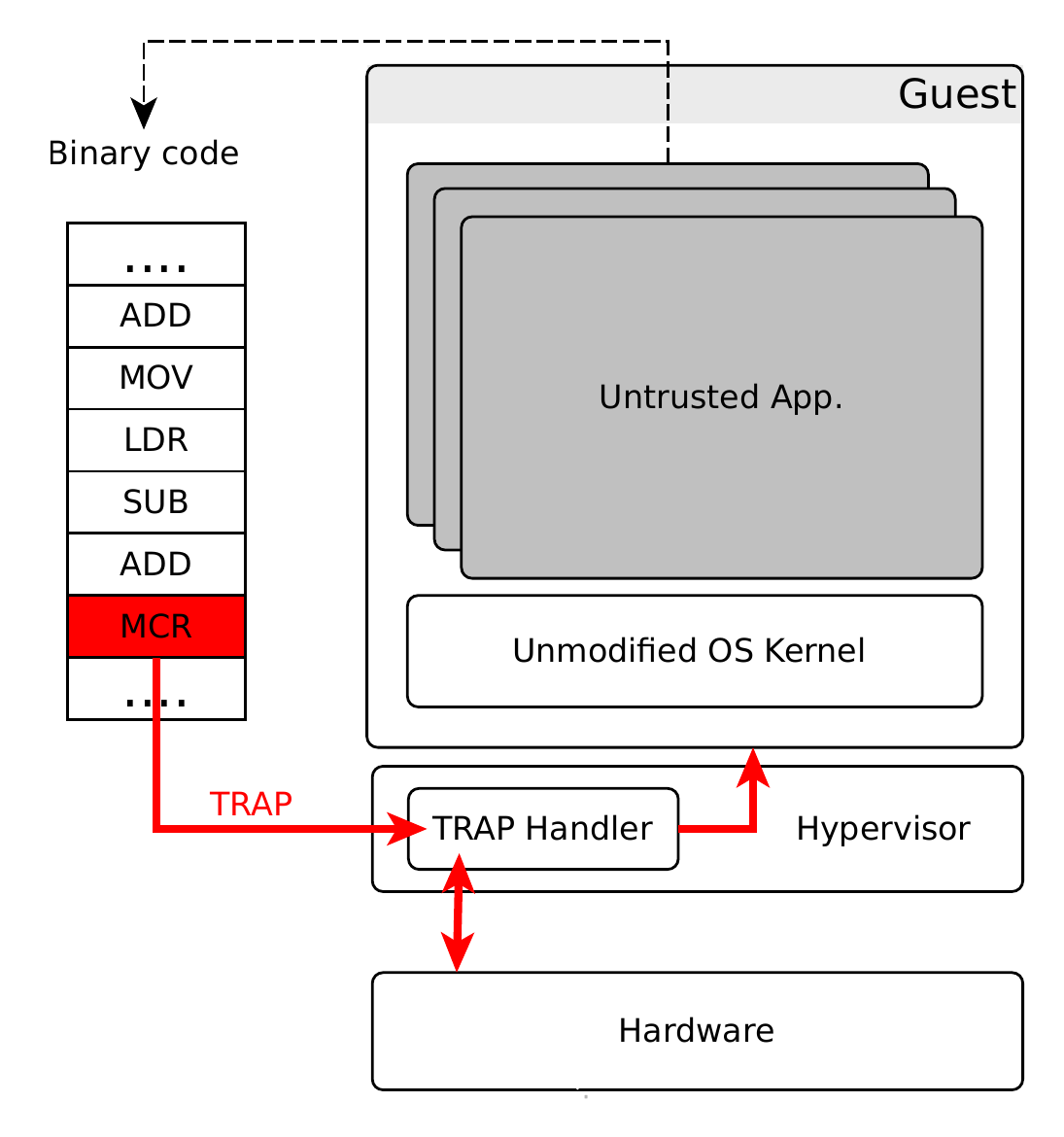}
\caption{Full Virtualization}
\label{fig:fullVirtualizationApproache}
\end{subfigure}
\begin{subfigure}[t]{0.5\textwidth}
\includegraphics[width=0.9\linewidth]{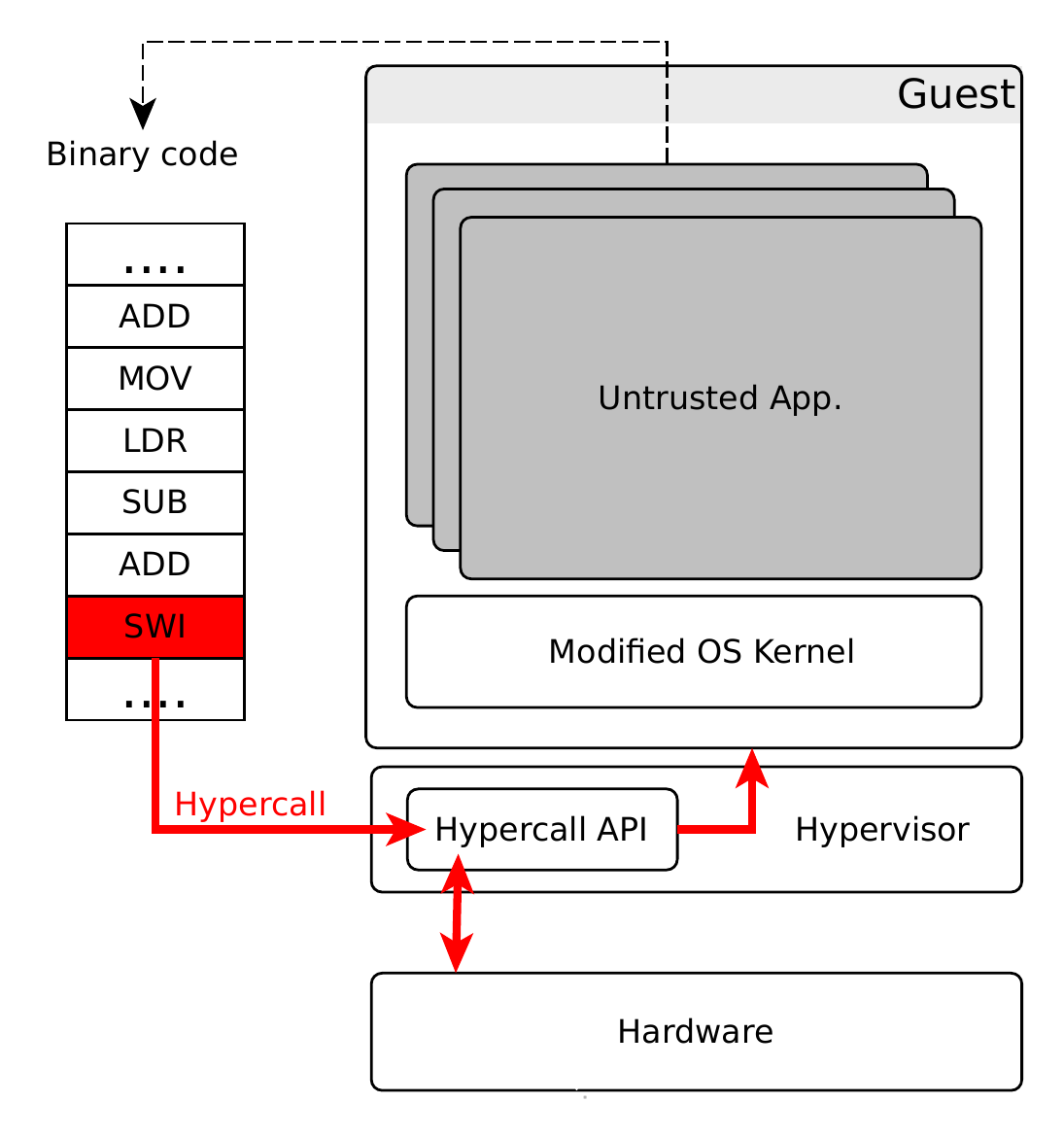}
\caption{Paravirtualization}
\label{fig:paraVirtualizationApproache}
\end{subfigure}
\caption{Different virtualization approaches.}
\end{figure}

\begin{itemize}
 \item Full virtualization is the technique of providing a guest software with the illusion of having sole access to the underlying hardware; i.e., the virtualization layer is transparent to the guest.
In this approach, the hypervisor resides at the highest privilege level and controls the execution of sensitive (or privileged) instructions. Whenever the deprivileged guest software tries to execute one of these sensitive 
instructions, the execution ``traps'' into the virtualization layer and the hypervisor emulates the execution of the instruction for the guest (c.f. Figure~\ref{fig:fullVirtualizationApproache}). The advantage of using
this approach is that  binaries  can execute on the hypervisor without any changes, neither the operating systems nor their applications need any adaptation to the virtualization layer. However, this technique 
increases complexity of the hypervisor design and introduces relatively high performance penalties.

 \item In contrast, in a paravirtualized system, the guest is aware of the virtualization layer. Using this technique, the execution of sensitive instructions in the guest software is replaced with an
explicit call, by invoking a \textit{hypercall} or a \textit{software interrupts}, to the hypervisor (c.f. Figure~\ref{fig:paraVirtualizationApproache}). Each hypercall is connected to a handler in
the hypervisor which is used to serve the requested service. Paravirtualization is proven to be more performant. However, it requires adaptation of the guest to the interface of the hypervisor, which can be a very
difficult task.
\end{itemize}

Paravirtualization can be implemented using the same hardware features that operating systems use. Nevertheless, efficient full virtualization of the system usually requires
support from hardware primitives such as the processor or I/O devices. Essential to full virtualization is the reconfiguration of the guests' privileges so 
that any attempt to execute  sensitive instructions traps into the hypervisor. This entails emulation of hardware functionalities, such as interrupt controller, within the hypervisor to 
allow execution of  hosted software inside partitions.
 
Modern processors provide features that can be used to simplify hypervisors design and increase their performance, e.g., extended privilege domains and 2-stage MMUs. However, since our main goal 
is  formal verification of the virtualization software we intentionally avoid using these features, which otherwise complicate formal analysis by shifting the verification burden from a flexible software 
to the fixed hardware.

In order for a hypervisor to host a general-purpose OS, it is generally necessary to allow the guest software to dynamically manage its internal memory hierarchy and to impose its own access restrictions. To achieve
this, a mandatory security property that must be enforced is the complete mediation of the MMU settings through virtualizing the memory management subsystem. In fact, since the MMU is the key functionality used by the
hypervisor to isolate the security domains, violation of this security property enables an attacker to bypass the hypervisor policies which could compromise the security of the entire system. This criticality is also
what makes a formal analysis of correctness a worthwhile enterprise.

Widely adopted solutions to virtualize the memory subsystem are shadow paging, nested paging, microkernels, and \textit{direct paging}. 
A hypervisor implemented using shadow paging keeps a copy of the guest page-tables in its memory to perform (intermediate-) virtual to physical address translation. 
This copy is updated by the hypervisor whenever the guest operates on its page-tables.
Nested paging, is a hardware-assisted virtualization technology which frees hypervisors from implementing the virtualization mechanism of the memory subsystem 
(e.g.,~\cite{hwang2008xen,McCoyd:2013:BHF:2470776.2471156}).
Direct paging was first introduced by Xen~\cite{xen} and is proved to show better
performance compared to other techniques. In paper~\ref{paper:JCS} we show a minimal and yet verifiable design of the direct paging algorithm and prove its functional correctness and the guarantees of isolation at the
machine code level.

\paragraph{PROSPER kernel}
Our implemented separation kernel is called PROSPER kernel, which  is a hypervisor developed using the paravirtualization technique to improve the security of embedded devices.
The hypervisor runs bare bone in the processor's most privileged mode, manages resource allocation and enforces access restrictions.
Implementation of the hypervisor targets BeagleBoard-xM and Beaglebone~\cite{bbone} equipped with an ARM Cortex-A8 processor (with no hardware virtualization extensions support) and allows execution of 
Linux (kernel 2.6.34 and 3.10) as its untrusted guest. 
Both user space applications and kernel services (or other trusted functionalities) on the PROSPER kernel execute in unprivileged mode. To this end, the hypervisor splits user mode in two virtual
CPU modes, namely \textit{virtual user mode} and \textit{virtual kernel mode}, each with their own execution context. The hypervisor is in charge of controlling context switching between these modes and
making sure that the memory configuration is setup correctly to enable separation of high- and low-criticality components.
In ARM these virtual modes can be implemented through ``domains''. These domains implement an access control regime orthogonal to the CPU's execution modes.
We added a new domain for the  hypervisor to reside in, and with the help of MMU configurations and the \gls{dacr}, the hypervisor can set the access
control depending on the active virtual guest mode.

The PROSPER kernel\footnote{A full virtualization variant of the PROSPER kernel was also developed by HASPOC project~\cite{7561034}, which supports additional features such as multicore
support and secure boot.} has a very small codebase, and its design is kept intentionally simple (e.g., the hypervisor does not support preemption) to make the formal verification of the 
hypervisor affordable.

\subsection{Secure Runtime Monitoring}

The increasing complexity of modern computing devices has also contributed to the development of new vulnerabilities in the  design and implementation of systems.
Among dangerous vulnerabilities are those that enable an adversary to impersonate trusted applications by either injecting malicious binaries into the executable memory of the applications or
causing anomalies in the system control flow to execute arbitrary code on target platforms.

\begin{figure} 
\begin{subfigure}[t]{0.5\textwidth}
 \includegraphics[width=0.95\linewidth]{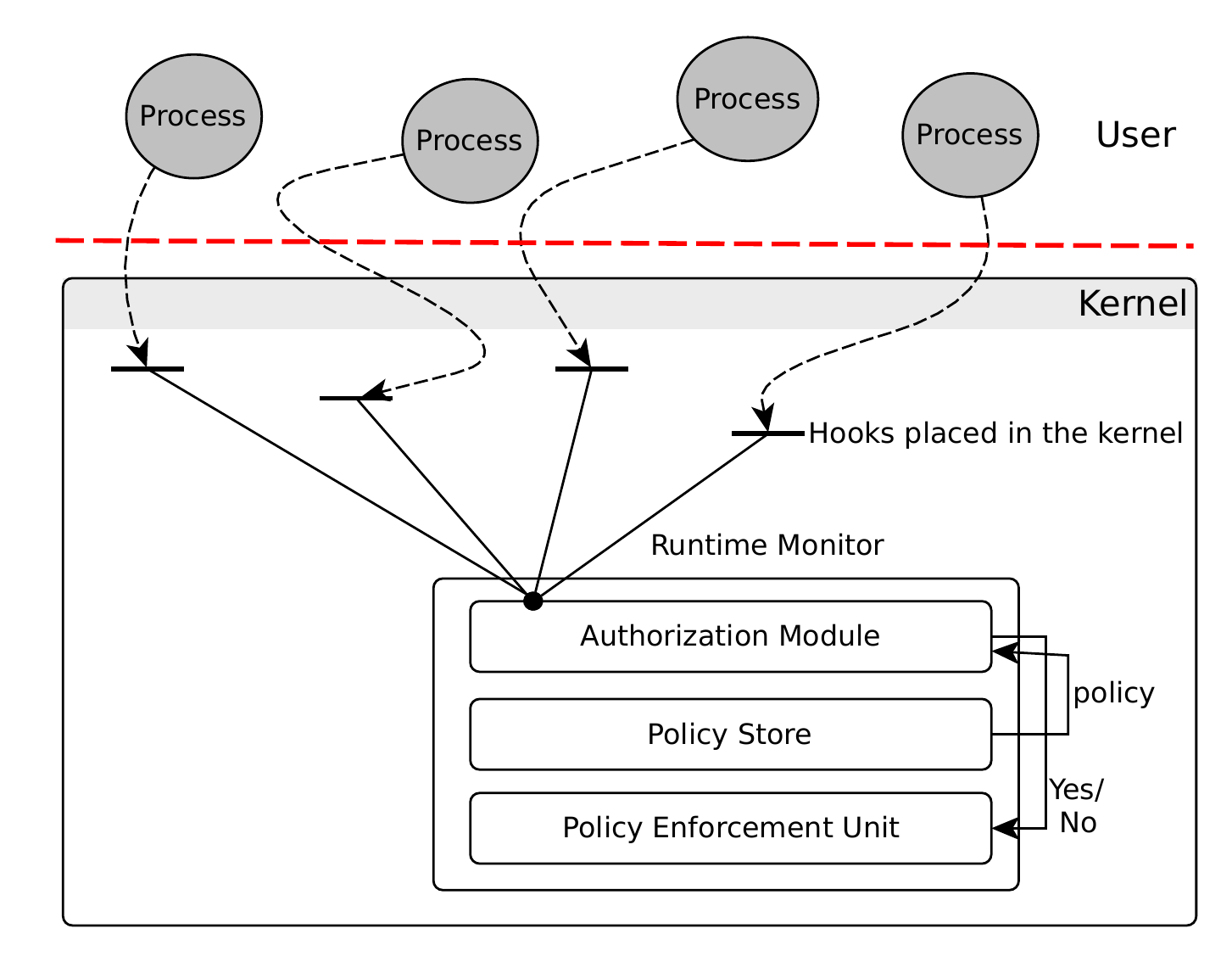}
 \caption{Classical runtime monitor}
 \label{fig:runtimeMonitorBasics}
\end{subfigure}
\begin{subfigure}[t]{0.5\textwidth}
\includegraphics[width=0.95\linewidth]{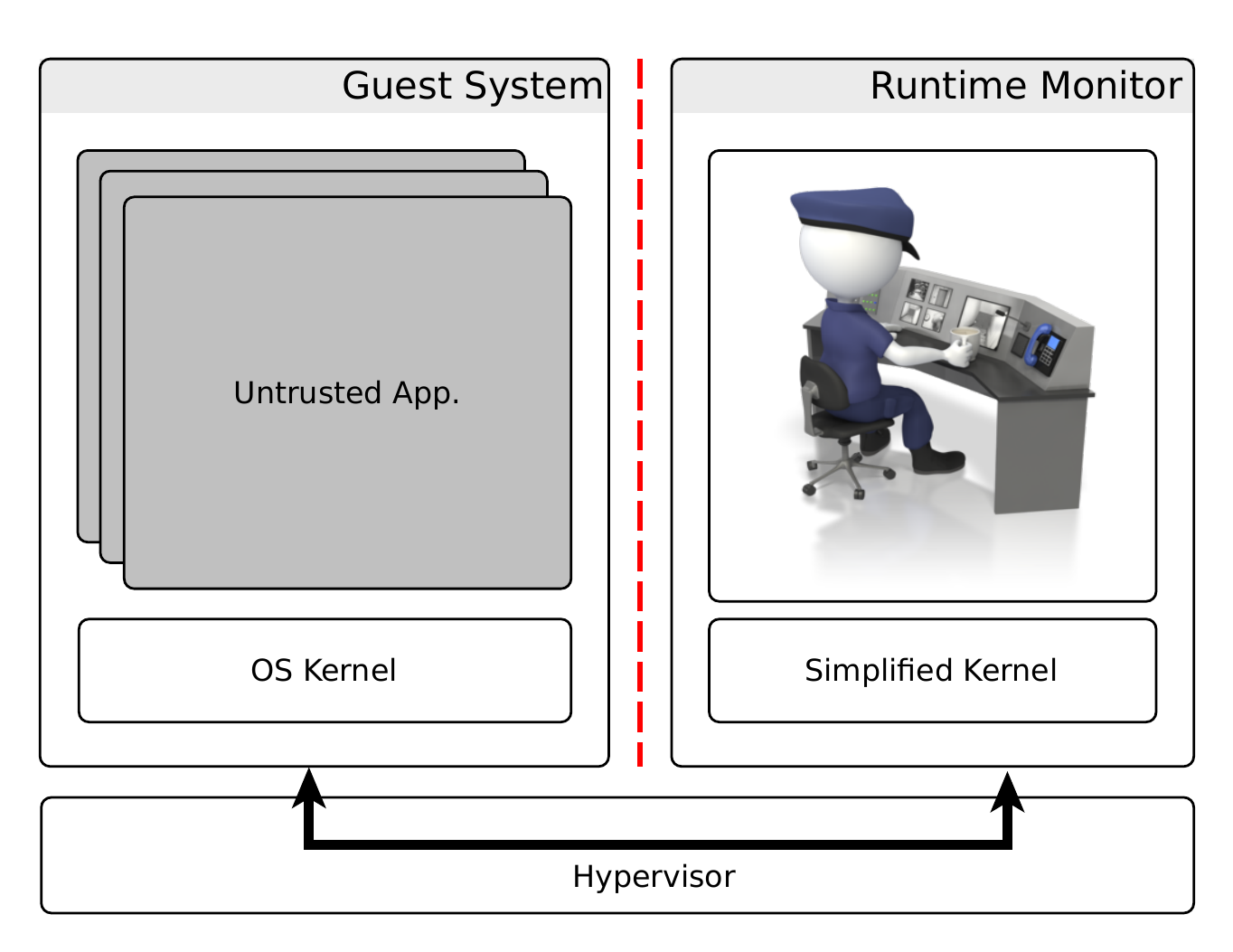}
 \caption{Secure runtime monitoring assisted with a hypervisor}
 \label{fig:runtimeMonitorVMI}
\end{subfigure}
\caption{Runtime monitor.}
\end{figure}

The success of these attacks mostly hinges on  unchecked assumptions that the system makes about executables in the memory. Therefore,
countermeasures against vulnerabilities of these types  include certifying all binaries that are safe to execute, and employing a monitoring mechanism which checks events at runtime to
prevent the execution of uncertified programs. 
A runtime monitor is a classical access enforcement method. A monitor checks validity of the requests to execute binaries
 through placing hooks (e.g., by patching \textit{system calls}) that invoke the monitor's authorization module  (c.f. Figure~\ref{fig:runtimeMonitorBasics}).
The main application domain for a runtime monitor is inside a rich execution environment such as a COTS OS, where the monitor runs in parallel with other applications in the system and
enforces some security policy.

Essential to making the monitoring process trustworthy and effective is to provide the monitor with a complete view of the system and to make it tamper resistant.
\gls{lsm}~\cite{Wright:2002:LSM:647253.720287} is an example of such a monitoring mechanism that is integrated into the Linux kernel. However, kernel-level techniques 
are not trustworthy as the hosting OS is itself susceptible to attacks~\cite{KargerS02}.

An interesting use-case scenario for isolation provided by a {\sk} is when the trusted isolated components are used as an aid for a commodity operating system to restrict 
the attack surface of the OS. In a virtualized environment the monitoring module can be deployed in a different partition (c.f. Figure~\ref{fig:runtimeMonitorVMI}). Since the hypervisor runs in 
most privileged mode, it has full control over the target operating system and can provide the monitor with a complete view of the events happening in the system. Such a virtualization assisted
monitoring was first introduced by Garfinkel and Rosenblum~\cite{DBLP:conf/ndss/GarfinkelR03}. \gls{vmi}, in the terminology of~\cite{DBLP:conf/ndss/GarfinkelR03}, places the monitoring 
subsystem outside of the guest software, thus making the monitoring module tamper proof. A further advantage of VMI-based solutions is that access enforcement can be done based on information
retrieved directly from the underlying hardware, which can not be tampered by an attacker. Security mechanisms that rely on the ability of observing the system state can also benefit from  
VMI's properties, e.g., isolation.

\subsection{Side-Channel Attacks}\label{background:attacksOnIsolation}
Despite ongoing efforts, developing trusted unbypassable mechanisms to achieve isolation remains a challenge. The problem arises due to the design of current hardware platforms and 
operating systems. Contemporary hardware platforms provide a limited number of resources such as  caches  that are shared among several processes by operating systems to enable
multitasking on the same processor.
While resource sharing is fundamental for the cost-effective implementation of system software, it is essential to do so carefully to avoid initiating 
unintentional channels  which may lead to the disclosure of sensitive information to unauthorized parties. This raises the potential of attack vectors that are not specified by the system specification
and some of which are available for user applications to exploit.

Broadly speaking, leakage channels are classified, based on the threat model, into two types:
\begin{itemize}
 \item[(i)] \textit{Covert-channels} that are channels used to deliberately transfer secret information to parties not allowed to access it by exploiting 
 hidden capabilities of system features, and
 \item[(ii)] \textit{Side-channels}  that refer to paths, which exist accidentally to the otherwise secure flow of data, for sensitive information to
 escape through.
\end{itemize}

In covert-channel attacks,  both sides of the communication are considered malicious. However, in
side-channels attacks, only the receiver has malicious intents and tries to get access to secret information through measuring  (unintended)
side effects of victim computations. Therefore, in COTS system software, we are mostly interested in studying side-channels.

Side-channels can be further subdivided into two groups, namely \textit{storage channels} and \textit{timing channels}.
Storage channels, in the current usage, are attacks conducted by exploiting aspects of the system that are directly observable by the adversary,
such as values stored in memory locations accessible by the attacker or registers content.
In contrast to storage channels, timing attacks rely on monitoring  variations in execution time to discover hidden hardware state. 
Timing channels are, in general, believed to have severe impact and they can occur even when the capability of the attacker to observe system 
resources is fully understood. However, the noise introduced by actuators operating on the system usually makes timing analysis hard.

One very important representative of side-channels are attacks based on measuring effects of caching on system performance.
Caches are hardware components that are widely adopted in computing devices and used to provide quicker response to memory requests to avoid wasting precious processor cycles.
Caches can hold recently computed data, not existing in the main memory, or they can be duplicates of original values in the memory. For
each memory access, if the requested data is in the cache (cache hit) the request can be served by simply reading the cache line containing data. 
If data is not present in the cache (cache miss), it has to be recomputed or loaded from its original storage location.
The MMU  controls accesses to caches, that is, it checks if data can be read from or written into the cache.

While enabling caches is important for performance, without proper management they can be used to extract sensitive information.
Cache timing side-channels are attack vectors that have been extensively studied, in terms of both exploits and countermeasures, cf.~\cite{StefanScheduling13,
Osvik:2006:CAC:2117739.2117741, AESattack,Kim:2012:SSP:2362793.2362804,GodfreyZ14,cock2014last}. However, cache usage has
pitfalls other than timing differentials. For instance, for the ARMv7 architecture, memory coherence may fail if the
same physical address is accessed using different cacheability attributes. This opens up for \gls{tocttou}\footnote{Time-Of-Check-To-Time-Of-Use is a class of attacks mounted by changing the victim system state between
the checking of a (security related) condition and the use of the results of that check.}-
like vulnerabilities since a trusted agent may check and later evict a
cached data item, which is subsequently substituted by an unchecked item placed in the main memory using an uncacheable alias. Moreover, an untrusted agent can similarly
use uncacheable address aliasing to measure which lines of the cache are evicted. This results in storage channels that are not visible in information flow analyses 
performed at the \gls{isa} level.

In practice, chip and IP manufacturers provide programming guidelines that guarantee memory coherence, and they routinely discourage the use of
mismatched memory attributes such as cacheability. However, in rich environments like hypervisors or OSs, it is often essential  to delegate the 
assignment of memory attributes to user processes that may be malicious, making it difficult to block access to the vulnerable features.

There are more side-channel attacks that deserve further studies, for instance, attacks based on analysis of
power, electromagnetic patterns, acoustic emanations.
However, exploring their impact on the system security is out the scope of this thesis.
Note also, in contrast to timing and storage channel attacks which can be conducted remotely~\cite{Brumley:2003:RTA:1251353.1251354} these attacks need physical access to the victim machine.
We are mainly interested in improving aspects of platform security that are relevant to the memory management subsystem.
In particular, in paper~\ref{paper:sp} we use the L1 data-cache to create low noise storage channels  to break isolation between system components.
Moreover, we show how these channels can be neutralized by exploiting proper countermeasures.

\subsection{ARM Architecture}

The ARM family processors are based on \gls{risc} architecture. The RISC architecture aims to provide a simple, yet powerful set of instructions that are able to execute 
in a single CPU cycle. This is achieved mostly by placing greater demands on the compiler to reduce the complexity of instructions that hardware executes.
In this thesis, we consider Harvard implementations of the ARM cores. The Harvard architecture uses separate buses, and consequently separates caches, for data and instructions
to improve performance (cf. Figure~\ref{fig:harvardArch}).

\begin{figure}[]
 \centering
 \includegraphics[width=0.5\linewidth]{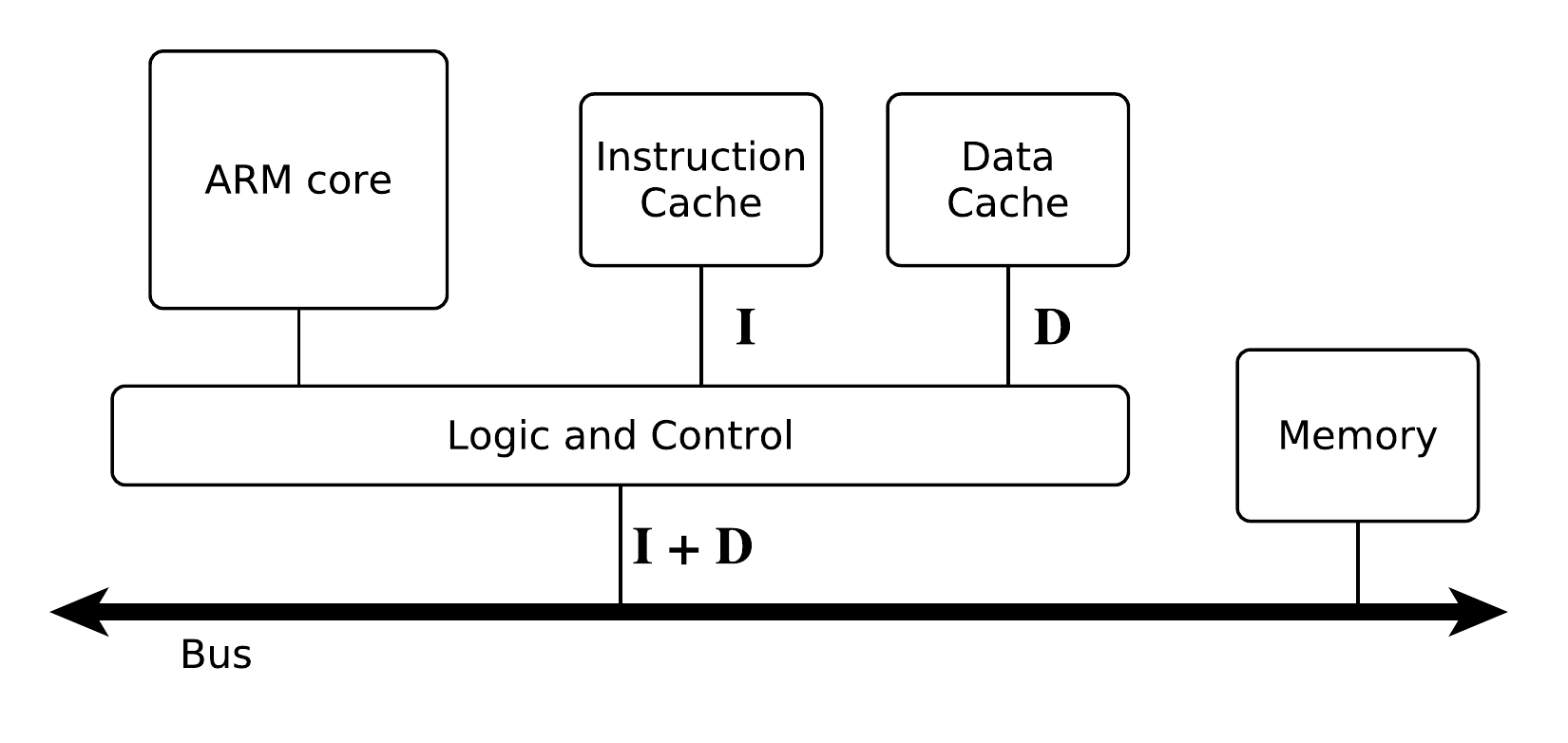}
 \caption{Harvard architecture with separate data and instruction caches.}
 \label{fig:harvardArch}
\end{figure}

The ARM architecture can operate in several execution modes; e.g., ARMv7 has seven modes. 
Among them, unprivileged user mode is used to execute user space applications, and the others are protected modes each having specific interrupt sources and are reserved for the privileged operations.
Privileged modes are accessible through the interrupt \textit{vector table}, which is a table associating interrupt handlers with the corresponding interrupt request.

ARM processors are \textit{load-store} architectures; i.e., the processor operates on data held in registers. Registers are hardware stores that act as the fast local memory and hold both data and addresses. 
Among others, general-purpose registers are accessible in all modes and of which the register 15 is the \textit{program counter} and contains the address of the next instruction that CPU will
execute, the register number 14 is called \textit{link register} and holds the return address of function calls, and 
the register 13 is used as the \textit{stack pointer}.
In addition to the general-purpose registers, ARM processors also include a number of \textit{control registers} that are used to determine the current execution mode, the active page-table,
and to control context switching between modes.

\section{Formal Verification}\label{sec:fromal:verification}
The increasing importance of {\sk}{s} in system security makes them an interesting target for attackers. This emphasizes the significance of applying formal methods 
to verify the correctness and isolation guarantees of these kernels. A key research contribution of this work is the formal verification
of the PROSPER kernel. We discuss our verification approach in papers~\ref{paper:JCS},\ref{paper:esorics},\ref{paper:csf} and provide more details on properties that we have verified.
In this section, we give a short exposition of our verification methodology and describe tools that are involved in this exercise.

Establishing trust on software can be achieved in a number of ways. We can make sure that the program design is fully understood and experienced developers are chosen to implement it;  we can 
conduct an extensive software testing procedure to check the quality of the program and its conformance with the system requirements; etc. Unfortunately, while these techniques are helpful to implement
high-quality programs with fewer bugs, still we keep finding errors in extensively adopted code, such as the binary search algorithm of the Java API~\cite{Norvig2006} and the open-source cryptographic software 
library OpenSSL~\cite{Heartbleed,Heartbleed2}. OpenSSL is behind many secure communication protocols over the Internet and the bug, dubbed Heartbleed, discovered in this cryptographic library could seriously compromise 
the security of systems using this protocol. Such incidents underline the need of using more rigorous methods to  verify the trustworthiness of programs that are security/safety-critical.

Over the last decades, formal verification has emerged as a powerful tool to provide enhanced trustworthiness to systems  software like hypervisors, microkernels, and separation kernels 
\cite{DBLP:conf/sosp/KleinEHACDEEKNSTW09,AlkassarHLSST09,WildingGRH10,dam2013formal, Heitmeyer:2008:AFM:1340674.1340715,zhao2011armor,SteinbergK10,GuVFSC11}. Formal verification provides strong guarantees backed by mathematical proofs
that all behaviors of a system meet some logical specification. 

There exist several approaches to formal verification with various degrees of automation. For example, \textit{model checking}~\cite{Emerson:2008:BMC:1423535.1423537} is a method of
automatically proving the correctness of a program based on a  logical specification, which expresses certain (temporal) properties of the program such as termination. While model checking is
a widely adopted verification approach in industry, it suffers from the \textit{state space explosion} problem\footnote{State space explosion refers to the problem that the memory needed to store the states required to 
model a system exceeds the available memory.},  and its use is limited to prove properties about programs with a finite state space. In what follows we briefly describe the techniques used to 
formally verify computer programs most relevant to the work presented in this thesis.

\newcommand{\htriple}[3]{\{\mathit{#1}\}~\mathit{#2}~\{\mathit{#3}\}}
\subsection{Static Program Analysis}
Program analysis is the process of examining a program to find bugs and to detect possible misbehavior. In general, research in program verification categorizes this process along two dimensions: dynamic vs. static,
and binary vs. source~\cite{Schwartz2014}. In static analysis, reasoning is done without actually running programs to  determine their runtime properties through analyzing the code structure. This is in contrast to 
the dynamic approach which verifies a program by inspecting values assigned to variables  while the program executes.

The main idea underlying static verification is to specify properties of programs by some \textit{assertions} and to prove that these assertions hold when the execution reaches them. 
Each assertion is a \textit{first-order} formula constructed using the program's constants, variables, and function symbols, and it describes logical properties of program variables.
In this approach, a program is a sequence of \textit{commands}  $\mathit{C}$ (or \textit{statements}). The correctness condition of each command is described by an annotation of the 
form $\htriple{P}{C}{Q}$, also called a \textit{Hoare-triple}~\cite{Hoare:1969:ABC:363235.363259}. We say the command $C$ is \textit{partially} correct 
if whenever ${C}$ is executed in a state satisfying the \textit{precondition} $\mathit{P}$ and if ${C}$ terminates,  then the state in which the execution of ${C}$ terminates must meet the \textit{postcondition} $\mathit{Q}$. 
By extending this notion, using an inference system, the entire program can be verified using this technique. The inference system is specific to the language of the program's source-code and consists of a set of axioms
and rules that allow to derive and combine such triples based on the operational semantics of the language.
Note that to prove the \textit{total correctness} of the program, proving that execution will eventually terminate is an additional proof obligation.

For a given postcondition, rules of this inference system also allow computing a \textit{weakest (liberal) precondition}. The weakest precondition ($wp$) computed  using this method
can be utilized to verify the triple $\htriple{P}{C}{Q}$ by checking the validity of a first-order predicate $P \Rightarrow wp(C,Q)$~\cite{Bradley:2007:CCD:1324777}. Such a predicate
often (e.g., if it is quantifier-free) can be resolved using an \text{SMT solver}.

\subsubsection{Binary Verification}
Verifying programs at the level of their source-code, while necessary, is not sufficient on its own. Indeed, most program verification (e.g., verification of a hypervisor) should reach the binaries (or machine-code)
of programs, as \textit{object code}
that  ultimately will execute on hardware. Machine-code needs to be examined to ensure that  properties established at the source-code level hold for programs binary as well. 

Binary analysis has its roots in work published by Goldstein and von Neumann~\cite{goldstine1947planning}, where they studied specification and correctness of machine-code programs.
Later, Clutterbuck and Carr\'{e}~\cite{6898} stressed the significance of formal analysis at the binary level and applied Floyd-Hoare-style verification condition generator to machine codes.
Similarly, Bevier~\cite{Bevier87averified} showed in his PhD thesis how  the kernel of an operating system can be verified down to its binary. Myreen~\cite{DBLP:conf/fmcad/MyreenGS08}  automates the whole
binary verification process inside the HOL4 theorem prover~\cite{hol4} by introducing a proof-producing decompilation procedure to transform machine codes to a function in the language of HOL4, which can be used
for reasoning.

At the binary level, machine state contains few components (e.g., memory and some registers and status bits), and instructions perform only very minimal and well-defined updates~\cite{Myreen2014}. Analysis of machine codes 
provides a precise account of the actual behavior of the code that  will execute on hardware. Programs' behavior at the level of source-code is generally undefined or vaguely defined; e.g., enabling some aggressive
optimization in C compilers  can lead to missing code intended to detect integer overflows~\cite{Wang:2013:TOS:2517349.2522728}. Moreover, low-level programs mix structured code (e.g., implemented in C) with assembly 
and use instructions (e.g., mode switching and coprocessor interactions) that are not part of the high-level language. This makes it difficult to use verification tools that target user space codes. A further 
advantage of using the binary verification approach is that it obviates the necessity of trusting compilers.

Despite the importance of machine-code verification, it suffers from some problems. Binary codes lack abstractions such as variables, types, functions, and control-flow 
structure, which makes it  difficult to use static analysis techniques to examine machine codes. Furthermore, the indirect jumps existing at the binary level prevent constructing a precise  \gls{cfg}
of programs, which are often used by static analyses. Tackling these sorts of problems researchers have been focused on developing techniques to over-approximate the CFG of programs~\cite{Bardin:2011:RCR:1946284.1946290}
and  to generate an intermediate representation of binaries  (or IL)~\cite{brumley2008analysis,Brumley:2011:BBA:2032305.2032342}. The IL representation is useful since it helps to mitigate complexities of modern
instruction sets and to restore required abstractions.
Paper~\ref{paper:JCS} elaborates on our solution to produce a certified IL representation of the PROSPER kernel and shows how we resolved indirect jumps at the kernel's binary to facilitate machine-code verification.

Binary verification can be automated to a large extent, for example, Bitblaze~\cite{Song:2008:BNA:1496255.1496257} and Binary Analysis Platform~\cite{Brumley:2011:BBA:2032305.2032342} are tools developed to automate 
verification of functional and safety properties at the machine-code level.

\subsubsection{Binary Analysis Platform}
\begin{figure}
 \centering
 \includegraphics[width=0.9\linewidth]{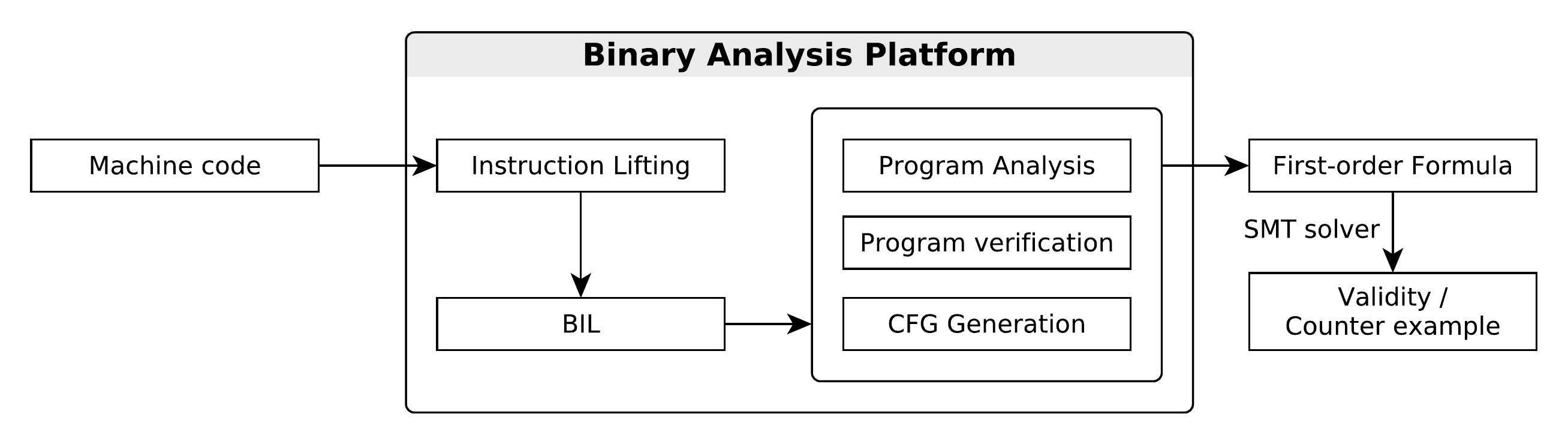}
 \caption{Binary Analysis Platform work-flow}
  \label{fig:bap}
\end{figure}

The \gls{bap}~\cite{Brumley:2011:BBA:2032305.2032342} (cf. Figure~\ref{fig:bap}) is a framework for analyzing and verifying binary codes. The BAP front-end includes tools to lift  input binaries to the \gls{bil}.
BIL provides an architecture independent representation which exposes all the side-effects of machine instructions.
BIL code is represented as an \gls{ast} that can be traversed and transformed using several methods to perform desired analysis.
 The back-end of BAP supports utilities to construct and analyze
control-flow graphs and program dependence graphs, to perform symbolic execution and to compute and verify contracts by generating the weakest preconditions.
The weakest precondition algorithm of BAP provides an effective method to speed up the verification of loop-free assembly fragments, by reducing the problem of verifying Hoare-triples to
proving a first-order formula. This formula can in many cases be validated by an external SMT solver~(e.g.,~\cite{DeMoura:2008:ZES:1792734.1792766}). 

\subsubsection{Satisfiability Modulo Theories Solver}
Deciding whether a formula that expresses a constraint\footnote{Constraints are formulas in \gls{cnf}.} has a solution (or a model) is a fundamental problem in the theoretical computer science. 
There are several problems that can be reduced to constraint satisfiability, including software and hardware verification,  type inference, static program analysis, etc.
Many of these problems, if expressed  using first-order Boolean formulas, can be solved by a \textit{Boolean Satisfiability} (a.k.a. SAT) solver.

A SAT solver determines whether a formula over boolean variables can be made true. For a first-order formula, this amounts to finding an interpretation of its variables,
function and predicate symbols which makes the formula true.
Generalizing concept of SAT solving, solvers for \gls{smt} can also be used to decide satisfiability of boolean formulas. SMT solvers
allow to include domain-specific theorems (e.g., real numbers or arrays theories) into the reasoning. This makes
\gls{smt} solvers more efficient for propositional formulas, but restricts application of these solvers to more specific areas.
 
 
 
\subsection{Interactive Theorem Proving}
\textit{Interactive theorem provers}, or proof-assistants, are computer programs commonly utilized as an aid for the human user to write and machine-check formal proofs.
Interactive provers rely on hints given by the user in constructing proofs rather than generating proofs fully  automatically.
In this approach, the user specifies the proof-structure  and provides some guidance communicated via a domain-specific language to the prover, while the machine checks the proofs  and uses the provided hints to
automate the process as much as possible. The automation can be achieved by proving proof slices (i.e., sub-goals) that are automatically inferable. 
In contrast to model checking, theorem proving can be used to verify programs with  probably infinitely many states and is the leading approach to deal with new verification challenges, such as verification
of system software.

\subsubsection{HOL4 Theorem Prover}
Among the most popular interactive theorem provers are Isabelle/HOL~\cite{Nipkow:2002:IPA:1791547}, HOL4 ~\cite{hol4}, and Coq~\cite{coq}.  HOL4 is a LCF-style~\cite{DBLP:books/sp/Gordon79,Plotkin:2000:PLI:345868} proof
assistant for Hight-Order-Logic built on a minimal proof kernel that  implements the axioms and basic inference rules. Higher-Order-Logic is an  extension of  first-order logic with
types and quantification over functions. HOL4 uses Standard ML as its meta-language~\footnote{The language used to implement the prover itself.} and provides a rich environment with a variety of libraries and theories 
to prove theorems and to implement proof tools. False statements have no possibility of being proved in HOL4. This is coded through the ML type system to force all proofs to pass the logical kernel of HOL4.

A user of HOL4 has the possibility of steering  the system via a number of constructs; namely (i) proof \textit{tactics}, which reduce a \textit{goal} (i.e. theorem to be proved) to simpler subgoals and are used to 
automate the process of theorem proving, (ii) proof \textit{rules}, which can be used to transform theorems to new ones, (iii) \textit{conversions}, which convert a logical expression into a theorem that establishes the 
equivalence  of the initial expression to the one computed by the applied conversion, and (iv) custom built-in ML programs. HOL4 uses backward reasoning to prove goals. Using this technique,
the input theorem is transformed into simpler subgoals by applying proper tactics.  Generated  subgoals can then be either directly discharged using a simplifier or will be further split into
simpler subgoals until they are discharged.

\newcommand{\intormem}[1]{\mathit{M_{#1}}}
\newcommand{\dmem}[1]{\mathit{M_{\{#1,d\}}}}
\newcommand{\bmem}[1]{\mathit{M_{\{#1,b\}}}}
\newcommand{\alphabet}{\mathit{E}}
\newcommand{\inevent}{\mathit{D}}
\newcommand{\exevent}{\mathit{P}}
\newcommand{\identifier}{\mathit{i}}
\newcommand{\transition}{\mathit{T}}
\newcommand{\SeqState}{\sigma}
\newtheorem{property}{Property}

\subsection{Security Properties}
In this subsection, we turn to formalize properties we have used to analyze the security of a hypervisor. The hypervisor plays the role of a {\sk} that provides a minimal software base to
enforce separation and control the flow of information within a single-core processor. Heitmeyer et al.~\cite{Heitmeyer:2008:AFM:1340674.1340715} showed that \textit{data-separation}
and \textit{information flow security} properties are strong enough to demonstrate isolation, modulo the explicit communication link, between partitions executing on the hypervisor. 

\subsubsection{Data Separation}
Informally, separation means prohibiting processes running in one partition from encroaching on protected parts of the system, such as the memory allocated to other guests or the hypervisor internal data-structures.
Separation, as discussed in Subsection~\ref{sec:platform:security:informal:separation}, can be expressed in terms of two properties, namely 
integrity and confidentiality or \textit{no-exfiltration} and \textit{no-infiltration} in the terminology of~\cite{Heitmeyer:2008:AFM:1340674.1340715}.

In the following, we first present formalization of the data-separation property as  presented in~\cite{Heitmeyer:2008:AFM:1340674.1340715}. 
Then, we instantiate the separation property with a model analogous to the model of our target platform to get properties similar to the ones we have shown in 
papers~\ref{paper:JCS},~\ref{paper:esorics}, and~\ref{paper:csf}.
To this end, we start with a gentle description of notions that are needed to understand the formalism.

We assume the behavior of the kernel is modeled as a state machine defined using a set of states $\SeqState \in \RealStateSpace$, an initial state $\SeqState_0$, an input alphabet $\alphabet$, and
a transition function $\transition: \RealStateSpace \times \alphabet \to \RealStateSpace$. 
Each partition is assigned an identifier $\identifier$ and a dedicated region of the memory $\intormem{\identifier}$ marked with the partition identifier~\footnote{While in~\cite{Heitmeyer:2008:AFM:1340674.1340715} an additional memory area is defined
that is shared among partitions, we skip presenting it here to simplify definition of the main properties. This also allows us to combine no-exfiltration with the kernel integrity property of~\cite{Heitmeyer:2008:AFM:1340674.1340715}.}.
The memory of partitions is further subdivided into two parts:  (i) a ``data memory'' area $\dmem{\identifier}$ which contains all the data-structures belonging to the partition $\identifier$,
and (ii) a number of ``input and output buffers'' $B_{\identifier} \in \bmem{\identifier}$ to communicate with other partitions.
The input alphabet consists of a number of internal actions $e_{\mathit{in}} \in \inevent$ and external actions $e_{\mathit{ext}} \in \exevent$.
The internal actions are used to invoke a service handler or to manipulate the data memory, and external actions are those
 that can be executed by fellow partitions or the kernel and have access to the communication buffers.
We use $\inevent_{i}$ ($\exevent_{i}$) to denote the internal (external) actions of the partition $i$ and $\inevent$ ($\exevent$) is the union of the internal (external) actions of all the existing
partitions in the system.
Moreover, the transition function $\transition$ transforms the system states by consuming the input events. 
Having defined this machinery, we can now proceed to give the formal account of no-exfiltration and no-infiltration properties.

No-exfiltration  guarantees the integrity of resources not allocated to the active partition (i.e., protected data).
This property is  defined in terms of the entire memory  $\intormem{}$, including the data memory and I/O buffers of all partitions, and says that: for every partition $i$,
event $e \in \inevent_{i} \cup \exevent_{i}$, and states $\SeqState$ and $\SeqState'$ such that $\SeqState'$ is reachable from $\SeqState$ by the event $e$, if 
a transition from $\SeqState$ to $\SeqState'$ changes the content of a memory location $m$, then $m$ is inside a memory region that is  modifiable by $i$.
\begin{property}(No-exfiltration)
 For all partition  $\identifier$, states $\SeqState$ and $\SeqState'$ in $\RealStateSpace$, event $e \in \inevent_{i} \cup \exevent_{i}$, and memory location $m \in \intormem{}$
 such that   $\SeqState' = \transition(\SeqState, e)$, if $m_{\SeqState} \neq m_{\SeqState'}$, then $m \in \intormem{\identifier}$.
\end{property}

On the other hand, \textit{no-infiltration} enforces confidentiality by ensuring that data processing of a partition is free of any influences from ``secret'' values stored in resources that 
are inaccessible to the partition.
No-infiltration is a 2-safety property~\cite{Clarkson:2010:HYP:1891823.1891830, Terauchi:2005:SIF:2156802.2156828} and requires reasoning about two parallel executions of the system. 
This property requires that for every partition $i$, event $e \in \inevent \cup \exevent$, and states $\SeqState_1$, $\SeqState_2$, $\SeqState'_1$, and $\SeqState'_2$ such that $\SeqState_2$ and $\SeqState'_2$ are, respectively,
reachable from $\SeqState_1$ and $\SeqState'_1$, if two executions of the system start in states having the same value in $m$ (which is located inside memory of the partition $i$),
after the event $e$, the content of $m$ should be changed consistently in both the final states.
\begin{property} (No-infiltration)
 For all partition  $\identifier$, states $\SeqState_1$, $\SeqState_2$, $\SeqState'_1$, and $\SeqState'_2$ in $\RealStateSpace$, and event $e \in \inevent \cup \exevent$, 
  such that $\SeqState_2 = \transition(\SeqState_1, e)$  and $\SeqState'_2 = \transition(\SeqState'_1, e)$, if for all $m \in \intormem{\identifier}$, $m_{\SeqState_1} = m_{\SeqState'_1}$ then
  it must hold that for all $m \in \intormem{\identifier}$, $m_{\SeqState_2} = m_{\SeqState'_2}$.
\end{property}

The no-exfiltration and no-infiltration properties impose constraints on the behavior of the active partition. On models allowing communication, an additional
property would be needed to restrict  effects of  the communication protocol on the partitions' state. \textit{Separation of control} expresses how information can flow  through
input/output buffers. In particular, this property says that the external events are only allowed to change the buffers of the receiving partition, and this write should not modify the private data memory of the receiver.
\begin{property} (Separation of Control)
 \label{prop:separation:of:control}
 For all partitions $\identifier$, $\identifier'$ and $\identifier''$, states $\SeqState$ and $\SeqState'$ in $\RealStateSpace$, and event $e \in \inevent \cup \exevent$ such that $\identifier'$ and
 $\identifier''$ are  the identifiers of the active partitions in $\SeqState$ and $\SeqState'$ respectively  and $\SeqState' = \transition(\SeqState, e)$, if $\identifier \neq \identifier'$ and 
 $\identifier \neq \identifier''$ then for all $m \in \dmem{\identifier}$, $m_{\SeqState'_1} = m_{\SeqState'_2}$.
\end{property}

This last property is not covered in this thesis. However, the writer co-authored the paper, ``\textit{Formal Verification of Information Flow Security for a Simple ARM-Based Separation Kernel}''~\cite{dam2013formal},
which presents the verification of the (simplified) PROSPER kernel for partitions communicating through \textit{message boxes} (a variant of input/output buffer). A property which subsumes Separation of Control and is
proved for the kernel in this verification exercise is to show that, while the hypervisor is able to change message boxes, it cannot modify the other state components belonging to guests.


The hypervisor we use here has the full control of the platform and provides for each guest a virtual space in much the same way that a process runs in an OS-provided virtual memory area.
The hypervisor hosts a few guests each assigned statically allocated and non-overlapping memory regions. 
Execution of the guests is interleaved and controlled by the hypervisor. 
That is, at the end of a specific time slice (when a timer interrupt happens) or when the active guest explicitly invokes a functionality of a fellow partition (through explicitly invoking a hypercall),
the hypervisor suspends the active guest and resumes one of idle guests (e.g., the one which hosts the functionality invoked by the active partition).  
The guests are virtualization aware (paravirtualized) software running entirely in unprivileged user mode and  
ported on the hypervisor to use the exposed APIs. Moreover, the only software that executes in privileged mode is the hypervisor and its execution cannot be interrupted, i.e., the hypervisor does not support preemption.

The models of our system in this thesis are (ARM-flavoured) instruction set architecture models. The memory in these models is partitioned into two segments, namely a code segment and a data segment.
Executable code (i.e., a sequence of instructions) resides in the memory code segment, and the processor fetches and executes instructions according to the value of a hardware store, called
\textit{program counter}. In such models the data memory $\dmem{-}$ is the aggregation of both the code and data memories, and  events can be replaced by execution of instructions, which makes
the notion of internal/external events superfluous. Also, in this section, we assume that the partitions on the hypervisor are non-communicating.

\paragraph{System state} A state $\SeqState \in \RealStateSpace$ in our model consists of the values of various  machine (i.e., the hardware platform) components such as registers (including both
 general-purpose and control registers), memory, and caches. 

\newcommand{\ltstransition}[5]{{#1 \rightarrow^{#3}_{#4} #5}}
\newcommand{\instructionSpace}{\mathit{Evnt}}
\newcommand{\naturalNumber}{\mathbb{N}}
\newcommand{\execution}{\mathcal{P}}
\newcommand{\ow}{\mathit{otherwise}}
\newcommand{\equivdef}[1]{\ensuremath{\stackrel{\text{#1}}{\equiv}}}
\tikzset{decorate sep/.style 2 args=
{decorate,decoration={shape backgrounds,shape=circle,shape size=#1,shape sep=#2}}}
\newcommand*{\DashedArrowi}{\mathbin{\tikz 
\draw [->,
line join=round,
decorate, decoration={
    zigzag,
    segment length=3,
    amplitude=1.2,post=lineto,
    post length=2pt
}]  (0,0.8ex) -- (1.0em,0.8ex);
}}
\newcommand{\WkTrsi}[1]{\DashedArrowi_{#1}}

\paragraph{Execution} 
An execution in this model is defined as a sequence of configurations from the state space $\RealStateSpace$. We represent the transition of states using a  deterministic \gls{lts}
$\ltstransition{}{}{n}{m}{} \subseteq \RealStateSpace \times \RealStateSpace$, where $n \in \naturalNumber$ is the number of taken steps, and $m \in \{0, 1\}$ determines the execution mode (i.e., either privileged 
mode $1$ or unprivileged mode $0$). Then, for states $\SeqState$ and $\SeqState'$ a transition from $\SeqState$ to $\SeqState'$ can be defined as the following:
if the number of taken steps is greater than zero $n > 0$ then there exist an intermediate state $\SeqState''$ which is reachable from $\SeqState$ in $n -1$ steps and $\SeqState'$ is the state immediately after $\SeqState''$,
otherwise $\SeqState$ and $\SeqState'$ are the same.

\[\ltstransition{\SeqState}{}{n}{m}{\SeqState'} \equivdef{def} 
 \begin{cases}
  \exists\ \SeqState''.\ \ltstransition{\SeqState}{}{n-1}{m}{\SeqState''}\ \land\ \ltstransition{\SeqState''}{}{1}{m}{\SeqState'} & : n > 0 \\
  \SeqState = \SeqState'                                                                                                                      & : n = 0
 \end{cases}
\]

In this transition system, single step transitions are denoted as $\ltstransition{\SeqState}{}{}{m}{\SeqState'} \equivdef{def}  \ltstransition{\SeqState}{}{1}{m}{\SeqState'}$, 
we use $\ltstransition{\SeqState}{}{*}{m}{\SeqState'}$ for arbitrary long executions, and if $\ltstransition{\SeqState}{}{}{m}{\SeqState'}$ then $\SeqState$ is in mode $m$.
Moreover, we use $\SeqState_0 \WkTrsi{} \SeqState_n$ to represent the weak  transition relation that holds if there is a finite execution 
$\ltstransition{\SeqState_0}{}{}{}{\cdots}\ltstransition{}{}{}{}{\SeqState_n}$ such that $\SeqState_n$ is in unprivileged mode and all the intermediate states $\SeqState_j$ for $0 < j < n$ are in privileged mode
(i.e., the weak transition hides internal states of the hypervisor).

For models having states consisting of components  other than memory, the no-exfiltration and no-infiltration properties as we defined above are not sufficient to ensure 
separation. The behavior of the system in such models does not solely depend on the memory and interference with other components can lead to unexpected misbehavior. Therefore, we try to extend
the definition of the security properties to accommodate these additional components. To this end, we give some auxiliary definitions.

\newcommand{\observation}{\mathit{O}}
\newcommand{\consistentStateSpace}{\mathit{Q}}

\begin{definition}(Observation) 
 For a given guest $g$ on the hypervisor, we define the guest observation $\observation_{g}$ as all state components that can affect its execution.
\end{definition}

\begin{definition}(Secure Observation) 
The remaining part of the state (i.e., the memory of other partitions and some control registers) which are not directly observable by the guest
constitute the secure observations $\observation_{s}$ of the state.
\end{definition}

\begin{definition} (Consistent State)
 We define consistent states as states in which value of components are
 constrained by a functional invariant. The invariant consists of properties that enforce well-definedness of states; e.g., there is no mapping in a page-table that permits guest accesses to the hypervisor memory. 
 Moreover, $\consistentStateSpace$ represents the set of all possible states that satisfy the functional invariant.
\end{definition}

The extended no-infiltration guarantees that instructions executed in unprivileged user mode and services executed inside the hypervisor  on behalf of the guest 
maintain equality of the guest observation if the execution starts in consistent states having the same view of the system.

\begin{property}
 Let $\SeqState_1  \in \consistentStateSpace$ and $\SeqState_2 \in \consistentStateSpace$ and assume that $\observation_g(\SeqState_1) = \observation_g(\SeqState_2)$, if 
 ${\SeqState_1} \WkTrsi{} {\SeqState'_1}$ and ${\SeqState_2}\WkTrsi{}{\SeqState'_2}$ then $\observation_g(\SeqState'_1) = \observation_g(\SeqState'_2)$.
\end{property}

Proving the no-infiltration property over the course of several instructions entails showing the inability of the guest software in changing the critical state components.
This prevents a guest from elevating its permissions to access resources beyond its granted privileges,
such as the memory allocated for the page-tables or the value of  control registers. This can be done by showing that
guest transitions preserve the consistency of states and that the secure observation remains constant in all reachable states. 

\begin{property}
 Let $\SeqState \in \consistentStateSpace$, if  $\ltstransition{\SeqState}{}{}{0}{\SeqState'}$  then $\observation_{s}(\SeqState) = \observation_{s}(\SeqState')$ and $\SeqState' \in \consistentStateSpace$.
\end{property}

Due to the interleaving of the hypervisor and guests executions, one has to check that context switching (i.e., changing processor's execution mode) to the hypervisor is done securely to 
prevent a guest from  gaining privileged access rights. More importantly, it must
be ensured that \textit{program counter} cannot be loaded with an address of the guest's choice, and all entry points into privileged mode should be in an exception handler. This  guarantees that the guest cannot 
execute arbitrary code in privileged mode. Likewise, it needs to be ensured that interrupts are correctly masked, and the return address belongs to the interrupted guest and is properly stored.
To check these properties, we define a \textit{switching convention} that serves to restore the context for the activated mode and checks that mode switching from unprivileged mode to privileged mode is done securely.

\begin{property}\label{prop:secureContextSwitching}
 Let $\SeqState \in \consistentStateSpace$ be an arbitrary state in unprivileged mode and $\SeqState'$ be a state in which the hypervisor is active. If $\ltstransition{\SeqState}{}{}{0}{\SeqState'}$
 then context switching complies with the \textit{switching convention}.
\end{property}

Additionally, it is also essential to show that execution of hypervisor handlers maintains the functional invariant. 
Since the hypervisor's transitions can break the invariant, we do not check the intermediate privileged steps. However, it must be ensured that the hypervisor yields control to the guest only in a state which satisfies
the invariant.

\begin{property}\label{prop:hypervisorPreserveInvariant}
 Let $\SeqState \in \consistentStateSpace$ is the state immediately after switching to the kernel, if  ${\SeqState}\WkTrsi{}{\SeqState'}$ then  $\SeqState' \in \consistentStateSpace$.
\end{property}

\subsubsection{Information Flow Security}
\textit{Information flow analysis}, for systems implementing the MILS concept, is the study of controlling the propagation of information among security levels. The primary objective of 
this study is to rigorously ensure that there are no ``illegal'' flows of high-criticality data to low-criticality observers.
\begin{figure} 
 \centering
 \includegraphics[width=0.62\linewidth]{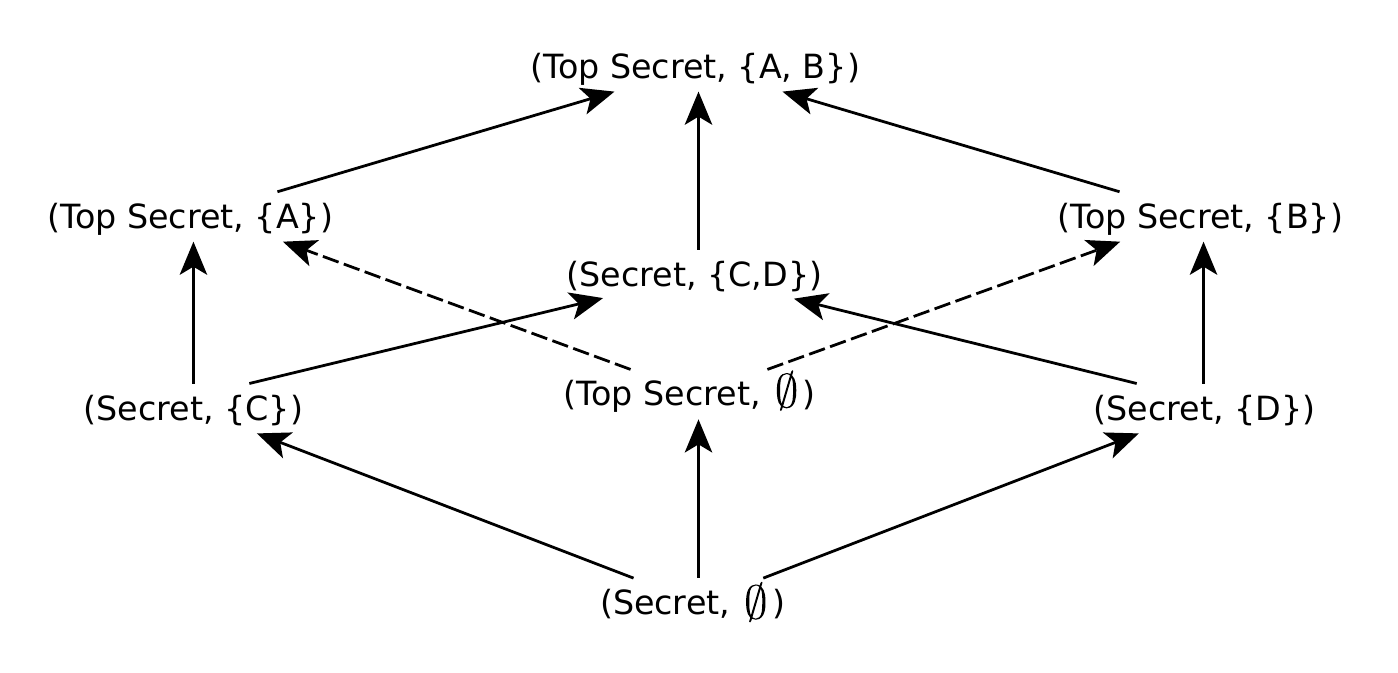}
 \caption{A lattice of security labels. Arrows show the intended information channels, moreover, $A$, $B$, $C$, and $D$ are objects in the system.}
 \label{fig:slattice}
\end{figure}
Denning~\cite{Denning:1976:LMS:360051.360056} used a lattice of security labels to perform this analysis. In her approach security labels indicate criticality level of objects (e.g., files or program variables) or
information receptacles, and the lattice structure represents  the \textit{information flow policy} within the system. An example of such a lattice is depicted in Figure~\ref{fig:slattice}.

The seminal paper by Goguen and Meseguer~\cite{Dblp:conf/sp/GoguenM82,6234812} was first to coin the concept of \textit{noninterference}.
They divide the system into a number of \textit{security domains} and information can flow between domains only according to the information flow policy of the system.
Goguen-Meseguer noninterference prevents actuators in one domain from learning about the occurrence of certain events in other ones.
Pure noninterference, which is also commonly known as no-infiltration, forbids  any flow of information among domains. This would ensure that
actions performed by one domain cannot influence subsequent outputs seen by another one, and thus the system does not leak confidential data.


Rushby defined the notion of \textit{intransitive noninterference}~\cite{csl-92-2} as a declassification of classical noninterference, which can accommodate the flow of information between
domains. In a later work, von Oheimb~\cite{vonOheimb2004} revisited  noninterference defined by Rushby and proposed concepts of \textit{nonleakage} and \textit{noninfluence} for state-based systems.
Nonleakage is a confidentiality property prohibiting domains from learning about private data of one another, but it does not forbid them from learning about the occurrence of transitions of other domains.
Noninfluence is defined as the combination of nonleakage and Goguen-Meseguer noninterference.


\subsection{Verification Methodology}
\newcommand{\states}{\mathit{states}}

We conclude this section by giving an overview of our proof strategy and drawing a connection  between tools that we have used for formal verification. Our strategy to analyze the PROSPER kernel, to examine its security
and correctness, consists of lifting the main body of  our reasoning to a high-level (design) model of the system,  which is derived from the real implementation of the kernel. The main purpose of defining such an
abstraction is to facilitate formal reasoning. This abstract model, which is also called \gls{tls} or \textit{ideal model}, defines the desired behavior of the system and serves as a framework to check the validity of the security 
properties~\cite{Roever:2008:DRM:1525579}.

The connection between different layers of abstraction relies on \textit{data refinement}~\cite{Back:1990:RCP:91930.91936,Roever:2008:DRM:1525579}. Refinement establishes the correspondence 
between the high-level (abstract) model  and real implementation, and it shows that results of the analysis on the abstract model can be transferred to the system implementation, too.  Since
verifying properties is in general easier on an abstract model, a refinement-based approach simplifies  the verification process.

Back was first to formalize the notion of \textit{stepwise refinement}, as a part of the \textit{refinement calculus}~\cite{Back78:thesis,Back:1998:RCS:551462}. 
This approach is later extended by Back and von Wright to support refinement of data representations, called \textit{data refinement}~\cite{Back:1990:RCP:91930.91936}.
Data refinement correctness is often validated by establishing either \textit{bisimulation} or \textit{forward simulation}~\cite{Lynch:1995:FBS:220262.220273}, as logical tools to show the behavioral
equivalence between two systems.


\begin{figure}[]
 \centering
 \includegraphics[width=0.4\linewidth]{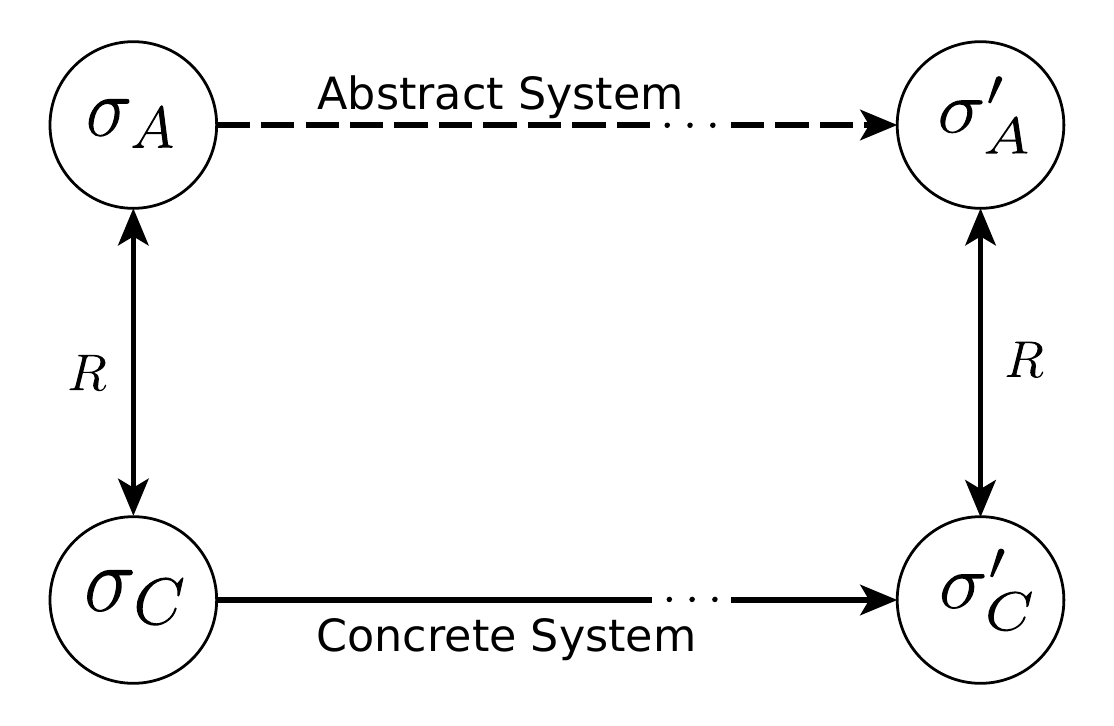}
 \caption{Forward simulation.}
 \label{fig:simulationRelation}
\end{figure}

For the given  concrete implementation $C$ and abstract model $A$ ,  the implementation $C$ refines $A$, or equivalently $A$  simulates $C$, if 
starting at the corresponding initial states, execution of $C$ and $A$ terminates in the corresponding final states (cf. Figure~\ref{fig:simulationRelation}).
The proof of such a refinement is usually done by defining an \textit{abstraction relation} $R \subseteq \RealStateSpace_{A} \times \RealStateSpace_{C}$\footnote{Establishing the simulation relation $R$ 
between two models usually entails showing several properties, among which demonstrating \textit{observational equivalence} between the models is an important property in this thesis.}
between states of the models and showing (by an inductive proof) that $R$ is a simulation relation  if $A$ and $C$ are advanced (\textit{step forward}) in parallel. 
Note that for deterministic systems whose executions depend on the initial state, since for every corresponding initial states there is only one pair of corresponding executions,
it is enough to prove the simulation once.

%

The validity of the simulation relation $R$  between the two models ensures that  the abstraction $A$ models all possible behavior of $C$; this is called \textit{soundness} of the simulation. However, the inverse relation does not 
necessarily hold.
On the other hand, the relation $R$ is called a bisimulation relation between the two models if, and only if, both the models are simulating each other\footnote{This condition does not necessarily
is true for  nondeterministic systems. That is, it is possible to have two machines that simulate each other, but are not bisimilar.}.
This requires showing \textit{completeness} of the relation; i.e., it must be proved that for every execution of $A$ there exists a simulating execution in $C$.


If the bisimulation relation is proved to hold for all executions of $A$ and $C$, any security property shown for $A$ also holds on all executions of the implementation $C$.
Additionally, hyperproperties (i.e., properties relating different executions of the same system) proven about $A$ can be transferred to $C$ as well.


\begin{figure}[]
 \centering
 \includegraphics[width=0.7\linewidth]{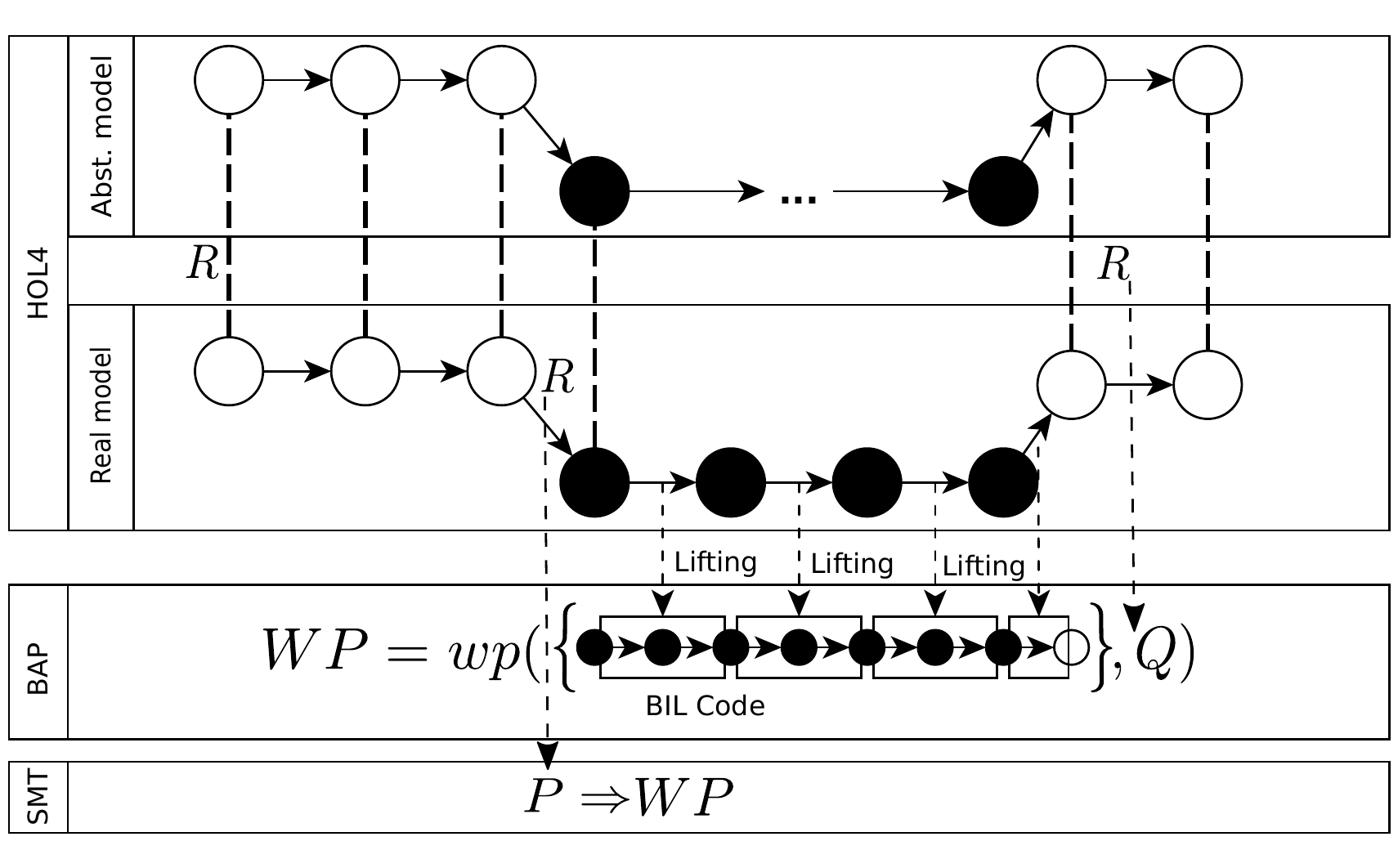}
 \caption{Executions of a real machine (middle), and the Top Level Specification (top) and the
relations between them. In addition the binary verification methodology (bottom) is depicted. Here, ``Lifting'' refers to the process of translating the binary code of the kernel to BAP's IL,
$\mathit{P}$ and $\mathit{Q}$  are pre- and post-conditions serving as the contract for a kernel handler, and $\mathit{WP}$ is the weakest precondition computed by taking into account effects of the handler machine code.}
 \label{fig:proofOverview}
\end{figure}

Figure~\ref{fig:proofOverview} illustrates our verification strategy. In particular,  it depicts the different layers of abstraction, how they are related, and the tools we have
used to verify properties at each layer.
White circles in Figure~\ref{fig:proofOverview} represent states in unprivileged user mode, and
black circles indicate states where the hypervisor is active. 
In this figure the ``Abstract model'' represents the top level specification, ``Real model'' is actual implementation of the system where each transition represents the execution of one
binary instruction, and the relation $R$ denotes the refinement relation.

In our approach, the (bi)similarity of unprivileged transitions in two models is established in HOL4. For the privileged transitions, however,
proof of the refinement relation $R$ is done using a combination of HOL4 and the BAP tools. 
Since during the execution of handlers no guest is active, internal hypervisor steps cannot be observed by guests.
Moreover, as the hypervisor does not support preemption, the execution of handlers cannot be interrupted. Therefore, we disregard internal states of the handlers and limit
the relation $R$ to relate only states where the guests are executing.
To show refinement, we use HOL4 to verify that the refinement relation transfers security properties to the Real model and
to prove a theorem that transforms the relational reasoning into a set of contracts for the handlers and guarantees
that the refinement is established if all contracts are satisfied. Contracts are derived from the
hypervisor specification and the definition of the relation $R$. Then we use BAP to check that the hypervisor code respects the contracts, which are expressed as Hoare triples.


Paper~\ref{paper:JCS} elaborates more on this pervasive approach that we have adopted to verify the hypervisor. 
Verification presented in paper~\ref{paper:esorics} is restricted to  proving the correctness of a runtime monitor only in an abstract model of the system.
Moreover, paper~\ref{paper:csf} shows how  properties transferred to an implementation model (using a similar
approach as of paper~\ref{paper:JCS}) can be further transferred down to a model augmented with additional hardware features, namely caches. 
\section{Summary}
A hypervisor is a system software that enables secure isolation of critical programs from less trusted (potentially malicious) applications coexisting on the same processor. Hypervisors reduce 
software portion of the system TCB to a thin layer which is responsible for isolating partitions and handling communication between them. The small codebase of hypervisors minimizes their attack
surface and enables the use of rigorous reasoning to prove their security and correctness. 

Commonly used properties in verification of hypervisors are data-separation  and information 
flow security. One way of verifying these properties is by constructing an abstract model of the system, proving the properties at this abstract level, and then transferring
the verified properties to the system's machine-code by showing a refinement relation. 

When modeling a system, one has to  take also into account hardware features such as caches which their effects on the system state can invalidate properties verified in the absence of these features.
\chapter{Related Work}\label{kappa:ch:relatedWorks}
\newcommand{\separation}{data-separation }
In this chapter, we review works on different aspects related to the main contributions of this thesis, namely formal verification of low-level execution platforms (Section~\ref{ch3:sec:verifcationOfIsolation}),
provably secure supervised execution (Section~\ref{ch3:sec:verifcationOfSuperviedExecution}), and attacks on isolation (Section~\ref{ch3:sec:attacksOnIsolation}).

Verification of system software is a goal pursued for several decades. Past works on formal verification of operating systems include
analysis of \capletter {Provably Secure Operating System} (PSOS)~\cite{Feiertag79thefoundations} and \capletter{UCLA Secure Unix}~\cite{Walker:1980:SVU:358818.358825}.
Neumann et al.~\cite{Feiertag79thefoundations} used the \gls{hdm}~\cite{robinson1977formal} to design an operating system with provable security properties.
\gls{hdm} is a software development methodology to facilitate formal verification of the design and implementation.
UCLA Secure Unix~\cite{Walker:1980:SVU:358818.358825} aimed to formally prove that \separation is properly enforced by the kernel of a Unix-based OS. Verification of UCLA Secure Unix was based on a refinement
proof from a top-level abstraction down to the Pascal code of the kernel, and verification assumed the correctness of underlying hardware. 
KIT~\cite{Bevier:1989:KSO:76141.76147}, for \capletter{kernel} of \capletter{isolated tasks}, is probably the first fully verified operating system. Software 
used in this exercise was a small idealized kernel which was proved to implement distributed communicating processes. 
Bevier~\cite{Bevier:1989:KSO:76141.76147} showed that properties verified at an abstract model of KIT could be refined to hold also for its real implantation.
Among others, the \textit{VFiasco}\cite{hohmuth2005vfiasco} and Robin~\cite{Tews:2008:FMM:1390859.1390960} projects were also conducted to  verify system level software.

\section{Verification of Security-Kernels}\label{ch3:sec:verifcationOfIsolation}
Next, we look at a selection of projects on verification of microkernels, separation kernels, and hypervisors; a rather comprehensive list of formally verified kernels can be found
in~\cite{DBLP:journals/corr/ZhaoSZL17}.

Formal verification of the seL4 microkernel~\cite{DBLP:conf/sosp/KleinEHACDEEKNSTW09} is probably the leading work in verification of system software.
The functional correctness of seL4 has been verified in Isabelle/HOL~\cite{Nipkow:2002:IPA:1791547} by proving a refinement. The refinement relation shows the correspondence between  different layers of abstractions
from a high-level abstract specification down to the machine code~\cite{Sewell:2013:TVV:2491956.2462183} of seL4.
Murray et al.~\cite{Murray:2012:NOS:2428345.2428357, murray2013sel4} later extended the verification of seL4 by showing its information flow security 
based on the notions of \textit{nonleakage} and \textit{noninfluence} as introduced in~\cite{vonOheimb2004}.
The verification of seL4 assumes a sequential memory model  and ignores leakages via cache timing and
storage channels. Nevertheless, the bandwidth of timing channels in seL4 and possible countermeasures were examined,
later, by Cock et al.~\cite{CockGMH14}. 

Heitmeyer et al.~\cite{Heitmeyer:2008:AFM:1340674.1340715} proposed a practical approach to security analyze separation kernels for embedded devices. 
The top-level specification used in that project was a state-machine model of the system which provides a precise description of the required behavior and is proved to enforce separation properly.
The information flow security of the kernel is shown by annotating its source-code by Floyd-Hoare style assertions and showing that non-secure information flows cannot occur.
Verification is done in the PVS theorem prover using a memory coherent (cacheless) model of the system. However, no machine-checked verification was directly done at the implementation level of the kernel.

The primary verification objective of the Verisoft project~\cite{Verisoft} was to achieve
a pervasive formal analysis of the system from hardware to programs running in unprivileged user mode. Verisoft aimed to dismiss the assumption of compiler and instruction set model correctness. 
Verisoft-XT~\cite{VerisoftXT} continued the verification started under the Verisoft project. Verisoft-XT initially  targeted the verification of the Microsoft Hyper-V hypervisor using the VCC tool~\cite{Cohen:2009:VPS:1616077.1616080}.
For this exercise, guests are modeled as full x64 machines, and caches are not transparent if the same memory location is  accessed in cacheable and uncacheable mode. 
The verification of Hyper-V was later
dropped to use the developed techniques in verification of an idealized hypervisor for a baby VAMP architecture~\cite{alkassar2010automated}. They fully automate the verification process using VCC and 
proved that the virtualization of the memory management subsystem is correct~\cite{Alkassar:FMCAD2010-b}.
Verisoft-XT also made contributions to multi-core concurrency and verification under weak memory models~\cite{Cohen:2010:TSO:2176728.2176759}.
A subproject of Verisoft-XT was to demonstrate the functional correctness of the PikeOS~\cite{baumann2011proving} microkernel at its source-code level. The verification is done in VCC 
by establishing a simulation relation between a top-level abstract model and  the real implementation of the system.

mCertiKOS~\cite{gu2015deep} is a hypervisor that uses the hardware virtualization extensions to virtualize the memory subsystem and to enforce isolation of partitions. mCertiKOS runs on a single-core processor,
and its functional correctness~\cite{gu2015deep} and noninterference property~\cite{Costanzo:2016:EVI:2908080.2908100} has been verified based on a  sequential memory model within the proof assistant Coq.
The follow-up work by Gu et al.~\cite{Gu:2016:CEA:3026877.3026928} extended the CertiKOS hypervisor to run on a multi-core processor. The Contextual functional correctness of the extend kernel also has been verified
using the Coq theorem prover. CertiKOS is implemented in ClightX, and they used CompCertX to compile and link kernel programs.

Among other related works, Ironclad~\cite{Ironclad} is a research project which applied formal methods to verify the system's entire software stacks, including
the kernel, user mode applications, device drivers, and cryptographic libraries. The verification of Ironclad to prove its functional correctness and information flow security has been performed on a  cacheless model
down to the assembly implantation of the system. Similarly, Barthe et al.~\cite{DBLP:conf/fm/BartheBCL11,Barthe:2012:CRO:2354412.2355248} applied formal reasoning to show security of an idealized 
hypervisor, they also included an abstract model of caches in their analysis and demonstrated how the isolation property can be verified when caches are enabled.

To formally verify the INTEGRITY-178B separation kernel~\cite{INTEGRITY} the GWV policy~\cite{greve2003separation}\footnote{The GWV security policy allows controlled  communication
of partitions executing on a separation kernel. GWV restricts effects on memory \textit{segments} in a partition to memory regions that are (i) associated with the current 
partition and (ii) allowed to interact with  the segments, according to the \capletter{direct interaction allowed} (dia) function.}
was extended to describe the dynamic scheduling of the kernel~\cite{Hardin:2010:DVM:1841184}. This extended policy
was also able to  capture the flow of information within the system completely.
The analysis is done by creating three specifications, namely  a functional  specification  which represents the functional interfaces of the system, a high- and low-level designs that are
the semi-formal representation of the system with different levels of details. The correspondence between the specifications was shown by a ``code-to-spec''  review process. The verification of INTEGRITY-178B to some
extent is accomplished using the ACL2 theorem prover.

\subsection{Trustworthy MMU Virtualization}

The memory subsystem is a critical resource for the security of low-level software such as OS kernels and hypervisors. Since devices like MMUs determine the binding of physical memory locations to 
locations addressable at the application level, circumventing the MMU provides a path for hostile applications to gain illicit access to protected resources. It is therefore of interest to develop methods for the MMU 
virtualization that enables complete mediation of MMU settings and can  protect trusted components.
Examples of systems that use virtualization of the memory subsystem to protect security-critical components include KCoFi~\cite{criswell2014kcofi} based on the 
\capletter{Secure Virtual Architecture} (SVA)~\cite{criswell2007secure}, Overshadow~\cite{Chen:2008:OVA:1346281.1346284}, Inktag~\cite{Hofmann:2013:ISA:2451116.2451146}, and Virtual Ghost~\cite{Criswell:2014:VGP:2541940.2541986}.

KCofi~\cite{criswell2014kcofi}, for \capletter{Kernel Control Flow Integrity}, is a system to provide complete Control-Flow Integrity protection for COTS OSs. KCofi relies on SVA~\cite{criswell2007secure} to trap
sensitive operations such as the execution of instructions that can change MMU configurations. SVA is a compiler-based virtual machine that can interpose between hardware and the operating system.
The functional correctness of KCofi has been verified in Coq. The verification covers page-table management, trap handlers, context switching, and signal delivery in KCofi.

Overshadow~\cite{Chen:2008:OVA:1346281.1346284} is a hypervisor to protect confidentiality and integrity of legacy applications from commodity operating systems.
Overshadow provides OS with only an encrypted view of application data. This encrypted view prevents the OS's kernel from performing unauthorized accesses to application data while enabling the kernel
to manage system resources. Overshadow yields a different view of the memory depending on the context performing a memory operation.

Techniques mostly used to virtualize the memory subsystem are shadow paging, nested paging, and microkernels. The functional correctness of mechanisms based on shadow page-tables has been verified 
in~\cite{Alkassar:FMCAD2010-b,alkassar2010automated,heiser2010okl4,CavalcantiD09}.
XMHF~\cite{Vasudevan:2013:DIV:2497621.2498126} and CertiKOS~\cite{gu2015deep} are examples of  verified hypervisors for  the x86 architecture that control memory operations through hardware virtualization
extensions.

XMHF~\cite{Vasudevan:2013:DIV:2497621.2498126}, or {e\capletter{X}tensible} and {\capletter{Modular Hypervisor Framework}}, is a formally verified security hypervisor,
which relies on hardware extensions to virtualize the memory subsystem and to achieve high performance.
The XMHF hypervisor design  does not support interrupts. This design choice enabled automated verification of the memory subsystem integrity
using model checking techniques.

\newcommand{\monitor}{\mathit{SM}}
\newcommand{\policySpace}{\mathbb{P}}
\newcommand{\Bin}{\mathbb{B}}
\section{Trustworthy Runtime Monitoring} \label{ch3:sec:verifcationOfSuperviedExecution}
\textit{Security kernels}~\cite{Ames:1983:SKD:1319721.1319947} or reference monitors are programs used traditionally to enforce access control policies in the kernel of mainstream operating systems. 
To fulfill the requirements of having an efficient monitoring subsystem (i.e., providing the monitor with a complete view of the system and  making it tamper-resistant), the monitor can be implemented as a loadable
kernel module. Nevertheless, kernel level solutions~\cite{Loscocco:2001:IFS:647054.715771, Wright:2002:LSM:647253.720287} are not reliable, as the kernel itself is vulnerable to  attacks, such as code injection
attacks and kernel rootkits.

An alternative is protecting the monitor by a more privileged layer, such as a hypervisor. The protection can be done through either making inaccessible from the hypervisor, 
 memory of the monitor, e.g.,~\cite{Sharif:2009:SIM:1653662.1653720}, or integrating the monitor into the virtualization layer~\cite{Seshadri:2007:STH:1294261.1294294, Kim:1994:DIT:191177.191183}.
The later, however, has the inconvenience of increasing the complexity of the hypervisor itself, invalidating the principle of keeping the TCB as minimal as
possible. Secvisor~\cite{Seshadri:2007:STH:1294261.1294294, franklin2008attacking} is a formally verified hypervisor based runtime monitor to preserve the guest's kernel integrity. Secvisor uses the MMU 
virtualization to protect the kernel memory and to defend itself against attacks. SecVisor was verified using the model checking approach to ensure that only user-approved code can execute in 
privileged mode of the CPU.

More reliable solutions to implement a secure runtime monitor are protecting the monitor using hardware extensions, e.g., ARM TrustZone~\cite{TrustZone}, or {\sk}. Example of softwares which use 
hardware extensions to protect the monitoring module includes~\cite{Payne:2008:LAS:1397759.1398072,Wang:2010:HLA:1849417.1849986}.
Lares~\cite{Payne:2008:LAS:1397759.1398072} is a runtime monitor built on a hypervisor. Lares consists of three components: a partition to execute a guest software, a monitoring module
deployed in a separate partition, and the hypervisor which creates isolated partitions and supplies the monitor with a complete view of the guest. 
Hypervision~\cite{Wang:2010:HLA:1849417.1849986} uses isolation at the hardware level to protect the monitoring subsystem and to have full control over the target kernel's memory management. Hypervision 
is entirely located within ARM TrustZone and uses hardware features to intercept security-critical events of the software that it protects.

\section{Attack On Isolation}\label{ch3:sec:attacksOnIsolation}
Creating an isolated, integrity-protected processing environment for security-critical computations is an active research area in platform security. Over the last years,
the advent of  new technologies~\cite{anati2013innovative,mckeen2013innovative,TrustZone}  has made this long-standing goal (almost) attainable. However,
attacks targeting these mechanisms raised concerns about security of
these  technologies themselves as the system \textit{root of trust}. There is a large body of papers that demonstrate attacks successfully conducted to subvert techniques employed to create trusted execution
environments. This section overviews a few works on exploiting vulnerabilities to compromise isolation solutions.

Hypervisor level attacks can be, in general, categorized into two groups: \textit{side-channel attacks} and \textit{hypervisor level malware attacks}.
Side-channels, as defined in Chapter~\ref{kappa:ch:background}, are paths that exist accidentally to the otherwise secure flow of data and through which confidential information can escape.
Side-channels are usually built using hidden features of system components; among others, memory and caches are extensively used  to construct such channels.
Jankovic et al.~\cite{Philippe-Jankovic2017} showed how to construct a side-channel between virtual machines on a hypervisor using the \textit{Flush+Reload} attack.
Apecechea et al.~\cite{Irazoqui:2014:FGC:2758335.2758750} used \textit{Bernstein's correlation attack}~\cite{Bonneau:2006:CTA:2105334.2105357} to create a side-channel between partitions executing on 
Xen and VMware~\cite{rosenblum1999vmwares} to extract the secret key of the AES cryptographic algorithm.
Similarly, in~\cite{Yarom:2014:FHR:2671225.2671271} the last level cache (i.e., L3 cache) is used to mount a side-channel attack.

Hypervisor level malware attacks are used mostly to detect the presence of the virtualization layer~\cite{6014696} and to  mimic the structure of a hypervisor~\cite{King:2006:SIM:1130235.1130383}. The later, 
also called \textit{hyperjacking}, installs as a bare-metal hypervisor and moves the victim software inside a partition without being detected. An example of this type of attacks is the 
\textit{Blue-Pill} root-kit developed by security researcher Joanna Rutkowska.

Hypervisors are not the only victim of attacks on isolation. There is an extensive list of attack techniques that target  hardware level solutions such as ARM TrustZone and Intel SGX.
TrustZone was first introduced in ARMv6  to create an isolated execution environment. However, it is vulnerability to attacks from unprivileged user space, to execute arbitrary code inside 
TrustZone's secure world, 
is shown in~\cite{shen2015exploiting}.
Intel \textit{Software Guard Extensions} is a technology to protect select data from disclosure or modification. Nevertheless, a study by Costan and Devadas~\cite{costan2016intel} and
similarly~\cite{DBLP:journals/corr/BrasserMDKCS17, Gotzfried:2017:CAI:3065913.3065915} revealed the vulnerability of this technology to  cache attacks and software side-channel attacks.
It is also acknowledged by  Intel~\cite{intelsgxvulnerability} that SGX does not provide protection against side-channel attacks including those that constructed through exploiting performance
counters.

\section{Summary}
Formal verification of system software has made great strides in recent years. Verifications are mostly done on models that are simpler than current processors. 
These verification exercises represent significant investments in development and verification. However, excluding hardware features such as caches, TLBs, and pipelines from the analysis makes
properties verified in the absence of these features not reliable in a richer setting.

Taking control of the memory management subsystem, e.g., through virtualizing the MMU, is a widely adopted technique to enforce isolation of components in the system. Many 
projects used memory isolation to deploy a runtime monitor to control system activities securely.
\chapter{Contributions}\label{kappa:ch:contributions}
The main contributions of this thesis can be split into three parts: 
(i) design and implementation of a security hypervisor, 
(ii) attacking the isolation guarantees of the hypervisor through constructing a low-noise cache storage side-channel, and
(iii) formal verification of the hypervisor down to its machine-code.
Our approach to secure the hypervisor is based on taking control of the memory management subsystem by designing an algorithm to virtualize the MMU. Furthermore, we 
use formal methods to verify the functional correctness and security of the virtualization mechanism.

This thesis consists of four articles, some of which originally published in peer-reviewed journals and conference proceedings. In this chapter, we give a summary of these papers
 together with a statement of the author's contributions. Some papers are revised to improve presentation and to include proofs of key theorems.
In particular, some details in paper~\ref{paper:esorics} concerning  the proof of important theorems have been changed from the original paper.

\section{Summary of Included Papers}
\paragraph{Paper A:  Provably secure memory isolation for Linux on ARM.} 
\begin{center}
\begin{minipage}{30em}
   \small{ Originally published as 
      \itshape {Roberto Guanciale, Hamed Nemati, Mads Dam, and Christoph Baumann. Provably secure memory isolation for Linux on ARM. Journal of Computer Security, 24(6):793-837, 2016. 
       }
    }
\end{minipage}
\end{center}

\paragraph{Content} 
Necessary for a hypervisor to host a general purpose operating system is permitting the guest software to  manage its internal memory hierarchy dynamically
and to impose its access restrictions. To this end, a  mandatory security requirement is taking control of the memory management subsystem by the hypervisor through virtualizing the MMU to 
mediate all its configurations. The MMU virtualization allows isolating security-critical components from a commodity OS running on the same processor. This enables the guest OS to
implement noncritical functionalities while critical parts of the system are adequately protected. In this paper, we present the design, implementation, and verification of a memory virtualization
mechanism. The hypervisor targets a common architecture for embedded devices, namely ARMv7 architecture, and it supports execution of Linux without requiring special hardware extensions. Our virtualization
approach is based on \textit{direct paging}, which is inspired by the paravirtualization mechanism of Xen~\cite{xen} and Secure Virtual Architecture~\cite{criswell2007secure}. We show that direct
paging can be implemented using a compact design that is suitable for formal verification. Moreover, using a refinement-based approach, we prove complete mediation along with memory isolation, and information 
flow correctness of the virtualization mechanism. Verification is done on a high-level model that augments a real machine state with additional components which represent the hypervisor
internal data structure. We also demonstrate how the verified properties at the high-level model of the system can be propagated to the hypervisor's machine-code.

The binary verification relies on Hoare logic and reduces the high-level  (relational) reasoning in HOL4 into verification of some contracts expressed in terms 
of Hoare triples $\{P\}\ C\ \{Q\}$.
To validate contracts, we compute the weakest precondition of binary fragments in the initial state to  ensure that the execution of the fragment
terminates in a state satisfying the postcondition.  We then prove that the precondition entails the weakest precondition.  
The validity of these contracts shows properties verified at high-level models are also valid for the machine-code of the hypervisor.

\begin{figure}
\centering
  \includegraphics[width=0.8\linewidth]{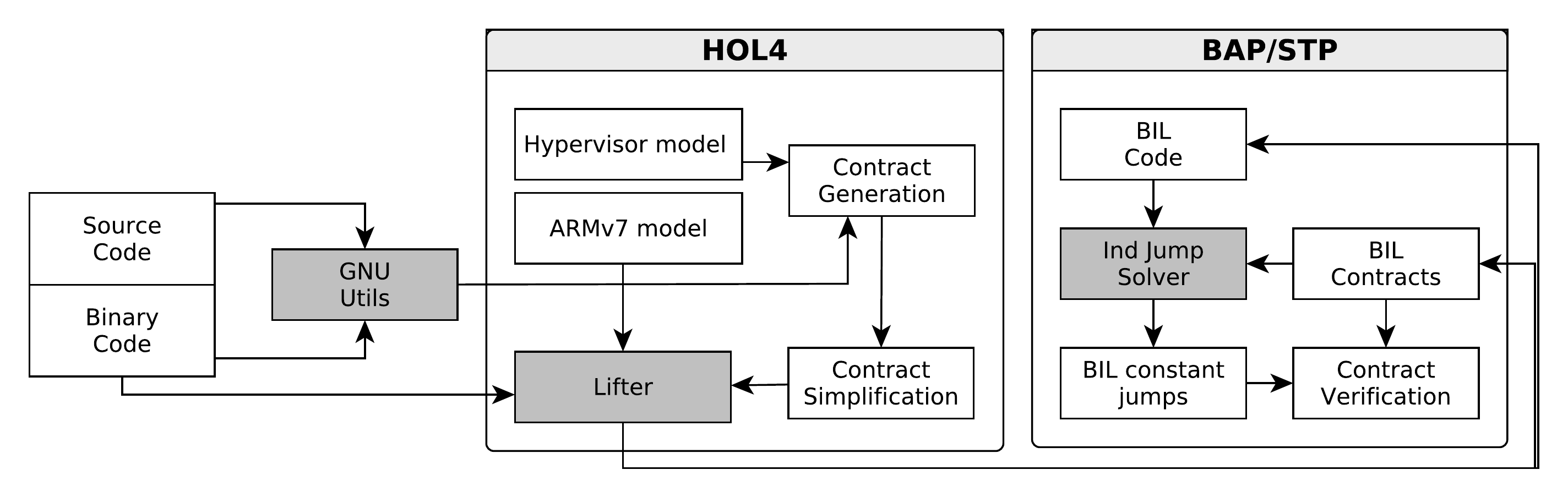}
  \caption{Binary verification workflow: \textit{Contract Generation}, generating pre- and post conditions based on the system specification and the refinement relation; \textit{Contract Simplification}, 
  massaging contracts to make them suitable for verification; \textit{Lifter}, lifting handlers machine-code and the generated contracts in HOL4 to BIL; \textit{Ind Jump Solver}, the procedure to resolve indirect jumps 
  in the BIL code; \textit{BIL constant jumps}, BIL fragments without indirect jumps; \textit{Contract Verification} using SMT solver to verify contracts. Here, gray boxes are depicting the tools that have been
  developed/extended to automate verification.}
  \label{fig:BAPworkFlow}
\end{figure}

The binary verification of the hypervisor is automated to a large extent using the BAP tool (cf. Figure~\ref{fig:BAPworkFlow}). To use BAP we \textit{lift} the ARMv7 assembly to the intermediate language of BAP. Lifting is
done by utilizing a tool developed in HOL4 that  generates for each expression a certifying theorem showing the equality of the expression's BIL fragment and the corresponding predicate in HOL4.
Several other tools have also been developed to automate the verification process and to optimize the weakest precondition generation algorithm of BAP. 
The optimization is needed to reduce the size of generated predicates, which otherwise can grow exponentially due to the number of instructions. Validating contracts by computing the weakest precondition relies on
the absence of indirect jumps. To fulfill this requirement, we  have implemented a simple iterative procedure that uses STP (an SMT solver)~\cite{DBLP:conf/cav/GaneshD07} to resolve indirect jumps in 
the code. Moreover, writing predicates and invariants usually need information on data types together with location, alignment, and size of data structure fields. Since machine-code (and BIL) lacks such information,
we developed a set of tools that integrate HOL4 and the \gls{gdb}~\cite{gdb} to extract the required information from the C source code and the compiled assembly.


In this paper, we also present several applications of the hypervisor. For instance, we show how a runtime monitor can be deployed on the hypervisor to check internal activities
of a Linux guest executing in a separate partition.

\paragraph{Statement of Contributions} For this paper, I was responsible for adapting the direct paging algorithm to the PROSPER kernel, implementing the memory management subsystem of the 
hypervisor, porting the Linux kernel on the hypervisor, and developing the new front-end for the BAP tool. The Linux port is done  jointly with other coauthors and our colleagues from \capletter{Swedish Institute} of
\capletter{Computer Scienc} (SICS). Further, I together with Roberto Guanciale did all the modeling and verification of the hypervisor. I also contributed to preparing the text.

\paragraph{Paper B: Trustworthy Prevention of Code Injection in Linux on Embedded Devices.} 
\begin{center}
\begin{minipage}{30em}
   \small{ Originally published as 
      \itshape {Hind Chfouka, Hamed Nemati, Roberto Guanciale, Mads Dam, and Patrik Ekdahl. Trustworthy prevention of code injection in Linux on embedded devices. In Computer Security - ESORICS 2015 - 20th European
        Symposium on Research in Computer Security, Vienna, Austria, September 21-25, 2015, Proceedings, Part I, pages 90-107, 2015.
       }
    }
\end{minipage}
\end{center}

\paragraph{Content} This paper demonstrate a use-case of isolation provided by the PROSPER hypervisor. In particular, the paper presents the design, implementation, and verification of a \gls{vmi}-based runtime monitor. 
The primary goal of the monitor is thwarting code injection attacks in an untrusted guest OS. The monitor is placed  in an isolated virtual machine  to
check internal activities of a Linux guest running in a separate partition. Deploying the monitor inside a dedicated partition has the advantage of decoupling 
the policy enforcement mechanism from the other hypervisor's functionalities and keeping the TCB of the system minimal.
The security policy enforced by the monitor is $W \oplus X$; i.e., each memory block within the Linux memory space can be either writable or executable but not both at the same time. 
This prevents buffer overflow attacks and guarantees that the untrusted guest is not able to directly modify executable memory blocks. 
The monitor uses a signature-based validation approach to check the authenticity of executables in the Linux memory. For this, the monitor keeps a database of valid signatures for known applications.
Whenever Linux tries to run an executable, the monitor intercepts this operation, due to the executable space protection policy, and checks if this executable has a valid signature in the 
database. If not, the monitor prohibits this operation.

The verification  is carried out on  an abstract model of the system to show that the monitor guarantees the integrity of the system.
The verification is performed using the HOL4 theorem prover and by extending the existing hypervisor model with the formal specification of the monitor.

\paragraph{Statement of Contributions} For this paper,  I contributed to the design, implementation, and verification of the monitor. The design and implementation are done together with Roberto Guanciale, and I
also contributed to the verification of the monitor and writing the paper text.

\paragraph{Paper C: Cache Storage Channels: Alias-Driven Attacks and Verified Countermeasures.} 
\begin{center}
\begin{minipage}{30em}
   \small{ Originally published as 
      \itshape {Roberto Guanciale, Hamed Nemati, Christoph Baumann, and Mads Dam. Cache storage channels: Alias-driven attacks and verified countermeasures. In IEEE Symposium on Security and Privacy, SP 2016, 
        San Jose, CA, USA, May 22-26, 2016, pages 38-55, 2016.
       }
    }
\end{minipage}
\end{center}

\paragraph{Content}
Resource sharing, while inevitable on contemporary hardware platforms, can impose significant challenges to security. First level data-cache is one such shared resource that is
transparent to programs running in user mode. However, it influences the system's behavior in many aspects.

Great as the importance of caches for the system performance is, caches are widely used to construct side-channels. These channels enable unauthorized parties
to gain access to confidential data through measuring caches' effect on the system state. As an instance, cache-timing channels are attack vectors built 
via monitoring variations in execution time, due to the presence of caches, and used in several cases to extract secret key of cryptographic algorithms.

This paper presents a new attack vector on ARMv7 architecture, which exposes a low-noise cache side-channel. 
The vector uses virtual aliasing with mismatched cacheability attributes and self-modifying code to build cache ``storage'' side-channels. We use this attack vector to break
the integrity of an ARMv7 hypervisor and its trusted services, such as a runtime monitor which was verified on a cacheless model, as presented in paper~\ref{paper:esorics} of this thesis.
Furthermore, we show how this vector can be used to attack the confidentiality of an AES cryptographic algorithm that was placed in the secure world of platform's TrustZone. 

To counter the cache storage side-channel attack we propose several countermeasures and implemented some of them to show their effectiveness. Additionally, we informally discuss how the formal verification 
of software which has been previously verified in a memory coherent (cacheless) model can be restored on platforms with enabled caches.

\paragraph{Statement of Contributions} Roberto Guanciale and I initially developed the idea of building cache storage side-channels using mismatched cacheability attributes. Then, I investigated more the
possibility of  constructing such channels and implemented the poof-of-concept attack based on the code provided by Arash Vahidi (at SICS). I was also responsible for developing the integrity attack
against the PROSPER kernel and implementing some of the countermeasures inside the hypervisor and the Linux kernel. Moreover, I contributed to discussions on fixing the formal verification of software
previously verified in a cacheless model, and writing the paper text.

\paragraph{Paper D: Formal Analysis of Countermeasures against Cache Storage Side Channels} 
\begin{center}
\begin{minipage}{30em}
   \small{
      \itshape {Hamed Nemati, Roberto Guanciale, Christoph Baumann, and Mads Dam. Formal Analysis of Countermeasures against Cache Storage Side Channels. 
       }
    }
\end{minipage}
\end{center}

\paragraph{Content} Over the past decades, formal verification has emerged as a powerful tool to improve the trustworthiness of security-critical system software. Verification of these systems, however, are mostly done on
models abstracting from low-level platform details such as caches and TLBs. In~\cite{DBLP:conf/sp/GuancialeNBD16} we have demonstrated several attacks built using caches against a verified hypervisor. Such 
attacks are facts showing that excluding low-level features from  formal analysis makes the verification of a system unreliable and
signify the need to develop verification frameworks that can adequately reflect the presence of caches.

We informally discussed in~\cite{DBLP:conf/sp/GuancialeNBD16} how to restore  verified  (on a cacheless model)  guarantees of software when caches are enabled.
In a similar vein, in paper~\ref{paper:csf} we formally show  how a countermeasure against cache attacks helps (1) restoring the integrity of a software and (2) proving the absence of cache storage side-channels. 
To ease the burden of verification task, we also try to reuse the previous analysis of the software on the cacheless model, as much as possible.

To carry out the analysis, we define a cacheless model, which is memory coherent by construction, and a cache-aware model. The cache-aware model is a generic ARM core augmented with caches and is sufficiently detailed to 
reflect the behavior of an untrusted application that can break memory coherence. We show for privileged transitions that  integrity and confidentiality can be directly transferred to the cache-aware model if the two models 
behave equivalently. 
However, since user level applications  are unknown and free to break their coherency, these properties cannot be transferred without overly restricting applications' transitions. Instead, we
revalidated security of unprivileged transitions in the cache-aware model.

Showing behavioral equivalence of the models entails proving that countermeasures are correctly implemented and are able to restore memory coherence. This, in turn, generates some proof obligations that can be imposed on the 
cacheless model, thus permitting to use existing tools~\cite{DBLP:conf/ccs/BalliuDG14,Brumley:2011:BBA:2032305.2032342,Song:2008:BNA:1496255.1496257} (mostly not available on a cache enabled model) to automate
verification to a large extent.

The main feature of our approach is the decomposition of proofs into (i) software dependent (ii) countermeasure dependent and (iii) hardware dependent parts. This helps to verify the countermeasures once and
independent of the software executing on the system and makes easier to adopt the verified properties for different software or hardware platforms.
We used HOL4 to machine-check our verification strategy to prove kernel integrity and user level properties. However, kernel confidentiality is so far mainly a pen-and-paper proof.

\paragraph{Statement of Contributions} Initially, I tried to develop a proof methodology to fix verification of a low-level execution platform on a cache-aware model, and I wrote the initial draft of the paper. Then,
all authors contributed to improve the technical development and writing the text.

\section{Further Publications}
In addition to the included papers, the following papers have been produced in part by the author of this thesis:

\begin{itemize}
 \item Mads Dam, Roberto Guanciale, Narges Khakpour, Hamed Nemati, Oliver Schwarz: Formal verification of information flow security for a simple arm-based separation kernel. ACM Conference 
 on Computer and Communications Security 2013: 223-234.
 
 \item Mads Dam, Roberto Guanciale, Hamed Nemati: Machine code verification of a tiny ARM hypervisor. TrustED@CCS 2013: 3-12.
 
 \item Hamed Nemati, Mads Dam, Roberto Guanciale, Viktor Do, Arash Vahidi: Trustworthy Memory Isolation of Linux on Embedded Devices. TRUST 2015: 125-142.
 
 \item Hamed Nemati, Roberto Guanciale, Mads Dam: Trustworthy Virtualization of the ARMv7 Memory Subsystem. SOFSEM 2015: 578-589.
\end{itemize}
   
\chapter{Conclusions}\label{kappa:ch:conclusions}
In this thesis, we explore the design, implementation, and verification of a security hypervisor as an enabler for isolation between system components.
The provided isolation allows building mixed-criticality systems which consolidate critical functionalities with convenience features on the same processor.
The hypervisor reduces the size of  the system's TCB, and its small codebase permitted us to  apply formal methods to show that the virtualization layer is functionally correct and behaves as specified. 
Further proofs are also used to ensure that the hypervisor preserves the system integrity and confidentially, guaranteeing that separated components are not able to affect each other.


\section{Contribution}
The contributions of this thesis are summarised as follows. 
In paper~\ref{paper:JCS}, we use direct paging to implement a provably secure virtualized memory management subsystem for embedded devices. We formally verify on a high-level 
model of the system that our approach is functionally correct, and it properly isolates partitions. By proving a refinement theorem, we then show how these properties can be  (semi-) automatically propagated to 
the binary code of the hypervisor. 
A use-case scenario for isolation obtained using memory virtualization  is presented in paper~\ref{paper:esorics}. In that paper, we show how a trustworthy runtime monitor can be securely deployed in a partition on
the hypervisor to prevent  code injection attacks within a fellow partition, which hosts a \gls{cots} operating system. The monitor inherits the security properties verified for the hypervisor, and we proved
that it enforces the executable space protection policy correctly.

In the related works chapter, we have seen that verification of system software is mostly done on models that are far simpler than contemporary processors. As a result, there are potential attack vectors that are not
uncovered by formal analysis. Paper~\ref{paper:sp} presents one such attack vectors constructed by measuring caches effects. The vector enables an adversary to breach isolation guarantees of a verified 
system on sequential memory model when is deployed on richer platforms with enabled caches. To restore integrity and confidentiality, we provide several countermeasures. 
Further, in paper~\ref{paper:csf}, we propose a new verification methodology to repair verification of the system by including data-caches in the statement of the top-level 
security properties. 

\section{Concluding Remarks and Future Work}
Whether the results presented in this work will contribute to any improvement in a system's overall security posture depends on many factors. In particular, it 
depends on characteristics of hardware, software executing on the platform, and capabilities of adversaries. In this thesis, we tried to address only a subset of
problems related to  platform security and soundness of our results are constrained by the hypervisor design, model of the processor in HOL4, and hardware platforms that we have used as
our testbeds. 


For instance, the model of the memory management subsystem which we have used in our formal analysis is manually derived from  ARM Architecture Reference Manuals~\cite{ARMV7,Cortexa8, ARMV8}.  ARM specifications are mostly written in a 
combination of natural language and pseudocode and leave some aspects of the processor behavior (e.g., effects of executing some system-level instructions) under-specified. This informal and imprecise description 
makes developing formal models from such specifications laborious and error-prone and potentially undermines results achieved based on these models.
ARM admits this problem and  declares the significance of having more precise and trustworthy specifications as a motivation behind their recently published machine-readable and executable specifications~\cite{Reid:2016:TSA:3077629.3077658}.
A future direction for our work is adopting models produced directly from such  machine-readable specifications, which also makes easier to keep up-to-date the formal model of the system as hardware changes.

Moreover, the complexity of today's hardware and system software is such that a verification approach allowing reuse of models and proofs as new features are added is
essential for formal verification of low-level execution platforms to be economically sustainable. 
The HASPOC project~\cite{7561034} showed how to adopt a compositional verification strategy to attain this goal.
They modeled system using several automatons, each with a separate specification, which interact via message passing. 
Such compositional approach allows delaying the implementation of detailed models for some components while making feasible to verify high-level security properties of the system like integrity and confidentiality.
A further advantage of decomposing the model of a system into separate components is that guarantees on constant parts can be reused when other parts change.

There are a few other issues that should be addressed before our hypervisor can get adopted in real-world scenarios. For example, the current hypervisor does not allow explicit communication
between partitions.  Enabling communication, however, makes formal verification challenging. The reason for this is that permitting the flow of information between system components makes it infeasible to use classical
noninterference properties like \textit{observational determinism}~\cite{McLean:1992:PNF:2699855.2699858} to show system confidentiality. In such a setting, proving isolation 
requires a careful analysis to ensure partitions cannot infer anything more that what is allowed by the information flow policy about one another's internal states, e.g., the secret key of a cryptographic algorithm 
stored in partition memory.

Finally, in this thesis, we restrict our experiments to software running on a single-core processor connected to first-level caches, namely L1- data and instruction caches. 
Further investigations would be needed to  understand the security impact of other components such second-level caches, pipelines, branch prediction unit, TLBs, and enabling multi-core 
processing on the platform, some of which are left for future work and discussed in more details in respective papers.

\part{Included Papers}

\setcounter{chapter}{0}
\renewcommand{\chaptername}{Paper}
\renewcommand{\thechapter}{\Alph{chapter}}

\chapter{Provably secure memory isolation for Linux on ARM}\label{paper:JCS}
\chaptermark{Provably secure memory isolation}
\backgroundsetup{position={current page.north east},vshift=1cm,hshift=-3cm,contents={\VerBar{OliveGreen}{3cm}}}
\BgThispage

\begin{center}
Roberto Guanciale, Hamed Nemati, Mads Dam, Christoph Baumann
\end{center}

\begin{abstract}
The isolation of security-critical components from an untrusted OS allows to both protect applications and to harden
the OS itself. Virtualization of the memory subsystem is a key component to
provide such isolation. We present the design, implementation and verification of a memory virtualization platform for 
ARMv7-A processors. The design is based on direct paging, an MMU virtualization mechanism previously introduced
by Xen. It is shown that this mechanism can be implemented using a compact design,
suitable for formal verification down to a low level
of abstraction, without penalizing system performance. The verification is performed using the HOL4 theorem prover and uses
a detailed model of the processor. We prove memory isolation along with
information flow security for an abstract top-level model of the virtualization mechanism. The abstract model is refined down
to a transition system closely resembling a C implementation. Additionally, it is demonstrated how the gap between the
low-level abstraction and the binary level-can be filled, using tools that check Hoare contracts.
The virtualization mechanism is demonstrated on real hardware via a hypervisor hosting Linux and 
 supporting a tamper-proof run-time monitor that provably prevents code
 injection in the Linux guest.
\end{abstract}

\newenvironment{defn}[1][Definition]{\begin{trivlist}
\item[\hskip \labelsep {\bfseries #1}]}{\end{trivlist}}
\newtheorem{corollary}{Corollary}
\newtheorem{proposition}{Proposition}
\newcommand{\HypState}{h}
\newcommand{\type}{\mathit{pgtype}}
\newcommand{\ActiveGuestReal}[1]{\mathit{act}_#1}
\newcommand{\ActiveGuestIdeal}[1]{\mathit{act}_#1}
\newcommand{\Trace}{\omega}
\newcommand{\project}{\mathit{prj}}
\newcommand{\map}{\mathit{map}}
\newcommand{\TraceSet}[2]{\mathit{tr}_{#1,#2}}
\newcommand{\TraceSetb}[1]{\mathit{tr}_{#1}}
\newcommand{\ArmState}{\sigma}
\newcommand{\UserMode}{\mathit{PL0}}
\newcommand{\KernelMode}{\mathit{PL1}}
\newcommand{\Sec}{S}
\newcommand{\PLm}{\mathit{pl}}
\newcommand{\req}{\mathit{req}}
\newcommand{\bregs}{\mathit{bregs}}
\newcommand{\uregs}{\mathit{uregs}}

\section{Introduction}
A basic security requirement for systems that allow software to
execute at different levels of security is memory isolation: The
ability to store a secret or to enforce data integrity within a designated part of memory
and prevent the contents of this memory to be affected by, or leak to,
parts of the system that are not authorised to access it.
Without the usage of special hardware, trustworthy memory isolation
is dependent on the OS kernel being correctly implemented. 
However, given the size and complexity of modern OSs, the vision of comprehensive and
formal verification of commodity OSs is as distant as ever.

An alternative to verifying the entire OS is to delegate critical
functionality to special low-level execution platforms such as
hypervisors, separation kernels, or microkernels. Such an approach has
some significant advantages. First, the size and complexity of the
execution platform can be made much smaller, potentially opening up
for rigorous verification. The literature has many recent examples of
this, in seL4 \cite{DBLP:conf/sosp/KleinEHACDEEKNSTW09}, Microsoft's
Hyper-V project~\cite{CavalcantiD09}, Green Hills' CC certified
INTEGRITY-178B separation kernel~\cite{INTEGRITY}, and the Singularity~\cite{hunt2007singularity} microkernel Second, the
platform can be opened up to public scrutiny and certification, independent of
application stacks.

Virtualization-like mechanisms can also be used to support various forms
of application hardening against untrusted OSs. Examples of this include
KCoFi \cite{criswell2014kcofi} based on the Secure Virtual Architecture~(SVA)~\cite{criswell2007secure}, Overshadow 
\cite{Chen:2008:OVA:1346281.1346284}, Inktag \cite{Hofmann:2013:ISA:2451116.2451146}, and Virtual Ghost \cite{Criswell:2014:VGP:2541940.2541986}.
All these examples rely crucially on memory isolation to provide the required security guarantees, typically by virtualizing the memory management unit (MMU) hardware.
MMU virtualization, however, can be exceedingly tricky to get right, motivating the use of formal methods for its verification.

In this paper we present an MMU virtualization API for the ARMv7-A
processor family and its formal verification down to the binary
level.
A distinguishing feature of our design is the use of direct
paging, a virtualization mechanism introduced by Xen \cite{xen} and used later with some variations by the SVA.
In direct paging, page tables are kept in guest memory and allowed to be read and directly manipulated by the
untrusted guest OS (when they are not in active use by the MMU). Xen demonstrated that this approach has better
performance than other software virtualization approaches (e.g. shadow page tables) on the x86 architecture. Moreover, 
since direct paging does not require shadow data structures, this approach has small memory overhead. 
The engineering challenge inherent to this project is to  design a minimal
API that (i) is sufficiently
expressive to host a paravirtualized Linux, (ii) introduces an
acceptable overhead and (iii) whose implementation is sufficiently
small to be subject to pervasive verification for a commodity CPU architecture such as ARMv7.
 

The security objective is to allow an untrusted guest system to operate freely, invoking the hypervisor at will, without
being able to access memory or processor resources for which the guest has not received static permission. In this paper
we describe the design, implementation, and evaluation of our memory virtualization API, and the formal verification of its security properties. The verification is performed using a formal model of the ARMv7 architecture \cite{DBLP:conf/itp/FoxM10}, 
implemented in the HOL4 interactive theorem prover.

The proof strategy is to establish a bisimilarity between the hypervisor
executing on a formal model of the ARMv7 instruction set architecture and the top level specification (TLS).
The TLS describes the desired behaviour of the system consisting of handlers implementing the virtualization mechanism and the behaviour of machine instructions executed by the
untrusted guest. The specification of the MMU virtualization API involves an abstract model state that is not represented in memory and thus by design invulnerable to direct guest access.
Due to the direct paging approach, however, the page tables that control the MMU are residing in guest memory and need to be modelled explicitly.
Hence, it is no longer self-evident that the desired memory isolation properties, no-exfiltration and no-infiltration in the terminology of \cite{Heitmeyer:2008:AFM:1340674.1340715}, 
hold for guests in the TLS, and an important and novel part of the verification is therefore to formally validate that these properties indeed hold.

To keep the TLS as simple and abstract as possible, the TLS addresses page tables directly using
their physical addresses. A real implementation cannot do this, but must use virtual addresses instead, in addition to managing its internal data structures.
To this end an implementation model is introduced,
%
which uses virtual addresses instead of
physical ones and stores the abstract model state
explicitly in memory. This provides a very low-level C-like model of handler
execution, directly reflecting all algorithmic features of the memory subsystem virtualization
implemented by the binary code of the handlers, on the real ARMv7 state, as represented by the HOL4 model.
We exhibit a refinement from the TLS to the implementation
model, prove its correctness, and show, as a corollary, that the memory isolation properties proved
for the TLS transfer to the implementation model. This constitutes the second part of the verification.

The next step is to fill the gap between the verification of this low-level
  abstraction and the binary level.
To accomplish this an additional refinement must be established. Using the same approach as \cite{TrustED},
we demonstrate how this can be achieved using 
a combination of theorem proving and tools that check contracts
for binary code. The machine code verification is then in charge of establishing
that the hypervisor code fragments respect these contracts, expressed as Hoare triples.
Pre and post conditions are generated semi-automatically starting from the specification of the low-level abstraction and
the refinement relation. They are then transferred to the
binary analysis tool BAP~\cite{Brumley:2011:BBA:2032305.2032342}, which is used to verify the hypervisor handlers at the assembly
level. Several tools have been developed to support this task,
including a lifter that transforms ARM code to the machine independent
language that can be analysed by BAP and a procedure to resolve indirect jumps.
The binary verification of the hypervisor has not been completed yet. However, 
we demonstrate the methodology outlined above by applying it to prove correctness of the binary code of one of the
API calls. The scalability of the approach has been shown in~\cite{dam2013formal}, where
it was used to verify the binary code of a complete separation kernel.

An alternative approach would be to focus the code verification at the C level.
First, such an approach does not directly give assurances at the ISA level,
which is our objective. 
This can be partly addressed by a certifying compiler such as
CompCert\cite{Leroy-Compcert-CACM}. However, system level code is currently not supported by such
compilers. Moreover, this type of code is prone to break the standard
C-semantics, for example by reconfiguring the MMU and changing the virtual
memory mapping of the program under verification as is the case here.

The verification highlighted three classes of bugs in the initial design of the virtualization mechanism:
\begin{enumerate}
\item Arithmetic overflows, bit field and offset mismatches, and signed operators where
the unsigned ones were needed.
\item Missing checks of self referencing page tables.
\item Approval of guest requests that cause unpredictable behaviours of the ARMv7 MMU.
\end{enumerate}
Moreover, the verification of the implementation model identified 
additional bugs exploitable by requesting the
validation of physical blocks residing outside the guest memory.
This last class of bugs was identified because the implementation model takes into
account the virtual memory mapping used by the handlers.
Finally, the binary code
verification identified a buffer overflow.

We report on a port of Linux kernel 2.6.34 and demonstrate the prototype implementation of a hypervisor for
which the core component is the verified MMU virtualization API.
The complete hypervisor augments the memory virtualization API by handlers that
route aborts and interrupts inside Linux.
Experiments  demonstrate that the hypervisor can run with reasonable performance on
real hardware (Beagleboard-xM based on the Cortex-A8 CPU). 
Furthermore an application scenario is demonstrated based on a trusted run-time
monitor. The monitor executes alongside the untrusted Linux system, 
enforces the W$\oplus$X  policy (no memory area can be writable and executable simultaneously) and uses code signing to prevent binary
code injection in the untrusted system.

\subsection{Scope and limitations}
The binary verification of the hypervisor has not been completed yet. However, 
we demonstrate the methodology outlined above by applying it to prove correctness of the binary code of one of the
API calls. The scalability of the approach has been shown in~\cite{dam2013formal}, where
it was used to verify the binary code of a complete separation kernel.
In Section~\ref{sec:limitation} we comment on the
tasks that are not automated and need to be manually accomplished to complete
the verification.


\section{Related Work}
The size and complexity of commodity
OSs make them susceptible to attacks that can bypass their
security mechanisms, as demonstrated in e.g. \cite{kernelvuln,KargerS02}.
The ability to isolate security-critical components from an untrusted OS allows non critical parts of a system to be implemented while the critical software remains adequately protected.  This isolation can be used both to protect applications from an untrusted OS as well as to protect the OS itself from internal threats.
For example, KCoFI~\cite{criswell2014kcofi} uses Secure Virtual Architecture~\cite{criswell2007secure} to isolate the OS from a 
run-time checker. The checker instruments the OS and monitors its activities
to guarantee the control-flow integrity of the OS itself. Related examples are application hardening frameworks such as Overshadow \cite{Chen:2008:OVA:1346281.1346284}, Inktag \cite{Hofmann:2013:ISA:2451116.2451146}, and Virtual Ghost \cite{Criswell:2014:VGP:2541940.2541986}. In all these cases some form of virtualization of the MMU hardware is a critical component to provide the required isolation guarantees.


Shadow page tables (SPT) is a common approach to MMU virtualization.  The virtualization layer maintains a shadow copy of page tables created and
maintained by the guest OS. The MMU uses only the shadow pages, which are updated after the virtualization layer validates the OS
changes. The Hyper-V hypervisor which uses shadow pages on x86, has been formally verified using the semi automated VCC tool \cite{CavalcantiD09}.
Related work \cite{alkassar2010automated,PSS12} uses shadow page tables  to provide full virtualization, including virtual memory, for  ``baby VAMP'', a simplified MIPS, using VCC. This work, along with later work~\cite{ACKP12} on TLB virtualization for an abstract mode of x64, has been verified using Wolfgang Paul's VCC-based simulation framework~\cite{CohenPS13}. Also, the OKL4-microvisor uses shadow paging to virtualize the memory subsystem~\cite{heiser2010okl4}. However, this hypervisor has not been verified.

Some modern CPUs provide native hardware support for virtualization.
The ARM Virtualization Extensions~\cite{arm:cortex15} augment
the CPU with a new execution mode and provide a two stage address translation.
These features greatly reduce the complexity of the virtualization layer~\cite{Varanasi:2011:HVA:2103799.2103813}.
XHMF~\cite{Vasudevan:2013:DIV:2497621.2498126} and CertiKOS~\cite{gu2015deep} are examples of verified hypervisors for the x86 architecture that control memory operations of guests .
using virtualization extensions.
%
The availability of hardware virtualization extensions, however, does not make software based solutions obsolete.
For example, the recent Cortex-A5 (used in feature-phones) 
and the legacy ARM11 cores (used in home network appliances and the 2014 ``New Nintendo 3DS'') do not make use of such extensions.
Today, the Internet of Things (IoT) and wearable computing are dominated by microcontrollers (e.g. Cortex-M). As the 
recent Intel Quark demonstrates, the necessity of executing legacy stacks (e.g. Linux) is pushing 
towards equipping these microcontrollers with an MMU. Quark and the upcoming ARMv8-R 
both support an MMU  and lack two stage page-tables. Generally, there is
  no universal sweet spot that reconciles the demands for low cost, low power
  consumption and rich hardware features.
For instance, solutions based on FPGAs and soft-cores such as LEON can benefit from software based virtualization by freeing gates not used for virtualization extensions to be used for application specific logic (e.g. digital signal processing, software-defined radio, cryptography). 

A virtualization layer provides to the guest OS an interface similar to the underlying hardware. An alternative approach is to execute
the commodity OS as a partition of a
microkernel, by mapping the OS threads directly to the microkernel threads, thus delegating completely
the process management functionality from the hosted OSes to the
microkernel (e.g. L$^4$Linux).
This generally involves an invasive and error-prone OS adaptation process, however.
The formal verification of seL4~\cite{DBLP:conf/sosp/KleinEHACDEEKNSTW09} demonstrated that a
detailed analysis of the security properties of a complete microkernel
is possible even at the machine code level~\cite{Sewell:2013:TVV:2491956.2462183}. Similarly, the Ironclad Apps framework \cite{Ironclad} hosts security services in a remote operating system. Its functional correctness and information flow properties are verified on the assembly level.

In order to achieve trustworthy isolation between partitions, more light-weight
solutions can also be employed, namely formally verified separation kernels
\cite{INTEGRITY, dam2013formal, baumann2011proving} and Software Fault Isolation
(SFI)~\cite{Wahbe:1993:ESF:168619.168635,zhao2011armor}
. The latter has the advantage over the former in that it is a software-only approach, not relying on common hardware components such as MMU and memory protection units (MPU). Nevertheless, both mechanisms are generally not equipped with the functionality needed to host a commodity OS. Conversely, formally verified processor architectures specifically designed with a focus on logical partitioning \cite{WildingGRH10} and information flow control \cite{azevedo2014verified} can be used to achieve isolation.

\subsection{Contributions}
We present a platform to virtualize the memory subsystem
of a real commodity CPU architecture: The ARMv7-A. The virtualization platform is based on direct paging, 
a virtualization approach inspired by the
paravirtualization mechanism of Xen~\cite{xen} and 
Secure Virtual Architecture~\cite{criswell2007secure}.
The design of the platform is sufficiently slim to enable 
its formal verification without penalizing the system performance.
The verification is performed down to a detailed model of the architecture,
including a detailed model of the ARMv7 MMU. This
enables our threat model to consist of an arbitrary guest
that can execute any ARMv7 instruction in user mode.
We prove complete mediation of the MMU configurations,
memory isolation of the hosted components, and information flow correctness. Additionally, we present our methodology for the binary verification of hypervisor code and report on first results. So far, one handler has been verified on the binary level. Completing the binary verification for all handlers is work in progress.
The viability of the platform is demonstrated via a prototype hypervisor that is 
capable of hosting a Linux system while provably isolating it
from other services. The hypervisor supports BeagleBoard-xM (a 
development board based on ARM Cortex-A8) and is used to benchmark the platform on real hardware. 
 As the main application it is shown how the virtualization mechanism can be used to support
a tamper-proof run-time monitor that prevents code injection in an untrusted Linux guest.


\section{Verification Approach}
In Figure~\ref{fig:trace} we give an overview of the entire verification flow
presented in this paper. In particular it depicts the different layers of
modelling, how they are related, and the tools used. This is discussed in more
detail in Section~\ref{sec:proof-engeneering}.

Our MMU virtualization API is designed for paravirtualization and
targets a commodity CPU (ARMv7-A). In such a scenario, the hosting CPU must provide
two levels of execution: privileged and unprivileged. 
The hypervisor is the only software component that is executed at the privileged level;
at this level the software has complete control of the underlying
hardware.
All other software components (including operating
system kernels, user processes, etc.) are executed in unprivileged mode;
direct accesses to the sensitive resources must be prevented and
all transitions to privileged mode are controlled through the use of exceptions and interrupts.

In addition to the MMU virtualization API 
itself, as part of the hypervisor, the system is intended to support two types of clients:
\begin{itemize}
\item An untrusted commodity OS guest (Linux)  running non-critical software (e.g. GUI, browser, server, games). 
\item A set of trusted services such as controllers that drive physical actuators, run-time monitors, sensor drivers, or cryptographic services.
\end{itemize}

An example computation of such system is shown in the row labelled ``Real
model'' of Figure~\ref{fig:trace}. White circles represent states in unprivileged execution level where the untrusted guest (either its kernel or one of its
user processes) are running. Gray circles represent
unprivileged states where one of the trusted services are in control.
Finally, black circles represent states in privileged
level where the hypervisor is active. Transitions between two
unprivileged states (e.g. $1 \rightarrow 2$) do not cause any exceptions.
The transition between the states 2
and 3 is caused by an exception, for example the
execution of a software interrupt. Finally, transitions from
privileged to unprivileged levels (e.g. $6 \rightarrow 7$) are caused by
instructions that explicitly change the execution level.

\begin{figure}
  \centering
  \includegraphics[width=0.8\linewidth]{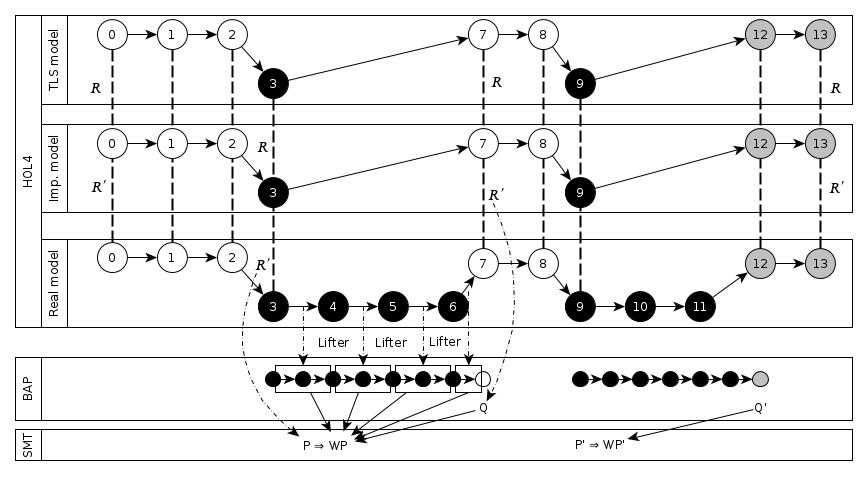}
  \caption{Executions of a real machine (middle), the implementation model (above), and the Top Level Specification (top) and the relations between them. 
  In addition the dependencies of the binary verification methodology (bottom) are depicted.}
  \label{fig:trace}
\end{figure}

\subsection{Attack Model}
Due to the size and complexity of a complete Linux system, a
realistic adversary model must consider the Linux partition compromised. 
For this reason, the attacker is an untrusted paravirtualized
Linux kernel and its user processes, that maliciously or due to an error
may attempt to gain access to resources outside the guest partition. 
Thus, the attacker is free to execute any CPU instruction in
unprivileged mode; 
it is initially not able to directly access the
coprocessor registers, and all attacker memory accesses are initially
mediated by the MMU.
However, by exploiting possible flaws in the hypervisor the attacker
may during the course of a computation gain such access to the
MMU configuration, something our security proof shows is in fact not possible.
In this work, we assume absence of external mechanisms that can directly modify
the internal state of the machine (e.g. external devices or physical
tampering). The analysis of temporal partitioning properties (e.g. timing channels as investigated in~\cite{cock2014last}) is also deliberately left out of this work. 

\subsection{Security Goals}
The verification must demonstrate that the low level platform
does not allow undesired interference between
guest and sensitive resources.
That is:
\begin{enumerate}
  \item The hypervisor 
must play the role of a security monitor of the
MMU settings. If \emph{complete mediation} of the MMU settings is violated, then
an attacker may bypass the hypervisor policies and compromise
the security of the entire system. 
We show this by proving that neither the untrusted
  guest nor the trusted services can directly change the MMU
  configuration.
\item Executions of an arbitrary guest cannot affect the ``trusted
  world'', i.e. the parts of the state the guest is not allowed to
  modify, such as memory of trusted services, system level
registers and status flags, and the hypervisor state. This is an integrity
property, similar to the no-exfiltration property
of~\cite{Heitmeyer:2008:AFM:1340674.1340715}.
  \item Absence of information flow from the trusted world to the
    guest, i.e. confidentiality, similar to no-infiltration
of~\cite{Heitmeyer:2008:AFM:1340674.1340715}.

\end{enumerate}

Note that these properties, as in
\cite{Heitmeyer:2008:AFM:1340674.1340715}, are qualitatively
different: The integrity property is a single-trace property, and concerns the
inability of the guest to directly write some other state
variables. Since it is under guest control when and how to invoke the
virtualization API, there are plenty of indirect communication
channels connecting guests to the hypervisor. For instance, a guest
decision to allocate or deallocate a page table affects large parts of
the hypervisor state,  without ever directly writing to any internal
hypervisor state variable. Enforcing this is in a sense the very
purpose of the hypervisor. On the other hand, the only desired effects
of hypervisor actions should be to allocate/deallocate, map, remap,
and unmap virtual memory resources, leaving any other observation a guest
may make unaffected, thus preventing the guest from extracting information from
inaccessible resources even indirectly. This is essentially a two-trace
information
flow property, needed to break guest-to-guest (or guest-to-service)
information channels in much the same way as intransitive
noninterference is used in \cite{murray2013sel4} to break
guest-to-guest channels passing through the scheduler in seL4.

In this work we establish these properties via successive refinements
that add more details (that in turn can highlight different misbehaviour of
the system) to the virtualization API, starting from an abstract model refining down to the binary code of the
low level execution platform.
We first demonstrate that the intended security property holds for the
most abstract model. 
At each refinement, the proof consist of (i) identifying a relation that is strong enough to transfer the security property from
the higher abstract model to the more real one (we call this a \emph{candidate relation}) and (ii)
demonstrating that the candidate relation actually satisfies the properties required from a refinement
relation. For the first task it turns out that one needs a bisimulation relation in order to transfer higher-order information flow properties like confidentiality. The latter task is reduced to subsidiary properties, which
have natural correspondences in previous kernel verification
literature~\cite{Heitmeyer:2008:AFM:1340674.1340715,INTEGRITY}:
\begin{itemize}
\item A malicious guest cannot violate isolation while it is executing. 
\item Executions of the abstract vs the more real model preserve the candidate bisimulation relation.
\end{itemize}
These two tasks are qualitatively different. The former task, due to our use of
memory protection, is really a noninterference-like
property of the hosting architecture rather than a property of the
hypervisor.
This property must hold independently of the hosted guest, which is unknown at verification time since the attacker can take complete
control of the untrusted Linux.
By contrast, the latter task consists in verifying at increasing levels of
detail the functional correctness of the individual handlers.

\subsection{Top Level Specification}
The first verification task focuses on establishing correctness of the
design of the virtualization API.
With this goal, in Section~\ref{sec:tls} we specify the desired
behaviour of the virtualization API as a transition system, called the
\textit{Top Level Specification}  (TLS).
This specification models unprivileged execution of an
arbitrary guest system on top of a CPU with MMU support,
alternating with abstract handler events. These events model
invocations of the hypervisor handlers as atomic transformations 
operating on an abstract machine state. Abstract states are real CPU states
extended by auxiliary (model) data that reflect the internal state of the hypervisor.  We refer
to this auxiliary data as the abstract hypervisor state. Handler
events represent the execution of several instructions at privileged
level, in response to exceptions or interrupts. Modelling handler
effects as atomic state transformations is possible, since the
hypervisor is non-preemptive, i.e. nested exceptions/interrupts are
ruled out by the implementation.

Since in direct paging the guest systems can directly manipulate
inactive page tables,
 the TLS needs to explicitly model page tables in
memory. This contrasts simpler models such as the one presented
in \cite{dam2013formal} where the hypervisor state was represented in the TLS using abstract model variables only. 
For this reason, establishing complete mediation, integrity, and confidentiality for the TLS is far from trivial.

\subsection{Implementation Model}\label{sec:method:refinement}
Extending the security properties to an actual implementation, however, requires additional work, for the following reasons:
\begin{itemize}
\item The TLS uses auxiliary data structures (the abstract hypervisor state) that are not stored inside the system memory.
\item The TLS accesses the memory directly using physical addresses.
\end{itemize}
As is common practice, the virtualization code executes under the same address translation as
the guest (but with different access permissions), in order to reduce the number of context switches
required. 
For this approach it is critical to verify that all low-level
operations performed by the hypervisor correctly implement the TLS
specification; 
these operations include reads and updates of the page tables, and
reads and updates of the hypervisor data structures. 
To show implementation soundness  we exhibit a refinement property
relating TLS states with states of the implementation. 
The refinement relation is proven to be preserved by all atomic hypervisor operations;
reads and updates of the page tables, reads and updates of the hypervisor data structures. In particular it is established that these virtual memory operations access the correct 
physical addresses and never produce any data abort exceptions.
Moreover, it is shown that the refinement relation directly transfers
both the integrity properties and the information flow properties of
the TLS to the implementation level.

\subsection{Binary Verification}
The last verification step consists in filling the gap between the
implementation and the binary code executed on the actual hardware.
This requires to exhibit a refinement 
relation between the implementation model and the real model of
the system (i.e.~where each transition represents the execution of one
binary instruction). 

Intuitively, internal hypervisor {steps} cannot be observed by
the guests, since during the execution of the handler no guest is active.
Moreover, {as} the hypervisor does not support preemption,
then the execution of handlers cannot be interrupted.
These facts permit to disregard internal states of the handlers and
limit the refinement to relate only states where the guests are executing.

Thus, the binary verification can be accomplished in three
steps:
(i) verification that the refinement relation directly transfers
the isolation properties to the real model,
(ii) verification of a top level theorem that transforms the
relational reasoning into a set of contracts for the
handlers and guarantees that the refinement
is established if all contracts are satisfied, and 
(iii) 
verification of the machine code.
The last step establishes if the hypervisor code
fragments respect the contracts, expressed as Hoare triples $\{P\}C\{Q\}$,
where $P$ and $Q$ are the pre/post conditions of the assembly fragment $C$.

\subsection{Proof Engineering}\label{sec:proof-engeneering}
We use Figure~\ref{fig:trace} and Table~\ref{tbl:summary}
to summarise the models, theorems and tools that
are described in the following sections.
We use three transition systems; the TLS
(Section~\ref{sec:tls}), 
the Implementation Model (Section~\ref{sec:implmodel}) and the
ARMv7 model (Section~\ref{sec:armv7}). These transition systems have been
defined in the HOL4 theorem prover and differ in the level of abstraction they
use to represent the hypervisor behaviour.
The three transition systems model guest behaviour
identically (e.g. transitions $0 \rightarrow 1$); these transitions 
obey the access privileges computed by the MMU
and satisfy properties~\ref{prop:User-No-Exfiltration} and
~\ref{prop:User-No-Infiltration} of Section~\ref{sec:armv7}.
These properties have been verified for a simplified MMU model in~\cite{khakpour2013machine}.

We use HOL4 to verify that the security properties hold for the TLS 
(Theorems~\ref{lem:invariant}, ~\ref{lem:mmu-integrity}, ~\ref{lem:switch},
~\ref{lem:No-Exfiltration} and ~\ref{lem:No-Infiltration} of
Section~\ref{sec:tls}). The reasoning used to implement the proofs in the
interactive theorem prover is summarised in Section~\ref{sec:tlsproof}.

The refinement ($\mathcal{R}$) between the TLS and the
implementation model is verified in HOL4 (Theorem~\ref{lem:refinement} of
Section~\ref{sec:implmodel}). We also use HOL4 to prove that the refinement
transfers the security properties of the TLS to the implementation model
(Corollary~\ref{lem:MLS-sec}). 

The refinement ($\mathcal{R'}$) between the implementation model 
and the real model is formally defined in HOL4,
allowing us to prove that the refinement transfers the security properties to
the ARMv7 model (Corollary~\ref{lem:MLS-binary}).

The verification of the refinement (Theorem~\ref{lem:binary-verification} of
Section~\ref{sec:ver:binary}) is only partial: we demonstrate
the verification of the binary code of the hypervisor only for a part of
the code-base and we rely on some assumptions in order to fill the semantic
gap between HOL4 and the external tools. 
We prove Theorem~\ref{lem:binary-verification} for non-privileged transitions in
HOL4 (i.e. transitions not involving the hypervisor code such as $1 \rightarrow
2$ and $12 \rightarrow 13$). 

For the hypervisor code, we show that the task can be partially automated  by
means of external tools. For this purpose  we use the HOL4 model of ARMv7 to
transform the binary code of the hypervisor (e.g. the code executed between
states 3 and 7 in the real model) to the input language of BAP
(represented in the figure by the arrow labelled ``Lifter'').
The usage of HOL4 for this task allows us to reduce the assumptions needed to
fill the gap between the HOL4 ARMv7 model and BAP, as described in
Section~\ref{sec:binary}. The methodology to complete the verification is the
following: given a hypervisor handler whose code has been translated to the BAP
code $C$, we use a 
HOL4 certifying procedure that generates a contract $\{P\}C\{Q\}$ starting
from the hypervisor implementation model and the refinement relation. The
certifying procedure yields a HOL4 theorem stating that the refinement relation 
$\mathcal{R'}$ is preserved if the hypervisor handler $C$ establishes the
postcondition $Q$ starting from the precondition $P$.
We use BAP to compute the weakest precondition $WP$ of the postcondition
$Q$ and the code $C$ and a finally an SMT solver checks that the weakest
precondition is entailed by the precondition.

\begin{table}
  \begin{tabular}{|l|p{5cm}|c|c|l|}
    \hline
    Artefact & Description & HOL4 & BAP & In\\
    \hline

     TLS & Model of the abstract design of the hypervisor + attacker (guest) & $\circ$ & & \cite{SOFSEM}\\
     \hline
     Implementation model & Low level model of the hypervisor + attacker & $\circ$ & & \cite{TrustED}\\
     \hline
     ARM model & Real model of the system & $\circ$ &  & \cite{DBLP:conf/itp/FoxM10,SOFSEM}\\
     \hline
    Properties~\ref{prop:User-No-Exfiltration} and ~\ref{prop:User-No-Infiltration} & Properties of the ARM instruction set (here only assumed) &
     $\circ$ & & \cite{khakpour2013machine}\\
     \hline
     \multicolumn{5}{|l|}{
     Properties of the TLS
     }\\
     \hline
     Theorem~\ref{lem:invariant} & Verification of the functional invariant & $\circ$ &  & \cite{SOFSEM}\\
     \hline
     Theorem~\ref{lem:mmu-integrity} & Verification of MMU integrity & $\circ$ &  &  \cite{SOFSEM}\\
     \hline
     Theorem~\ref{lem:switch} & Verification of no context switch & $\circ$ &  &  \cite{SOFSEM}\\
     \hline
     \multirow{ 2}{*}{\parbox{3cm}{Theorems~\ref{lem:No-Exfiltration} and ~\ref{lem:No-Infiltration}}} 
      & Verification of No exfiltration + No infiltration = isolation 
                           & $\circ$ &  &  \cite{SOFSEM}\\
     \hline
     \multicolumn{5}{|l|}{
     Properties of the Implementation model
     }\\
     \hline
     Theorem~\ref{lem:refinement} & Verification of Refinement & $\circ$ &  &  \cite{TrustED}\\
     \hline
     Corollary~\ref{lem:MLS-sec} & Verification of MMU integrity + No exfiltration + No infiltration & $\circ$ &  &  \cite{TrustED}\\
     \hline
     \multicolumn{5}{|l|}{
     Properties of the real model
     }\\
     \hline
     Theorem~\ref{lem:binary-verification} & Refinement.
     For non-privileged transitions proved in HOL4. One of the API function proved using BAP & $\circ$ &  $\circ$ &  Here \\
     \hline
     Corollary~\ref{lem:MLS-binary} & Verification of MMU integrity + No exfiltration + No infiltration & $\circ$ &  &  Here \\
     \hline
     \multicolumn{5}{|l|}{
     Miscellaneous
     }\\
     \hline
     Lifter &  Translation of ARMv7 binary to BIL &  $\circ$ &  $\circ$ & \cite{TrustED,dam2013formal}\\
     \hline
     Certifying procedure & Generates a contract starting from the  model of one of the API function and the refinement relation  & $\circ$ &  &  Here \\
     \hline
     Indirect jump solver & Computes all possible target of indirect jumps for a BIL loop free program. Here extended and re-implemented as BAP extension & & $\circ$ & \cite{TrustED}
     \\
     \hline
  \end{tabular}
  \caption{
    List and first appearance of models, theorems and tools. 
  }
  \label{tbl:summary}
  
\end{table}


\section{The ARMv7 CPU}\label{sec:armv7}
ARMv7 is the currently dominant processor architecture in embedded devices.
Our verification relies on the HOL4 model of ARM developed at
Cambridge~\cite{DBLP:conf/itp/FoxM10}. The use of a theorem prover
allows the verification goals to be stated in a manner which is faithful
to the intuition, without resorting to approximations and abstractions
that would be needed when using a fully automated tool such as a model
checker. Furthermore, basing the verification on the Cambridge ARM
model lends high trustworthiness to the exercise: The Cambridge model
is well-tested and phrased in a manner that retains a high
resemblance to the pseudocode used by ARM in the architecture
reference manual \cite{ARMV7}. The Cambridge model has been extended
by ourselves to include MMU functionality. The resulting model gives a
highly detailed account of the ISA level instruction semantics at the
different privilege levels, including relevant MMU coprocessor
effects. It must be noted that the Cambridge ARM model assumes
linearizable memory, and so can be used out of the box only for
processor and hypervisor implementations that satisfy this property, for instance through adequate cache flushing as discussed in Section~\ref{sec:cache}. 

We outline the HOL4 ARMv7 model in sufficient detail to make the formal results
presented later understandable.
An ARMv7 machine state is a record $$\sigma = \tuple{\uregs,\bregs,\coregs,\mem}\in\Sigma\ ,$$
where $\uregs$, $\bregs$, $\coregs$, and  $\mem$, respectively, represent
the user registers, banked registers (used for handling exceptions), coprocessors, and
memory.
The function $\mode(\sigma)$ returns the current  privilege execution
mode in the  state $\sigma$, which can be either $\UserMode$
(unprivileged or user mode, used by the guest) or $\KernelMode$ (privileged mode, used by the
hypervisor). The memory is the function $\mem \in 2^{32} \rightarrow
2^8$. The coprocessor registers $\mathit{\coregs}$ control the MMU. 

System behaviour is modelled by the state transition relation
$\to_{l \in \{\UserMode, \KernelMode\}} \subseteq \RealStateSpace
\times \RealStateSpace$, where a transition is performed by  the execution of an ARM instruction.
Unprivileged transitions ($\ArmState \to_{\UserMode} \ArmState'$)
start from and end in states that are in unprivileged execution mode
(i.e. $\priv(\ArmState) = \priv(\ArmState') = \UserMode$).
All the other transitions ($\ArmState \to_{\KernelMode} \ArmState'$)
involve at least one state in privileged level.
The raising of an exception is modelled by a transition that enables the
level $\KernelMode$.
An exception can be raised because: (i) a software interrupt (SWI) is
executed, (ii) the current instruction is undefined, (iii) a memory access
is attempted that is disallowed by the MMU, or (iv) an hardware interrupt is
received. 
Whenever an exception occurs, the CPU
disables the interrupts  and jumps to a predefined address in the vector
table to transfer control to the corresponding exception handler.

The ARMv7 MMU uses a two level translation scheme.
The first level (L1) consists of a 4096 entry table that divides up to 4GB of memory into 1MB sections.
These sections can either point to an equally large region of physical memory or to a level 2 (L2)
page table with 256 entries that maps the 1MB section into 4 KB physical pages. 
MMU behaviour is modelled by the function $mmu(\ArmState, \PLm, va, \req)$, which
takes a state $\ArmState$, a privilege level, a virtual address
$va$ and an access request $\req \in \{rd, wt,ex\}$
 (representing read, write and execute accesses)
and yields $pa \in 2^{32} \cup \{\bot\}$, where $pa$ is the
translated physical address or an access denied. The ARMv7 documentation
describes the possibility of unpredictable behaviour due to erroneous
setup of the MMU through coprocessor registers and page tables.
In this work the hypervisor completely mediates the MMU configuration
and aims to rule out this kind of behaviour.

In the ARM architecture \emph{domains} provide a discretionary access control
mechanism. This mechanism is orthogonal to the one provided by CPU execution modes.
There are sixteen domains, each on activated independently in one of the coprocessor registers $\mathit{\coregs}$. 
The page tables map each virtual page/section to one of the domains and the MMU forbids accesses to a page/section
if the corresponding domain is not active.

The state transition relation queries the MMU
whenever a virtual address is accessed, and raises an
exception if the requested access mode is not allowed. 
To describe the security properties guaranteed by an ARMv7 CPU
we introduce some auxiliary definitions.

\begin{definition}[Physical memory access rights] The predicate $\mathit{mmu}_{ph}$ 
takes a state $\ArmState$, the privilege level $\PLm$, a physical
address $pa$ and an access permission $\req \in \{rd, wt,ex\}$ and holds
if the access permission is granted for physical address $pa$.
$$\mathit{mmu}_{ph}(\ArmState, \PLm, pa, \req)\ \Leftrightarrow\ \exists va .\ mmu(\ArmState, \PLm, va, \req) = pa$$
\end{definition}
The ARMv7 MMU mediates accesses to the virtual memory, enabling or forbidding
operations to virtual addresses. Intuitively, a physical address $pa$ can be
read (written) if it exists at least a virtual addresses $va$ that can be read
(written) and that is mapped to $pa$ according with the current page tables.
\begin{definition}[Write-derivability]\label{def:mmu-der} We say that a state $\ArmState'$ is \emph{write-derivable} from a state $\ArmState$ in privilege level $\PLm$ if their memories differ only for physical addresses that are writable in $\PLm$. 
  $$wd(\sigma, \sigma', \PLm) \ \Leftrightarrow\ \forall pa.\ \sigma.\mathit{mem}(pa) \neq \sigma'.\mathit{mem}(pa) \Rightarrow \mathit{mmu}_{ph}(\ArmState, \PLm, pa, wt)\ .$$
\end{definition}
According with the MMU configuration in $\sigma$, only a subset of physical
addresses are writable ($\mathit{mmu}_{ph}(\ArmState, \PLm, pa, wt)$).
Write-derivability
identifies the set of states that can be
produced by changing the memory content of an arbitrary number of
such physical addresses with arbitrary values.
\begin{definition}[MMU-equivalence] We say that two states are \textit{MMU-equivalent} if for any virtual
address $va$ the MMU yields the same translation and the same access
permissions.
$$\ArmState \equiv_{\mathit{mmu}} \ArmState' \ \Leftrightarrow\ \forall va, \PLm, \req.\ mmu(\ArmState, \PLm, va, \req) = mmu(\ArmState', \PLm, va, \req)$$
\end{definition}
Informally, two states are \textit{MMU-equivalent} if their MMUs are configured exactly in
the same way.
\begin{definition}[MMU-safety]\label{def:mmu-safe} Finally, a state is \textit{MMU-safe} if it has the same MMU behaviour as any state with the same coprocessor registers whose memory differs only for addresses that are writable in $\UserMode$.
$$\mathit{mmu}_s(\sigma) \ \Leftrightarrow\ \forall \sigma'.\ \sigma.\coregs=\sigma'.\coregs \land wd(\sigma, \sigma', \UserMode) \Rightarrow (\sigma \equiv_{\mathit{mmu}} \sigma')$$
\end{definition}
A state is \textit{MMU-safe} if there is no way to change the MMU configuration
(i.e. the page tables) by writing into addresses that are writable in non-privileged
mode. That is the MMU configuration prevents non-privileged SW to tamper the
page tables.

An ARMv7 processor that obeys the access privileges computed by the MMU satisfies the following two properties:
\begin{property}[ARM-integrity]\label{prop:User-No-Exfiltration}
  Assume $\sigma \in \RealStateSpace$ with $\mode(\ArmState) = \UserMode$.
  If $\sigma \to_{\UserMode} \sigma'$ and $\mathit{mmu}_s(\sigma)$
  then
  $wd(\ArmState, \ArmState', \UserMode)$ and $\ArmState.\coregs= \ArmState'.\coregs$, i.e., unprivileged steps from MMU-safe states can only lead into write-derivable states and do not affect the coprocessor registers.
\end{property}

Note, that the MMU-safety prerequisite is not redundant here, because single instructions in ARM may result in a series of write operations, e.g., for ``store pair'' and unaligned store instructions. If the MMU configuration was not safe from manipulation in unprivileged mode, then such a series of writes could lead to an intermediate MMU configuration granting more write permissions than the initial one and the resulting state would not be write-derivable from $\ArmState$.

\begin{property}[ARM-confidentiality]\label{prop:User-No-Infiltration}
  Let $\sigma_1, \sigma_2 \in \RealStateSpace$ with $\mode(\ArmState_1) =
  \mode(\ArmState_2) = \UserMode$, and let $A$ contains all physical addresses
  accessible in $\ArmState_1$, i.e.,  $A \supseteq \{ pa \mid \exists \req.\
  \mathit{mmu}_{ph}(\sigma_1, \UserMode,  pa, \req) \}$.
  Suppose that $\sigma_1.\uregs = \sigma_2.\uregs$, 
  $\sigma_1.\coregs= \sigma_2.\coregs$,
  $\sigma_1 \equiv_{\mathit{mmu}} \sigma_2$,
and $\forall pa\in A.\ \sigma_1.\mathit{mem}(pa) = \sigma_2.\mathit{mem}(pa)$.
  If $\sigma_1 \to_{\UserMode} \sigma'_1$, there exists $\sigma'_2$ such that
  $\sigma_2 \to_{\UserMode} \sigma'_2$, $\mathit{mmu}_s(\sigma_1)$, and $\mathit{mmu}_s(\sigma_2)$
  then
  $$\sigma'_1.\uregs = \sigma'_2.\uregs\ , 
  \sigma'_1.\coregs= \sigma'_2.\coregs
\ \text{and}\ \ \forall pa \in A.\ \sigma'_1.\mathit{mem}(pa) = \sigma'_2.\mathit{mem}(pa)\ .$$
\end{property}
Intuitively, Property~\ref{prop:User-No-Infiltration} establishes that in MMU-safe configurations unprivileged transitions only can access information stored in the registers and in the part of memory that is readable in $\UserMode$ according to access permissions. 
Within this paper we take Properties~\ref{prop:User-No-Exfiltration} and \ref {prop:User-No-Infiltration} for granted. In~\cite{khakpour2013machine} the authors validated the HOL4 ARMv7
model against these properties assuming an identity-mapped address translation. Extending the result for an arbitrary but MMU-safe page table setup is currently nearing completion.


\newcommand{\GuestMemType}{0}
\newcommand{\SecMemType}{1}
\newcommand{\DataType}{\textit{data}}
\newcommand{\LTwoType}{\textit{L2}}
\newcommand{\LOneType}{\textit{L1}}

\section{The Memory Virtualization API}\label{sec:architecture}
 
%
%

The memory virtualization API is designed for the ARMv7-A architecture\footnote{In practice,
  the presented design also supports the ARMv6 and ARMv5
  architectures. } and assumes neither hardware  virtualization extensions nor TrustZone~\cite{arm:trustzone}
support. 
To properly isolate the trusted components from the untrusted guest, which hosts a commodity OS, the memory virtualization subsystem needs to provide two main functionalities:
\begin{itemize}
  \item Isolation of memory resources used by the trusted components.
  \item Virtualization of the memory subsystem to enable the untrusted OS to dynamically manage its own memory hierarchy, and to enforce access restrictions.
\end{itemize}
The physical memory region allocated to each type of client is statically defined.  Inside its own region the guest OS is free to manage its own memory, and the virtualization API does not provide any additional guarantees for the security of the guest OS kernel against attacks from its user processes.
However, using trusted services such as a run-time monitor it is possible to provide provable security guarantees to the guest OS, for instance to enforce the W$\oplus$X policy or to secure software updates, as explained in Section~\ref{sec:applications}.

\begin{figure}[h]
  \begin{center}
  \includegraphics[width=6cm]{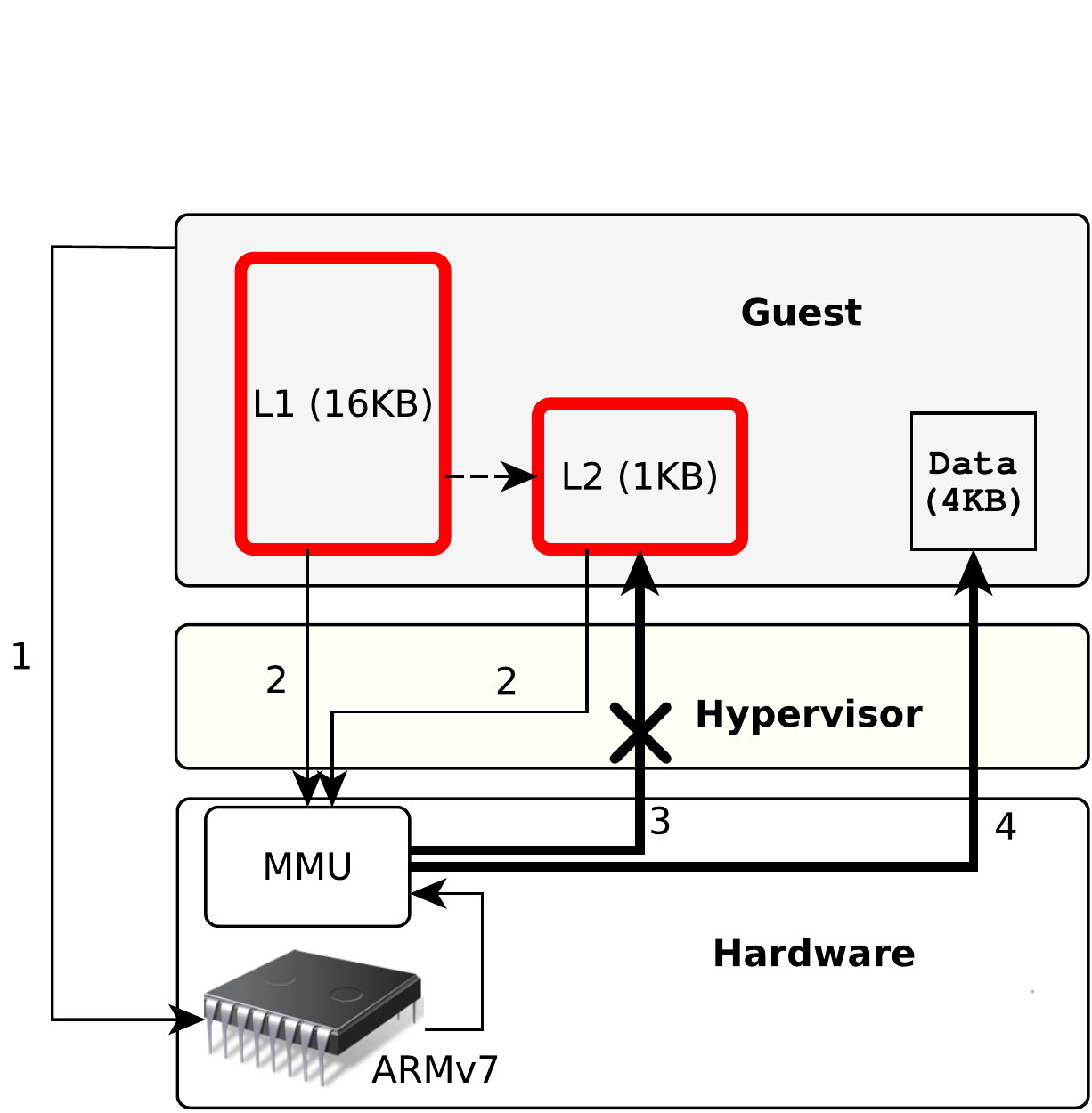}
  \end{center}
\caption{Direct paging:
  1) guest writes to virtual memory are mediated by the MMU as usual;
  2) page tables are allocated in guest memory;
  3) the hypervisor prevents writable mappings to guest memory regions holding
  page tables, forbidding the guest to directly modify them;
  4) the hypervisor allows writable mappings to data blocks in guest memory. 
 }
\label{fig:directPaging}
\end{figure}

\subsection{Memory Management}
The virtual memory layout is defined by a set of page tables that reside in physical memory. 
The configuration of these page tables is security-critical and must not be directly manipulated by untrusted parties.
At the same time, the untrusted Linux kernel needs to manage its memory layout, which requires constant access to the page tables.
Hence the hypervisor must provide a secure access mechanism, which we refer to as virtualizing the memory subsystem.

We use direct paging~\cite{xen} to virtualize the memory subsystem. 
Direct paging allows the guest to allocate the page tables inside its
own memory and to directly manipulate them while the tables are not in active use by the MMU.
Once the page tables are activated, the hypervisor must guarantee that
further updates are possible only via the virtualization API to
modify, allocate and free the page tables.

Physical memory is fragmented into  blocks of 4 KB. Thus, a 32-bit
architecture has $2^{20}$ physical blocks.
Since L1 and L2 page tables have size 16 KB and 1 KB respectively,
 an L1 page table is stored in four
contiguous physical blocks and a physical block can
contain four L2 page tables. 
We assign a type to each physical block, that can be:
 \begin{itemize}
  \item \textit{data}: the block can be written by the guest.
  \item \textit{L1}: contains part of an L1 and is not writable in unprivileged mode.
  \item \textit{L2}: contains four L2 and is not writable in unprivileged mode.
 \end{itemize}
The virtualization API shown in Figure~\ref{APIfig} is very similar to the MMU interface of the Secure Virtual Architecture \cite{criswell2007secure}
and consists of 9 hypercalls that select, create, free, map, or unmap memory blocks or page tables.

Figure \ref{fig:directPaging} indicates the address translation procedure and the connection between components of memory subsystem.
\begin{figure}
\centering
\begin{tabular}{|l|l|}
\hline
\textit{switch} & Select the active L1\\
\hline
\textit{L1create} & Create page table of type L1 \\
\hline
\textit{L2create} & Create page table of type L2 \\
\hline
\textit{L1free} & Change the type of an L1 block to \DataType \\
\hline
\textit{L2free} &  Change the type of an L2 block to \DataType \\
\hline
\textit{L1unmap} & Clear an entry of an L1 page table \\
\hline
\textit{L2unmap} & Clear an entry of an L2 page table \\
\hline
\textit{L1map}  & Set an entry of an L1 page table \\
\hline
\textit{L2map} & Set an entry of an L2 page table \\ \hline
\end{tabular}
\caption{The virtualization API of the hypervisor to support direct paging.}
\label{APIfig}
\end{figure}
\subsection{Enforcing The Page Type Constraints}\label{sec:arch:constraints}
Each API call needs to validate the page type,
 guaranteeing that page tables are write-protected.
This is illustrated in Figure~\ref{fig:invariant}.
\begin{figure}
  \centering
  \raisebox{-0.5\height}{\includegraphics[width=0.5\linewidth]{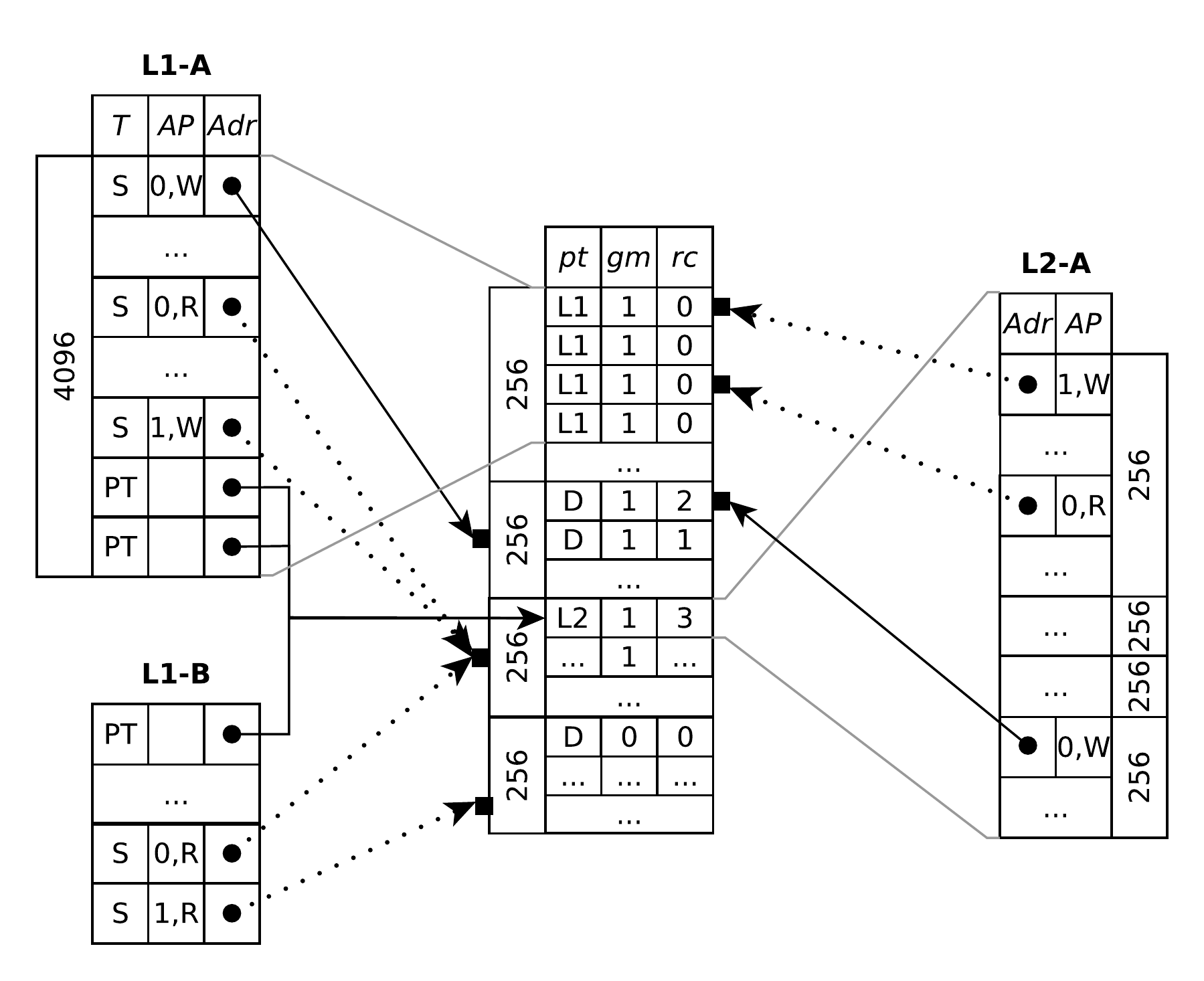}}
  ~
  \raisebox{-0.5\height}{\includegraphics[width=0.47\linewidth]{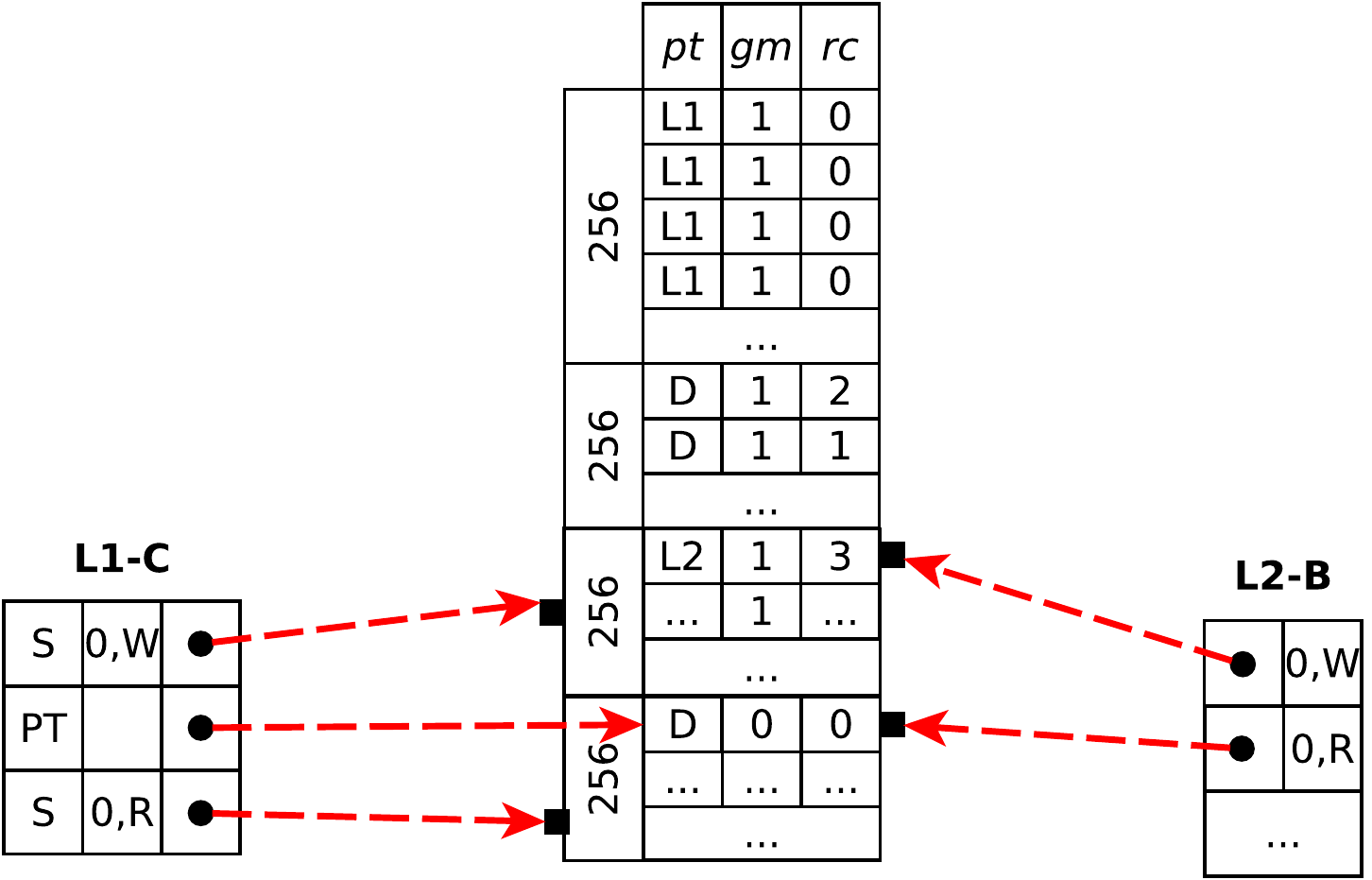}}
\caption{Direct-paging mechanism. We use solid arrows to represent the
L2 page table references and unprivileged write permissions, dotted arrows to represent other allowed
references, and dashed arrows for references violating the page table policy.
}
\label{fig:invariant}
\end{figure}
The table in the centre represents the physical memory and stores the
virtualization data structures  for each physical block;
the page type ($pt$), a flag informing if the block is allocated to
guest memory ($gm$) and a reference counter ($rc$). 

The four top most blocks contain an L1
page table, whose 4096 entries
are depicted by the table \textit{L1-A}. The top entry of the page
table is a section descriptor ($T=S$) that grants write permission to the guest
($\mathit{AP} = (\GuestMemType, w)$).
This entry's address component ($\mathit{Adr}$) points to the second physical section, which consists of
256 physical blocks. Two more section descriptors of L1-A are depicted in
the table: the first one grants read-only permission to the guest ($\GuestMemType, r$), the second descriptor prevents any
guest access and enables write permission for the privileged mode
($\SecMemType, w$).
The last two entries of L1-A are PT-descriptors. Each entry points
to an L2 page table in the same physical block described by table L2-A and containing four L2 page tables.
 
The API calls manipulating an L1 enforce the following policy:
\begin{quotation}Any section descriptor that allows the guest to access the memory
must point to a section for which every physical block resides in the guest memory
space.
Moreover, if a descriptor  enables a guest to write then each block must be typed $\DataType$. Finally, all PT-descriptors must point to 
physical blocks of type L2.
\end{quotation}

The Figure depicts two additional L1 page tables; \textit{L1-B} and
\textit{L1-C}.
The type of a physical block containing \textit{L1-B} can be
transformed to \LOneType~by invoking \textit{L1create}.
On the other hand, a block containing \textit{L1-C}
 is rejected by \textit{L1create} since the block contains three entries that violate the policy.
In fact,
\begin{itemize}
\item the first descriptor
grants guest write permission over a section which has at least one non
data block, in this case L2,
\item the second section descriptor allows the guest to access a
section of the physical memory in which there exists a block that is outside
the guest memory,
and 
\item the third entry is a PT-descriptor, but points to a physical
block that is not typed L2.
\end{itemize}
The first setting clearly breaks MMU-safety, since the guest is now able to write directly to a page table, circumventing the complete mediation of MMU configurations by the hypervisor. The second situation compromises confidentiality and possible integrity of the system if the guest has write access to the block outside its own memory. The third issue may again break MMU-safety if the referenced block is a writable data block. In case the referenced block contains (part of) another L1 page table this setting can lead to unpredictable MMU behaviour, since the L1 page table entries have a different binary format than the expected L2 entries.

The table \textit{L2-A} depicts the content of a physical block typed
L2 that contains four L2 page tables, each consisting of 256 entries.
Each hypercall that manipulates an L2 enforces the following policy:
\begin{quotation}If any entry of the four L2 page tables grants access permission to the
guest then the pointed block must be in the guest memory.
If the entry also enables guest write access then the pointed
block must be typed $\DataType$. \end{quotation}
For example a block containing \textit{L2-B} is rejected by
\textit{L2create},
 since the block contains at least two entries that violate the policy and thus threaten MMU-safety and integrity (in case of the first entry shown) as well as confidentiality (in case of the second one).

A naive run-time check of the page-type policy is not efficient, since
it requires to re-validate the L1 page table whenever the
\textit{switch} hypercall is invoked.
To efficiently enforce that only $\DataType$ blocks 
can be written by the guest, the hypervisor maintains a reference counter,
tracking for each block the sum of:
\begin{enumerate}
  \item The number of descriptors
    providing writable access in user mode to the block.
  \item The number of PT-descriptors that point
     to the block.
\end{enumerate}
 The intuition is that a hypercall can change the type of a
 physical block (e.g. allocate or free a page table) only if the
 corresponding reference counter is zero.
 Lemmas~\ref{lem:inv-sound-type} and ~\ref{Lemma:ref_cnt_counters} in Section
\ref{sec:tlsproof}  demonstrate that this approach is sound and that the page table policy described above is sufficient to guarantee MMU-safety.

\begin{figure}
  \begin{center}
  \includegraphics[width=0.6\linewidth]{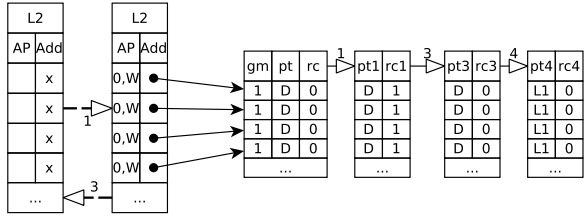}
  \end{center}
\caption{Spawning a new process using the virtualization API. The guest (1) requests a writable mapping of four physical blocks to allocate a new L1 page table. After (2) setting up the table in memory (not shown), it asks (3) to remove the writable mapping, (4) to interpret the blocks as an L1 table, and (5) to make this one the new active L1 page table (not shown).}
\label{fig:example}
\end{figure}
In Figure~\ref{fig:example} we exemplify how an OS can use the API to
spawn a new process. The OS selects four blocks from
its physical memory to allocate a new L1 page table. 
We assume that initially the OS has no virtual
mapping that enables it to access this part of the memory (i.e. the
reference counter $rc$ of these blocks is zero and the type $pt$
is \textit{data}).
\begin{enumerate}
\item Using \textit{L2map}, the OS requests to change an existing L2 page
table, establishing a writable mapping to the four blocks. The hypercall
increases the reference counter accordingly (i.e. $rc1=1$).
\item Without any mediation of the hypervisor, the OS uses the new
mapping to write the content of the new L1
page table.
\item Using \textit{L2unmap}, the guest removes the mapping established
in (1) and decreases the reference counters
(i.e. $rc3=0$).
\item The guest invokes \textit{L1create}, requesting the page table to be validated and 
the block type changed to L1.
The request is granted only if the reference counter is zero, guaranteeing that
there does not exist any mapping in the system that allows the guest to
directly write the content of the page table.
\item Finally, the OS invokes \textit{switch} to perform the context
switch and to activate the new L1.
\end{enumerate}
The example demonstrates some of the principles used to design
the minimal API:
(i) the address of the page tables are chosen by the
guest, thus we do not need to change the OS allocators,
(ii) the preparation of the page table can be done by the OS without
mediation of the hypervisor,
(iii) the content of the page table is not copied into the hypervisor
memory, 
thus reducing memory accesses and memory overhead and
not requiring dynamic allocation in the
hypervisor, 
(iv) tracking the reference counter is used to guarantee the absence of page tables
granting the guest write access to another page table, thus we can allow context
    switches among all created L1s without needing to re-validate
    their content.

\subsection{Integrity of the Hypervisor Memory Map}
When an exception is raised, the CPU redirects execution flow to a fixed location according to the exception vector.
In ARMv7, subsequent instructions are executed in privileged mode but under the same virtual memory mapping as the interrupted guest. 
The hypervisor must enforce that the memory mapping of the exception vector, handler code, and hypervisor data structures is
accessible during an exception without being modifiable by the guests. To this end, the hypervisor maintains its own static virtual memory mapping in a master page table and mirrors the corresponding regions to all L1s of the guest (with restricted access permissions).

  
%

\subsection{Hypervisor Accesses to Guest Page Tables}
The hypervisor APIs must be able to read and write the page tables
allocated by the guest, in order to check the
soundness of the requests and to apply the corresponding changes. 
The naive solution requires the hypervisor to change the current page
table, enabling a hypervisor master page table whenever the guest
memory must be accessed  and then re-enabling the original page
table before the guest is restored.
This solution is expensive as it requires to flush TLB and
caches.
A solution tailored for Unixes can rely on the injective
mapping built by the guest, which can be used by the hypervisor
to access the guest kernel memory. In our settings the hosted
guest is not trusted, thus this solution cannot guarantee that the
injective mapping is obeyed by the guest.
Some ARMv7 CPUs support special coprocessor instructions for virtual-to-physical
address translation. These instructions can be used to validate the
guest injective mapping  at run-time. 
However, this approach is platform dependent and can result in
  nested exceptions that complicate the architecture and verification
  of the hypervisor.
Instead, our design reserves a subset of the virtual 
address space for hypervisor use. The hypervisor master page table
is built so that this address space is always mapped according to an
injective translation (1-to-1) allowing the hypervisor 
 to easily compute the virtual address for each physical address in the guest memory, similar to the direct memory maps supported by FreeBSD \cite{mckusick2004design} and Linux \cite{bovet2005understanding}. 
As with the hypervisor code and data structures, 
 these regions are mirrored in all guest L1 tables.

  

\subsection{Memory Model and Cache Effects}\label{sec:cache}
Hypervisors are complex software interacting directly with many low level hardware components, like processor, MMU, etc. 
Furthermore, there are hardware pieces that, while being invisible to the software layer, still can affect the system behaviour in many aspects.
For example, the memory management unit relies on a caching mechanism, which is used to speed up accesses to page table entries. Basically, a data-cache is a shared resource between all partitions and it thus affects and is affected by activities of each partition. 
Consequently, data-caches may cause unintended interaction between software components running on the
same processor, which can lead to cross-partition information leakage. 

Moreover, for the ARMv7 architecture cache usage may cause sequential consistency to fail if the same physical address is accessed using different cacheability
attributes. This opens up for TOCTTOU\footnote{TOCTTOU -- Time Of Check To Time Of Use}-like vulnerabilities since a trusted agent may
check and later evict a cached data item, which is subsequently substituted by an unchecked item placed in the main memory using an 
uncacheable alias. Furthermore, an untrusted agent can similarly use uncacheable address aliasing to easily measure which lines of the cache are evicted. This results in storage channels that are not visible in information flow analyses performed at the ISA level.

As an example (Figure~\ref{cache:attacks}), 
the guest can use an uncacheable virtual alias of a page table entry in physical memory to bypass the page type constraints and install a potentially harmful page table.
If the cache contains a valid page table entry PTE A for the physical address from a previous check by the hypervisor and this cache entry is clean (i.e., it will not be written back to memory upon eviction),
the guest can (1) store an invalid (i.e. violating the page table policy) page table entry PTE B in a data page and
(2) request the data page to become a page table.
If the guest write is (3) directly applied to the memory, bypassing the cache using
a uncacheable virtual address, and (4) the
hypervisor accesses the same physical location through the cache, then the
hypervisor potentially validates stale data (5).
At a later point in time, the validated data PTE A is evicted from the cache and not written back to memory since it is clean. Then (6) the MMU will use the invalid page table
containing PTE B instead and its settings become untrusted.

\begin{figure}
{\small
  
  \begin{center}
\includegraphics[width=0.6\linewidth]{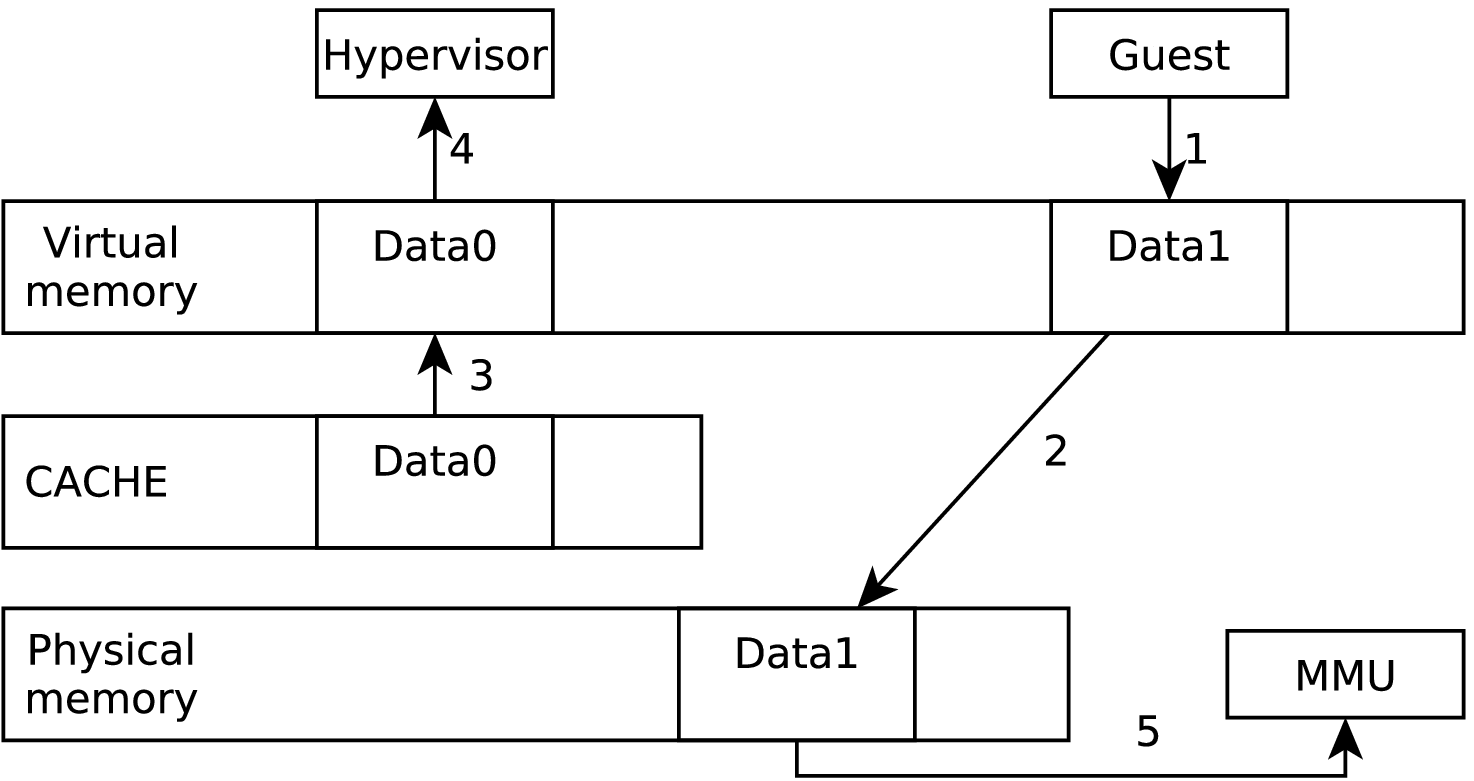}
  \end{center}
  \caption{Integrity threat due to incoherent memory caused by mismatched cacheability attributes. PTE A is validated by the hypervisor (4) but PTE B will be used as a page table entry for the guest (6).}
  \label{cache:attacks}
}
\end{figure}

This kind of behaviour undermines the properties assured by formal verification
that assumes a sequentially consistent model. In this paper, to counter this threat
we use a naive solution; we prevent memory incoherence by cleaning the
complete cache before accessing data stored by the guest. 
Clearly, this can introduce a substantial performance overhead, as shown in
Section~\ref{sec:benchmark}. In~\cite{DBLP:conf/sp/GuancialeNBD16}, we demonstrate more efficient
countermeasures to such threats and propose techniques to fix the verification. 


\section{Formalizing the Proof Goals}\label{sec:goal}


\subsection{The Top Level Specification}~\label{sec:tls}
A state of the Top Level Specification (TLS) is a tuple $\tuple{\ArmState, \HypState}$, 
consisting of an ARMv7 state $\ArmState$  and an abstract hypervisor state $\HypState$.
An abstract hypervisor state  has
the form $\tuple{\pgtype, \pgrefs}$ where $\pgtype$ indicates memory block types and $\pgrefs$ maintains
reference counters. Specifically,  $\pgtype \in 2^{20} \rightarrow \{D,L1,L2\}$ tracks the type of each
4 KB physical block; a block can either be (D) memory writable from the guest or data page, (L1) contain a L1 page table or (L2) contain a L2 page table.
The map $\pgrefs \in 2^{20} \rightarrow 2^{30}$ tracks the references to each physical block, as described in Section~\ref{sec:architecture}.

The TLS interleaves standard unprivileged transitions with abstract handler invocations. Formally, the TLS transition relation $\tuple{\sigma, h} \to_{i \in \{0,1\}} \tuple{\sigma', h'}$ is defined as follows:
\begin{itemize}
  \item If $\sigma \to_{\UserMode} \sigma'$  then $\tuple{\sigma, h} \to_0 \tuple{\sigma', h}$;
    instructions executed in unprivileged mode that do not
    raise exceptions behave equivalently to the standard ARMv7
    semantics and do not affect the abstract hypervisor state.
  \item If $\sigma \to_{\KernelMode} \sigma'$ and $\priv(\sigma) =
    \UserMode$ then $\tuple{\sigma, h} \to_1 H_a(\tuple{\sigma', h})$;
    whenever an exception is raised, the hypervisor is
    executed, modelled by the abstract handler $H_a$. The abstract handler
    always yields a state whose execution mode is unprivileged.
\end{itemize}
In our setup the trusted services and the untrusted guest are both executed in unprivileged mode.
To distinguish between these two partitions, we use ARM domains. 
We reserve the domains 2-15 for the secure services.

\begin{definition}[Secure services]
  Let $\sigma \in \Sigma$, the predicate $\Sec(\sigma)$ identifies if the active partition is the one hosting the secure services: the predicate holds if at least one of the reserved domains (2-15) is enabled in the coprocessor registers $\mathit{\coregs}$ of $\sigma$.
\end{definition}

%
%
Intuitively, guaranteeing spatial isolation means confining the guest
to manage a part of the physical memory available for the guest uses. 
In our setting, this part is determined statically and identified by the set of physical addresses $G_m$. 
Clearly, no security property can be guaranteed if the system
starts from a insecure state; for example the guest must not be allowed to change the MMU behaviour by directly writing the page tables.
For this reason we introduce a system invariant $I(\tuple{\ArmState,
  h})$ that is used to constrain the set of secure initial states
of the TLS. The set of all possible TLS
states that satisfy the invariant is denoted by $\IdealStateSpace_I$. Then one needs to show:
\begin{theorem}[Invariant preserved]\label{lem:invariant}
  Let $\tuple{\sigma,h} \in \IdealStateSpace_I$ and $i\in\{0,1\}$.
  If 
  $\tuple{\sigma,h} \to_i \tuple{\sigma',h'}$
  then
  $I(\tuple{\sigma',h'})$.
\end{theorem}
Section~\ref{sec:tlsproof} elaborates the definition of the invariant and
summarises the proof of the Theorem.

Complete mediation (MMU-integrity) is demonstrated by showing
that neither the guest nor the secure services are able to directly change the content of the page tables and affect the address translation mechanism.
\begin{theorem}[MMU-integrity]\label{lem:mmu-integrity}
  Let $\tuple{\sigma,h} \in \IdealStateSpace_I$.
  If $\tuple{\sigma,h} \to_0 \tuple{\sigma',h'}$
  then
  $\sigma \equiv_{mmu} \sigma'$.
\end{theorem}

We use the approach of~\cite{Heitmeyer:2008:AFM:1340674.1340715} to
analyse the hypervisor data separation properties.
The observations of the guest in a state $\tuple{\sigma,h}$ is represented by the
structure $O_g(\tuple{\sigma,h}) = \tuple{ \sigma.\mathit{uregs},
\mathit{mem}_g(\sigma), \sigma.\coregs}$ of user registers,  
guest memory $\mathit{mem}_g(\sigma)$ (the restriction of $\sigma.\mathit{mem}$ to $G_m$), and coprocessor registers. The latter are visible to the guest since they directly
affect guest behaviour by controlling the address translation, and do not contain any
information the guest should not be allowed to see. 
Evidently, however, all writes to the coprocessor registers
must be mediated by the hypervisor.

The remaining part of the state (i.e. the content
of the memory locations that are not part of the guest memory, banked
registers) and, again, the coprocessor registers constitute the secure observations $O_s(\tuple{\sigma,h})$
of the state, which guest transitions are not supposed to affect.

The following theorem demonstrates that the context switch between the untrusted guest and the trusted services is not possible without the mediation of the hypervisor. The proof is straightforward, since $\Sec$ only depends on coprocessor registers that are not accessible in unprivileged mode.
\begin{theorem}[No context switch]\label{lem:switch}
  Let $\tuple{\sigma,h} \in \IdealStateSpace_I$.
  If $\tuple{\sigma,h} \to_0 \tuple{\sigma',h'}$
  then
  $\Sec(\sigma) = \Sec(\sigma')$.
\end{theorem}

The \emph{no-exfiltration} property guarantees that 
a transition executed by the guest 
does not modify the secure resources:

\begin{theorem}[No-Exfiltration]\label{lem:No-Exfiltration}
  Let $\tuple{\sigma,h} \in \IdealStateSpace_I$.
  
  If $\tuple{\sigma,h} \to_0 \tuple{\sigma',h'}$ and $\neg \Sec(\sigma)$
  then
  $O_s(\tuple{\sigma,h}) = O_s(\tuple{\sigma',h'})$.
\end{theorem}

The \emph{no-infiltration} property  is a
non-interference property
guaranteeing that guest instructions and hypercalls
executed on behalf of the guest do not depend on any information stored in
resources not accessible by the guest.
\begin{theorem}[No-Infiltration]\label{lem:No-Infiltration}
  Let
  $\tuple{\sigma_1,h_1} ,\tuple{\sigma_2,h_2}  \in \IdealStateSpace_I$, $i \in\{0,1\}$,
  and assume  that 
  $O_{g}(\tuple{\sigma_1,h_1}) = O_{g}(\tuple{\sigma_2,h_2})$,
  $\neg \Sec(\sigma_1)$, and  $\neg \Sec(\sigma_2)$.

  If $\tuple{\sigma_1,h_1} \to_i \tuple{\sigma'_1,h'_1}$ and
  $\tuple{\sigma_2,h_2} \to_i \tuple{\sigma'_2,h'_2}$
  then
  $O_{g}(\tuple{\sigma'_1,h'_1}) = O_{g}(\tuple{\sigma'_2,h'_2})$.
\end{theorem}

\subsection{The Implementation Model}\label{sec:implmodel}
%
A critical problem of verifying low level platforms is that 
intermediate states of a handler that reconfigures  the MMU can break the semantics
of the high level language (e.g. C).
For example, a handler can change a page table and (erroneously) unmap the region of virtual memory
where the handler data structure (or the code) are located.
For this reason we introduce
the implementation model, that is sufficiently detailed to expose
misbehaviour of the hypervisor accesses to the observable part of
the memory (i.e. page tables, guest memory and internal data structures).
The implementation interleaves standard unprivileged transitions
and hypervisor functionalities. In contrast to the TLS, these
functionalities now store their internal data in system memory,
accessed by means of virtual addresses. 
In practice, in the implementation model the hypervisor
functionalities are expressed as executable specifications, yet they are very
close to the execution of the actual hardware at instruction semantics level.
We demonstrate these differences by comparing two fragments of the TLS
and the implementation specifications.

The TLS models the update of a guest page table entry
as $\sigma'.mem = write_{32}$ $(\sigma.mem, pa, desc)$, where $pa$ is the
physical address of the entry, $desc$ is a word
representing the new descriptor and $write_{32}$ is a function that
yields a new memory having four consecutive bytes updated.
At the implementation level the same operation is represented as

\begin{alltt}
if \negat \mmu{}(\sigmas, PL1, Gpa2va(pa), wt)
 then \botm
 else \wrt(\sigmas.mem,
         \mmu{}(\sigmas,PL1,Gpa2va(pa), wt),desc)  
\end{alltt}
%
\noindent where $\mathit{mmu}$ is the formal model of the ARMv7 MMU (introduced in Section~\ref{sec:armv7}) and
$\mathit{Gpa2va}$ is the function used by the hypervisor to compute the
virtual address of a physical address that resides in guest memory. This
function is statically defined and is the inverse of the 
injective translation established by the hypervisor master page table.

The implementation can fail to match the TLS for two reasons: (i) the
current page table can prevent the hypervisor from accessing the computed
virtual address, and then the implementation terminates in a failing state (denoted
by $\bot$), (ii) the current address translation does not respect the expected injective
mapping, thus $\mathit{mmu}(\sigma, \KernelMode, Gpa2va(pa), \mathit{wt}) \neq pa$ and
the implementation writes in an address that differs from the one updated by the
TLS.

The next example shows the difference between accesses to the reference counter in the TLS and at implementation level.
The TLS models this operation as $h.\pgrefs(b)$, where $b$ is the
physical block. The implementation models the same operation using memory offsets as follows:

\begin{alltt}
 if \negat \mmu{}(\sigmas, PL1, \mathunderline{tbl}{va} + 4*b, rd)
  then \botm
  else \rd(\sigmas.mem,
         \mmu{}(\sigmas, PL1, \mathunderline{tbl}{va} + 4*b, rd)) 
           & 0xCFFFFFFF
\end{alltt}


This representation is directly reflected in the hypervisor code. For each block, the page type (two
bits)  and the reference counter (30 bits) are placed contiguously in
a word. These words form an array, whose initial virtual address is $tbl_{va}$.

The handlers are represented by a HOL4 function $H_r$ from
ARMv7 states to ARMv7 states. The function is the
executable specification of the various exception handlers including
the MMU mapping/remapping/unmapping functionalities and is composed sequentially of the functional specifications for the corresponding code segments.

\newcommand{\impto}{\rightarrowtail}

Then, the state transition relation $\impto_{i \in \{0,1\}} \subseteq \RealStateSpace  \times
(\RealStateSpace \cup \{\bot\})$  determines the implementation behaviour as follows:
\begin{itemize}
  \item If $\sigma \to_{\UserMode} \sigma'$  then $\sigma \impto_0 \sigma'$;
    instructions executed in unprivileged mode that do not
    raise exceptions behave according to the standard ARMv7
    semantics.
  \item If $\sigma \to_{\KernelMode} \sigma'$ and $\mode(\sigma) =
    \UserMode$ then $\sigma \impto_1 H_r(\sigma')$;
    whenever an exception is raised, the hypervisor is
    executed and its behaviour is modelled by the function $H_r$.
    The function yields either a state whose execution mode is unprivileged or
    $\bot$.
\end{itemize}

%
To show implementation soundness  we exhibit a refinement property
relating abstract states $\tuple{\sigma_1,h}$ to real states $\sigma_2$.
The refinement relation  $\mathcal{R}$ (that is
  left-total and surjective with the exception of the faulty state 
$\bot$)
requires that: (i) the registers and coprocessors contain the same
value in both states, (ii) the guest memory contains the same values in both
states, (iii) the hypervisor data structures of the TLS state are projected into a part of hypervisor memory, and (iv)
the reserved virtual addresses are always mapped in the same way as they are mapped in the
master page table.
Observations of the guest $O^{r}_{g}$ and secure observations $O^{r}_{s}$
are defined on real states using the hypervisor data structure mapping in analogy
with the corresponding observations on abstract states defined above. We require the refinement relation $\mathcal{R}$ to be a bisimulation relation, that is preserved by computations of the abstract and implementation model.

\begin{theorem}[Implementation refinement]\label{lem:refinement}
  Let $\tuple{\ArmState_1, h} \in \IdealStateSpace_I$ and $\ArmState_2
  \in \RealStateSpace$ such that $\tuple{\ArmState_1, h}\ \mathcal{R} \
  \ArmState_2$. Let $i\in\{0,1\}$.
  \begin{itemize}
  \item If $\ArmState_2 \impto_i \ArmState'_2$ then exists $\tuple{\ArmState'_1, h'}$ such that
  $\tuple{\ArmState_1, h} \to_i \tuple{\ArmState'_1, h'}$ and 
    $\tuple{\ArmState'_1, h'}\ \mathcal{R} \ \ArmState'_2$.
  \item If $\tuple{\ArmState_1, h} \to_i \tuple{\ArmState'_1, h'}$ then exists $\ArmState'_2$ such that $\ArmState_2 \impto_i \ArmState'_2$ 
   and 
  $\tuple{\ArmState'_1, h'}\ \mathcal{R} \ \ArmState'_2$.
  \end{itemize}
\end{theorem}

Finally we show that the security property of the TLS and the
refinement relation directly transfer the 
mmu-integrity/no-exfiltration/no-infiltration to the implementation. 
We use $\RealStateSpace_I$ to
represent the space of consistent implementation states: States $\ArmState_2$ such
that if $\tuple{\ArmState_1,h}\  \mathcal{R}\ \ArmState_2$ then
$I(\tuple{\ArmState_1,h})$. 
\begin{corollary}[Implementation security transfer]\label{lem:MLS-sec}
  Let $\ArmState_1, \ArmState_2 \in \RealStateSpace_I$, $i\in \{0,1\}$,
  $O^r_g(\ArmState_1) = O^r_g(\ArmState_2)$:
  \begin{itemize}
    \item if $\ArmState_1 \impto_0 \ArmState'_1$ then
      $\ArmState_1 \equiv_{mmu} \ArmState'_1$
    \item if $\ArmState_1 \impto_0 \ArmState'_1$ and  $\neg \Sec(\sigma_1)$ then
      $O^r_s(\ArmState_1) = O^r_s(\ArmState'_1)$
    \item if $\ArmState_1 \impto_i \ArmState'_1$,
      $\ArmState_2 \impto_i \ArmState'_2$, and 
      $\neg \Sec(\sigma_1)$ and  $\neg \Sec(\sigma_2)$
      then
      $O^r_g(\ArmState'_1) = O^r_g(\ArmState'_2)$
  \end{itemize}
\end{corollary}

\subsection{Binary Code Correctness}\label{sec:ver:binary}
In the ARMv7 model of Section~\ref{sec:armv7} the untrusted guest, the trusted services
and the hypervisor share the CPU, and the hypervisor behaviour is
modelled by the execution of its binary instructions.

Intuitively, internal hypervisor states cannot be observed by
the guest, since (i) during the execution of the handler the guest is
not active, (ii) the hypervisor does not support preemption and
(iii) the handlers do not raise nested exceptions.
For this reason, we introduce a weak transition relation, which hides
states that are privileged.
We write $\ArmState_0 \rightsquigarrow_i \ArmState_n$ if there is a finite execution
$\ArmState_0 \to_i \dots \rightarrow \ArmState_n$ such that
$\mode(\ArmState_n) = \UserMode$ and $\mode(\ArmState_j) = \KernelMode$ for $0 <j < n$.

Our goal is to exhibit a refinement property
relating implementation states $\sigma_1$ and real states $\sigma_2$.
The refinement relation  $\mathcal{R'}$ (that is left-total with the exception of the faulty state 
$\bot$ and surjective)
requires that: (i) the registers and coprocessors contain the same
value in both states, (ii) the guest memory contains the same values in both
states, (iii) the memory holding the hypervisor data structures 
contains the same values in both states. For the observations of the guest $O^r_{g}$ on real states, the same definition as for the implementation model are used, i.e.~the guest can observe the same addresses in both models. Again the refinement is required to establish a bisimulation.

\begin{theorem}[Real Refinement]\label{lem:binary-verification}
  Let $\ArmState_1, \ArmState_2  \in \RealStateSpace_I$ such that
  $\ArmState_1\ \mathcal{R'} \ \ArmState_2$. Let $i\in\{0,1\}$.
\begin{itemize}
   \item If $\ArmState_2 \rightsquigarrow_i \ArmState'_2$ then exists $\ArmState'_1$ such that
  $\ArmState_1 \to_i \ArmState'_1$ and 
    $\ArmState'_1\ \mathcal{R'} \ \ArmState'_2$.
  \item If $\ArmState_1 \impto_i \ArmState'_1$ then exists $\ArmState'_2$ such that $\ArmState_2 \rightsquigarrow_i \ArmState'_2$ and  $\ArmState'_1\ \mathcal{R'} \ \ArmState'_2$.
\end{itemize}
\end{theorem}

Finally one must show that the security properties are transferred to the real model. 
\begin{corollary}[Real Security Transfer]\label{lem:MLS-binary}
  Let $\ArmState_1, \ArmState_2 \in \RealStateSpace_I$, $i\in \{0,1\}$,
  $O^r_g(\ArmState_1) = O^r_g(\ArmState_2)$:
  \begin{itemize}
    \item if $\ArmState_1 \rightsquigarrow_0 \ArmState'_1$ then
      $\ArmState_1 \equiv_{mmu} \ArmState'_1$
    \item if $\ArmState_1 \rightsquigarrow_0 \ArmState'_1$ and  $\neg \Sec(\sigma_1)$ then
      $O^r_s(\ArmState_1) = O^r_s(\ArmState'_1)$
    \item if $\ArmState_1 \rightsquigarrow_i \ArmState'_1$,
      $\ArmState_2 \rightsquigarrow_i \ArmState'_2$, and 
      $\neg \Sec(\sigma_1)$ and  $\neg \Sec(\sigma_2)$
      then
      $O^r_g(\ArmState'_1) = O^r_g(\ArmState'_2)$
  \end{itemize}
\end{corollary}

\subsection{Execution Safety and End-to-End Information Flow Security}
Note that we do not prove explicitly execution safety.
The reason is that the transition relations of the ARM CPU
and the TLS are left-total.
Left-totality for the ARM CPU depends on the fact that
the physical CPU never halts (with the exception of the privileged ``wait''
instruction that is never used by the hypervisor).
Left-totality for the TLS holds because the virtualization API is modelled by
HOL4 total functions over TLS states;
every function is equipped with a termination proof, which is either
automatically inferred by the theorem prover or has been manually verified.
The only transitions that can yield a dead state ($\bot$) are the 
hypervisor transitions of the implementation model, due to incorrect memory
accesses. Proving that this model can never reach the state $\bot$ is part of
the proof of Theorem~\ref{lem:refinement}. It makes use of Lemma~\ref{ref:mem},
which shows that all hypervisor memory accesses are correct.

We do not prove standard end-to-end information flow properties because their
definitions depend on the actual trusted services. This is often the case when
two components are providing services to each other. For example, if the trusted
service is the run-time monitor of Section~\ref{sec:applications}, then it should be able to directly read the
memory of the untrusted Linux (to compute the signatures). Additionally, the
trusted service can be allowed to affect the behaviour of the guest, for
example by rebooting it or by changing its process table if a malware is
detected.

 However, our verification results enable the enforcement of end-to-end security
 by properly restricting the capability of the trusted services. In fact, 
 these services are executed non-privileged, thus their execution is constrained
 by properties~\ref{prop:User-No-Exfiltration} and
 ~\ref{prop:User-No-Infiltration}.
Moreover, their memory mapping is static, is configured in the master page table
of the hypervisor, and is independent of the guest configuration. If complete isolation is needed, it is sufficient to configure these entries of the master page table properly, use properties~\ref{prop:User-No-Exfiltration} and
 ~\ref{prop:User-No-Infiltration} together with Theorem~\ref{lem:mmu-integrity} to prove
 that the trusted services cannot affect and are independent of the guest resources.
This enables the trace properties to be established and consequently to obtain end-to-end security.

\section{TLS Consistency}\label{sec:tlsproof}
We proceed to describe the strategy for proving the TLS consistency properties
of Section~\ref{sec:tls}. To this end we summarise the structure of the system
invariant.
The system invariant consists of two predicates:
 one ($RC$) ensures soundness of the reference counter, and
the other ($TC$) guarantees that the state of the system is well typed.

The reference counter is sound (i.e. $RC\tuple{\ArmState, h}$) if
for every physical block $b$, the reference counter $h.\pgrefs(b)$ is
equal to $\sum_{i \in \{0 \dots  2^{20}-1\}} \mathit{count}(\tuple{\ArmState, h}, i, b)$, where $count$ is a
function that counts the number of references from the block $i$ to
the block $b$, according to the reference-counter policy:
\begin{itemize}
 \item if $b$ is a data block and $i$ is a page table, i.e.~$h.\pgtype(b) = \Data$ and $h.\pgtype(i) \neq \Data$, then $count$ is
   the number of page table descriptors in $i$ that are writable in non-privileged mode and
   that point to $b$,
 \item if $b$ is a L2 page table and $i$ is a L1 page table then $count$ is
   the number of page table descriptors in $i$ that use the table $b$
   to fragment the section, and
 \item if $i$ is not a page table, i.e.~$h.\pgtype(i) = \Data$, then $\mathit{count}(\tuple{\ArmState, h}, i, b)=0$.
\end{itemize}

A system state is well typed ($TC\tuple{\ArmState, h}$), if
the MMU is enabled, the current L1 page table is inside a physical block of type L1,
and each physical block $b$ that does not have type
\textit{data} ($h.pgtype(b) \neq \Data$) contains a sound page table ($\sound(\tuple{\ArmState, h}, b)$) and
resides in the guest memory ($pa\in G_m$ for all $pa$ such that $pa[31:12]=b$).
The predicate $\sound$ ensures that
(i) no unpredictable setting is allowed,
(ii) page tables grant write access only to blocks with type \textit{data},
(iii) page tables forbid any access in $\UserMode$ to
blocks outside the guest memory, and (iv) every L1 page table descriptor points 
to a block typed L2.
Section~\ref{sec:arch:constraints} and Figure~\ref{fig:invariant}
exemplify these constraints.

The proofs of the theorems of Section~\ref{sec:tls} have been obtained using the HOL4 theorem
prover and the lemmas are described in the following.

\begin{lemma}[Invariant implies MMU-safety]\label{lem:mmu-safe}
 If $\tuple{\sigma, h} \in \IdealStateSpace_I$  then $\mathit{mmu}_s(\sigma)$.
\end{lemma}
Lemma \ref{lem:mmu-safe} demonstrates an important property of the system invariant;
a state that satisfies the invariant has the same MMU behaviour as any state whose memory differs
only for addresses that are writable in unprivileged mode.
The proof of the lemma depends on the formal model of the ARMv7 MMU (but not on its instruction set); there the MMU behaviour is determined by coprocessor registers and the contents of the active L1 and referenced L2 page tables. The invariant guarantees that the active L1 page table of $\sigma$ resides in four consecutive blocks that have type L1 and every page table descriptor in this table points to a block typed L2. Moreover, only \textit{data} blocks may be writable in unprivileged mode and write attempts to other blocks will be rejected. We examine a state $\sigma'$ that is write-derivable in unprivileged mode from $\sigma$, but has the same coprocessor registers, selecting the same active L1 page table. Since the content of the page tables is unchanged, the MMU in $\sigma'$ behaves exactly like in $\sigma$.

\paragraph{\normalfont\textbf{Proof sketch of Theorem~\ref{lem:mmu-integrity}}}
To demonstrate MMU-integrity the ARM-integrity property is used. By definition
of the TLS transition relation, if $\tuple{\sigma,h} \to_0 \tuple{\sigma',h'}$ then
(in the ARM model) $\sigma \to_{\UserMode} \sigma'$.
Moreover, Lemma~\ref{lem:mmu-safe} guarantees $\mathit{mmu}_s(\sigma)$. Thus,
Property~\ref{prop:User-No-Exfiltration} can be used to conclude that
$wd(\ArmState, \ArmState', \UserMode)$ and $\ArmState.\coregs = \ArmState'.\coregs$, i.e. $\ArmState'$ is a state write-derivable from $\sigma$ and coprocessor registers have not changed.
Finally, it suffices to apply Definition~\ref{def:mmu-safe} (MMU-safety) to show
that $\sigma \equiv_{mmu} \sigma'$.

\begin{lemma}[Guest isolation]\label{lem:user-mem}
  Let $\tuple{\sigma, h} \in \IdealStateSpace_I$.
  For every physical address $pa$ and access request $req$
  if 
  $\neg \Sec(\sigma)$ and
  $\mathit{mmu}_{ph}(\sigma, \UserMode, pa, req)$
  then $G_m(pa)$.
\end{lemma}
The proof of Lemma~\ref{lem:user-mem} uses the formal ARMv7 MMU model and directly follows from the invariant. In particular, part (iii) of predicate $\sound$ demands that entries of a page table grant access permissions to the guest only if the entry points to a physical address that is inside the guest memory.

\paragraph{\normalfont\textbf{Proof sketch of Theorem~\ref{lem:No-Exfiltration}}}
Similar to the proof of Theorem~\ref{lem:mmu-integrity},
the definition of the transition relation and Lemma~\ref{lem:mmu-safe} yield that $\tuple{\sigma,h} \to_0 \tuple{\sigma',h'}$ implies $h=h'$, $\sigma \to_{\UserMode} \sigma'$ and $\mathit{mmu}_s(\sigma)$.
Then Property~\ref{prop:User-No-Exfiltration} gives $\ArmState.\coregs= \ArmState'.\coregs$ and $wd(\ArmState, \ArmState', \UserMode)$,
meaning that (according to the contraposition of Definition~\ref{def:mmu-der}) the memories of $\ArmState$  and $\ArmState'$ contain
the same value for every physical address that is not writable in mode $\UserMode$ in $\ArmState$.
By Lemma~\ref{lem:user-mem} the guest can only obtain write permissions to the physical addresses belonging to its own memory, thus
the memories of $\ArmState$  and $\ArmState'$ have the same value for every physical address not in $G_m$. Moreover banked registers cannot be changed in unprivileged mode. Consequently, $O_s(\tuple{\sigma,h}) = O_s(\tuple{\sigma',h'})$ holds as claimed.

\paragraph{\normalfont\textbf{Proof sketch of Theorem~\ref{lem:No-Infiltration}}}
We proceed separately for unprivileged and privileged
transitions. For unprivileged transition the ARM-confidentiality property is used.
As proven above, from the definition of the transition relation,
Lemma~\ref{lem:mmu-safe}, $\tuple{\sigma_1,h_1} \to_0 \tuple{\sigma'_1,h'_1}$, and 
$\tuple{\sigma_2,h_2} \to_0 \tuple{\sigma'_2,h'_2}$ we obtain
$h_1=h'_1$, $h_2=h'_2$, $\sigma_1 \to_{\UserMode} \sigma'_1$,
$\sigma_2 \to_{\UserMode} \sigma'_2$, $\mathit{mmu}_s(\sigma_1)$ and
$\mathit{mmu}_s(\sigma_2)$. 
Since $O_{g}(\tuple{\sigma_1,h_1}) = O_{g}(\tuple{\sigma_2,h_2})$,
the user registers, guest memories (i.e. the content for addresses in $G_m$),
and coprocessor registers are the same in $\sigma_1$ and $\sigma_2$.
The definition of $\mathit{mmu}_s(\sigma_1)$ yields $\sigma_1\equiv_{\mathit{mmu}}\sigma_2$.
Moreover, Lemma~\ref{lem:user-mem} shows that the guest can obtain an access permission
only to the physical addresses in $G_m$, thus
the memories of $\sigma_1$ and $\sigma_2$ contain the same value for every address in
$G_m \supseteq \{ pa \mid \exists \req.\ \mathit{mmu}_{ph}(\sigma_1, \UserMode,  pa, \req) \}$.
This enables Property~\ref{prop:User-No-Infiltration}, which in turn justifies that
$\sigma'_1.\uregs = \sigma'_2.\uregs\ \ \text{and}\ \ \forall pa \in G_m.\ \sigma'_1.mem(pa) = \sigma'_2.mem(pa)$, 
i.e. the guest observations in $\sigma'_1$ and $\sigma'_2$ are the same.

The proof of the Theorem~\ref{lem:No-Infiltration} for hypervisor
transitions has been obtained by performing relational analysis. The function $H_a$
accesses only three state components: the hypervisor data structures (i.e. $h$),
the user registers and the memory (in order to validate page tables).
The function $H_a$ is symbolically evaluated on the states $\tuple{\sigma_1, h_1}$ and $\tuple{\sigma_2, h_2}$;
whenever $H_a$ updates an intermediate variable, it must be demonstrated that the value of the variable is the
same in both executions,
whenever  $H_a$ modifies a state component (e.g. memory or register), it must be demonstrated that the
equivalence of guest observation is preserved.
These tasks are completely automatic for assignments that only depend on intermediate variables and user registers.
For every assignment that depends on memory accesses, a new verification condition is generated to require
that the accessed addresses are the same in both executions and to guarantee that
such address is in the guest memory. Finally, these verification conditions are verified,
showing that $H_a$ never accesses memory outside $G_m$.

Finally, we prove Theorem~\ref{lem:invariant} by showing that the
invariant is preserved first by guest transitions (Lemma \ref{lem:inv-usr})
and then by the abstract handlers (Lemma \ref{lem:inv-hyper}). 

\begin{lemma}[Invariant vs guest]\label{lem:inv-usr}
  Let $\tuple{\sigma, h} \in \IdealStateSpace_I$.
  If $\tuple{\sigma, h} \to_0 \tuple{\sigma', h'}$ then 
  $I(\tuple{\sigma', h'})$.
\end{lemma}
This lemma demonstrates that the invariant is
preserved by guest transitions. Its proof depends on the ARM-integrity property.
It is straightforward to show that the invariant only
depends on the content of the physical blocks that are not typed
$\Data$ and the hypervisor data (i.e. $h$ and $h'$).
Similar to the proof of Theorem~\ref{lem:No-Exfiltration}, the definition
of the transition relation, Lemma~\ref{lem:mmu-safe}
and Property~\ref{prop:User-No-Exfiltration}
guarantee that if $\tuple{\sigma,h} \to_0 \tuple{\sigma',h'}$ then $h=h'$, $\sigma \to_{\UserMode} \sigma'$, $\mathit{mmu}_s(\sigma)$ and $wd(\ArmState, \ArmState', \UserMode)$.
As in the proof of Lemma~\ref{lem:mmu-safe} it is shown
that in $\sigma$ every block that is not typed $\Data$ is not writable,
concluding that the invariant is preserved.

\begin{lemma}[Invariant vs hypervisor]\label{lem:inv-hyper}
  Let $\tuple{\sigma, h} \in \IdealStateSpace_I$.
  If $\tuple{\sigma, h} \to_1 \tuple{\sigma', h'}$ then 
  $I(\tuple{\sigma', h'})$.
\end{lemma}
The lemma demonstrates that the invariant is preserved by the handler
functionalities and shows the functional correctness of
the TLS design. By definition, if $\tuple{\sigma, h} \to_1
\tuple{\sigma', h'}$ then there exists $\sigma''$ such that
$\sigma \to \sigma''$, $\mode(\sigma'') = \KernelMode$ and $\tuple{\sigma',
h'} = H_a(\tuple{\sigma'', h})$.
 Similar to
the proof of Lemma~\ref{lem:inv-usr},
Property~\ref{prop:User-No-Exfiltration} is used to guarantee that
the invariant is preserved by this transition:
$I(\tuple{\sigma'', h})$.
Then we show that the invariant is preserved by the abstract handler
$H_a$.

This verification task requires the introduction of several supporting
lemmas. The idea is that according to the input request, the abstract handler only changes a small
part of the system state. For instance, when $H_a$ maps a section, only the current L1 page table is modified, the contents of other blocks are unchanged.
In order to demonstrate that the invariant is indeed preserved for
the parts of the state that are not affected by $H_a$, we introduce a number of additional lemmas.
These lemmas are sufficiently general to be used to verify
different virtualization mechanisms that involve direct paging and
they prove the intuition that the type of a
block can be safely changed when its reference counter is zero. 
\begin{definition}
Let $h$ and $h'$ be two abstract hypervisor states.
The predicate $type_s(h,h')$ holds if and only if $h.\type(b) \neq h'.\type(b)$ implies $h.\refs(b) = 0$ for all blocks $b$.
\end{definition}

Changing the type of a block can affect the soundness of page tables that reference that block. The following lemma expresses the key property that soundness of page tables is preserved for all type changes of other blocks, as long as the reference counters of that blocks are zero:
\begin{lemma}\label{lem:inv-sound-type}
  Assume $I\tuple{\ArmState,h}$ and $type_s(h,h')$. For every block $b$ such that
  $h.\type(b) = h'.\type(b)$,  if $\sound(\tuple{\ArmState, h}, b) $ then $\sound(\tuple{\ArmState,h'}, b)$.
\end{lemma}
The proof of Lemma \ref{lem:inv-sound-type} hinges on the fact that type changes can only break parts (ii) and (iv) of the page table soundness condition. However, if the type is only changed for blocks that are not referenced by any page table, soundness is preserved trivially.

We exemplify the usage of Lemma~\ref{lem:inv-sound-type} when proving
Lemma~\ref{lem:inv-hyper}. Assume that $H_a$ is allocating a new L2 page table in the block $b'$
(i.e. changing the type of $b'$ from $\Data$ to $L2$). This operation can break soundness of any other block
$b$. In fact, $b$ can be a page table containing a writable mapping to $b'$, thus $b$ is sound in
$\tuple{\ArmState, h}$ but is unsound in $\tuple{\ArmState, h'}$. The side condition $type_s(h,h')$ ensures
that this case cannot occur: to safely allocate a new page table, the reference counter of $b'$ must be zero, thus
$b$ cannot contain a writable mapping to $b'$. 

Similarly, the following lemma shows that, if the page type is changed
only for blocks with zero references, then for all other page tables, the number of references is unchanged.
\begin{lemma}\label{Lemma:ref_cnt_counters}
  Assume $I\tuple{\ArmState,h}$ and $type_s(h,h')$. 
  For all blocks $b,b'$ if
  $h.\type(b) = h'.\type(b)$ then
   $\countt(\tuple{\ArmState, h}, b, b') =   \countt(\tuple{\ArmState, h'}, b, b') $.
\end{lemma}

Finally we use the following lemma to show that the well-typedness of a
block and its counted outgoing references are
independent from the content of the other physical blocks.
\begin{lemma}\label{Lemma:unchanged_counters}
  Let $\sigma, \sigma' \in \Sigma$
  such that $I\tuple{\ArmState,h}$. 
  If $\ArmState$ and $\ArmState'$  have the
  same memory content for the block $b$ then
  $\sound(\tuple{\ArmState', h}, b)$ and
  for every block $b'$
  $\countt(\tuple{\ArmState', h}, b, b') =
   \countt(\tuple{\ArmState', h}, b, b')$.
\end{lemma}
For every functionality of the virtualization API (see Figure~\ref{APIfig}), Lemmas~\ref{lem:inv-sound-type}, ~\ref{Lemma:ref_cnt_counters} and~\ref{Lemma:unchanged_counters} help to limit the proof of Lemma~\ref{lem:inv-hyper} to only checking the well-typedness and soundness of the reference counter for the blocks that are affected by $H_a$.

\paragraph{\normalfont\textbf{Proof of Theorem~\ref{lem:invariant}}} The theorem directly
follows from Lemmas~\ref{lem:inv-usr} and~\ref{lem:inv-hyper}.


\section{Refinement}\label{sec:refinement}
To verify the implementation refinement relation (i.e. prove Theorem~\ref{lem:refinement}) we proved two auxiliary lemmas:
\begin{lemma}[Real MMU]\label{ref:mmu}
  Let $\tuple{\ArmState_1, h}\in\IdealStateSpace_I$ and $\ArmState_2
  \in \RealStateSpace$. If $\tuple{\ArmState_1, h}\ \mathcal{R} \
  \ArmState_2$ then
  $\ArmState_1 \equiv_{mmu} \ArmState_2$.
\end{lemma}
The Lemma shows that TLS and implementation states have the same MMU
configuration. Its proof uses that the system invariant requires page
tables to be allocated inside the guest memory, whose content is
the same in the TLS and implementation states. Moreover, coprocessor registers contain the same data.

\begin{lemma}[Hypervisor page tables]\label{ref:mem}
  Let $\tuple{\ArmState_1, h} \in \IdealStateSpace_I$ and $\ArmState_2
  \in \RealStateSpace$. If $\tuple{\ArmState_1, h}\ \mathcal{R} \
  \ArmState_2$ then:
  \begin{enumerate}
    \item For all $pa$ and $\mathit{req}$, if $pa \in G_m$ then $mmu(\ArmState_2,
      \mathit{Gpa2va}(pa), \KernelMode, \mathit{req}) =  pa$.
    \item For every block $b$ and access request $\mathit{req}$,
      $mmu(\ArmState_2, tbl_{va} + 4*b, \KernelMode, \mathit{req}) = tbl_{pa} + 4*b$, where $tbl_{pa}$ is the physical address where the hypervisor data structure is allocated.
  \end{enumerate}
\end{lemma}
The Lemma shows that the implementation is always able to access the guest memory
and the hypervisor data structures, and that the computed physical
addresses match the expected values.

\paragraph{\normalfont\textbf{Proof sketch of Theorem~\ref{lem:refinement}}}
To prove that the refinement is preserved by all possible transitions
 we verify independently the guest and
hypervisor transitions.
For guest transitions, Theorem~\ref{lem:No-Exfiltration} (No-Exfiltration) and
Lemma~\ref{lem:mmu-safe} (MMU-safety) guarantee
that the guest can change neither the memory outside $G_m$ nor
the page tables. Thus it is sufficient to show that the physical addresses
of the hypervisor data structures are outside the guest memory. 
Moreover, Theorem~\ref{lem:No-Infiltration} (No-Infiltration) guarantees that the
guest transition is not affected by a part of the state that is not
equivalent in $\tuple{\ArmState_1, h}$ and $\ArmState_2$.
For the hypervisor transition we used a compositional approach. First,
we verified that all low-level operations (i.e. reads and
updates of the page tables, reads and updates of the hypervisor data
structures) preserve the refinement relation. Then we compose these
results to show that the TLS and implementation transitions behave equivalently.


\paragraph{\normalfont\textbf{Proof sketch of Corollary~\ref{lem:MLS-sec}}}
The proof depends on the fact the the relation $\mathcal{R}$ is left-total and surjective.
Proving that the security properties of the TLS are transferred to the
implementation model is simplified by the definition of the
refinement relation.
For example, Lemma~\ref{ref:mmu} and Theorem~\ref{lem:mmu-integrity} are used to show that the MMU configuration cannot be changed by the untrusted guest.
Assume $\ArmState_2 \impto_{0} \ArmState'_2$ and let 
$\tuple{\ArmState_1, h}$ be a TLS state such that
$\tuple{\ArmState_1, h} \to_{0} \tuple{\ArmState'_1, h'}$ and
$\tuple{\ArmState_1, h} \mathcal{R} \ \ArmState_2$.
Since the refinement is preserved by all transitions (Theorem~\ref{lem:refinement}), exists $\tuple{\ArmState'_1, h'}$ such that
$\tuple{\ArmState'_1, h'} \mathcal{R} \ \ArmState'_2$.
Lemma~\ref{ref:mmu}  yields $\ArmState_1 \equiv_{\mathit{mmu}} \ArmState_2$
and Theorem~\ref{lem:mmu-integrity} (MMU-integrity) guarantees that 
$\ArmState_1 \equiv_{\mathit{mmu}} \ArmState'_1$,
thus $\ArmState_2 \equiv_{\mathit{mmu}} \ArmState'_1$.
Finally, Lemma~\ref{ref:mmu} yields $\ArmState'_1 \equiv_{\mathit{mmu}} \ArmState'_2$, thus $\ArmState_2 \equiv_{\mathit{mmu}} \ArmState'_2$.
Similar reasoning is used to prove that properties \emph{no-exfiltration} and
\emph{no-infiltration} are transferred to the implementation model,
by showing that, if two TLS states have the same observations
(i.e. $O_g(\tuple{\ArmState_1, h}) = O_g(\tuple{\ArmState'_1, h'})$ or $O_s(\tuple{\ArmState_1, h}) = O_s(\tuple{\ArmState'_1, h'})$) and
the states are refined by two implementation states (i.e. $\tuple{\ArmState_1, h} \mathcal{R} \ \ArmState_2$ and $\tuple{\ArmState'_1, h'} \mathcal{R} \ \ArmState'_2$),
then the two implementation states have the same observations (i.e. $O^r_g(\ArmState_2) = O^r_g(\ArmState'_2)$ or $O^r_s(\ArmState_2) = O^r_s(\ArmState'_2)$).

\section{Binary Verification}\label{sec:binary}

Binary analysis is key requirement to ensure security of low-level software platform, like hypervisors.
Machine code verification obviates the necessity of trusting the compilers. Moreover, low level programs mix
structured code (e.g. implemented in C) with assembly and use instructions (e.g.
mode switches and coprocessor interactions) that are not part of
the high level language, thus making difficult to use verification tools that target user level code.

For our hypervisor the main goal of the verification of the binary code
is to prove Theorem~\ref{lem:binary-verification}.
This verification relies on Hoare logic
and requires several steps.
The first step (Section~\ref{sec:binary:contracts})
is transforming the relational  reasoning into a set of contracts for the
hypervisor handlers and guaranteeing that the refinement
is established if all contracts are satisfied.
Let $C$ be the binary code of one of the handlers,
the contract $\{P\}C\{Q\}$ states that
if the precondition $P$ holds in the starting state of $C$, then the
postcondition $Q$ is guaranteed by $C$. 

Then, we adopt a standard semi-automatic strategy
to verify the contracts. First, the weakest liberal precondition $\mathit{WLP}(C,Q)$ is computed on the
starting state, then it is verified that
 the precondition $P$ implies
the weakest precondition. 

The second verification step (computation of weakest preconditions)
 can be performed directly in HOL4
using the ARMv7 model. However, this task requires a significant
engineering effort. We adopted a more practical
approach, by using the Binary Analysis Platform (BAP)~\cite{Brumley:2011:BBA:2032305.2032342}.
The BAP tool-set provides
platform-independent utilities to extract control flow graphs
and program dependence graphs, to perform symbolic execution
and to perform weakest-precondition calculations. These utilities reason
on the BAP Intermediate Language (BIL), a small and formally specified
language that models instruction evaluation as compositions of
variable reads and writes in a functional style.

The existing BAP front-end to translate ARM programs to BIL
lacks several features required to handle our binary code: Support of
ARMv7, management of processor status registers,
banked registers for privileged modes and coprocessor registers. 
For this reason we developed a new front-end, which is presented in
Section~\ref{sec:binary:lifter}, that converts an ARMv7 assembly
fragment to a BIL program.

The final verification step consists of checking that the precondition $P$ implies
the weakest precondition. This task can be fully automated
if the predicate $P \Rightarrow \mathit{WLP}(C, Q)$ is equivalent to a
predicate of the form $\forall \vec x.A$ where $A$ is quantifier free.
The validity of $A$ can then be checked using a Satisfiability Modulo
Theory (SMT) solver that supports bitvectors to handle operations on words.
In this work, we used STP~\cite{DBLP:conf/cav/GaneshD07}.

An alternative approach for binary verification is to use the ``decompilation''
  procedure developed by Myreen~\cite{DBLP:conf/fmcad/MyreenGS08}. This procedure takes an
  ARMv7 binary and produces a HOL4 function that behaves equivalently (i.e. implements 
the same state transformation).
This result allows to lift the verification
of properties of assembly programs to reasoning on HOL4
functions. However, the latter task can be expensive due to the lack of
automation in HOL4.

\subsection{Soundness of the Verification Approach}

  The procedure described here to establish the functional correctness of the hypervisor code relies on four main arguments.
  \begin{enumerate}
  \item The HOL4 procedures that evaluate the effects of a given instruction in
    the ARMv7 model specify the updates to the processor state correctly.
    We use the ARMv7 step theorems to guarantee the correctness of this task.
  \item The lifter transforms this state update information into an equivalent
    list of single-variable assignments in BIL. The correctness of this part of
    the lifter is an assumption for now.
  \item The expressions in each update of a processor component are correctly
    translated to BIL expressions in the list of assignments, preserving their
    semantics. This has been proven for our lifter.
  \item The binary code fragment that is lifted is actually executed on the ARMv7 hardware.
  \end{enumerate}
  The last argument relies on the fact that the boot loader places the
  unmodified hypervisor image to the right place in memory. This is another
  assumption since we do not verify our boot loader. Furthermore there must not
  be self-modifying code. The easiest way to enforce this is to partition the
  hypervisor memory via its page table into data and code region and prove an
  invariant that the first is non-executable but the latter is non-writeable.
  There is no such protection against self-modifying code in the hypervisor at
  the moment. Finally, one needs to show that the binary code is not
  interrupted, thus proving that the hypervisor is in deed non-preemptive. We do
  not have a full proof of the statement, but there are provisions in the lifter
  to show the absence of ARMv7 interrupts and exceptions. 

  For system call and unknown instructions, the lifter generates BIL statements
  that always fail, i.e., one can only verify programs in BAP that do not use
  such instructions. We follow the same approach for fetches, jumps, and memory
  instructions accessing constant addresses which are not mapped in the
  hypervisor's static page table. Thus such operations cannot produce pre-fetch
  or data aborts. Additional care has to be taken to distinguish data and code
  regions to avoid permission faults due to writes to the code region or fetches
  from the data region, however there are no such checks at the moment. Indirect
  jumps are solved dynamically based on the lifted BIL program (see
  Sect.~\ref{sec:tools}) and for any jump to a location that is not defined in
  the program, i.e., not in the region accessible by the hypervisor, analysis
  with BAP will give an error. For dynamic memory accesses the lifter is able to
  insert assertions that the corresponding address is mapped, however the
  feature is currently not activated automatically. At last, the reception
  of external interrupts should not be re-enabled during hypervisor execution. Currently this invariant is not
  checked in the code verification but it could be easily added as an assertion
  between every instruction.

\subsection{The Contracts Generation}\label{sec:binary:contracts}
Let $C$ be the binary code of one of the handlers,
we define the precondition $P$ and the postcondition $Q$ such that
the contract subsumes the refinement:
\begin{itemize}
  \item $P(\ArmState_2) = $ exists $\ArmState_1$ such that
    $\ArmState_1 \mathcal{R'} \ArmState_2$ 
  \item $Q(\ArmState'_2, \ArmState_2) = $ for all $\ArmState_1,
    \ArmState'_1$ if $\ArmState_1 \mathcal{R'} \ArmState_2$ and 
    $\ArmState_1 \impto_1 \ArmState'_1$ then 
    $\ArmState'_1 \mathcal{R'} \ArmState'_2$
\end{itemize}
These contracts are not directly suitable 
for the verification of the binary code because the contracts 
quantify on states ($\ArmState_1$ and $\ArmState'_1$) that
are in relation with the pre-state ($\ArmState_2$) and post-state ($\ArmState'_2$)
of the binary code.
We developed an HOL4 inference procedure specific for the structure of our hypervisor.
The output of the procedure is a
proof guaranteeing that the original contract $\{P\}C\{Q\}$
is satisfied if a ``simplified'' contract $\{P'\}C\{Q'\}$ is met.
That is, 
for every $\ArmState_2, \ArmState'_2$ the predicate
$P'(\ArmState_2) \Rightarrow Q'(\ArmState'_2, \ArmState_2)$
implies 
$P(\ArmState_2) \Rightarrow Q(\ArmState'_2, \ArmState_2)$.

This procedure makes heavy use of the simplification rules
and decision procedures of HOL4.
We informally summarise how this procedure works for the memory resource.
The precondition $P'$ is generated by transferring the hypervisor invariant $I$
from the abstract model down to the real model.
This is possible because (i)  $\mathcal{R'}$ constrains the memory holding the
hypervisor data structures to be the same in $\ArmState_2$ and
$\ArmState_1$, (ii)  $\mathcal{R}$ (the refinement between the abstract model and the implementation model)
guarantees that this memory area contains a projection of the hypervisor data structures in the TLS
state, (iii) on the TLS state the hypervisor invariant holds.

For the postcondition $Q'$ we proceed as follows.
If $\ArmState_1 \to_1 \ArmState'_1$ then $\ArmState'_1 =
H_r(\ArmState_1)$. Let $A$ be the set of memory addresses that are
constrained by the refinement relation $\mathcal{R}'$ and let $B$ be the set
of addresses that are modified by $H_r$.
The set $B$ is usually easy to identify in HOL4, thanks to its
symbolic execution capability and the lemmas that have been already proven for
the tasks of Section~\ref{sec:refinement}.
For each handler we demonstrate that $B \subseteq A$.

For every address $a \in (\bar B) \cap A$ (namely addresses
constrained by the refinement relation that are not updated) 
we add to $Q'$ the constraint $\ArmState'_2.\mathit{mem}(a) =
\ArmState_2.\mathit{mem}(a)$. This uses $\ArmState_1.\mathit{mem}(a) = \ArmState_2.\mathit{mem}(a)$ and the refinement for the address
$a$ ($\ArmState'_1.\mathit{mem}(a) = \ArmState'_2.\mathit{mem}(a)$).

For every address $a \in B$  we make use of the HOL4 rewriting engine to
obtain a naive symbolic execution of the handler specification.
We use HOL4 to symbolically compute $H_r(\ArmState_1)$ then we use the
precondition $\ArmState_1 \mathcal{R'} \ArmState_2$ to rewrite the
result and make sure that this is expressed only in terms of
$\ArmState_2$. Let $\mathit{exp}$ be the resulting expression, we add to $Q'$
the constraint $\ArmState'_2.\mathit{mem}(a) = \mathit{exp}.\mathit{mem}(a)$.

When the symbolic execution is too complex (e.g. too many outcomes are
possible according to the initial state), we split the verification by
generating multiple contracts $\{P_1\}C\{Q_1\}, \dots,
\{P_n\}C\{Q_n\}$, where $P_i = P(\ArmState_2) \wedge A_i(\ArmState_2)$
and $\bigvee_i A_i$ is a valid  formula (i.e. all possible cases are
taken into account).

\subsection{Translation of ARMv7 to BIL}\label{sec:binary:lifter}
The target language of the ARMv7 BAP front-end is BIL, a simple single-variable
assignment language tailored to model the behaviour of assembly
programs and developed to be platform independent.
A BIL program is a sequence of statements. Each statement
 can affect the system state by assigning
the evaluation of an expression to a variable, (conditionally or unconditionally) modifying the control flow, 
terminating the system in a failure state if an assertion does
not hold and unconditionally halting the system in a
successful state. 
The BIL data types for expressions and variables include boolean,
words of several sizes and memories. The
main constituent of BIL statements are expressions, that include
constants, standard bit-vector binary and unary operators,
and type casting function. Additionally, an expression 
can read several words from a memory or generate a new
memory by changing a word in a given one.

We developed the new
front-end on top of the HOL4 ARM model (see Section~\ref{sec:armv7}),
so that the soundness of the transformation from an ARM assembly
instruction to its corresponding BIL program relies on the correctness
of the ARM model used in HOL4 and not on a possibly different formalization of
ARMv7. Our approach is illustrated in Figure~\ref{fig:lifter}.

\begin{figure}
  \centering
  \includegraphics[scale=0.45]{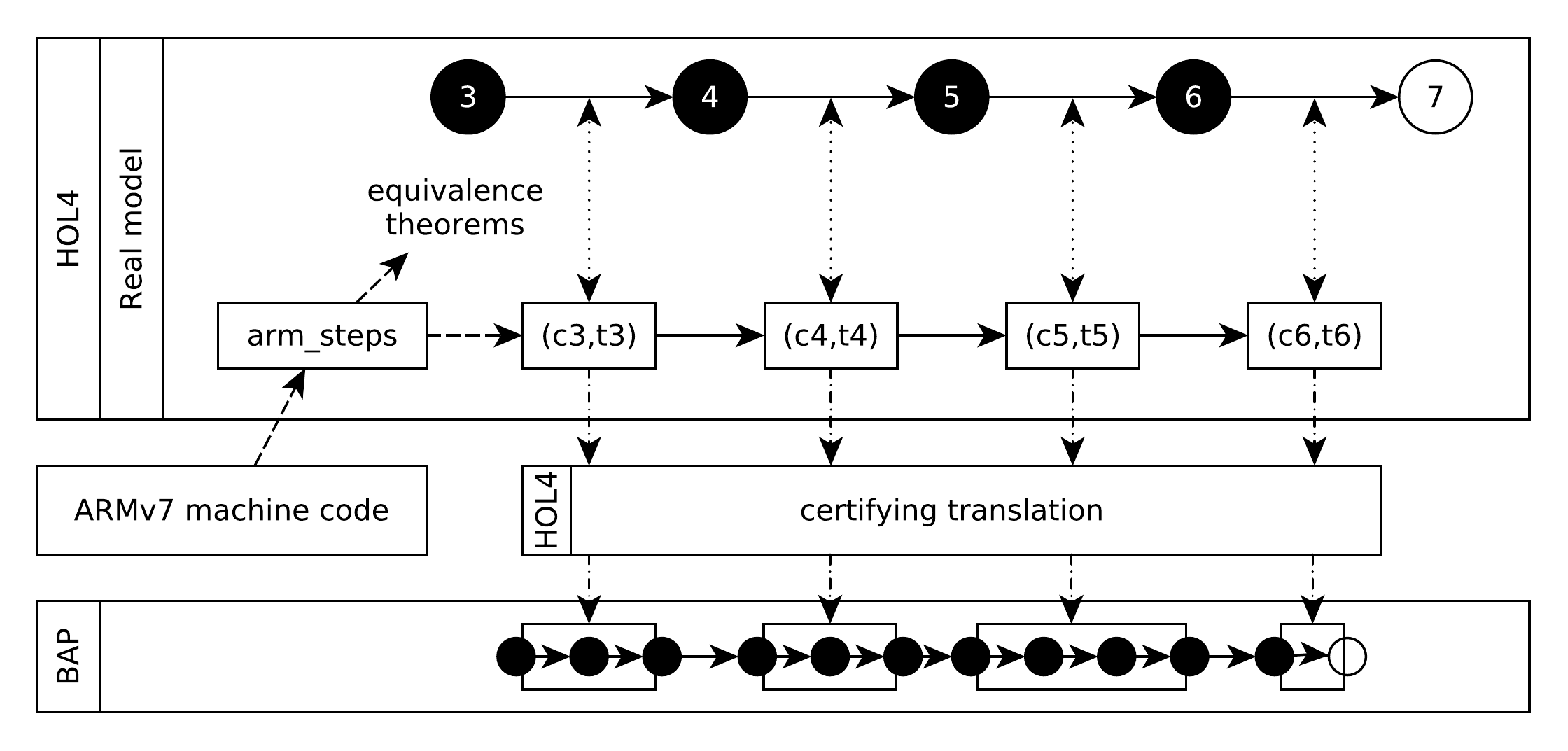}
  \caption{Lifting machine code to BIL using the HOL4 ARMv7 model. The $\mathit{arm\_steps}$ function translates machine instructions into steps $st_i$ consisting of guards $c_i$ and transition functions $t_i$. Their effect is equivalent to the hypervisor computation in the real model (states 3-7, cf. Fig.~\ref{fig:trace}). Each step $st_i$ is in turn translated into equivalent BIL code.}
  \label{fig:lifter}
\end{figure}

The HOL4 model provides a
mechanism to  statically compute the effects of an instruction via the
$\mathit{arm\_steps}$ function.
Let $\mathit{inst}$ be 
the encoding of an instruction, then
$\mathit{arm\_steps}(\mathit{inst})$ returns the possible execution steps
$\{\mathit{st}_1,\dots,\mathit{st}_n\}$.
Each step $\mathit{st}_i=(c_i,t_i)$ consists of the condition $c_i$ that enables the
transition and the function $t_i$ that transforms the starting state
into the next state.
The function $\mathit{arm\_steps}$ is a certifying HOL4 procedure, since its output is a theorem demonstrating that for every $\ArmState \in \Sigma$ if 
$\mathit{fetch}(\ArmState)=\mathit{inst}$ and $c_i(\ArmState)$ holds then
$\ArmState \to_{\KernelMode} t_i(\ArmState)$.
For standard ARM decoding the function $\mathit{fetch}$
reads four bytes from memory starting from 
the address pointed to by the program counter.

The translation procedure involves the following steps,
(i) mapping HOL4 ARM states to BIL states and
(ii) for each instruction of the given assembly fragment 
 producing the BIL fragment that emulates the $\mathit{arm\_steps}$ outputs.
To map an ARM state to the corresponding BIL state
we use a straightforward approach. A BIL variable
is used to represent a single component of the machine
state: for example, the variable $\mathit{R0}$ represents the register number zero
and the variable $\mathit{MEM}$ represents the system memory.

To transform an ARM instruction to the corresponding BIL fragment
we need to capture all the possible effects of its execution
 in terms of affected registers, flags and memory locations.
The generated BIL fragment should simulate the
behaviour of the instruction executed on an ARM machine.
Therefore,  to obtain a BIL fragment for an instruction we need to translate the
predicates $c_i$ and their corresponding transformation functions
$t_i$.
This task is accomplished using symbolic evaluation of the predicates and
the transformation functions. The input of the evaluation is a symbolic
state in which independent variables are used to represent each state
register, flag, coprocessor register and memory.
This approach allows us to obtain a one-to-one mapping between the
symbolic state variables and the BIL state variables.
To transform a predicate $c_i$, we apply the predicate to a symbolic
ARMv7 state, thus
obtaining a symbolic boolean expression in which free-variables are a subset of the symbolic state variables. Similarly, to map a transformation function $t_i$, we apply $t_i$ to
a symbolic state, thus obtaining
a new state in  which each  register, flag and affected memory
location is assigned a symbolic expression. 
Intuitively, for each instruction we produce the following BIL fragment:
\begin{lstlisting}[mathescape]
label GUARD_1
cjmp $|b_1|$, EFFECT_1, GUARD_2
...
label GUARD_N
cjmp $|b_n|$, EFFECT_n, ERROR
label ERROR
assert false
\end{lstlisting}
Where ``cjmp''  is the BIL instruction for conditional jump and
$|b_i|$ is a BIL boolean expression obtained by translating the
symbolic evaluation of $c_i$.
Then, for each step $i$ we symbolically evaluate the transformation
$t_i$ and for each field (i.e. memory locations, registers, flags and coprocessor
registers) that has been updated
we transform the resulting symbolic expression and assign it to
the corresponding BIL variable, generating a fragment
\begin{lstlisting}[mathescape]
label EFFECT_i
var_1 = $|\mathit{exp}_1|$
...
var_n = $|\mathit{exp}_n|$
\end{lstlisting}

The described lifting procedure is straightforward. However, its
soundness depends on the correct transformation of HOL4 terms
(e.g. $|b_n|$ and  $|\mathit{exp}_n|$) to BIL expressions. Since 
the number of HOL4 operators that occur in the generated expressions is
huge, we cannot rely on a simple syntactical transformation to obtain
a robust conversion of them to BIL.
Moreover, the transformation of HOL4 terms to BIL expressions is
used to convert the pre/post conditions of our contracts from HOL4 to BAP. 
For this reason, we formally modelled in HOL4 the BIL expression
language (by providing a deep embedding of BIL expression in HOL4) and 
the translation procedure $\mathit{liftExp}$ certifies its output:
\[
\mathit{liftExp}(\mathit{exp}) = (\mathit{exp'}, \vdash \mathit{exp}=\mathit{exp}')
\]
In particular, the translation procedure procedure yields a theorem demonstrating that
the HOL4 input term $\mathit{exp}$ is equivalent to the BIL expression  $\mathit{exp}'$.

In order to dynamically generate the certifying theorem, the 
translation procedure is implemented in ML, which is the 
HOL4 meta language. 
The translation syntactically analyses and deconstructs the input
expressions to select the theorems to use in the HOL4 conversion
and rewrite rules. For terms composed by nested expressions the
procedure acts recursively.

\subsection{Supporting Tools}\label{sec:tools}
To compute the weakest precondition of a program is necessarily to statically know the
control flow graph (CFG) of the program. This means that the algorithm depends 
on the absence of indirect jumps. Even if the hypervisor C-code avoids
their explicit usage (e.g. by not using function pointers), the
compiler introduces an indirect jump for each function exit point (e.g. the instruction at the address \verb|0x20C| in Figure~\ref{fig:IndirectJumpExample},  is an indirect jump).
Solving an indirect jump (i.e. enumerating all possible locations that can be target of the jump)
is depending on checking the correctness of other
properties of the application
(e.g. 
 the link register, which is usually used to track the return address of functions, can be pushed and popped from the stack, thus making
 the correctness of the control flow dependent on the integrity of the
 stack itself).
Since we are interested in solving indirect jumps of code
fragments that must respect contracts (Hoare triples
$\{P\}C\{Q\}$), we implemented a simple iterative procedure that uses
STP to discover all possible indirect jump targets under the contract
precondition $P$.

 \begin{enumerate}
   \item \label{IND_JMP_START}
     The CFG of the of 
     $C$ fragment is computed using BAP. From the CFG, the list $L$ of reachable
     addresses containing an indirect jump is extracted
   \item For each address $a \in L$, the code fragment $C$ is modified as follows:
     \begin{enumerate}
       \item let $\mathit{exp_a}$ be the expression used in the indirect jump
       \item the indirect jump is substituted with an assertion, which requires
         $\mathit{exp_a}$ to be different from a fresh variable $\mathit{fv}_a$; if such
         assertion fails, i.e.~$\mathit{exp}_a=\mathit{fv}_a$, the modified fragment $C$ terminates with a fault,
         otherwise it correctly terminates
     \end{enumerate}
   \item the new fragment has no indirect jump; the weakest precondition $\mathit{WP}$ of the postcondition $\mathit{true}$
     (i.e. correct termination) is computed
   \item  the SMT solver searches for an assignment of the free variables (including all
     $\mathit{fv}_i$) that invalidates $P \Rightarrow \mathit{WP}$
   \item if the SMT solver discovers a counterexample which involves the indirect jump at the address $a$,
     then it also discovers a possible target for this jump via selected assignment of the variable $\mathit{fv}$.
     Let $\mathit{exp}$ be the expression used in the indirect jump.
     The fragment $C$ is transformed by substituting the indirect jump with a conditional statement;
     if $\mathit{exp}$ is equal to $\mathit{fv}$ then jump to the fixed address $\mathit{fv}$, otherwise
     jump to the expression $\mathit{exp}$: \verb|jmp exp| will be transfed into \begin{lstlisting}
cjmp exp == fv; value; new_label
label new_label: jmp exp
     \end{lstlisting}
     
   \item if the SMT solver does not find a counterexample, then every indirect jump is either
     unreachable or all its possible targets have been discovered.
     The fragment $C$ is transformed by substituting every indirect jump with
     an assertion that always fails (\verb|assert false|).
   \item The procedure is restarted. Note that the inserted conditional statements prevent that
     the discovered assignments of $\mathit{fv}_a$ can be used to invalidate the formula
     by the SMT solver in the next iteration.
 \end{enumerate}
In order to handle the greater complexity of the hypervisor code respect to the separation
kernel verified in~\cite{dam2013formal}, we re-engineered this tool as a BAP plug-in. 
A particular problem that we face is that the CFG can contain
loops if the same internal function of the hypervisor is called twice from different
points in the program. Integrating the procedure with BAP allowed us to reuse
the existing loop-unfolding algorithms to break these artificial loops.

 \newsavebox{\IndJmpASM}
\begin{lrbox}{\IndJmpASM}
\begin{lstlisting}
0x100 bl #0x200  

0x104 ...
0x108 bl #0x200  

0x10C ...
// function
0x200 push LR    


0x204 str R1, R2 

0x208 pop LR     


0x20C b LR
...
\end{lstlisting}
\end{lrbox}

\newsavebox{\IndJmpBIL}
\begin{lrbox}{\IndJmpBIL}
\begin{lstlisting}
label pc_0x100 PC = 0x100;
  LR = PC+4; jmp pc_0x100
label pc_0x104 ...
label pc_0x108 PC = 0x108;
  LR = PC+4; jmp pc_0x100
label pc_0x10C ...

label pc_0x200 PC = 0x200;
  mem=store32(mem, SP, LR);
  SP=SP-4
label pc_0x204 PC = 0x204;
  mem=store32(mem, R1, R2);
label pc_0x208 PC = 0x208;
  LR=load32(mem, SP+4);
  SP=SP+4
label pc_0x20C PC = 0x20C;
  jmp LR
\end{lstlisting}
\end{lrbox}

 \begin{figure}
 \centering
 \begin{subfigure}[b]{0.3\linewidth} 
     \usebox{\IndJmpASM}
     \caption{Assembly}
     \label{fig:IndirectJumpExample-arm}
 \end{subfigure}\hspace{1cm}
 \begin{subfigure}[b]{0.4\linewidth} 
     \usebox{\IndJmpBIL}
     \caption{BIL}
     \label{fig:IndirectJumpExample-bap}
 \end{subfigure}
   \caption{Indirect jump example}
   \label{fig:IndirectJumpExample}  
 \end{figure}

 \begin{figure}
 \begin{subfigure}[t]{0.25\linewidth} 
   \raisebox{0.25\height}{\includegraphics[width=0.9\linewidth]{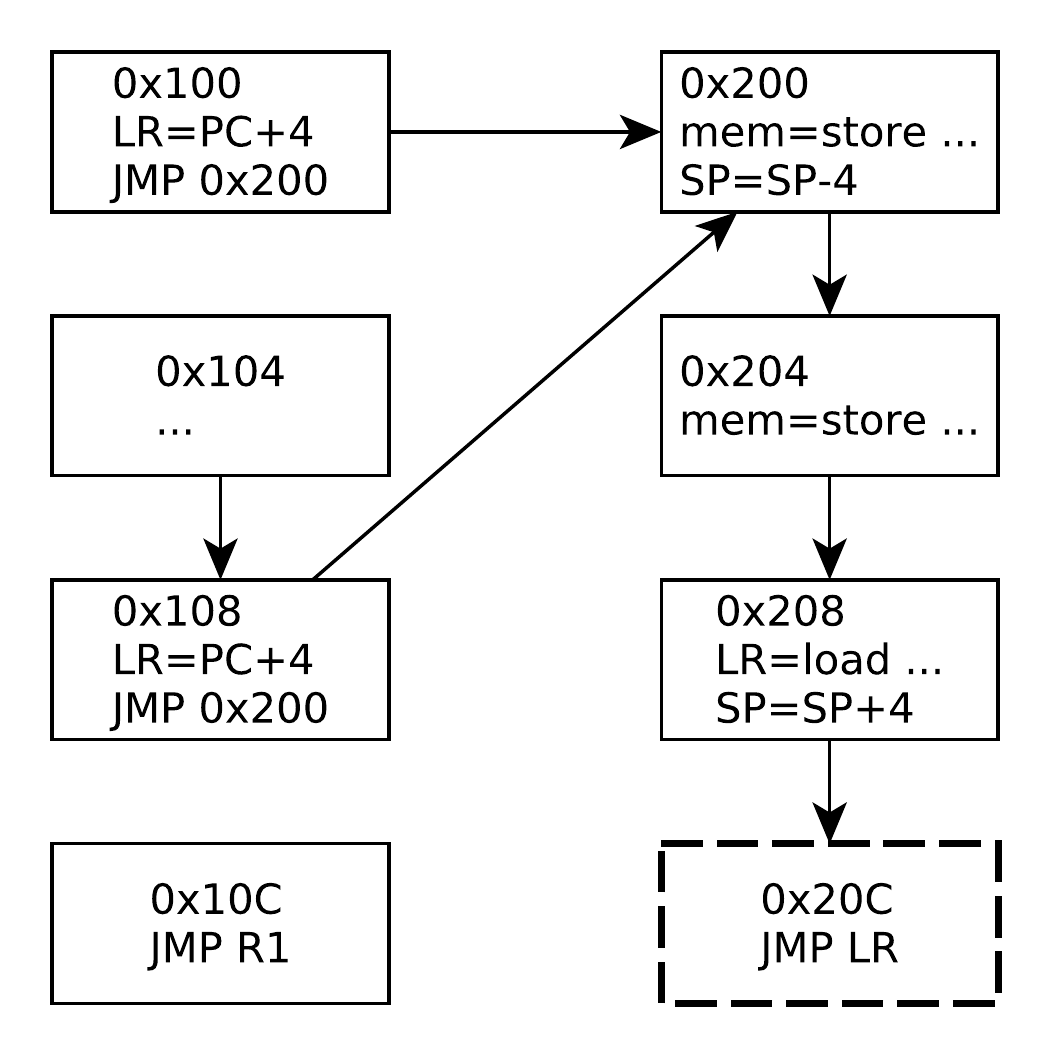}}
     \caption{}
     \label{fig:IndirectJumpExample-cfg1}
  \end{subfigure}\hspace{0.25cm}
 \begin{subfigure}[t]{0.25\linewidth} 
   \includegraphics[width=1\linewidth]{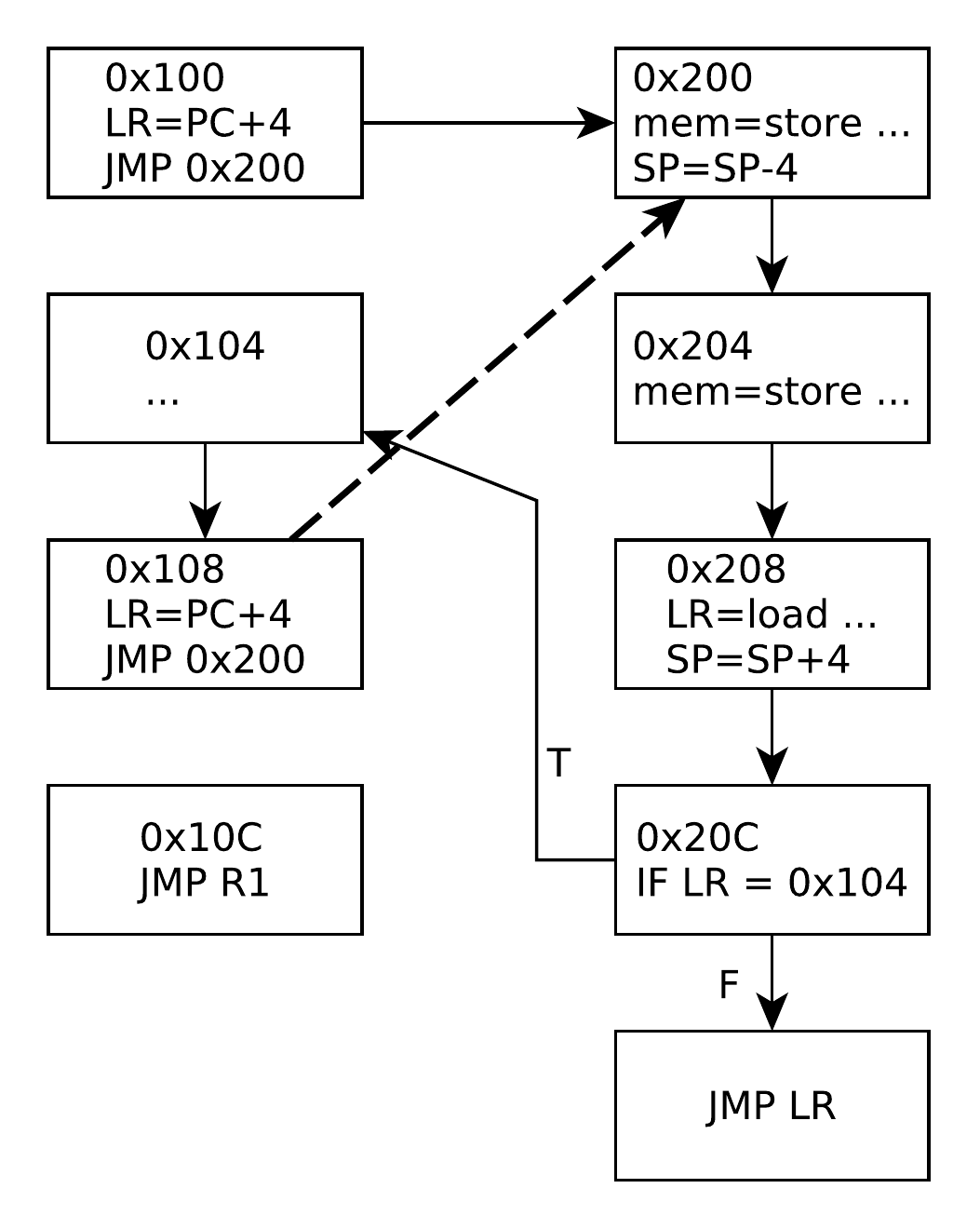}
     \caption{}
     \label{fig:IndirectJumpExample-cfg2}
  \end{subfigure}\hspace{0.25cm}
 \begin{subfigure}[t]{0.45\linewidth} 
   \includegraphics[width=1\linewidth]{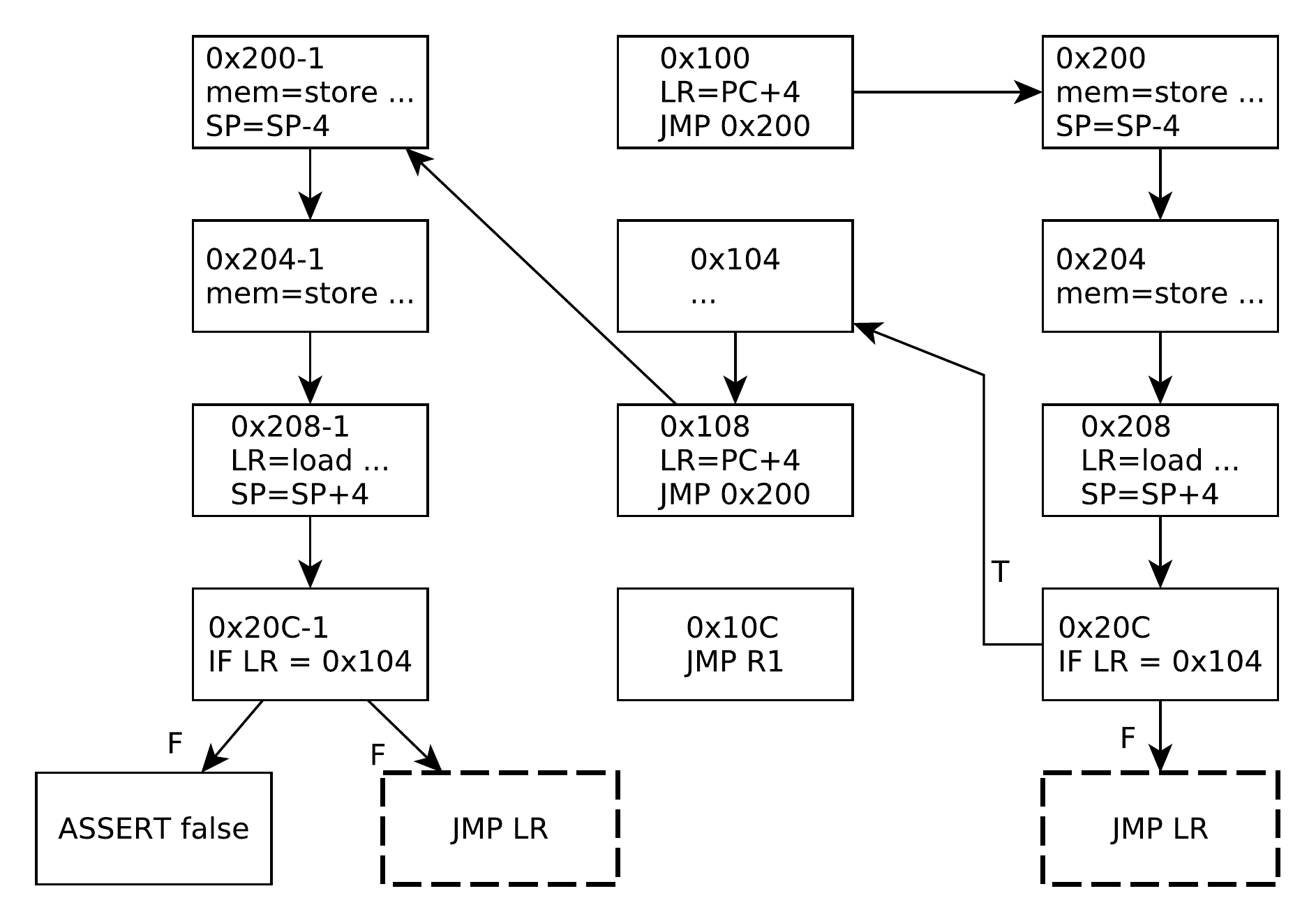}
     \caption{}
     \label{fig:IndirectJumpExample-cfg3}
  \end{subfigure}\\
 \begin{subfigure}[b]{0.45\linewidth} 
   \includegraphics[width=1\linewidth]{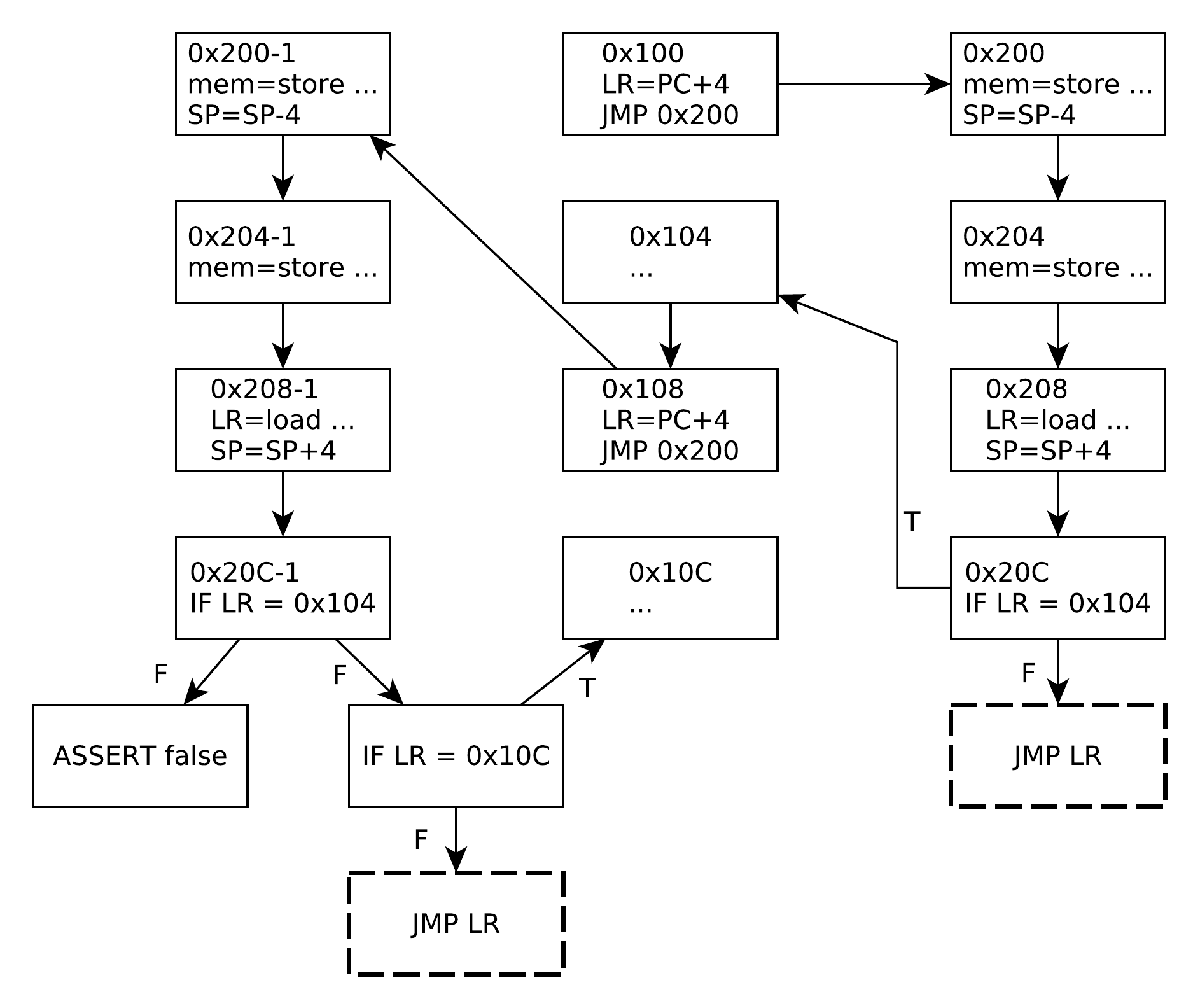}
     \caption{}
     \label{fig:IndirectJumpExample-cfg4}
  \end{subfigure}\hspace{1cm}
 \begin{subfigure}[b]{0.45\linewidth} 
   \includegraphics[width=1\linewidth]{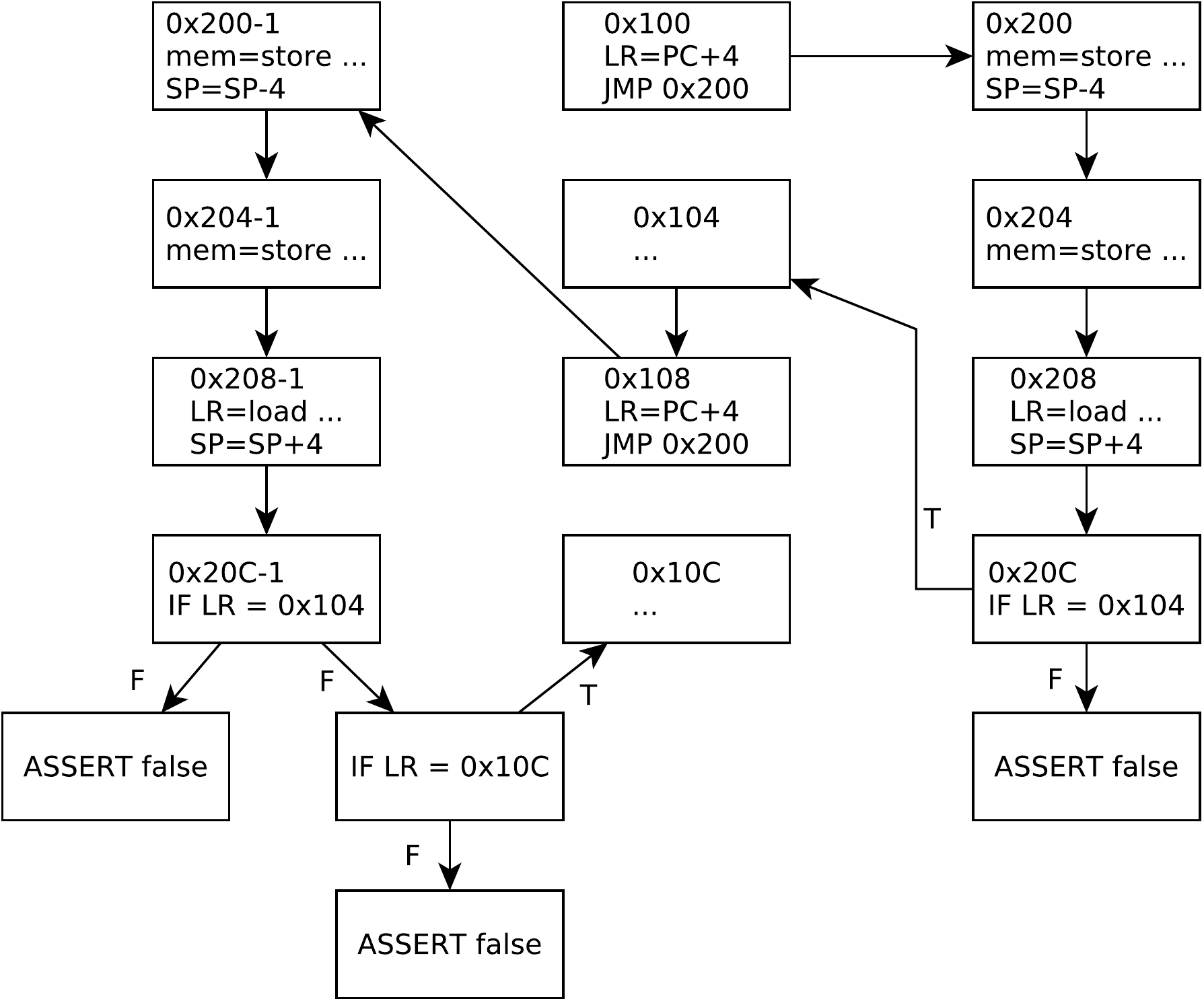}
     \caption{}
     \label{fig:IndirectJumpExample-cfg5}
  \end{subfigure}
   \caption{Execution of the indirect jump solver}
   \label{fig:IndirectJumpExample1}  
 \end{figure}

We use figures~\ref{fig:IndirectJumpExample} and~\ref{fig:IndirectJumpExample1} to demonstrate the algorithm. 
The assembly program (Figure~\ref{fig:IndirectJumpExample-arm}) contains a function at
\verb|0x200| , which is invoked twice (from \verb|0x100| and \verb|0x108|).
This function push the link register in the stack (\verb|0x200|), writes the
content of the register \verb|R2| into the memory pointed by \verb|R1|
(\verb|0x204|), pop the link register in the stack (\verb|0x208|)  and returns 
(\verb|0x20C|).
We
assume that the precondition used is strong enough to ensure correct
manipulation of the stack (e.g. the value of the stack pointer \verb|SP| and the
value of register \verb|R1| used as pointer in the instruction at \verb|0x204| are
distant at least one word). 
Figure~\ref{fig:IndirectJumpExample-bap} and ~\ref{fig:IndirectJumpExample-cfg1}
depict the BIL translation of the program and its initial CFG
respectively. The CFG has only one reachable indirect jump (in
\verb|0x20C|), whose expression is \verb|LR|.
The SMT solver discovers a possible target for this jump (in this case
\verb|0x104|) and the program is transformed by substituting the indirect jump
with a conditional statement, obtaining CFG is depicted in
Figure~\ref{fig:IndirectJumpExample-cfg2}.
This CFG has an artificial loop due to the two invocations of the
same function. Figure~\ref{fig:IndirectJumpExample-cfg3} depicts the CFG
obtained by unrolling the loop once. The program has now two reachable indirect
jumps, the procedure is repeated and the SMT solver discovers that \verb|0x10C|
is a possible target of the jump in \verb|0x20C-1|. The CFG is transformed as
Figure~\ref{fig:IndirectJumpExample-cfg4}. This CFG has still two 
indirect jumps. However, the SMT solver discovers that there is no assignment to
the initial variables of the program that enables the activation of these jumps. Thus all
indirect jumps have been resolved, the remaining ones are unreachable and are
suppressed, obtaining the CFG in Figure~\ref{fig:IndirectJumpExample-cfg5}.

In addition to solving indirect jumps, effective application of the verification strategy required the
implementation of several tools and optimisation of the weakest
precondition algorithm of BAP.
Weakest preconditions can grow exponentially with regard to the number of
instructions. 
Even though this problem cannot be solved in general, we
can handle the most common case for ARM
binaries, namely the sequential composition of several conditionally
executed arithmetical instructions. This pattern matches the
optimisation performed by the compiler to avoid small
branches.
We improved the BAP weakest precondition algorithm by adding a
simplification function that identifies these cases.
For some fragments of the code this straightforward strategy
strongly reduced the size of the precondition; e.g. for one fragment
consisting of 27 C lines compiled to 35 machine instructions
the size of the precondition has been reduced from 8 GB to 15 MB.

Furthermore, machine code (and BIL) lacks information on data types (except for
the native types like word and bytes) and represents the whole memory
as a single array of bytes. Writing predicates and invariants
is complex because their definition depends on location, alignment and size of
data-structure fields. Moreover, the behaviour of compiled code often depends
on the content of static memory used to represent constant values of the high level language.
We developed a set of tools that integrate HOL4 and
GDB to extract information from the C source
code and the compiled assembly.
With the support of these tools we are able to write the 
invariants and contracts of the hypervisor independently
of the actual symbol locations and data structure offsets 
produced by the compiler. 

Figure \ref{fig:binaryVerif} summarises the work-flow of our binary verification approach.
\begin{figure}
  \centering
  \includegraphics[width=0.9\linewidth]{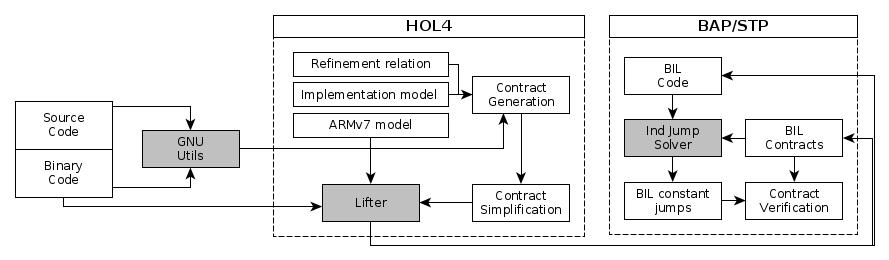}
  \caption{Binary verification work-flow:           
           \textit{Contract Generation}, generating pre and post conditions based on the specification of the low-level abstraction and the refinement relation; \textit{Contract Simplification},
            massaging contracts to make them suitable for verification; \textit{Lifter}, lifting handlers machine code and the generated contracts in HOL4 to BIL, \textit{Ind Jump Solver}, procedure to resolve
            indirect jumps in the BIL code; \textit{BIL constant jumps}, BIL fragments without indirect jumps; \textit{Contract Verification} using SMT solver to verify contracts. 
            Here, grey boxes are depicting the tools that have been developed to automate the verification as much as possible.}
  \label{fig:binaryVerif}
\end{figure}

\subsection{Limitations}\label{sec:limitation}
The binary verification of the hypervisor has not been completed yet due to some
time consuming tasks that require better automation.
First,
the inference procedure of Section~\ref{sec:binary:contracts} uses the HOL4
simplification rules and decision procedures, however it is
not completely automatic and must be adapted for every handler.
Without taking into account the specificity of each handler,
a naive procedure can easily generates contracts that cannot be handled by SMT
solvers. For every handler, we manually specialize the procedure to generate
contracts that have no quantifier in the precondition and only
universal quantifiers in the postcondition. 

Further complexity arises due to presence of loops.
In theory, loops can be automatically handled by
unfolding, since all loops in the hypervisor code
iterate over fixed and limited ranges (e.g. the number of descriptors
in a page table). Practically, this increases the size of the
code (1024 times for handlers working on L2, and up to 4096*256
for handlers on L1) beyond the limit of programs that
can be analyzed with BAP; thus the majority of loops 
must be manually handled.

By design, every loop in the hypervisor is also present in the
specification.
Let $C=C_1;while(B)\{C_2\};C_3$ be the handler fragment and
let $H_{r}(\ArmState) = let\ \ArmState_1 := H_{1}(\ArmState)\ in\ 
let\ \ArmState_2 := FOR(b, H_2, \ArmState_1)\ in\
H_3(\ArmState_2)
$ be the specification.
The problem of verifying that the refinement is preserved
(i.e. if $\ArmState \mathcal{R'} \ArmState'$,
and $C(\ArmState)$ is the state produced by the program $C$,
and  $H_{r}(\ArmState')$ is the state produced by the specification $H_r$ then
$C(\ArmState) \mathcal{R'} H_{r}(\ArmState')$)
is reduced in verifying three refinements:
\begin{itemize}
  \item $\ArmState \mathcal{R'} \ArmState'$ implies 
    $C_1(\ArmState) \mathcal{R'}_1 H_1(\ArmState')$
  \item $\ArmState \mathcal{R'}_1 \ArmState'$ implies 
    $C_2(\ArmState) \mathcal{R'}_1 H_2(\ArmState')$
  \item $\ArmState \mathcal{R'}_1 \ArmState'$ implies 
    $C_3(\ArmState) \mathcal{R'} H_3(\ArmState')$
\end{itemize}
that is, a new refinement relation/invariant $\mathcal{R'}_1$ must be identified
 for the loop.
 This usually  means identifying register allocation, allocations of variables
 on the stack etc. Due to lack of tools and integration with the compiler, this
 task is manually performed and requires to additionally specialize 
 the inference procedure of Section~\ref{sec:binary:contracts}.


\section{Implementation}
The implementation of the hypervisor demonstrates the feasibility of our approach. The actual implementation targets BeagleBoard-xM (which is equipped with an ARM Cortex-A8) and supports the execution of Linux as the untrusted guest.
The hypervisor executes both the untrusted guest and
  the trusted services in unprivileged mode, and their execution is
  cooperatively scheduled.
Theorems~\ref{lem:invariant},~\ref{lem:mmu-integrity} and~\ref{lem:switch} guarantee that the main security properties of the system (i.e.~the correct setup of the page tables) cannot be violated by either the guest or the trusted services. Moreover, the untrusted guest cannot directly affect the trusted services or directly extract information from their states (Theorems~\ref{lem:No-Exfiltration} and ~\ref{lem:No-Infiltration}).
This isolation is achieved by the complete mediation of the MMU settings and the allocation of the ARM domains 2-15 to the secure services. 
This approach limits the number of secure services to fourteen. However, this
mechanism has the benefit of using the same page tables for both the guest and
the trusted services (by reserving an area of the hypervisor virtual memory for
the latter). This reduces the cost of context switch, since TLB and caches do
not need to be cleaned. If more trusted services are needed, a separate page
tables can be used. 

The core of the hypervisor is the virtualization of the memory subsystem. This
is provided by the handlers that are the subject of the verification and that
are modelled by the transformations $H_a$ and $H_r$ (Section~\ref{sec:goal}).
This core have been extended with additional handlers to provide further
functionalities, which are needed to host a complete OS and to implement
useful secure services. 
Since these additional handlers are not involved in
the virtualization of the memory subsystem, establishing that they preserve the
invariant (Theorem~\ref{lem:invariant}) usually requires only to demonstrate
that they do not directly change the physical blocks that contain the page
tables and their memory safety.

\subsection{Linux Support}\label{sec:linux}
The Linux kernel 2.6.34
has been modified to run on top of the hypervisor.
This task required modification of architecture-dependent parts of the Linux kernel like 
execution modes, low-level exception routines and page table management. 
High-level OS functions such as process, resource and memory manager, file system, and networking did not require any modifications.
This also introduce the additional handlers of the hypervisor that are not part
of the formal verification.

\paragraph{CPU Privilege Modes}
In the absence of hardware supports, like virtualization extension,
the target CPU includes only two execution modes: privileged and
unprivileged (user).
As for other approaches based on paravirtualization,
since the hypervisor executes as privileged, the Linux kernel has been modified to execute 
as unprivileged. To separate kernel and user applications, the hypervisor manages two separate unprivileged execution contexts:
virtual user and virtual kernel modes. In x86 these virtual modes can be implemented by
segmentation. This approach is not possible for CPUs that do not provide this feature (e.g. x86 64-bit and ARM). 
Instead, we reserve the ARM domain 0 for the kernel virtual mode. Whenever the
guest kernel requests a switch to virtual user mode (invoking the dedicated hypercall) we disable  the domain 0, thus any 
access to the kernel virtual addresses generates a fault. 

Note that the main security goal here is not to guarantee this
OS-internal isolation, but to maintain the separation between the virtualized components (such as the Linux guest vs. secure data or services residing in non-guest memory).

\paragraph{CPU Exceptions} 
CPU exceptions such as aborts and interrupts change the processor mode to privileged.
These exceptions must therefore be handled in the hypervisor, which after 
validation can forward them to the unprivileged exception handlers of the Linux kernel.
The hypervisor supplies the kernel exception handlers with some 
privileged data needed to correctly service an on-going exception
(e.g. for pre-fetch abort, the privileged fault address and fault
status registers are forwarded to the guest).
 The exception handlers in the Linux kernel have thus been slightly
 modified to support this.
Among the exceptions that are forwarded to the Linux kernel there are the hardware
interrupts delivered by the timer. This allows Linux to implement an internal
time based scheduler.

\paragraph{Memory Management}
To paravirtualize the kernel, we modified the architecture dependent layer of its memory management. In the modified Linux all accesses to the coprocessor registers or to the active page
tables are done by issuing the proper hypercalls. The architecture independent layer of the memory management has been left unmodified. In order to speed up the execution of Linux,
a minimal emulation layer has been moved from the Linux kernel into the hypervisor itself. 
This layer reduces the overhead by translating a guest request into a sequence
of invocations of the APIs that virtualize the MMU . Since the emulation layer
accesses page tables only through the virtualization API, showing memory safety of this
component is sufficient to extend the coverage of the verification.


\subsection{Run-time Overhead}\label{sec:benchmark}
The port of the Linux kernel 2.6.34 on the hypervisor allows us to present a rough
comparison of our approach with existing paravirtualized hypervisors for the ARM architecture~\cite{iqbal2009overview}.
The purpose of the evaluation is more to demonstrate that our approach
actually runs with reasonable efficiency.
A serious evaluation is out of scope of this work. It requires a more optimised implementation, and a more
comprehensive evaluation.
 
The run-time evaluation is done using LMBench~\cite{McVoy:1996:LPT:1268299.1268322} running on Linux 2.6.34 with and without virtualization. 
The outcome, measured on an ARMv7-A Cortex-A8
system (BeagleBoard-xM~\cite{beageboardxm}), is presented in
Table~\ref{tab:performanceLMB}.
The significant virtualization overhead for the fork benchmarks is due to a large number of simple operations (in this case, write access to a page-table) being replaced with a large number of expensive hypercalls. It may be possible to reduce this overhead with minimal optimisation (e.g. batching).
In Table~\ref{tab:performanceLMB} we also report measures
from~\cite{iqbal2009overview}, where the authors compare several
existing hypervisors for ARM. 
We point out that these performance numbers have been obtained from different sources, testing different ARM cores, boards and hosted Linux
kernels. Hence we do not claim to be able to draw any hard conclusions from these figures about the
relative performance of the hypervisors or their underlying architectures.
With the purpose of demonstrating that the hypervisor can run
efficiently real applications, we also measured the overhead
introduced when executing tar, dd and several compression tools.

The second column reports the latency for the version of the hypervisor that aggressively flushes the caches (i.e.
the caches are
completely clean and invalidated whenever an exception or an interrupt is raised, while the hypervisor in the first column
limits cache flushes to the cases of context switch).
This naive approach guarantees that
the actual CPU respects the fully sequential memory model, but introduces severe performance penalties especially in the application benchmarks.
Less conservative approaches (e.g.~evicting only the necessary physical
addresses or forcing the page tables to be allocated in memory regions that are
always cacheable) can be adopted for some processor implementations, but they
require a more fine-grained modelling including caches and and an adaptation of
the verification approach for their justification, as discussed in~\cite{DBLP:conf/sp/GuancialeNBD16}.

\begin{table}[]
\tiny
\centering
\label{tab:microbench}
\begin{tabular}{l|r|r||r|r|r}

Benchmark                            & Hypervisor & \shortstack{Aggressive\\cache flushes}                  & L4Linux & Xen   & OKL4  \\  
\hline
null syscall                         & 329\%      & 332\%                                                   & 3043\%  & 150\% & 60\%  \\  
read                                 & 160\%      & 181\%                                                   & 844\%   & 90\%  & 15\%  \\  
write                                & 193\%      & 201\%                                                   & 877\%   & 85\%  & 24\%  \\  
stat                                 & 83\%       & 84\%                                                    & 553\%   &       & 7\%   \\  
fstat                                & 118\%      & 122\%                                                   & 945\%   &       & 42\%  \\  
open/close                           & 121\%      & 119\%                                                   & 433\%   &       &       \\  
select(10)                           & 78\%       & 84\%                                                    & 4461\%  &       & 14\%  \\  
sig handler install                  & 237\%      & 245\%                                                   & 1241\%  &       & 16\%  \\  
sig handler overhead                 & 226\%      & 237\%                                                   & 1281\%  & 82\%  & -14\% \\  
protection fault                     & 40\%       & 39\%                                                    & 975\%   &       & 67\%  \\  
pipe                                 & 168\%      & 3073\%                                                  & 450\%   & 74\%  & 31\%  \\  
fork+exit                            & 195\%      & 1861\%                                                  & 950\%   & 247\% & 8\%   \\  
fork+execve                          & 187\%      & 1787\%                                                  & 591\%   & 239\% & 5\%   \\  
pagefaults                           & 435\%      & 8740\%                                                  & 567\%   &       &       \\  
\end{tabular}
   \caption{
Latency benchmarks. LMBench micro benchmarks  for the Linux kernel v2.6.34 running naively on BeagleBoard-xM, paravirtualized on the hypervisor without
    cache flushing (\textit{Hypervisor}), with aggressive flushing
    (\textit{Aggressive cache flushes}), and the other hypervisors (\textit{L4Linux}, \textit{Xen}, \textit{OKL4}).
    Figures in the three last columns have been obtained from different ARM cores, boards and hosted Linux kernels}
   \label{tab:performanceLMB}
\end{table}

\begin{table}[]
\tiny
\centering
\label{tab:appbench}
\begin{tabular}{l|r|r}
Applications                         & Hypervisor & \shortstack{Aggressive\\cache flushes}                \\                                      
\hline 
tar (500 KB)                           & 0\%        & 171\%                                                  \\
tar (1 MB)                             & 0\%        & 108\%                                                  \\
dd (10 MB)                             & 100\%      & 1000\%                                                 \\
dd (20 MB)                             & 79\%       & 932\%                                                  \\
dd (40 MB)                             & 76\%       & 1061\%                                                 \\
jpg2gif(5 KB)                         & 0\%        & 117\%                                                  \\
jpg2bmp(5 KB)                         & 0\%        & 175\%                                                  \\
jpg2bmp(250 KB)                       & 0\%        & 27\%                                                   \\
jpg2bmp(750 KB)                       & -1\%       & 24\%                                                   \\
Jpegtrans(270', 5 KB)                 & 0\%        & 700\%                                                  \\
Jpegtrans(270', 250 KB)               & 14\%       & 300\%                                                  \\
Jpegtrans(270', 750 KB)               & 11\%       & 176\%                                                  \\
Bmp2tiff(90 KB)                      & 0\%        & 500\%                                                  \\
Bmp2tiff(800 KB)                     & 0\%        & 300\%                                                  \\
Ppm2tiff(100 KB)                     & 0\%        & 600\%                                                  \\
Ppm2tiff(250 KB)                     & 0\%        & 700\%                                                  \\
Ppm2tiff(1.3 MB)                     & 50\%       & 350\%                                                  \\
Tif2rgb(200 KB)                      & 200\%      & 1100\%                                                 \\
Tif2rgb(800 KB)                      & 25\%       & 575\%                                                  \\
Tif2rgb(1.200 MB)                    & 31\%       & 462\%                                                  \\
sox(aif2wav -r 8000 --bits 16 100 KB) & 50\%       & 600\%                                                  \\
sox(aif2wav -r 8000 --bits 16 500 KB) & 75\%       & 350\%                                                  \\
sox(aif2wav -r 8000 --bits 16 800 KB) & 83\%       & 267\%                                                  \\
\end{tabular}
\caption{
Latency benchmarks. Application benchmarks for the Linux kernel v2.6.34 running natively on BeagleBoard-xM, paravirtualized on the hypervisor without
    cache flushing (\textit{Hypervisor}), with aggressive flushing (\textit{Aggressive cache flushes}).}
    \label{tab:performanceApp}
\end{table}
\begin{table}
  \tiny
  \center
\begin{tabular}{c|r|r|r}
  Processes  & \shortstack{Direct Paging\\256 MB} & \shortstack{Direct
    Paging\\1 GB} & Shadow page table \\
\hline
  32  & 56 KB  & 224 KB  & 608 KB \\
  64  & 64 KB  & 256 KB  & 1216 KB \\
  128 & 72 KB  & 288 KB  & 2432 KB 
\end{tabular}
  \caption{
Memory footprint. Comparison of memory usage of \textit{Shadow page table} and \textit{direct paging}.}
  \label{tab:MemoryUsage}
\end{table}

\subsection{Memory Footprint}
The main difference between our proposal and the existing verified hypervisors is the MMU virtualization mechanism. 
The direct paging approach requires a table which contains at most
$\mem_{\fnsize} / block_{\fnsize}$ entries, where $\mem_{\fnsize}$ is the total
available physical memory and $\block_{\fnsize}$ is the minimum page size
(here, 4 KB). 
Each entry in this table uses $2 + \log_2 \ max_{\fnref}$ bits, with
the first two bits used to record entry type and $max_{\fnref}$ being the
maximum number of references pointing to the same page. 
Assuming this number is bound by the number of processes, Table~\ref{tab:MemoryUsage} indicates 
the memory overhead introduced by direct paging.

It should be noted that on ARMv7, most operating systems including
Linux dedicate one L1 page to each process and at least three L2
pages to map the stack, the executable code and the heap.
Then the OS itself has a minimum footprint of $16\text{ KB}
+ 3*1 \text{ KB}$ per process. This footprint is doubled if the underlying
hypervisor uses shadow page tables.

\section{Evaluation}
The hypervisor is implemented in C (and some assembly) and consists of 4529 lines of code (LOC).
Excluding platform dependent parts, the hypervisor core is no larger than 2066 LOC.
The virtualization of the memory subsystem consists of 1200 LOC.
To paravirtualize Linux we changed 1025 LOC of its kernel,
950 in the ARM specific architecture folder and 75 in init/main.c.
The paravirtualization is binary compatible with existing userland
applications, thus we do not need to recompile either hosted applications
or the libc.
For comparison, the only other hypervisor that implements direct
paging
is the Xen hypervisor, which consists of 100 KLOC and its design is not suitable
for verification.
Instead, the small code base of our hypervisor makes it easier to
experiment with different virtualization paradigms and enables formal verification of its correctness.
The formal specification consists of 1500 LOC of HOL4 and
intentionally avoids any high level construct, in order to make the HOL4 model
as similar as possible to the C implementation, at the price of
increasing the verification cost. The complete proof consists of 18700 
LOC of HOL4.


 The verification highlighted a number of bugs in the initial design of the APIs:
 (i) arithmetic overflow when updating the reference counter, caused by
 not preventing the guest to create an unbounded number of references to a
 physical block,
 (ii) bit field and offset mismatch,
 (iii) missing check that a newly allocated page table prevents the
 guest to overwrite the page table itself,
 (iv) usage of the signed shift operator where the unsigned one was
 necessary and
 (v) approval of guest requests that cause unpredictable MMU behaviour.
 Moreover, the verification of the implementation model identified three
 additional bugs exploitable by the guest by requesting the validation
 of page tables outside the guest memory.
Finally, the methodology described in Section~\ref{sec:binary} has been
experimented in the verification of the binary code of one of the
hypercalls. This experiment identified a buffer overflow in the binary
code that was missing in implementation model.
The HOL4 model uses a 10-bit variable to store an untrusted parameter
which is later used to index the entries of a page table. The binary
code uses a 32-bit registers to store the same parameter, thus
causing an overflow when accessing the L2 page table if the received
parameter is bigger than 1023.
The bug has been fixed by sanitising the input using the mask
\verb|parameter = parameter & 0x3ff|.

The project was conducted in three steps. 
The design, modelling and verification of the APIs for MMU virtualization required nine
person months. Here, the most expensive tasks have been the
verification of Theorems~\ref{lem:invariant} and ~\ref{lem:refinement}.
The C implementation of the APIs and
the Linux port has been accomplished in three person months. 
While the implementation team was completing the Linux port the
verification team started the verification of the refinement, which has taken three months so far. This work is continuing, in order to complete the verification from the HOL4 implementation level down to assembly.


\section{Applications}\label{sec:applications}
Applications of the hypervisors include the deployment of trusted cryptographic services and trusted controllers.
In the first scenario, the hypervisor core is extended with the
handlers required to implement message passing. These handlers allow
(i) Linux to send a message to the trusted service,  (ii) the trusted
service to reply with an encrypted message and (iii) the two
partitions to cooperatively schedule themselves. The isolation
properties guarantee that the untrusted guest cannot access the
cryptographic keys stored in the memory of the trusted services.
The second scenario includes a device (e.g. a physical sensor) whose IO is  memory mapped. The guest is forbidden to access the memory where the IO registers are mapped, thus guaranteeing that the trusted controller is the only subject capable of directly affecting the device. The complete Linux system can be used to provide a rich and complex user interface (either graphical or web based) for the controller logic without affecting its security.

The MMU virtualization solution demonstrated here can be used by other ARM-based software platforms than the hypervisor reported above. A fully fledged hypervisor (e.g. XEN) can use our approach to support hardware that lacks virtualization extensions (e.g. Cortex-A8, Cortex-A5, ARM11).
The mechanism can also be used by compiler-based virtual machines and unikernels, which need to monitor the memory configuration and protect it from the rest of the system (e.g. SVA uses a non-verified implementation of direct paging). 
Customers of cloud infrastructures can also benefit from our approach (see Figure~\ref{fig:cloud}).
In this setting, if the virtualization extensions are available,
the most privileged execution mode is controlled by the software platform
managed by the cloud provider (e.g. a hypervisor).
Thus, these extensions cannot be used by the customer to isolate its untrusted Linux from its own trusted services. In this setup, our mechanism can be used to fulfil this requirement.

\begin{figure}
  \centering
  \begin{subfigure}[b]{0.45\linewidth}
    \includegraphics[width=1\linewidth]{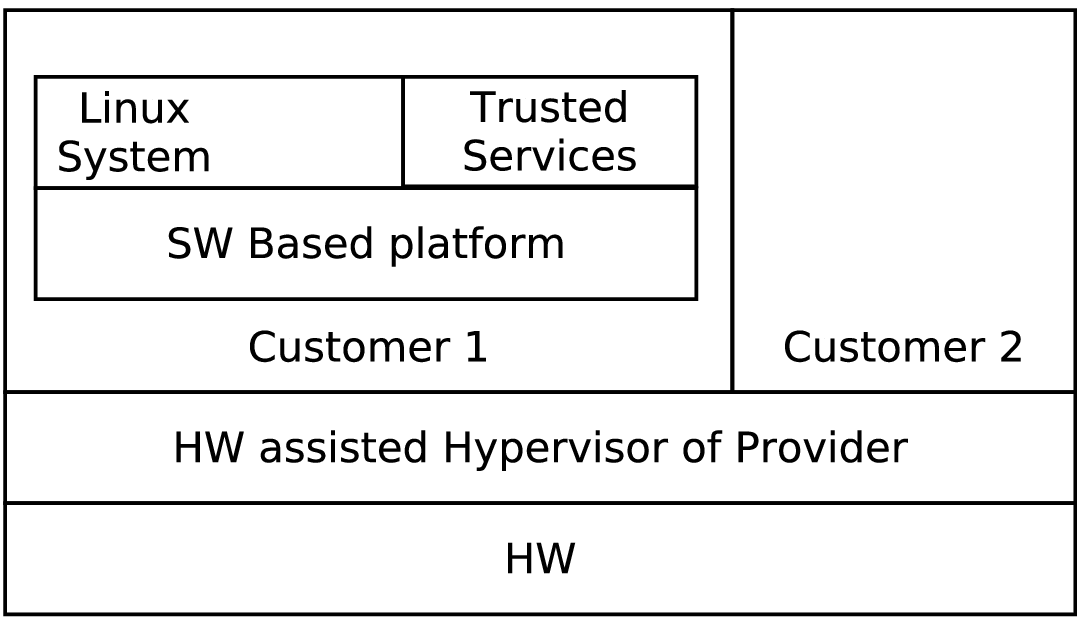}
    \caption{Usage of SW-based virtualization in a cloud platform}
    \label{fig:cloud}
  \end{subfigure}
  \qquad
  \begin{subfigure}[b]{0.45\linewidth}
    \includegraphics[width=1\linewidth]{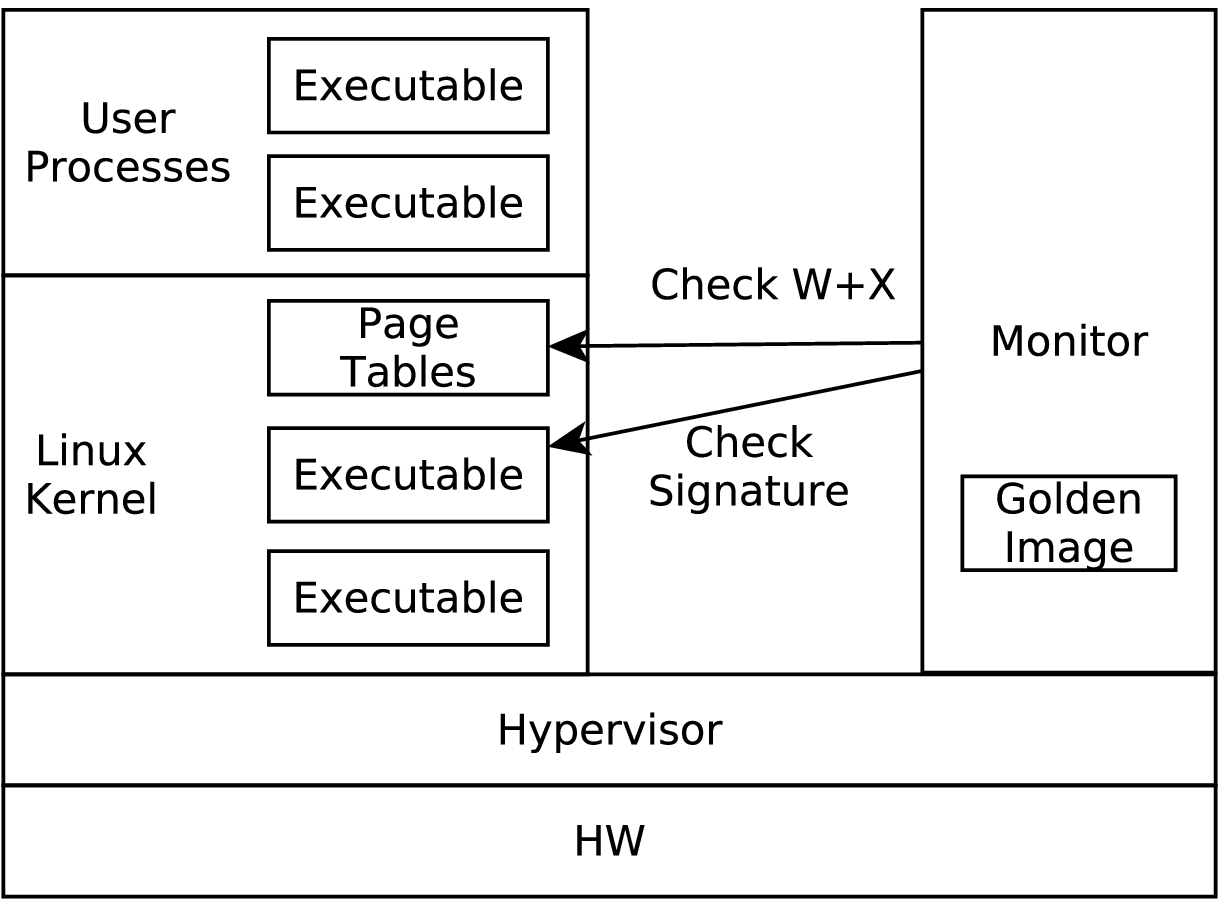}
    \caption{Deployment of a run-time monitor preventing code injection}
    \label{fig:monitor}
  \end{subfigure}
  \caption{Applications of the secure virtualization platform.}
  \label{fig:apps}
\end{figure}

An interesting application of 
isolating platforms is the external protection of an untrusted commodity OS from internal threats, as demonstrated in~\cite{criswell2014kcofi}.
Trustworthy components are deployed together and properly isolated from the application OS (see Figure~\ref{fig:monitor}).
These components are used as an aid for the application OS to 
restrict its own attack surface, by guaranteeing the impossibility of certain malicious
behaviours. In~\cite{DBLP:conf/esorics/ChfoukaNGDE15}, we show that this approach can be used to
implement an embedded device that hosts a Linux system provably free of binary
code injection. Our goal is to formally prove that the target system prevents all forms of binary code
injection even if the adversary has full control of the hosted Linux and 
no analysis of Linux is performed. 

The system is implemented as a run-time monitor. The monitor
forces an untrusted Linux system to obey the executable space
protection policy (usually represented as $W \oplus X$); a memory area
can be either executable or writable, but cannot be both. 
The protection of executable space allows the monitor to
intercept all changes to the executable code performed by a user
application or by the Linux kernel itself. On top of this infrastructure, we
use standard signature checking to prevent code injection.
Here, integrity of an executable physical block stands for
the block having a valid signature. Similarly, the integrity of the system code
depends on the integrity of all executable physical blocks. The valid signatures
are assumed to be known by the run-time monitor. We refer to this information
as the ``golden image'' (GI) and it is held by the monitor.

We configured the hypervisor to support the following interaction protocol:
\begin{enumerate}
\item For each hypercall invoked by a guest, the hypervisor forwards the request to the monitor. 
\item The monitor validates the request based on its validation mechanism.
\item The monitor reports to the hypervisor the result of the hypercall validation. 
\end{enumerate}
Since the  hypervisor supervises the changes of the page tables, 
the monitor is able to intercept all the memory layout modifications.
This makes the monitor able to know whether a physical block is writable:
if there exists at least one virtual mapping pointing to such block and having
writable access permission. Similarly it is
possible to know which physical block is executable.

Then the signature checking is implemented in the obvious way:
whenever Linux requests to change a page table (i.e. causing to change
the domain of the executable code) the monitor (i) identifies
the physical blocks that can be made executable by the request, (ii)
computes the block signature and (iii) compares the result with the
content of the golden image.
This policy is sufficient to prevent code injection
that is caused by changes of the memory layout setting, due to 
the  hypervisor forwarding to the monitor all requests to change 
the page tables.


Figure~\ref{fig:monitor} depicts the architecture of the system;
both the run-time monitor and the untrusted Linux are deployed as two guests
of the hypervisor. 
Using a dedicated guest on top
of the hypervisor permits to decouple the enforcement of the security policies
from the other hypervisor functionalities, thus keeping the
trusted computing base minimal.
Moreover, having the security policy wrapped inside a guest supports both
the tamper-resistance and the trustworthiness of the monitor. In fact, the monitor can take advantage from the isolation properties
provided by the hypervisor. This avoids malicious
interferences coming from the other guests (for example from a process of
an OS running on a different partition of the same machine). Finally,
decoupling the run-time security policy from the other functionalities of the
hypervisor makes the formal specification and verification of the monitor
more affordable.

The formal model of the system (i.e. consisting of the hypervisor,
the monitor and the untrusted Linux) is built on top of the models
presented in Section~\ref{sec:tls}.
Here we
leave unspecified the algorithm used to sign and check
signatures, so that our results can be used for different intrusion
detection mechanisms. 
The golden image $\mathit{GI}$ is a finite set of signatures $\{s_1,
\dots, s_n\}$, where the signatures are selected from a domain $S$.
We assume the existence of a function $\mathit{sig} : 2^{4096*8}
\rightarrow S$ that computes the signature of the content of a block.
The system behaviour is modelled by the following rules:

\begin{enumerate}
 \item 
$\begin{array}{ccccc}
  \frac{\tuple{\sigma, h} \to_{0} \tuple{\sigma', h'}}
       {\tuple{\sigma, h, \mathit{GI}} \to_{0} \tuple{\sigma', h', \mathit{GI}}}
\end{array}$
 \item 
$\begin{array}{ccccc}
  \frac{\tuple{\sigma, h} \to_{1} \tuple{\sigma', h'}
       \ \mathit{validate}(\mathit{req}(\tuple{\sigma, h}), \tuple{\sigma, h, \mathit{GI}})
       }
       {\tuple{\sigma, h, \mathit{GI}} \to_{1} \tuple{\sigma', h', \mathit{GI}}}
\end{array}$
 \item
$\begin{array}{ccccc}
  \frac{\tuple{\sigma, h} \to_{1} \tuple{\sigma', h'}
       \ \neg \mathit{validate}(\mathit{req}(\tuple{\sigma, h}), \tuple{\sigma, h, \mathit{GI}})
       }
       {\tuple{\sigma, h, \mathit{GI}} \to_{1} \epsilon(\tuple{\sigma, h, \mathit{GI}})}
\end{array}$
\end{enumerate}

User mode transitions (e.g. Linux activities)  require neither the
hypervisor nor the monitor intermediation.
Theorem~\ref{lem:No-Exfiltration} justifies the fact
that, by construction, the transitions executed by the untrusted component cannot 
affect the monitor state; (i) the golden image is constant and (ii)
the monitor code can be statically identified and abstractly modelled.
The executions in privileged mode require the 
intermediation of the monitor. If the monitor validates the request,
then the standard behaviour of the hypervisor is executed. Otherwise
the hypervisor performs a special
operation to reject the request, by reaching the state that is returned by a function $\epsilon$.
Hereafter, the function 
$\epsilon$ is assumed to be the identity. 
Alternatively, 
$\epsilon$ can transform the state so that the requestor is informed
about the rejected operation, by updating the user registers according
to the desired calling convention.
The function $\mathit{validate}(\mathit{req}(\tuple{\sigma, h}), \tuple{\sigma, h, \mathit{GI}})$
represents the validation mechanism of the monitor, which checks at
run-time possible violations of the security policies. 

To formalize the top level goal of our verification
we introduce some auxiliary notations. 
The ``working set'' identifies the physical blocks that host
executable binaries and their corresponding content.
Let $\sigma$ be a machine state. The working set of $\sigma$ is defined as
\[
\begin{array}{l}
  \mathit{WS}(\sigma) = \{\tuple{\mathit{bl}, \mathit{content}(\mathit{bl}, \sigma)} \mid
  \exists \mathit{pa}. \mathit{mmu}_{ph}(\sigma, \UserMode, \mathit{pa}, \mathit{ex}) \wedge \mathit{pa} \in \mathit{bl}
  \}
\end{array}
\]

By using a code signing approach, we say that the integrity of a
physical block is satisfied if the signature of the block's content belongs
to the golden image.
  Let $\mathit{cnt} \in 2^{4096*8}$ be the 4 KB content of a physical block $bl$
   and $\mathit{GI}$ be the golden image
\[
  \mathit{int}(\mathit{GI}, \mathit{bl}, \mathit{cnt}) = \mathit{sig}(\mathit{bl}, \mathit{cnt}) \in \mathit{GI}
\]
Notice that our security property can be refined to fit different
anti-intrusion mechanisms. For example, $\mathit{int}(\mathit{GI}, \mathit{bl}, \mathit{cnt})$ can be
instantiated with the execution of an anti-virus scanner.

The system state is free of malicious code injection if
the signature check is satisfied for the whole executable code. That is:
  Let $\sigma$ be a machine state, $bl$ be a physical block and $\mathit{GI}$
  be the golden image
\[
  \mathit{int}(\mathit{GI}, \sigma)\ \Leftrightarrow\ \forall \tuple{\mathit{bl},\mathit{cnt}} \in \mathit{WS}(\sigma)\  .\  \mathit{int}(\mathit{GI}, \mathit{bl}, \mathit{cnt})
\]

Finally, in~\cite{DBLP:conf/esorics/ChfoukaNGDE15} we demonstrate our top level goal: No code injection can succeed.
\begin{theorem}
  If $\tuple{\sigma, h, \mathit{GI}}$ is a state reachable from the initial
  state of the system $\tuple{\sigma_0, h_0, \mathit{GI}}$ then 
  $\mathit{int}(\mathit{GI}, \sigma)$
\end{theorem}

We implemented a prototype of the system.
The monitor code consists of 720 lines of C and 100 lines have been
added to the hypervisor to support the needed interactions among the
hosted components.

\section{Concluding Remarks}
We have presented a memory virtualization platform for ARM based on direct paging,  an approach inspired by the paravirtualization mechanism of Xen
\cite{xen}, and the Secure Virtual Architecture \cite{criswell2007secure}. The platform has been verified down to a detailed model of a commodity CPU architecture (ARMv7-A), and we have
shown a hypervisor based on the platform capable of hosting a Linux system while provably isolating it from other services. The hypervisor has been implemented on real hardware and shown 
to provide promising performance, although the benchmarks presented here are admittedly preliminary. The verification is done with respect to a top-level model that augments a real machine
state with additional model components. The verification shows complete mediation, memory isolation, and information flow correctness with respect to the top-level model. As the main 
application we demonstrated how the virtualization mechanism can be used to support a provably secure run-time monitor for Linux that provides secure updates along with the W$\oplus$X policy.

The main precursor work on formally verified MMU virtualization uses the simulation-based
approach of Paul et al \cite{alkassar2010automated,PSS12,ACKP12}.
In \cite{alkassar2010automated,PSS12} shadow page tables are used to provide full virtualization, including virtual memory, for  ``baby VAMP'', a simplified MIPS, using VCC. 
Full virtualization is generally more complex than
the paravirtualization approach studied in the present paper, but the machine model is simplified,
information flow security is not supported by the simulation
framework, and neither applications nor implementation on real hardware are reported.
In \cite{ACKP12} the same simulation-based approach is used to study TLB
virtualization on an abstract version of the x64 virtual memory architecture.
Other related work on verification of microkernels and hypervisors
such as seL4 \cite{DBLP:conf/sosp/KleinEHACDEEKNSTW09} or the Nova project \cite{SteinbergK10} does not address MMU virtualization in detail. It may be argued that
the emergence of hardware based virtualization support makes software MMU virtualization obsolete. We argue that this is not the case. First, many
platforms remain or are currently in development that do not yet support virtualization extensions, second, many application hardening frameworks such as 
Criswell et al.~\cite{Criswell:2009:MSL:1855768.1855774}, KCoFi~\cite{criswell2014kcofi}, Overshadow \cite{Chen:2008:OVA:1346281.1346284}, Inktag \cite{Hofmann:2013:ISA:2451116.2451146} and Virtual Ghost \cite{Criswell:2014:VGP:2541940.2541986} rely on some form of MMU virtualization for their
internal security, and third, some use cases, e.g.~in cloud scenarios, could make good use of software based MMU virtualization to harden VMs without
relying on cloud provider hardware.

Our results are not yet complete. The MMU virtualization approach does not
support DMA.  To securely enable DMA
the behaviour of the specific DMA controller must be formally modelled
(in~\cite{schwarz2014formal} the authors describe a framework for such extensions and
establish Properties~\ref{prop:User-No-Exfiltration} and~\ref{prop:User-No-Infiltration} for the resulting model)
and the hypervisor must (i) mediate all accesses to the memory area where
the controller's  registers are mapped, (ii) enable a DMA channel
only if the pointed physical blocks is \DataType  \ and (iii) update
the reference counters accordingly. 
Several embedded platforms are equipped with IOMMUs, that provide HW support to
isolate/confine external peripherals that use DMA. However SW based isolation of
DMA is still interesting since it can be used in the scenarios where these HW
extensions are not available (e.g. CortexM microcontrollers), they are not
accessible (e.g. when they are managed by a cloud provider), or in time
critical applications since the page walks introduced by the IOMMU can slow down the
peripheral and make worst case execution time analysis more difficult.

A tricky problem concern the treatment of unpredictable behaviour in the ARMv7
architecture. The Cambridge ISA model~\cite{DBLP:conf/itp/FoxM10} maps transitions resulting in
unpredictable behaviour to $\bot$. We
ignore this for the following reason.
Our verification shows that unpredictable behaviour never arises during
hypervisor code execution. This is so since the ARMv7 step theorems 
used by the lifter are defined only for predictable instructions, and since
our invariant guarantees that the MMU configuration is always well defined.
As a result unpredictable behaviour can arise only during non-privileged
execution, the analysis of which we have in effect deferred to other work~\cite{schwarz2014formal}. 

Finally more work is needed to properly reflect caches, TLBs, and, further down the line, multi-core. The soundness of the current implementation depends on
the type of data cache, and on flushing the cache when needed, in order to support a linearizable memory model. To enable more aggressive optimisation, 
and to fully formally secure our virtualization framework on processors with weaker cache guarantees, the model must be extended to reflect cache behaviour.



\chapter[Trustworthy Prevention of Code Injection in Linux on Embedded  Devices]
{\texorpdfstring{Trustworthy Prevention of Code Injection in Linux on Embedded \\ Devices}{Trustworthy Prevention of Code Injection in Linux on Embedded  Devices}}\label{paper:esorics}
\chaptermark{Trustworthy Prevention of Code Injection}
\backgroundsetup{position={current page.north east},vshift=1cm,hshift=-5cm,contents={\VerBar{RoyalBlue}{2cm}}}
\BgThispage
\begin{center}
Hind Chfouka, Hamed Nemati, Roberto Guanciale, Mads Dam, Patrik Ekdahl
\end{center}

\begin{abstract}

We present MProsper, a trustworthy system to prevent code
injection in Linux on embedded devices. 
MProsper is a formally verified run-time monitor, which
forces an untrusted Linux to obey the executable space
protection policy; a memory area
can be either executable or writable, but cannot be both. 
The executable space protection allows the MProsper's monitor to
intercept every change to the executable code performed by a user
application or by the Linux kernel. On top of this infrastructure, we
use standard code signing to prevent code injection.
MProsper is deployed on top of the Prosper hypervisor and is
implemented as an isolated guest. Thus MProsper inherits the security
property verified for the hypervisor: (i) Its code and data cannot be
tampered by the untrusted Linux guest and (ii) all changes to the
memory layout is intercepted, thus enabling MProsper to completely
mediate every operation that can violate the desired security property.
The verification of the monitor has been
performed using the HOL4 theorem prover and by extending the existing
formal model of the hypervisor with the formal specification of the
high level model of the monitor.
\end{abstract}

\newcommand{\accreq}{\mathit{accreq}}
\newcommand{\relRule}[4]{\{#1:#2\to #3 \}}
\newcommand{\ArmStateVar}{\sigma}
\newcommand{\mmudisbl}{\sigma.\SCTLR = 0}
\newcommand{\mmuenbl}{\sigma.\SCTLR \neq 0}
\newcommand{\descType}[3]{$#1(#2) = #3$}
\newcommand{\accPermChecker}[3]{#3(#1(#2))}
\newcommand{\RecCone}[1]{{#1}.c1}
\newcommand{\RecCtwo}[1]{{#1}.c2}
\newcommand{\RecMem}[1]{{#1}.mem}
\newcommand{\RecTTBR}[1]{{#1}.TTBR_{zero}}
\newcommand{\RecPgType}[1]{{#1}.\mathit{pgtype}}
\newcommand{\RecPgRefs}[1]{{#1}.rc}
\newcommand{\MaxRefCnt}{MAX}
\newcommand{\StaticMemType}{G_{mem}}
\newcommand{\MemTypeVar}{t}
\newcommand{\LOneSection}{SEC}
\newcommand{\LOnePt}{PT}
\newcommand{\SmallPageType}{SP}
\newcommand{\Types}[3]{#1 \vdash #2 : #3}
\newcommand{\inGuestMem}[1]{\Types \StaticMemType {#1} \GuestMemType}
\newcommand{\ptType}[3]{\Types {\RecPgType{#1}} {#2} {#3}}
\newcommand{\readDesc}[3]{$read_{#1}(#2, \sigma.mem, #3)$}
\newcommand{\readDescLOne}[3]{read_{L_1}(#1, #2, #3)}
\newcommand{\readDescLTwo}[3]{read_{L_2}(#1, #2, #3)}
\newcommand{\pgType}[1]{#1.type}
\newcommand{\pgAdd}[1]{#1.pa}
\newcommand{\pgPhBlock}[1]{#1.blk}
\newcommand{\secPhBlock}[1]{#1.sec}
\newcommand{\pgAP}[1]{#1.ap}
\newcommand{\funApl}[2]{$#1 #2$}
\newcommand{\Unmapped}{\bullet}
\newcommand{\HyperStateVar}{h}
\newcommand{\TlsStateVar}{s}
\newcommand{\PgTypeVar}{\tau}
\newcommand{\PgRefVar}{\rho}
\newcommand{\TlsState}[2]{\left \langle {#1},{#2} \right \rangle}
\newcommand{\refUpdateDec}[3]{\left\{\begin{array}{l l}
						#1(#2) - 1& #3\\
						#1(#2)\ & \mathit{otherwise}
					   \end{array}
					 \right.}					 
\newcommand{\refUpdateInc}[3]{\left\{\begin{array}{l l}
						#1(#2)++ & #3\\
						#1(#2)\ & \mathit{otherwise}
					   \end{array}
					 \right.}
\newcommand{\replc}[2]{#1 := #2}
\newcommand{\LabelBlk}{\mathit{bl}}
\newcommand{\LabelVar}{\alpha}
\newcommand{\LabelUser}{0}
\newcommand{\LabelSwitchLOne}[1]{switchL1(#1)}
\newcommand{\LabelUnmapEntryLTwo}[2]{unmapEntryL2\tuple{#1, #2}}
\newcommand{\LabelMapLTwo}[4]{mapL2\tuple{#1, #2, #3, #4}}
\newcommand{\LabelCreateLOne}[1]{createL1(#1)}
\newcommand{\LabelCreateLTwo}[1]{createL2\(#1)}
\newcommand{\LabelFreeLOne}[1]{freeL1(#1)}
\newcommand{\LabelFreeLTwo}[1]{freeL2(#1)}
\newcommand{\LabelMapSecLOne}[1]{mapSecL1(#1)}
\newcommand{\LabelMapPtLOne}[1]{mapPtL1(#1)}
\newcommand{\LabelUnmapEntryLOne}[1]{unmapEntryL1(#1)}
\newcommand{\tlsTrans}[1]{\xrightarrow[]{#1}}
\newcommand{\inferTrans}[3]{#1 \tlsTrans{#2} #3}
\newcommand{\LTwoAP}[1]{l2\_acc\_perm(#1)}
\newcommand{\LOneSecAP}[1]{l1\_acc\_perm(#1)}
\newcommand{\MakeSPDesc}[2]{[#1,#2]}
\newcommand{\writeEntry}[4]{update\_entry(#1, #2, #3, #4)}
\newcommand{\ApGuestWritable}[1]{(0,wt) \in #1}
\newcommand{\ApGuestReadable}[1]{(0,rd) \in #1}
\newcommand{\OpLabel}[2]{\overset{#2}{#1}}
\newcommand{\TlsTrans}[1]{\OpLabel{\rightarrow}{#1}}
\newcommand{\MemTypeEquiv}[1]{\OpLabel{\equiv}{\StaticMemType:{#1}}}
\newcommand{\MemBlockUnchanged}[1]{\OpLabel{\equiv}{#1}}
\newcommand{\MmuEquiv}{\OpLabel{\equiv}{mmu}}
\newcommand{\Invariant}{\mathcal{I}}
\newcommand{\InvStates}{\mathcal{Q}_{\Invariant}}
\newcommand{\InvType}{\mathcal{I}_T}
\newcommand{\InvTypeOne}{\mathcal{I}_{T_1}}
\newcommand{\InvTypeTwo}{\mathcal{I}_{T_2}}
\newcommand{\InvRC}{\mathcal{I}_C}
\newcommand{\RcCount}{cnt}

\section{Introduction}
Even if security is a critical issue of IT systems,
commodity OSs are not designed with security in mind. Short time to
market, support of legacy features, and adoption of binary blobs are only
few of the reasons that inhibit the development of secure commodity
OSs. Moreover, given the size and complexity of modern OSs, 
the vision of comprehensive and formal verification of them  is as
distant as ever. At the same time the necessity of adopting commodity
OSs can not be avoided; modern IT systems require complex network
stacks, application frameworks etc. 

The development of verified low-level execution platforms for system
partitioning (hypervisors~\cite{CavalcantiD09,SOFSEM}, separation
kernels~\cite{INTEGRITY,dam2013formal}, or
microkernels~\cite{DBLP:conf/sosp/KleinEHACDEEKNSTW09}) has enabled an efficient strategy to develop systems with provable security properties without having to verifying the entire software. The idea is to  partition the system into small and trustworthy
components with limited functionality running alongside  large commodity software components that provide little or no assurance. 
For such large commodity software it is not realistic to restrict the
adversary model. For this reason, 
the goal is to show, preferably using formal verification, that the architecture satisfies the desired
security properties, even if the commodity software is completely compromised.

An interesting usage of this methodology is when the 
trustworthy components are used as an aid for the application OS to 
restrict its own attack surface, by proving the impossibility of certain malicious
behaviors. In this paper, we show that this approach can be used to
implement an embedded device that hosts a Linux system provably free of binary
code injection. Our goal is to formally prove that the target system prevents all forms of binary code
injection even if the adversary has full control of the hosted Linux and 
no analysis of Linux itself is performed. This is necessary to make the verification feasible, since
Linux consists of million of lines of code and even a high level 
model of its architecture is subject to frequent changes.

Technically, we use Virtual Machine Introspection (VMI).
VMI is a virtualized architecture, where an untrusted guest is monitored by an external observer.
VMI has been proposed as a solution
to the shortcomings of network-based and host-based intrusion
detection systems. Differently from
network-based threat detection, VMI monitors the
internal state of the guest.
Thus, the VMI does not depend on information obtained from monitoring
network packets which may not be accurate or sufficient.
Moreover, differently from
host-based threat detection, VMIs place the
monitoring component outside of the guest, thus 
making the monitoring itself tamper proof. 
A further benefit of VMI monitors is that they can rely on
trusted information received directly from the underlying hardware,
which is, as we show, out of the attackers reach.

Our system, MProsper, 
is implemented as a run-time monitor. The monitor
forces an untrusted Linux system to obey the executable space
protection policy (usually represented as $W \oplus X$); a memory area
can be either executable or writable, but cannot be both. 
The protection of executable space allows MProsper to
intercept all changes to the executable code performed by a user
application or by the Linux kernel itself. On top of this infrastructure, we
use standard code signing to prevent code injection.

Two distinguishing features of MProsper are its execution on top of
a formally verified hypervisor (thus guaranteeing  integrity) and
the verification of its high level model (thus demonstrating that the
security objective is attained). To the best of our knowledge this
is the first time the absence of binary code injection
has been verified for a commodity OS. The verification of the monitor has been
performed using the HOL4 theorem prover and by extending the existing
formal model of the hypervisor \cite{SOFSEM} with the formal specification of the
monitor's run-time checks. 

The paper is organized as follows: Section~\ref{sec:background}
introduces the target CPU architecture (ARMv7A), the architecture of
the existing hypervisor and its interactions with the hosted Linux
kernel, the threat model and the existing formal models;
Section~\ref{sec:design} describes the MProsper architecture and
design, it also elaborates on the additional software required to host
Linux;
Section~\ref{sec:model} describes the formal model of the monitor and
formally states the top level goal: absence of code injection;
Section~\ref{sec:verification} presents the verification strategy,
by summarizing the proofs that have been implemented in HOL4;
Section~\ref{sec:eval} demonstrates the overhead of MProsper through
standard microbenchmarks, it also presents measures of the code
and proof bases; finally, Sections~\ref{sec:related} and~\ref{sec:conclusions} present
the related work and the concluding remarks.

\newcommand{\msound}{\mathit{sound}}
\newcommand{\mcontent}{\mathit{content}}
\newcommand{\mmap}{\mathit{map}}
\newcommand{\munmap}{\mathit{unmap}}
\newcommand{\mfree}{\mathit{free}}
\newcommand{\mcreate}{\mathit{create}}
\newcommand{\mswitch}{\mathit{switch}}
\newcommand{\mvalidate}{\mathit{validate}}
\newcommand{\mWS}{\mathit{WS}}
\section{Background}\label{sec:background}
\subsection{The Prosper Hypervisor}
The Prosper hypervisor supports
the execution of an untrusted Linux guest~\cite{SOFSEM} along with 
several trusted components.
The hosted Linux is paravirtualized; both applications and kernel are
executed unprivileged (in user mode)
while privileged operations are delegated to the hypervisor, which is
invoked via hypercalls. 
The physical memory region allocated to each component
is statically defined. The hypervisor guarantees spatial isolation of the
hosted components; a component can not directly affect (or be affected by)  the
content of the memory regions allocated to other components.
Thus, the interactions among the hosted components are possible only
via controlled communication channels, which are supervised by the
hypervisor.

The Prosper hypervisor and the MProsper monitor target the ARMv7-A
architecture, which is the most widely adopted instruction set architecture in
mobile computing.
In ARMv7-A, the virtual memory is configured via page tables that reside in
physical memory. The architecture provides two levels of page
tables, in the following called L1s and L2s.
These tables represent the configuration of the
Memory Management Unit (MMU) and define the access permissions to the
virtual memory.
As is common among modern architectures, the entries of ARMv7 page tables
 support the NX (No eXecute) attribute: an instruction can be
 executed only if it is fetched from 
a virtual memory area whose NX bit is not set.
Therefore, the system executable code is a subset of the content
of the physical blocks that have at least an executable virtual mapping.

To isolate the components, the hypervisor takes control of the MMU and configures
 the pagetables so that no illicit access is possible. 
This MMU configuration can not be static; 
the hosted Linux must be able to reconfigure the layout of its own
memory (and the memory of the user programs). 
For this reason the hypervisor virtualizes the memory subsystem.
This virtualization consists of a set of APIs 
that enable Linux to request the creation/deletion/modification of a
page table and to switch the one currently used by the MMU.

Similarly to Xen~\cite{xen}, the virtualization of the memory subsystem 
is accomplished by direct paging.
Direct paging allows the guest to allocate the page tables inside its
own memory and to directly manipulate them while the tables are not in active use by the MMU.
Once the page tables are activated, the hypervisor must guarantee that
further updates are possible only via the virtualization API.

The physical memory is fragmented into  blocks of 4 KB. Thus, a 32-bit
architecture has $2^{20}$ physical blocks.
We assign a type to each physical block, that can be:  \textit{data}: the block can be written by the guest,
$\LOneType$: contains part of an $\LOneType$ page table and should not be
writable by the guest, 
$\LTwoType$: contains four $\LTwoType$ page tables and should not be
writable by the guest.
We call the $\LOneType$ and $\LTwoType$ blocks ``potential'' page
tables, since the hypervisor allows to select only these memory areas
to be used as page tables by the MMU.

\begin{table}[t]
\begin{center}
\begin{tabular}{|l|l|}
        \hline
  \textbf{request $r$} & \textbf{DMMU behavior}\\
  \hline
  $\mswitch(bl)$ & makes block $bl$ the active page table\\
  \hline
  $\mfree_{\LOneType}(bl)$ and $\mfree_{\LTwoType}(bl)$ & frees block $bl$, by
  setting its type to $\DataType$\\
  \hline
  $\munmap_{\LOneType}(bl, idx)$,  $\munmap_{\LTwoType}(bl, idx)$ &
 unmaps entry $idx$ of the page table\\
 & stored in block $bl$\\
\hline
$\mathit{link}_{\LOneType}(bl, idx, bl')$
&
maps entry $idx$ of block $bl$ to point the \\
&  L2 stored in $bl'$
\\
  \hline
  $\mmap_{\LTwoType}(bl, idx, bl', ex, wt, rd)$  and
  & map entry $idx$ of block $\ bl\ $  to point to\\
  $\mmap_{\LOneType}(bl, idx, bl', ex, wt, rd)$    
  & block $bl'$ and  granting rights $ex, wt, rd$\\
  & to user mode
  \\
  \hline
  $\mcreate_{\LTwoType}(bl)$  and $\mcreate_{\LOneType}(bl)$
  &
makes block $bl$ a potential $\ \LTwoType/\LOneType$, by\\
&setting its type to $\LTwoType/\LOneType$
\\
\hline
\end{tabular}
\caption{DMMU API}
\label{tbl:dmmu}
\end{center}
\end{table}

Table~\ref{tbl:dmmu} summarizes the APIs that
manipulate the page tables. The set of these functions is called
DMMU. 
Each function validates the page type,
guaranteeing that page tables are write-protected.
A naive run-time check of the page-type policy is not efficient, since
it requires to re-validate the L1 page table whenever the
\textit{switch} hypercall is invoked.
To efficiently enforce that only blocks typed $\DataType$
can be written by the guest the hypervisor maintains a reference counter,
which tracks for each block the sum of descriptors
providing access in user mode to the block.
The intuition is that a hypercall can change the type of a
physical block (e.g. allocate or free a page table) only if the
corresponding reference counter is zero.

\begin{figure}
\center
\includegraphics[width=0.5\linewidth]{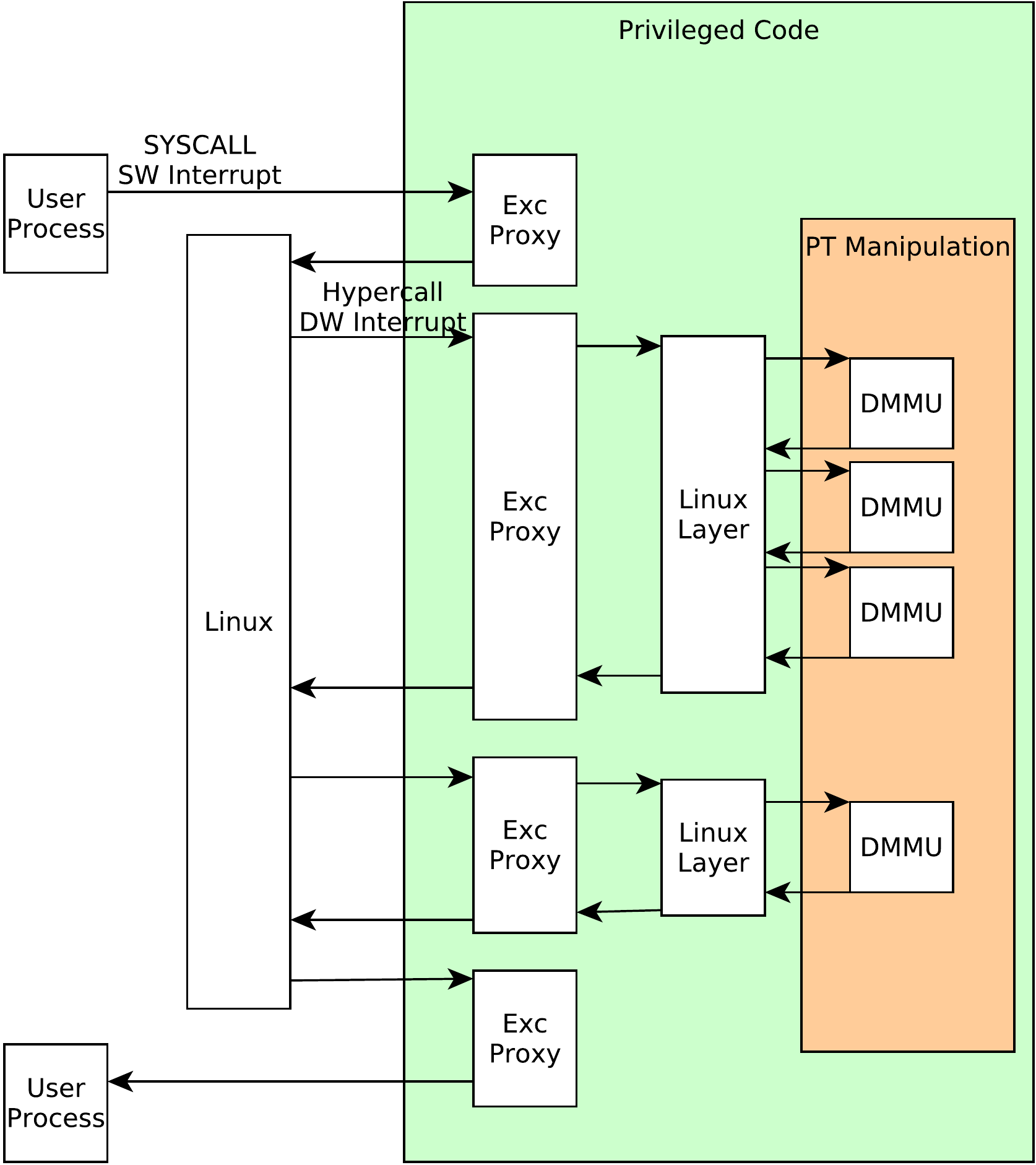}
\caption{Hypervisor Architecture}\label{fig:architecture:hyper}
\end{figure}

A high level view of the hypervisor architecture is depicted in
Fig.~\ref{fig:architecture:hyper}.
The hypervisor is the only component that is executed in
privileged mode. It logically consists of three layers: (i)
an interface layer (e.g. the exception handlers) that is independent
from the hosted software, (ii) a Linux specific layer and (iii) a
critical core (i.e. the DMMU), which is the only component that manipulates the
sensible resources (i.e. the page tables).
Fig.~\ref{fig:architecture:hyper} demonstrates
the behavior of the system when a user process in the Linux guest spawns a new process.

This design has two main benefits: (i) the critical
part of the hypervisor is small and does not depend on the
hosted software and (ii) the Linux-specific layer enriches the
expressiveness of the hypercalls, thus reducing the number of context
switches between the hypervisor and the Linux kernel.
From a verification point of view, to guarantee security of the complete
system it is not necessary to verify functional correctness of the Linux
layer; it suffices to verify that this layer never changes directly the
sensitive resources and that its execution does not depend on
the sensitive data. These tasks can be accomplished using
sand-boxing techniques~\cite{sehr2010adapting} or 
tools for information flow analysis~\cite{DBLP:conf/ccs/BalliuDG14}.

\subsection{The Attack Model}
The Linux guest is not trusted, thus we take into account an 
attacker that has complete control of the partition that hosts 
Linux. The attacker can force user programs and the Linux kernel to follow
arbitrary flows and to use arbitrary data. 
The attacker can invoke the hypervisor, including the DMMU API, through software interrupts and exceptions.
Other transitions into privileged memory are prevented by the hypervisor.  
The goal of the attacker is ``code injection'', for example using a buffer overflow to inject malicious
executable code. This attack is normally performed
by a malicious software that is able to write code into a data storage area of
another process, and then cause this code to be executed.

In this paper we exemplify our monitor infrastructure using code signing.
Signing the system code is a widely used approach to confirm
the software author and guarantee (computationally speaking) that the
code has not 
been altered or corrupted, by use of a cryptographic hash.
Many existing code signing systems rely on a public key
infrastructure (PKI) to provide both code authenticity and integrity.
Here we use code signing to define integrity of the system code:
integrity of an executable physical block stands for the block having
 a valid signature.
Similarly, the integrity of the system code depends on the integrity of 
all executable physical blocks.
The valid signatures are assumed to be known by the runtime monitor.
We refer to this information as the ``golden image'' (GI) and it is held
by the monitor. 

In order to make injected code detectable, we also assume that 
the attacker is computationally bound; it can not modify the
injected code to make its signature compliant with the golden image.
We stress that our goal is not to demonstrate the security properties
of a specific signature scheme. In fact the monitor can be equipped
with an arbitrary signature mechanism and the signature mechanism
itself is just one of the possible approaches that can be used to
check integrity of the system code. For this reason we do not
elaborate further on the computational power of the attacker.

\subsection{Formal Model of the Hypervisor}\label{sec:hypermodel}
Our formal model is built on top of the existing HOL4 model for
ARMv7~\cite{DBLP:conf/itp/FoxM10}. This has
been extended with a detailed formalization of the ARMv7 MMU, so that every
memory access uses virtual addresses and respects the constraints
imposed by the page tables.

An ARMv7 state is a record $\sigma =
\tuple{\regs,\coregs,\mem}\in\Sigma$, where $\regs$, $\coregs$ and $\mem$, respectively, represent the registers, coprocessors and memory. 
In the  state $\sigma$, the function $\mode(\sigma)$ determines the current privilege execution mode, 
which  can be either $\UserMode$ (user mode, used by Linux and the
monitor) or
$\KernelMode$ (privileged mode, used by the hypervisor). 

The system behavior is modeled by a state transition relation $\xrightarrow{l
  \in \{\UserMode, \KernelMode\}} \subseteq \RealStateSpace \times
\RealStateSpace$, representing
the complete execution of a single ARM instruction. 
Non-privileged transitions  ($\ArmStateVar \xrightarrow{\UserMode} \ArmStateVar'$) start and end in $\UserMode$ states.
All the other transitions ($\ArmStateVar \xrightarrow{\KernelMode} \ArmStateVar'$) involve at least one state in privileged level.
A transition from $\UserMode$ to $\KernelMode$ is done by raising
an exception, that can be caused by software interrupts,
illegitimate memory accesses, and hardware interrupts.

The transition relation queries the MMU to translate the virtual
addresses and to check the access permissions.
The MMU is represented by the function $mmu(\ArmStateVar, PL, va, \accreq)
\rightarrow pa \cup \{\bot\}$: it takes the state
$\ArmStateVar$, a privilege level $PL$, a virtual address
$va \in 2^{32}$ and the requested access right $\accreq \in \{rd, wt,
ex\}$,
for \textit{readable}, \textit{writable}, and \textit{executable} in
non-privileged respectively, and returns either the corresponding physical address $pa \in
2^{32}$ (if the access is granted) or a fault ($\bot$).

In~\cite{SOFSEM} we show that a system hosting the hypervisor
resembles the following abstract model.
A system state is modeled by a tuple $\tuple{\ArmStateVar, \HyperStateVar}$, consisting of an ARMv7 state $\ArmStateVar$  and an abstract hypervisor state $\HyperStateVar$,
of the form $\tuple{\PgTypeVar, \PgRefVar_{ex},  \PgRefVar_{wt}}$.
Let $\LabelBlk \in 2^{20}$ be the index of a physical block and $t \in
\{D,L1,L2\}$, 
$\PgTypeVar \vdash \LabelBlk : t$ tracks the type of the
block and  $\PgRefVar_{ex}(\LabelBlk), \PgRefVar_{wt}(\LabelBlk) \in
2^{30}$ track the reference counters: the
number of page tables entries (i.e. entries of
physical blocks typed either $L1$ or $L2$) that map to the physical
block $\LabelBlk$  and are executable or writable respectively.

The transition relation for this model is 
$\tuple{\sigma, h} \xrightarrow{\alpha} \tuple{\sigma', h'}$, where
$\alpha \in \{0,1\}$, and is
defined by the following inference rules:
\begin{itemize}
  \item if 
    $\sigma \xrightarrow{\UserMode} \sigma'$
    then 
    $\tuple{\sigma, h} \xrightarrow{0} \tuple{\sigma', h}$
    ; instructions executed in non-privileged mode that do not
    raise exceptions behave equivalently to the standard ARMv7
    semantics and do not affect the abstract hypervisor state.
  \item if $\sigma \xrightarrow{\KernelMode} \sigma'$ 
    then $\tuple{\sigma, h} \xrightarrow{1} H_r(\tuple{\sigma', h})$,
    where $r=req(\sigma')$;
    whenever an
    exception is raised, the hypervisor is
    invoked through a hypercall, and the reached state is resulting
    from the execution  of the handler $H_r$
\end{itemize}
Here, $req$ is a function that models the hypercall calling
conventions; the target hypercall is identified by
the first register of  $\sigma$, and the other registers provide the hypercall's
arguments. The handlers $H_r$ formally model the behavior of the
memory virtualization  APIs of the hypervisor (see Table~\ref{tbl:dmmu}).

Intuitively, guaranteeing spatial isolation means confining the guest
to manage a part of the physical memory available for the guest uses. 
In our setting, this part is determined statically and identified by the predicate
$G_m(bl)$, which holds if the physical block $bl$ is part of the
physical memory assigned to the guest partition.
Clearly, no security property can be guaranteed if the system
starts from a non-consistent state; for example the guest can not be allowed to change the MMU behavior by directly writing the page tables.
For this reason we introduce a system invariant $I_H(\tuple{\sigma,  h})$ that is used to constrain the set of consistent initial states. Then the hypervisor guarantees that the invariant is preserved by every transition:
\begin{preposition}\label{prep:hypervisor:invariant}
  Let $I_H(\tuple{\sigma,h})$.  If 
  $\tuple{\sigma,h} \xrightarrow{i} \tuple{\sigma',h'}$
  then
  $I_H(\tuple{\sigma',h'})$.
\end{preposition}

We use the function $\mcontent: \Sigma \times 2^{20} \rightarrow
2^{4096*8}$  that returns the content of a physical block in a system
state as a value of $4$~KB.
Proposition~\ref{prep:hypervisor} summarizes some of the security
properties verified in~\cite{SOFSEM}:
the untrusted guest can not directly change 
(1) the memory allocated to the other components,
(2) physical blocks that contain potential page tables,
(3) physical blocks whose writable reference counter is zero and
(4) the behavior of the MMU.

\begin{preposition}\label{prep:hypervisor}
  Let $I_H(\tuple{\sigma,(\tau, \rho_{wt}, \rho_{ex})})$.
  If 
  $\tuple{\sigma,(\tau, \rho_{wt}, \rho_{ex})} \xrightarrow{0} \tuple{\sigma',h'}$
  then:
\begin{itemize}[] 
   \item \label{prep:hypervisor:guesttblocks} 
-- For every $bl$ such that $\neg G_m(bl)$ then $content(bl, \sigma) = content(bl, \sigma')$
          \item \label{prep:hypervisor:ptblocks}
-- For every $bl$ such that $\tau(bl) \neq \DataType$ then
  $\mcontent(bl, \sigma) = \mcontent(bl, \sigma')$
  \item \label{prep:hypervisor:wtblocks}
-- For every $bl$ if $\mcontent(bl, \sigma) \neq \mcontent(bl, \sigma')$ then
  $\rho_{wt}(bl) > 0$
  \item \label{prep:hypervisor:mmu-safe}
-- For every $va,PL,acc$  we have $mmu(\sigma, va, PL, acc) = mmu(\sigma', va, PL, acc)$
\end{itemize}

\end{preposition}

\section{Design}\label{sec:design}

We configured the hypervisor to support the interaction protocol
of Figure~\ref{fig:protocol}; the monitor mediates accesses
to the DMMU layer.
Since the  hypervisor supervises the changes of the page tables
the monitor is able to intercept all modifications to the memory layout.
This makes the monitor able to know if a physical block is writable: This is
the case 
if there exists at least one virtual mapping pointing to the block with a
guest writable access permission. Similarly it is
possible to know if a physical block is executable.
Note that the identification of the executable code (also called
``working set'') does not rely
on any information provided by the untrusted guest. Instead, the monitor
only depends on HW information, which can not be tampered by an attacker.

\begin{figure}
  \begin{boxedminipage}{\textwidth}
    \begin{enumerate}
    \item For each DMMU hypercall invoked by a guest, the hypervisor forwards the hypercall's request to the monitor. 
    \item The monitor validates the request based on its validation mechanism.
    \item The monitor reports to the hypervisor the result of the hypercall validation. 
    \end{enumerate}
  \end{boxedminipage}
  \caption{The interaction protocol between the Prosper hypervisor and the monitor}
  \label{fig:protocol}
\end{figure}

The first policy enforced by the monitor is code signature:
Whenever Linux requests to change a page table (i.e. causing to change
the domain of the working set) the monitor (i) identifies
the physical blocks that can be made executable by the request, (ii)
computes the block signature and (iii) compares the result with the
content of the golden image.
This policy is sufficient to prevent code injection
that are caused by changes of the memory layout setting, due to 
the  hypervisor forwarding to the monitor all requests to change 
the page tables.

However, this policy is not sufficient to guarantee integrity of the
working set. In fact, operations that modify the content of a physical
block that is executable can violate the integrity of the executable code.
These operations cannot be intercepted by the monitor, since they are not
supposed to raise any hypercall. In fact, a simple write operation in
a block typed $\DataType$
does not require the hypervisor intermediation since no
modification of the memory layout is introduced. To prevent code
injections performed 
 by writing malicious code in an executable area of the memory, the
 monitor enforces the executable
space protection policy $W \oplus X$, preventing physical blocks from being simultaneously writable and executable.
As for the hypervisor, a naive run-time check of the executable
space protection is not efficient. Instead, we reuse the hypervisor
reference counters: we accept a hypercall that makes a block
executable (writable) only if the writable (executable) reference counter
of the block is zero.

An additional complication comes from the Linux architecture. An
unmodified Linux kernel will not survive the policies enforced by
the monitor, thus its execution will inevitably fail. For example,
when a user process is running there are at least two virtual memory
regions that are mapped to the same physical memory where the process
executable resides: (i) the user ``text segment'' and (ii) the
``kernel space''  (which is an injective map to the whole physical
memory).
When the process is created, Linux requests to set the text
segment as executable and non writable. However, Linux does not
revoke its right to write inside this memory area using its kernel
space. This
setting is not accepted by the monitor, since it violates $X
 \oplus W$, thus making it impossible to execute a user process.

Instead
of adapting a specific Linux kernel we decided to implement
a small emulation layer that has two functionalities:
\begin{itemize}
  \item It proxies all requests from the Linux layer to the
    monitor. If the emulator receives a request that can be rejected
    by the monitor (e.g. a request setting as writable a memory region
    that is currently executable) then the emulator (i) downgrades the
    access rights of the request (e.g. setting them as non writable)
    and (ii) stores the information about the  suspended right  in a 
    private table.
  \item It proxies all data and prefetch aborts. The monitor looks up
    in the private table to identify if the abort is due
    to an access right that has been previously downgraded by the
    emulator. In this case the monitor attempts (i) to downgrade the
    existing mapping that conflicts with the suspended access right
    and (ii) to re-enable the suspended access right.
\end{itemize}
Note that a malfunction of the emulation layer does not affect the
security of the monitor. Namely, we do not care if the emulation
layer is functionally correct, but only that it does not access
sensible resources directly. 

\begin{figure}[t]
\center
\includegraphics[width=0.5\linewidth]{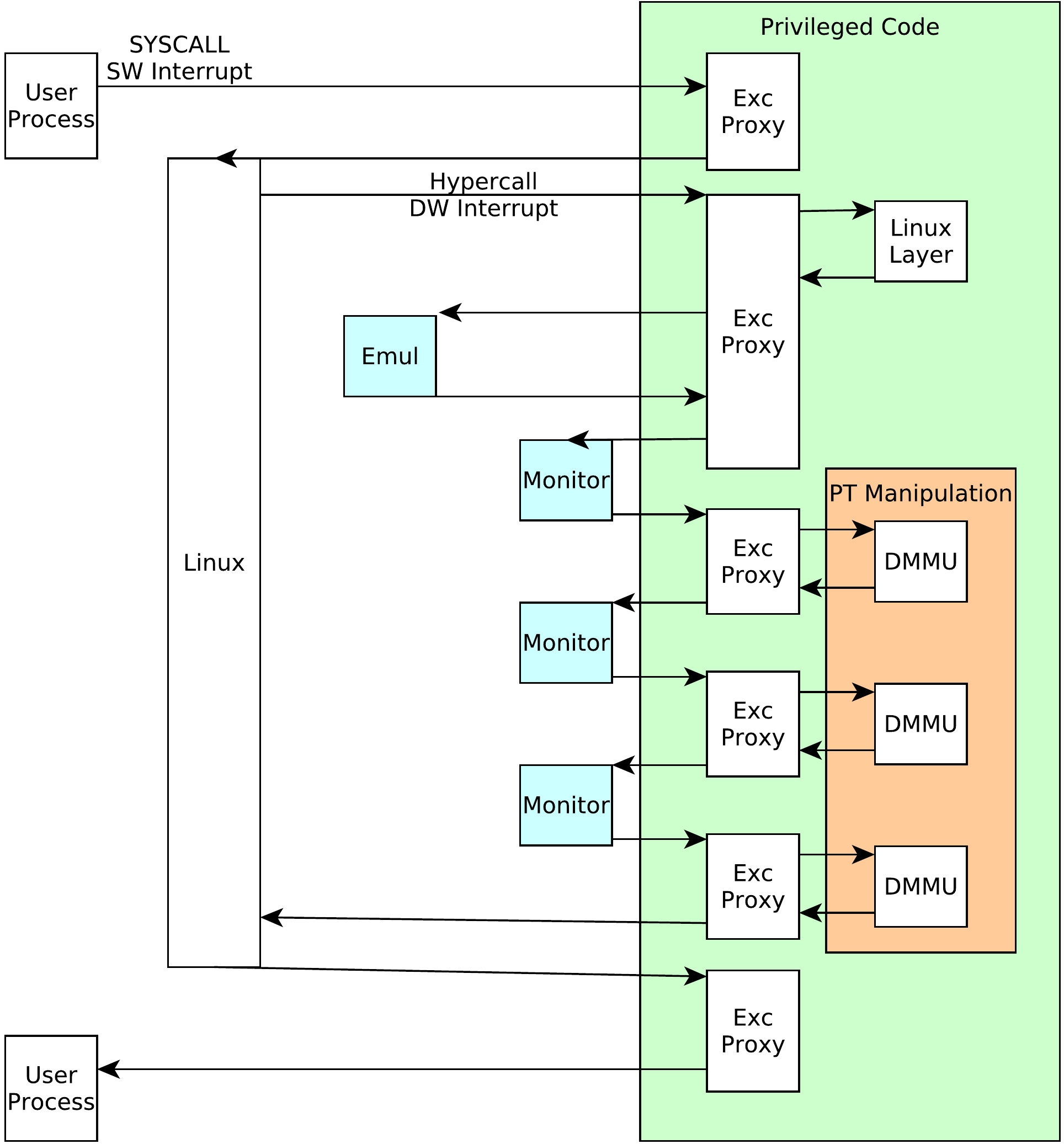}
\caption{MProsper's Architecture}\label{fig:architecture:final}
\end{figure}

Fig.~\ref{fig:architecture:final} depicts the architecture of
MProsper.
Both the runtime monitor and the emulator are deployed as two guests
of the Prosper hypervisor. 
The Linux layer prepares a list of requests in  a
buffer shared with the emulation guest.
After the Linux layer returns, the hypervisor
activates the emulation guest, which manipulates the requests (or adds
new ones) as discussed before. Then the hypervisor iteratively asks
the monitor to validate one of the pending requests and upon success it
commits the request by invoking the corresponding DMMU function.

Using a dedicated guest on top
of the hypervisor permits to decouple the enforcement of the security policies
from the other hypervisor functionalities, thus keeping the
trusted computing base minimal.
Moreover, having the security policy wrapped inside a guest supports both
the tamper-resistance and the trustworthiness of the monitor. In fact, the monitor can take advantage from the isolation properties
provided by the hypervisor. This avoids malicious
interferences coming from the other guests (for example from a process of
an OS running on a different partition of the same machine). Finally,
decoupling the run-time security policy from the other functionalities of the
hypervisor makes the formal specification and verification of the monitor
more affordable.

\newcommand{\integrity}{\mathit{integrity}}

\section{Formal Model of MProsper}\label{sec:model}
The formal model of the system (i.e. consisting of the hypervisor,
the monitor and the untrusted Linux) is built on top of the models
presented in Section~\ref{sec:hypermodel}.
Here we
leave unspecified the algorithm used to sign and check
signatures, so that our results can be used for different intrusion
detection mechanisms. 
The golden image $GI$ is a finite set of signatures $\{s_1,
\dots, s_n\}$, where the signatures are selected from a domain $S$.
We assume the existence of a function $sig : 2^{4096*8}
\rightarrow S$ that computes the signature of the content of a block.
The system behavior is modeled by the following rules:

\begin{enumerate}
 \item 
$\begin{array}{cc}
  \frac{\tuple{\sigma, h} \xrightarrow{0} \tuple{\sigma', h'}}
       {\tuple{\sigma, h, GI} \xrightarrow{0} \tuple{\sigma', h', GI}}
\end{array}$

 \item
$\begin{array}{cc}
  \frac{\tuple{\sigma, h} \xrightarrow{1} \tuple{\sigma', h'}
       \ \mvalidate(req(\tuple{\sigma, h}), \tuple{\sigma, h, GI})
       }
       {\tuple{\sigma, h, GI} \xrightarrow{1} \tuple{\sigma', h', GI}}
\end{array}$
       
 \item
$\begin{array}{cc}
  \frac{\tuple{\sigma, h} \xrightarrow{1} \tuple{\sigma', h'}
       \ \neg \mvalidate(req(\tuple{\sigma, h}), \tuple{\sigma, h, GI})
       }
       {\tuple{\sigma, h, GI} \xrightarrow{1} \epsilon(\tuple{\sigma, h, GI})}
\end{array}$
\end{enumerate}

User mode transitions (e.g. Linux activities)  require neither
hypervisor nor  monitor intermediation.
Proposition~\ref{prep:hypervisor:guesttblocks} justifies the fact
that, by construction, the transitions executed by the untrusted component can not 
affect the monitor state; (i) the golden image is constant and (ii)
the monitor code can be statically identified and abstractly modeled.
Executions in privileged mode require monitor 
intermediation. If the monitor validates the request,
then the standard behavior of the hypervisor is executed. Otherwise
the hypervisor performs a special
operation to reject the request, by reaching the state that is returned by a function $\epsilon$.
Hereafter, the function 
$\epsilon$ is assumed to be the identity. 
Alternatively, 
$\epsilon$ can transform the state so that the requestor is informed
about the rejected operation, by updating the user registers according
to the desired calling convention.

The function $\mvalidate(req(\tuple{\sigma, h}), \tuple{\sigma, h, GI})$
represents the validation mechanism of the monitor, which checks at
run-time possible violations of the security policies. In Table~\ref{tbl:monitor:sec-policies} we briefly
summarize the policies for the different access requests.
Here, $PT$ is a function that
yields the list of mappings granted by a page table,
where each mapping is a tuple $(vb,pb,wt,ex)$ containing the virtual block
mapped ($vb$), the pointed physical block ($pb$) and the unpriviledged
rights to execute ($ex$) and write ($wt$).
The rules in Table~\ref{tbl:monitor:sec-policies} are deliberately
more abstract than the ones modeled in HOL4 and are used to
intuitively present the behavior of the monitor. For example, the
function $PT$ is part of the hardware model and is not explicitly used
by the monitor code, it is instead more similar to an iterative
program. This makes our verification more difficult, but it also makes the monitor
model as near as possible to the actual implementation, enabling further
verification efforts that can establish correctness of the implementation.

Note that the monitor always checks that a mapping is not writable and
executable simultaneously.
Furthermore, if a mapping grants a writable access then the executable
reference counter of the pointed physical block must
be zero, guaranteeing that this mapping does not conflict (according
with the executable space protection policy) with any other allocated
page table. Similarly, if a mapping grants an executable access,
then the writable reference counter of the pointed block must be zero.

\begin{table}[t]
\begin{tabular}{|l|l|}
        \hline
\textbf{request 
  $r$} & 
\textbf{$\mvalidate(r, \tuple{\sigma, (\tau, \rho_{wt}, \rho_{ex}), GI})$ holds iff}\\
  \hline
  $\mswitch(bl)$ & always\\
  \hline
  $\mfree_{L1}(bl)$ and $\mfree_{L2}(bl)$ & always\\
  \hline
  $\munmap_{L1}(bl, idx)$, $\munmap_{L2}(bl, idx)$, &  $\rho_{ex}(bl) =  0$ \\ and $link_{L1}(bl, idx, bl')$ &
\\
  \hline
  $\mmap_{L2}(bl, idx, bl', ex, wt, rd)$  and
  & $\msound_{W \oplus X}(wt,ex,\rho_{wt}, \rho_{ex}, bl')\ $ and \\
  $\mmap_{L1}(bl, idx, bl', ex, wt, rd)$ 
  & $\msound_{S}(ex,bl',\sigma, GI)
  \wedge \rho_{ex}(bl) = 0$
  \\
  \hline
  $\mcreate_{L2}(bl)$  and $\mcreate_{L1}(bl)$
  & $\forall (vb,pb,wt,ex) \in PT(\mcontent(bl, \sigma))$. \\
  & $\>\>\>\msound_{W \oplus X}(wt,ex,\rho_{wt}, \rho_{ex}, pb)$ and \\
  & $\>\>\>\msound_{S}(ex,pb,\sigma, GI)$\\ 
  & \\
  & $\forall (vb',pb',wt',ex') \in PT(\mcontent(bl, \sigma)).$ \\
  & $\>\>\>\textit{no-conflict}(vb,pb,wt,ex)(vb',pb',wt',ex')$

\\
\hline
\multicolumn{2}{l}{
 where}\\
\multicolumn{2}{l}{
  $\msound_{W \oplus X}(wt,ex,\rho_{wt}, \rho_{ex}, bl) = 
  \left (
  \begin{array}{l}
(ex \Rightarrow \neg wt \wedge \rho_{wt}(bl) = 0)\ \wedge\ \\
(wt \Rightarrow \neg ex \wedge \rho_{ex}(bl) = 0)
\end{array}
\right )
 $
\vspace{1em}
}
 \\
\multicolumn{2}{l}{
  $\msound_{S}(ex,bl,\sigma, GI) = 
(ex \Rightarrow \integrity(GI,bl,\mcontent(bl,\sigma)))
 $
 \vspace{1em}
}
 \\
\multicolumn{2}{l}{
  $\textit{no-conflict}(vb,pb,wt,ex)(vb',pb',wt',ex') = 
\left (
\begin{array}{l}
(vb \neq vb' \wedge pb = pb') \Rightarrow \\
(ex \Rightarrow \neg wt' \wedge wt \Rightarrow \neg ex')
\end{array}
\right )
 $
}
\end{tabular}
\caption{Security policies for the available access requests}
\label{tbl:monitor:sec-policies}
\end{table}

To formalize the top goal of our verification
we introduce some auxiliary notations. 
The working set identifies the physical blocks that host
executable binaries and their corresponding content.
\begin{definition}
  Let $\sigma$ be a machine state. The working set of $\sigma$ is defined as
\[
  \mWS(\sigma) = \{\tuple{bl, content(bl, \sigma)} \mid
  \exists pa,va. mmu(\sigma, \UserMode, va, ex) = pa \wedge pa \in bl
  \}
\]
\end{definition}

By using a code signing approach, we say that the integrity of a
physical block is satisfied if the signature of the block's content belongs
to the golden image.
\begin{definition}
  Let $cnt \in 2^{4096*8}$ be the 4KB content of a physical block $bl$
   and $GI$ be the golden image. Then $ \integrity(GI, bl, cnt)$ if, and
only if, $sig(bl, cnt) \in GI$
\end{definition}
Notice that our security property can be refined to fit different
anti-intrusion mechanisms. For example, $\integrity(GI, bl, cnt)$ can be
instantiated with the execution of an anti-virus scanner.

The system state is free of malicious code injection if
the signature check is satisfied for the whole executable code. That is:
\begin{definition}
  Let $\sigma$ be a machine state, $bl$ be a physical block and $GI$
  be the golden image. Then $\integrity(GI, \sigma)$ if, only only if, for all $ \tuple{bl,cnt} \in \mWS(\sigma)$, $\integrity(GI, bl, cnt)$
\end{definition}

Finally, we present our top level proof goal: No code injection can succeed.
\begin{preposition}
  If $\tuple{\sigma, h, GI}$ is a state reachable from the initial
  state of the system $\tuple{\sigma_0, h_0, GI}$ then 
  $\integrity(GI, \sigma)$
\end{preposition}

\section{Verification Strategy}\label{sec:verification}

Our verification strategy consists of introducing a state invariant $I(s)$ that is preserved by any possible transition and demonstrating that the invariant
guarantees the desired security properties.

\begin{definition}\label{def:monitor:invariant}
$I(\sigma, (\tau, \rho_{wt}, \rho_{ex}), GI)$ holds if
\[
\begin{array}{l}
I_H(\sigma, (\tau, \rho_{wt}, \rho_{ex})) \wedge 
\mbox{} \\
\forall\ bl\ .\ (\neg (\tau(bl) = \DataType)) \Rightarrow
\forall\ (vb,pb,wt,ex) \in PT(\mcontent(bl, \sigma)) .\\
\hskip0.5cm 
\msound_{W \oplus X}(wt,ex,\rho_{wt}, \rho_{ex}, pb) \wedge 
\msound_{S}(ex,pb,\sigma, GI)
\end{array}
\]
\end{definition}

Clearly, the soundness of the monitor depends on the soundness of the
hypervisor, thus $I$ requires that the hypervisor's invariant $I_H$
holds.
Notice that the invariant constrains not only the page tables
currently in use, but it constrains all potential page tables, which
are all the blocks that have type different from $\DataType$. This allows to
speed up the context switch, since the guest simply re-activates a page
table that has been previously validated.
Technically, the invariant guarantees 
protection of the memory that can be potentially executable 
and the correctness of the corresponding signatures.

We verified independently that the invariant is preserved by
unprivileged transitions (Theorem~\ref{lem:monitor:invariant:user})
and by privileged transitions
(Theorem~\ref{lem:monitor:invariant:hyper}). Moreover, 
Lemma~\ref{lem:monitor:invariant:safe} demonstrates that the
monitor invariant guarantees there is no malicious content in the
executable memory.
\begin{lemma}
  \label{lem:monitor:invariant:safe}
  If $I(\tuple{\sigma, (\tau, \rho_{wt}, \rho_{ex}), GI})$ then
  $\integrity(GI, \sigma)$.
\end{lemma}
\begin{proof}
The proof is straightforward, following from $\msound_{S}$ of every
block that can be executable according with an arbitrary potential page table.
\end{proof}

Theorem~\ref{lem:monitor:invariant:user} demonstrates that the
invariant is preserved by instructions executed by the untrusted
Linux. This depends on Lemma~\ref{lem:monitor:invariant:user-ws},
which shows that the invariant forbids user transitions to change the
content of the memory that is executable.
\begin{lemma}
  \label{lem:monitor:invariant:user-ws}
  Let $\tuple{\sigma, (\tau, \rho_{wt}, \rho_{ex}), GI} \xrightarrow{0} \tuple{\sigma', h', GI'}$ and $I(\tuple{\sigma, h, GI})$ then
\[
\forall bl\ .\ ( \neg (\tau(bl) = \DataType)) \Rightarrow
\left (
\begin{array}{l}
  PT(\mcontent(bl, \sigma')) = PT(\mcontent(bl, \sigma)) \wedge\\
  \forall (vb,pb,wt,ex) \in PT(\mcontent(bl, \sigma'))\ .\\
  \hskip1cm
  (ex \Rightarrow \mcontent(pb, \sigma) = \mcontent(pb, \sigma'))
\end{array}
\right )
\]
\end{lemma}
\begin{proof}
 Proof is straightforward and we split it in two parts. First we show that page tables remain constant after user transitions and then 
 we prove that executable block cannot be changed by the user.
 From Proposition~\ref{prep:hypervisor} we know since the hypervisor invariant holds in the state $\tuple{\sigma, h, GI}$, user transitions are not allowed
 to change the page tables. Thus the user transitions preserve mappings of all page tables in the memory.
 
 By the monitor's invariant we know that in $\tuple{\sigma, h, GI}$ all the mappings in page tables comply with the policy $sound_{W \oplus X}$. That is, 
 if a mapping in a page table grants executable permissions on block to the user, the block must be write protected against unprivileged accesses. Moreover,
 we know that page tables remain the same after an unprivileged transition and the user cannot change his access permission. This proves that the content of
 executable block is not modifiable by unprivileged transitions.
\end{proof}

\begin{theorem}
  \label{lem:monitor:invariant:user}
  If $\tuple{\sigma, h, GI} \xrightarrow{0}
  \tuple{\sigma', h', GI'}$ and $I(\tuple{\sigma, h, GI})$ then $I(\tuple{\sigma', h',GI'})$.
\end{theorem}
\begin{proof}
From the inference rules we know that $h'=h$, $GI'=GI$ and that
the system without the monitor behaves as
$\tuple{\sigma, h} \xrightarrow{0} \tuple{\sigma', h}$.
Thus, Proposition~\ref{prep:hypervisor:invariant} can be used to
guarantee that the hypervisor invariant is preserved ($I_H(\sigma',
h')$).

If the second part of the invariant is violated then there must exist
a mapping in one (hereafter $bl$) of the allocated page tables that is compliant
with the executable space protection policy in $\sigma$ and 
violates the policy in $\sigma'$. Namely, $\mcontent(bl, \sigma')$
must be different from $\mcontent(bl, \sigma)$. This contradicts
Proposition~\ref{prep:hypervisor:ptblocks}, since the type of
the changed block is not data ($\tau(bl) \neq
\DataType$).

Finally we must demonstrate that every potentially executable block
contains a sound binary. Lemma~\ref{lem:monitor:invariant:user-ws}
guarantees that the blocks that are potentially executable are the
same in $\sigma$ and $\sigma'$ and that the content of these blocks is
unchanged. Thus is sufficient to use the invariant $I(\sigma, h, GI)$, to
demonstrate that the signatures of all executable blocks are correct.
\end{proof}

To demonstrate the functional correctness of the monitor
(Theorem~\ref{lem:monitor:invariant:hyper} i.e. that the invariant is
preserved by privileged transitions) we introduce two auxiliary lemmas:
Lemma~\ref{lem:monitor:invariant:hyper-ws} shows that the monitor
correctly checks the signature of pages that are made executable.
Lemma~\ref{lem:monitor:invariant:hyper-helper2} expresses that
executable space protection is preserved for all hypervisor data
changes, as long as a block whose reference counter
(e.g. writable; $\rho'_{wt}$) becomes non zero has the other reference counter
(e.g. executable; $\rho_{ex}$) zero. 

\begin{lemma}
  \label{lem:monitor:invariant:hyper-ws}
If $\tuple{\sigma, h, GI} \xrightarrow{1}
  \tuple{\sigma', (\tau', \rho'_{wt}, \rho'_{ex}), GI'}$ and $I(\tuple{\sigma, h, GI})$ then for all $bl$,
  $ \tau'(bl) \neq \DataType \Rightarrow
  \forall (vb',pb',wt,ex) \in PT(\mcontent(bl, \sigma')). \
  \msound_{S}(ex,pb',\sigma', GI)$.
\end{lemma}
\begin{proof}
 Proof is done by case analysis on the type of the performed operation and its the validity. If the requested operation is not valid then the proof is trivial, since
 $\epsilon$ is the identity function and it yields the same state. Therefore, $\msound_{S}$ holds.
 
 However, if the request is validated by th monitor and it is committed by the hypervisor we prove Lemma~\ref{lem:monitor:invariant:hyper-ws} by doing case analysis on the type of
 the performed operation. Let assume that the requested operation is $map_{L2}$. 
 \begin{itemize}
  \item[(i)] If the requested mapping does not set the executable flag $ex$, the working set would not be changed by committing the request $\mWS(\sigma) = \mWS(\sigma')$ and the 
  predicate $\msound_{S}$ holds trivially.
  \item[(ii)] If the requested mapping sets the executable flag $ex$ (i.e. the  mapping grants an executable access to a physical block $pb$), from validity of the request
  we know that  $\msound_{S}(ex,pb,\sigma, GI)$. Thus adding this mapping preserves the soundness of state and the lemma holds in this case as well.
 \end{itemize}

 Proof for the other operations can be done similarly. 
\end{proof}

\begin{lemma}
  \label{lem:monitor:invariant:hyper-helper2}
  If $\tuple{\sigma, h, GI} \xrightarrow{1} \tuple{\sigma', (\tau', \rho'_{wt}, \rho'_{ex}), GI'}$ and assuming that:
  \begin{itemize}
   \item[(i)] $I(\tuple{\sigma, (\tau, \rho_{wt}, \rho_{ex}), GI})$, 
   \item[(ii)] $\forall\ bl.\ (\rho_{ex}(bl) = 0 \wedge \rho'_{ex}(bl) > 0) \Rightarrow (\rho_{wt}(bl) = 0)$, and
   \item[(iii)] $\forall\ bl.\ (\rho_{wt}(bl) = 0 \wedge \rho'_{wt}(bl) > 0) \Rightarrow (\rho_{ex}(bl) = 0)$.
  \end{itemize}  
  For all blocks $bl$, if   $\msound_{W \oplus X}(wt,ex,\rho_{wt}, \rho_{ex}, bl)$ then $\msound_{W \oplus X}(wt,ex,\rho'_{wt}, \rho'_{ex}, bl)$
\end{lemma}
\begin{proof}
 Since in the initial state the invariant holds, we know that the state $\tuple{\sigma, h, GI}$ complies with the policy $W \oplus X$ . Thus for each block $bl$ only one of its 
 counters $\rho_{ex}(bl)$ and $\rho_{wt}(bl)$ can be greater than zero. This implies that only one of the assumptions $(ii)$, $(iii)$ is true. Let assume that in $\tuple{\sigma, h, GI}$
 the block $bl$ is executable, therefore $\rho_{wt}(bl) = 0$. If in  $\tuple{\sigma, h', GI}$, the block $bl$ is still executable, its writable permission has not been changed and 
 $\msound_{W \oplus X}(wt,ex,\rho'_{wt}, \rho'_{ex}, bl)$ holds. However, its permissions has been changed in $\tuple{\sigma, h', GI}$ it has to be happened according to the assumption 
 $(iii)$, this means that $\msound_{W \oplus X}(wt,ex,\rho'_{wt}, \rho'_{ex}, bl)$ will be hold true. 
\end{proof}

\begin{theorem}
  \label{lem:monitor:invariant:hyper}
  If $\tuple{\sigma, h, GI} \xrightarrow{1} \tuple{\sigma', h', GI'}$ and $I(\tuple{\sigma, h, GI})$ then $I(\tuple{\sigma', h',GI'})$.
\end{theorem}
\begin{proof}
When the request is not validated
($\neg \mvalidate$) then the proof is trivial, since 
$\epsilon$ is the identity function.

If the request is validated by the monitor and committed by the
hypervisor, then the inference rules guarantee that $GI'=GI$ and that
the system without the monitor behaves as
$\tuple{\sigma, h} \xrightarrow{0} \tuple{\sigma', h}$.
Thus, Proposition~\ref{prep:hypervisor:invariant} can be used to
guarantee that the hypervisor invariant is preserved ($I_H(\sigma',
h')$). Moreover, Lemma~\ref{lem:monitor:invariant:hyper-ws}
demonstrates that the $\msound_{S}$ part of the invariant holds.

The proof of the second part (the executable space protection)
 of the invariant is the most challenging
task of this formal verification. This basically establishes the
functional correctness of the monitor and that its run-time policies
are strong enough to preserve the invariant (i.e. they enforce protection of the potentially executable space).
Practically speaking, the proof consists of several cases: one for
each possible request. The structure of the proof for each case is
similar. For example, for $r=\mmap_{\LTwoType}(bl, idx, bl', ex, wt, rd)$, we (i)
prove that the hypervisor (modeled by the function $H_{r}$) only
changes entry $idx$ of the page table stored in block $bl$ (that is, all
other blocks that are not typed $\DataType$ are unchanged), (ii)  we
show that only the counters of physical block $bl'$ are changed, and (iii) we
establish the hypothesis of
Lemma~\ref{lem:monitor:invariant:hyper-helper2}.
This enables us to infer $\msound_{W \oplus X}$ for the unchanged blocks
and to reduce the proof to only check the correctness of the entry
$idx$ of the page table in the block $bl$.
\end{proof}

Finally, Theorem~\ref{thm:top-goal} composes our results,
demonstrating that no code injection can succeed.
\begin{theorem}\label{thm:top-goal}
  Let $\tuple{\sigma, h, GI}$ be a state reachable from the initial
  state of the system $\tuple{\sigma_0, h_0, GI_0}$ and 
  $I(\tuple{\sigma_0, h_0, GI_0})$, then 
  $\integrity(GI, \sigma)$ holds
\end{theorem}
\begin{proof}
  Theorems~\ref{lem:monitor:invariant:user} and
  ~\ref{lem:monitor:invariant:hyper} directly show that the invariant
  is preserved for an arbitrary trace. Then,
  Lemma~\ref{lem:monitor:invariant:safe} demonstrates that every
  reachable state is free of malicious code injection.
\end{proof}

\section{Evaluation}\label{sec:eval}
The verification has been performed using the HOL4 interactive theorem
prover.
The specification of the high level model of the monitor adds
710 lines of HOL4 to the existing model of the hypervisor.
This specification is intentionally low level and does not depend on any 
high level theory of HOL4.
This increased the difficulty of the proof (e.g., it musts handle finite
arithmetic overflows), that consists of 4400 lines of HOL4.
However, the low level of abstraction allowed us to
directly transfer the model to a practical implementation and 
to identify several bugs of the original design.
For example,
the original policy for $link_{L1}(bl, idx, bl')$ did not contain
the condition $\rho_{ex}(bl) = 0$, allowing to violate the integrity of the
working set if a block is used to store an L1 page table and is itself 
executable.

The monitor code consists of 720 lines of C and the emulator
consists of additional 950 lines of code. Finally, 100 lines have been
added to the hypervisor to support the needed interactions among the
hosted components.

We used LMBench to measure the overhead introduced on user processes
hosted by Linux.
We focused on the benchmarks ``fork'', ``exec'' and
``shell'', since they require the creation of new processes and thus
represent the monitors worst case scenario. 
As macro-benchmark, we measured in-memory compression of two
data streams.
The benchmarks have been
executed using Qemu to emulate a Beagleboard-Mx. Since we are not
interested in evaluating a specific signature scheme, we computed the
signature of each physical 
block as the xor of the contained words, allowing us to focus
on the overhead introduced by the monitor's infrastructure.
Table~\ref{tbl:benchmarks} reports the benchmarks for different prototypes of the monitor, thus enabling to compare the overhead introduced by different design choices. 
``No monitor'' is the
base configuration, where neither the monitor or the emulation layer
are enabled. In ``P Emu'' the emulation layer is enabled and deployed
as component of the hypervisor. This benchmark is used to measure the
overhead introduced by this layer, which can be potentially removed at
the cost of modifying the Linux kernel. 
In ``P Emu + P Mon'' both the monitor and the emulation layer are
deployed as privileged software inside the hypervisor.
Finally, in ``P Emu + U Mon'' the monitor
is executed as unprivileged guest.

\begin{table}[t]
\scriptsize
\centering
\begin{tabular}{l|r|r|r|r|r}
Benchmark & fork & exec & shell & \multicolumn{2}{c}{tar -czvf 1.2 KB/12.8 MB}\\
        \hline
No monitor    & 0.010 & 0.010 & 0.042 & \hspace{3.5em} 0.05 & 20.95 \\
P Emu         & 0.011 & 0.011 & 0.048 & 0.09 & 21.05 \\
P Emu + P Mon & 0.013 & 0.013 & 0.054 & 0.10 & 21.02 \\
P Emu + U Mon & 0.017 & 0.017 & 0.067 & 0.11 & 20.98 \\
\end{tabular}
\caption{Qemu benchmarks [in second]}
\label{tbl:benchmarks}
\end{table}
\section{Related Work}\label{sec:related}
Since a comprehensive verification of commodity SW is not possible, 
it is necessary to architect systems so that the trusted computing
base for the desired properties is small enough to be verified, and that
the untrusted code cannot affect the security properties.
Specialized HW (e.g. TrustZone and TPM)
has been proposed to support this approach and has been used to implement secure storage and
attestation.
The availability of platforms like hypervisors and microkernels extended the adoption of this
approach to use cases that go beyond the ones that can be handled
using static HW based solutions.

For example, in~\cite{klein2010verified} the authors use the seL4 microkernel to
implement a secure access controller (SAC) with the purpose
of connecting one front-end terminal to either of two back-end networks one at a
time.
The authors delegate the complex (and non security-critical) functionalities (e.g. IP/TCP routing, WEB front-end) to
untrusted Linuxes, which are isolated by the microkernel from a small
and trusted router manager. The authors describe how the system's information
flow properties  can be verified disregarding the behavior of the
untrusted Linuxes.

Here, we used the trustworthy components to help the insecure Linux to restrict its own attack 
surface;
i.e. to prevent binary code injection.
Practically, our proposal uses 
Virtual Machine Introspection (VMI), which has been first introduced by
Garfinkel \textit{et al.}~\cite{Azmandian:2011:VMM:2007183.2007189}
and Chen \textit{et al.}~\cite{Chen:2001:VBR:874075.876409}.
Similarly to MProsper, other proposals (including
Livewire~\cite{Azmandian:2011:VMM:2007183.2007189},
VMWatcher~\cite{Jiang:2007:SMD:1315245.1315262} and
Patagonix~\cite{litty2008hypervisor}) use VMI, code signing and executable space
protection to prevent binary code injection in commodity OSs. 
However, all existing proposals rely on untrusted hypervisors
and their designs have not been subject of formal verification.

Among others non trustworthy  VMIs, hytux~\cite{hytux},
SecVisor~\cite{Seshadri:2007:STH:1294261.1294294} and
NICKLE \cite{Riley:2008:GPK:1433006.1433008} focus on protecting
integrity of the sole guest kernel.
SecVisor establishes a trusted channel with the user, which must
manually confirm all changes to the kernel.
NICKLE uses a \textit{shadow memory} to keep copy of authenticated
modules and guarantees that any instruction fetch by the kernel is
routed to this memory.

OpenBSD 3.3 has been one of the first OS enforcing executable space protection ($W \oplus X$).
Similarly, Linux (with the PaX and Exec Shield patches), NetBSD and Microsoft's OSs (using Data Execution
Prevention (DEP)) enforce the same policy.
However, we argue that due to the size of the modern kernels, trustworthy
executable space protection can not be achieved without the external
support of a trusted computing base.
In fact,  an attacker targeting the kernel can circumvent 
the protection mechanism, for example using \textit{return-oriented
  programming}~\cite{Shacham:2007:GIF:1315245.1315313}. 
The importance of enforcing executable space protection from a
privileged point of view (i.e. by VMI)
is also exemplified by~\cite{liakh2010analyzing}. Here, the authors
used model checking techniques to identify several misbehaviors of
the Linux kernel that violate the desired property.

\section{Concluding Remarks}\label{sec:conclusions}
We presented a trustworthy code injection prevention
system for Linux on embedded devices. 
The monitor's trustworthiness is based on
two main principles
(i) the trustworthy hypervisor guarantees the
monitor's tamper resistance and that all memory operations that
modify the memory layout are mediated,
(ii) the formal verification of design demonstrates  that the
top security goal is guaranteed by the run-time checks executed by the
monitor. 
These are distinguishing features of MProsper, since
it is the first time that absence of binary code injection has been
verified for a commodity OS.

Even if the MProsper's formal model is not yet at the level of the actual binary
code executed on the machine, this verification effort is important to
validate the monitor design; in fact we were able to spot security
issues that were not dependent on the specific implementation of the
monitor. The high level model of the monitor is actually a state transition
model of the implemented code, operating on the actual ARMv7 machine
state.  Thus the
verified properties can be transferred to the actual implementation by using standard
refinement techniques (e.g. ~\cite{Sewell:2013:TVV:2491956.2462183}).

Our ongoing work include the development of a end-to-end secure infrastructure, where an
administrator can remotely update the software of an embedded
device. Moreover, we are experimenting with other run-time binary analysis techniques that go beyond code signature checking: for example an
anti-virus scanner can be integrated with the monitor, enabling to
intercept and stop self-decrypting malwares.

\clearpage

\thispagestyle{empty}
\cleartoleftpage
\thispagestyle{empty}
\section*{ARM Data-Cache Terminology}

In the following we give a short summary of some of the terms used in the next two papers and a figure illustrating the structure of a data-cache in the ARM architecture:

\begin{figure}[h!]
 \centering
 \includegraphics[width=0.4\linewidth]{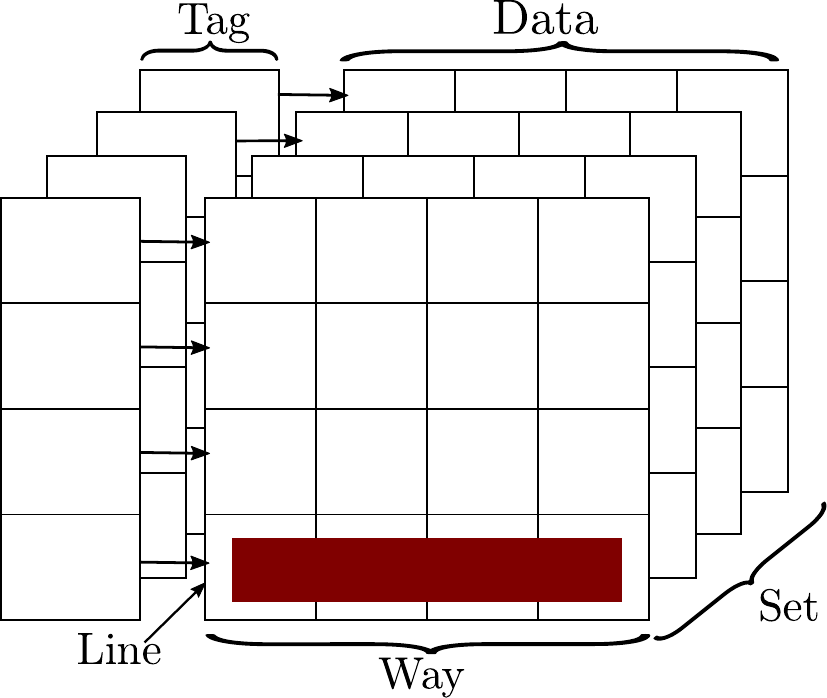}
 \captionsetup{labelformat=empty}
 \caption{Data-cache structure for the ARM architecture.}
 \label{fig:cache-term}
\end{figure}

\textbf{Tag: } The tag is the part of an address (mostly physical address) stored in the cache which identifies the corresponding memory address associated to a line in the cache.

\textbf{Line: } A line in the cache contains a block of contiguous words for the memory. Each line include also a \textit{valid flag} (which indicates if the line contains a valid data) and a \textit{dirty flag}
(which indicates if the line has been changed since it was read from memory).

\textbf{Set: } Memory addresses are logically partitioned into sets of lines that are congruent w.r.t. a set index; usually set index depends on either virtual addresses (then the cache
is called virtually indexed) or physical addresses (then the cache is called physically indexed).

\textbf{Way: } The cache contains a number of ways which can hold one corresponding line for every set index.

\textbf{Cache Flushing: } Flushing means writing back into the memory a cache line and cleaning the dirty flag.

\textbf{Eviction: } Eviction means writing back into the memory a cache line and invalidating it in the cache.

\textbf{Memory Coherence: } Is the problem of ensuring that a value read from a memory location (using either a cacheable or an uncacheable address) is always the most recently written value to that location.

\chapter{Cache Storage Channels: Alias-Driven Attacks and Verified Countermeasures}\label{paper:sp}
\chaptermark{Cache Storage Channels}
\backgroundsetup{position={current page.north east},vshift=1cm,hshift=-7cm,contents={\VerBar{Purple}{2cm}}}
\BgThispage

\begin{center}
Roberto Guanciale, Hamed Nemati, Mads Dam, Christoph Baumann
\end{center}
\newenvironment{paraenum}{\begin{inparaenum}[(1)]}{\end{inparaenum}}

\begin{abstract}
  Caches pose a significant challenge to formal proofs of security for code executing on application processors, as the cache access pattern of security-critical services may leak secret information. This paper reveals a novel attack vector,
  exposing a low-noise cache storage channel that can be exploited by
  adapting well-known timing channel analysis techniques. The vector can
  also be used to attack various types of security-critical 
  software such as hypervisors and application security monitors. 
  The attack vector uses virtual
  aliases with mismatched memory attributes and self-modifying code to
  misconfigure the memory system, allowing an attacker to place 
  incoherent copies of the same physical address into the caches and observe
  which addresses are stored in different levels of cache. We design
  and implement three different attacks using the new vector on
  trusted services and report on the discovery of an 128-bit key from
  an AES encryption service running in TrustZone on Raspberry Pi
  2. Moreover, we subvert the integrity properties of an ARMv7
  hypervisor that was formally verified against a cache-less model. We
  evaluate well-known countermeasures against the new attack vector and
  propose a verification methodology that allows to formally prove the effectiveness of defence mechanisms on the binary code of the trusted software. 
\end{abstract}

\newcommand{\ArmMachineVar}{\ArmState}
\newcommand{\RelP}{\mathrel{\mathcal{P}}}
\section{Introduction}

Over the past decade huge strides have been made to realise the long-standing vision of formally verified execution platforms, including hypervisors~\cite{CavalcantiD09,SOFSEM}, separation kernels~\cite{dam2013formal,INTEGRITY}, and microkernels~\cite{DBLP:conf/sosp/KleinEHACDEEKNSTW09}.
Many of these platforms have been comprehensively verified, down to machine code~\cite{DBLP:conf/sosp/KleinEHACDEEKNSTW09} and Instruction Set Architecture (ISA)~\cite{dam2013formal} levels, and provide unprecedented security and isolation guarantees.

Caches are mostly excluded from these analyses. The verification of both seL4~\cite{KleinAEMSKH14} and  the Prosper kernels~\cite{dam2013formal,SOFSEM} assume that caches are invisible and ignore timing channels.
The CVM framework from the Verisoft project~\cite{AlkassarHLSST09} treats caches only in the context of device management~\cite{HillebrandRP05}. For the verification of user processes and the remaining part of the kernel, caches are invisible. %
Similarly, the Nova \cite{SteinbergK10,TewsVW09} and CertiKOS
\cite{GuVFSC11} microvisors do not consider caches in their formal
analysis. 

How much of a problem is this? It is already well understood that caches are one of the key features of modern
commodity processors that make a precise analysis of, e.g., timing and/or power consumption exceedingly difficult, and that this can be exploited to mount timing-based side channels, even for kernels that have been fully verified \cite{CockGMH14}. These channels, thus, must be counteracted by model-external means, e.g., by adapting scheduling intervals~\cite{StefanScheduling13} or cache partitioning~\cite{Raj:2009:RMI:1655008.1655019,Kim:2012:SSP:2362793.2362804}. 

The models, however, should preferably be sound with respect to the features that \textit{are} reflected, such as basic memory reads and writes. Unfortunately, as we delve deeper into the Instruction Set Architecture we find that this expectation is not met: Certain configurations of the system enable an attacker to exploit caches to build storage channels. Some of these channels are especially dangerous since they can be used to compromise both confidentiality and integrity of the victim, thus breaking the formally verified properties of isolation.

The principle idea to achieve this, is to break coherency of the memory system by deliberately not following the programming guidelines of an ISA. In this report we focus on two programming faults in particular:
\begin{enumerate}
  \item Accessing the same physical address through virtual aliases with mismatched cacheability attributes.
  \item Executing self-modifying code without flushing the instruction cache.
\end{enumerate}
Reference manuals for popular architectures (ARM, Power, x64) commonly warn that not following such guidelines may result in unpredictable behaviour. However, since the underlying hardware is deterministic, the actual behaviour of the system in these cases is quite predictable and can be reverse-engineered by an attacker.

The first fault results in an incoherent memory configuration where cacheable and uncacheable reads may see different values for the same physical address after a preceding write using either of the virtual aliases. Thus the attacker can discover whether the physical address is allocated in a corresponding cache line. For the second fault, jumping to an address that was previously written without flushing the instruction cache may result in the execution of the old instruction, since data and instruction caches are not synchronised automatically. By carefully selecting old and new instructions, as well as their addresses, the attacker can then deduce the status of a given instruction cache line.

Obtaining this knowledge, i.e., whether certain cache lines contain attacker
data and instructions, is the basic principle behind the Prime+Probe flavor of
access-driven timing channel attacks~\cite{Tromer:2010:ECA:1713125.1713127}. This type of attack can
be adapted using the new attack vector. The main advantage of this approach is
that the cache storage channels presented here are both more stealthy, less
noisy, and easier to measure than timing channels. Moreover, an 
incoherent data cache state can be used to subvert the integrity of trusted services that depend on untrusted inputs. 
Breaking the memory coherency for the inputs exposes vulnerabilities that enable
a malicious agent to bypass security monitors and possibly to compromise the integrity of the trusted software.

The attacks sketched above have been experimentally validated in three realistic scenarios.
We report on the implementation of a prototype that extracts a 128-bit key from an AES encryption
service running in TrustZone on Raspberry Pi 2.
We use the same platform to implement a process that extracts the exponent of a modular exponentiation
procedure executed by another process.
Moreover, implementing a cache-based attack we subverted the integrity properties of an ARMv7 hypervisor that was formally verified against a cache-less model. The scenarios are also used to evaluate several existing
countermeasures against cache-based attacks as well as new ones that are targeted to the alias-driven attack vector.

Finally, we propose a methodology to repair the formal analysis of the trusted
software, reusing existing means as much as possible. Specifically,
we show (1) how a countermeasure helps restoring integrity of a previously formally
verified software and (2) how to prove the absence of cache storage side channels.
This last contribution includes the adaptation of an existing tool \cite{DBLP:conf/ccs/BalliuDG14} to analyse the binary code of the trusted software.

\section{Background}\label{sec:background}
Natural preys of side-channel attacks are implementations of cryptographic
algorithms, as demonstrated by early works of Kocher~\cite{Kocher:1996:TAI:646761.706156} and
Page~\cite{Page02theoreticaluse}.
In cache-driven attacks, the adversary exploits the caches to acquire knowledge about
the execution of a victim and uses this knowledge to infer the victim's internal variables.
These attacks are usually classified in three groups, that differ by
the means used by the attacker to gain knowledge.
In ``time-driven attacks'' (e.g.~\cite{Tsunoo03cryptanalysisof}), the attacker,
who is assumed to be able to trigger an encryption, 
measures (indirectly or directly) the execution time of the victim and
uses this knowledge to estimate the number of cache misses and hits
of the victim.
In ``trace-driven attacks'' (e.g.~\cite{Aciicmez:2006:TCA:2092880.2092891,Page02theoreticaluse,zhang2012cross}), the adversary
has more capabilities: he can profile the cache activities during the
execution of the victim  and thus observe the cache effects
of a particular operation performed by the victim.
This highly frequent measurement can be possible due to the
adversary being interleaved with the victim by the scheduler of the operating system
or because the adversary executes on a separate core and monitors a shared cache.
Finally, in ``access-driven attacks'' (e.g.~\cite{Neve:2006:AAC:1756516.1756531,Tromer:2010:ECA:1713125.1713127}), the
attacker determines the cache indices modified by the victim. This
knowledge is obtained indirectly, by observing cache side effects of victim's
computation on the behaviour of the attacker.

In the literature, the majority of trace and access driven attacks
use timing channels as the key attack vector. These vectors rely on time
variations to load/store data and to
fetch instructions in order to estimate the cache activities of the victim:
the cache lines that are evicted, the cache misses, the cache hits, etc.

Storage channels, on the other hand, use system variables to
carry information. The possible presence of these channels raises 
concerns, since they invalidate the results of formal verification.
The attacker can use the storage channels without the support of an
external measurement (e.g. current system time), so there is no external variable
such as time or power consumption that can be manipulated by the victim to close the channel
and whose accesses can alert the victim about malicious intents.
Moreover, a storage channel can be less noisy than timing channels that are 
affected by scheduling, TLB misses, speculative execution,
and power saving, for instance. Finally, storage channels can pose risk
to the integrity of a system, since they can be used to bypass reference
monitors and inject malicious data into trusted agents.
Nevertheless, maybe due to the practical complexities in implementing these
channels, few works in literature address cache-based storage channels.

One of the new attack vectors of this paper is based on
mismatched cacheability attributes and has pitfalls other than
enabling access-driven attacks.
The vector opens up for Time Of Check To Time Of Use (TOCTTOU) like
vulnerabilities. A trusted agent may check data stored in the cache
that is not consistent with the data
that is stored in the memory by a malicious software.
If this data is later evicted from the cache, it
can be subsequently substituted by the unchecked item placed in the
main memory.
This enables an attacker to bypass a reference monitor, possibly
subverting the security property of formally verified software.

Watson \cite{Watson07} demonstrated this type of vulnerability
for Linux system call wrappers. 
He uses concurrent memory accesses, using preemption to change the
arguments to a system call in  user memory after they were validated.
Using non-cacheable aliases one
could in the same way attack the Linux system calls that read from
the caller's memory.
A further victim of such attacks is represented by run time monitors.
Software that dynamically loads untrusted modules
often uses Software-based Fault Isolation
(SFI)~\cite{Wahbe:1993:ESF:168619.168635,Silver:1996:IAS:867744} to isolate untrusted
components from the trusted ones.
If an on-line SFI verifier is used (e.g. because the
loaded module is the output of a just-in-time compiler), then caches can be
used to mislead the verifier to accept stale data.
This enables malicious components to break the  SFI assumptions and
thus the desired isolation.

In this paper we focus on scenarios where the victim and the attacker
are hosted on the same system. An instance of such scenarios consists of
a malicious user process that attempts to compromise either another
user process, a run-time monitor or the operating system itself. In a cloud environment,
the attacker can be a (possibly compromised) complete operating system
and the victim is either a colocated guest, a virtual machine
introspector or the underlying hypervisor.
Further instances of such scenario are systems that use specialised hardware to
isolate security critical components from untrusted operating systems.
For example, some ARM processors implement TrustZone~\cite{TrustZone}. This mechanism can be used to
isolate and protect the system components that implement remote attestation, trusted
anchoring or virtual private networks (VPN).
In this case, the attacker is either a compromised operating system
kernel or an untrusted user process threatening a TrustZone application.

\section{The New Attack Vectors: Cache Storage Channels}\label{sec:attacks}

Even if it is highly desirable that the presence of caches is transparent to program behaviour, this is usually not the case unless the system
configuration satisfies some architecture-specific constraints.
Memory mapped devices provide a trivial example: If the address representing the output register of a
memory mapped UART is cacheable, the output of a program is never visible on the serial cable, since the output characters are overwritten in the cache instead of being sent to the physical
device. These behaviours, which occur due to misconfigurations of the system, can raise to security threats.

To better understand the mechanisms that constitute our attack vectors,
we summarise common properties of modern architectures. 
The vast majority of general purpose systems use set-associative caches:
\begin{enumerate}[(i)]
  \item Data is transferred between memory and cache in blocks of fixed
size, called cache lines.
  \item The memory addresses are logically partitioned into sets of lines
that are congruent wrt. a set index; usually set index depends on
either virtual addresses (then the cache is called virtually
indexed) or physical addresses (then the cache is called
physically indexed);
  \item The cache contains a number of ways which can hold one
corresponding line for every set index.
  \item  A cache line stores both the data, the corresponding
physical memory location (the tag) and a dirty flag
(which indicates if the line has been changed since it was read from
memory).
\end{enumerate}

Caches are used by
processors to store frequently accessed information and thus to
reduce the number of accesses to main memory.
A processor can use separate instruction and data caches in a Harvard
arrangement (e.g. the L1 cache in ARM Cortex A7) or unified caches
(e.g. the L2 cache in ARM Cortex A7).
Not all memory areas should be cached; for instance, accesses to addresses
representing registers of memory mapped devices should always be directly sent to the main memory subsystem.
For this reason,
modern Memory Management Units (MMUs) allow to configure, via the page
tables, the caching policy on a per-page basis,
allowing a fine-grained control over if and how areas of memory are cached.

In Sections~\ref{sec:attack:conf:dcache},~\ref{sec:attack:int:dcache} and~\ref{sec:attack:conf:icache}
we present three new attack vectors that depends on misconfigurations
of systems  and caches. These attacks exploit the following
behaviours:
\begin{itemize}
\item Mismatched cacheability attributes; if the data cache reports
  a hit on a memory location that is marked as non-cacheable, the cache might access the
  memory disregarding such hit. ARM calls
  this event ``unexpected cache hit''.
\item Self-modifying code; even if the executable code is updated,
  the processor might execute the old version of it if this has been stored in the
  instruction cache.
\end{itemize}
The attacks can be used to threaten both confidentiality and
integrity of a target system. Moreover, two of them use new
storage channels suitable to mount access driven attacks. This is
particularly concerning, since so far only a noisy timing channel could be used to launch attacks of this kind, which makes real implementations
difficult and slow.
The security threats are particularly severe whenever the attacker is
able to (directly or indirectly) produce the misconfigurations that
enable the new attack vectors, as described in Section~\ref{sec:attack:scenario}.

\begin{figure}[t]
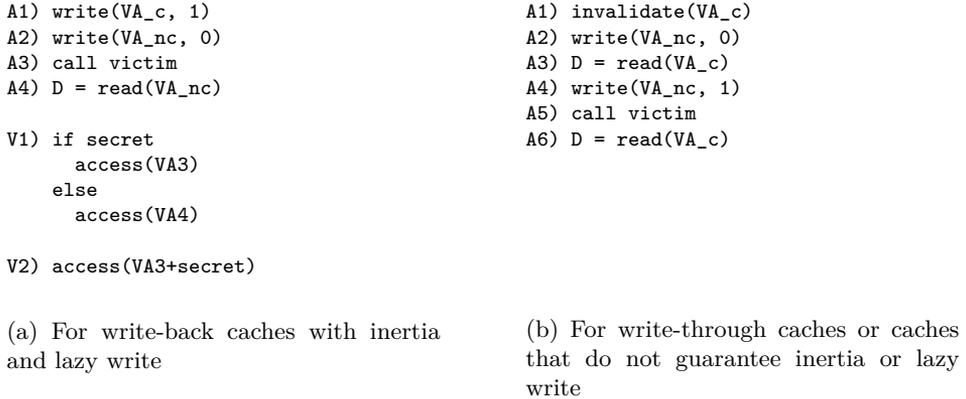

\centering
\begin{subfigure}[t]{0.45\textwidth}
 \begin{Verbatim}[fontsize=\footnotesize]
A1) write(VA_c, 1)
A2) write(VA_nc, 0)
A3) call victim
A4) D = read(VA_nc)

V1) if secret
      access(VA3)
    else
      access(VA4)

V2) access(VA3+secret)
 \end{Verbatim}
  \caption{For write-back caches with inertia and lazy write}
  \label{fig:conf1a}
\end{subfigure}\hspace{1cm}
\begin{subfigure}[t]{0.45\textwidth}
\begin{Verbatim}[fontsize=\footnotesize]
A1) invalidate(VA_c)
A2) write(VA_nc, 0)
A3) D = read(VA_c)
A4) write(VA_nc, 1)
A5) call victim
A6) D = read(VA_c)





\end{Verbatim}
  \caption{For write-through caches or caches that do not guarantee inertia or lazy write}
  \label{fig:conf1b}
\end{subfigure}
\caption{Confidentiality threat due to data-cache}
\end{figure}

\subsection{Attacking Confidentiality Using Data-Caches}\label{sec:attack:conf:dcache}
Here we show how an attacker can use 
mismatched cacheability attributes to mount
access-driven cache attacks; i.e.~measuring which data-cache lines
are evicted by the execution of the victim.

We use the program in Figure~\ref{fig:conf1a} to demonstrate the
attacker programming model.
For simplicity, we assume that the cache is physically indexed,
it has only one way and that it uses the write allocate/write back policy.
We also assume that the attacker can access the virtual addresses
$va_c$ and $va_{nc}$,  both pointing to the physical address $pa$;
$va_{c}$ is cacheable while $va_{nc}$ is not.
The attacker writes $1$ and $0$ into the virtual addresses $va_{c}$
and $va_{nc}$ respectively, then it invokes the victim. After the
victim returns, the attacker reads back from the address $va_{nc}$.


Let $idx$ be the line index corresponding to the address $pa$.
Since $va_c$ is cacheable, the instruction in $A1$ stores the value $1$ in the cache
line indexed by $idx$, the line is flagged as dirty and its tag is set to $pa$.
When the instruction in $A2$ is executed, since $va_{nc}$ is non-cacheable, the system ignores the
``unexpected cache hit'' and the value $0$ is directly written into the memory,
bypassing the cache.
Now, the value stored in main memory after the execution of the victim depends
on the behaviour of the victim itself;
if the victim accesses at least one address
whose line index is $idx$, then the dirty line is evicted and the value $1$ is
written back to the memory;
otherwise the line is not evicted and the physical memory still contains the value $0$ in $pa$.
Since the address is non-cacheable, the value that is read from $va_{nc}$ in $A4$ depends on the victim's behaviour.

This mechanism enables the attacker to probe if the line index $idx$ is
evicted by the the victim. If the attacker has available a
pair of aliases (cacheable and non-cacheable) for every cache line index,
the attacker is able to measure the list of cache lines that are
accessed by the victim, thus it can mount an access-driven cache
attack.
The programs V1 and V2 in Figure~\ref{fig:conf1a} exemplify two victims of such
attack; in both cases the lines evicted by the programs depend on a
confidential variable $\mathit{secret}$ and the access-driven cache
attack can extract some bits of the secret variable.

{Note that we assumed that the data cache (i) is ``write-back'', 
(ii) has ``inertia'' and (iii) uses ``lazy write''.
That is, (i) writing is done in the cache and the write access to the memory is postponed,
(ii) the cache evicts a line only  when the corresponding space is needed to
store new content, and 
(iii) a dirty cache line is written back only when it is evicted.}
This is not necessarily true%
{; the cache can be write-through or it can
  (speculatively) write back and
clean dirty lines when the memory bus is
unused.
}
Figure~\ref{fig:conf1b} presents an alternative attack whose
success does not depend on this assumption.
The attacker ($A1$-$A3$) stores the value $0$ in the cache,
by invalidating the corresponding line, writing $0$ into the memory
and reading back the value using the cacheable virtual address.
Notice that after step $A3$ the cache line is clean, since the
attacker used the non-cacheable virtual alias to write the value $0$.
Then, the attacker writes $1$ into the memory, bypassing the
cache.
The value read in $A6$ using the cacheable address depends
on the behaviour of the victim;
if the victim accesses at least one address
whose line index is $idx$, then the cache line for $pa$ is evicted and
the instruction in $A6$ fetches the value from the memory, yielding
the value $1$;
otherwise the line is not evicted and the cache still contains the value $0$ for $pa$.

\subsection{Attacking Integrity Using Data-Caches}\label{sec:attack:int:dcache}
\begin{figure}[t]
\centering
\begin{subfigure}[t]{0.2\textwidth}
\begin{Verbatim}[fontsize=\footnotesize]
V1) D = access(VA_c)
A1) write(VA_nc, 1)
V2) D = access(VA_c)
V3) if not policy(D)
       reject
    [evict VA_c]
V4) use(VA_c)
\end{Verbatim}
\end{subfigure}
\caption{Integrity threat due to data-cache}
\label{fig:int}
\end{figure}
Mismatched cacheability attributes may also produce integrity threats,
by enabling an attacker to modify critical data in an unauthorized or undetected manner.
Figure~\ref{fig:int} demonstrates an integrity attack. 
Again, we assume that the data-cache is direct-mapped, that it is physically indexed and that its write policy is write allocate/write back.
For simplicity, we limit the discussion to the L1 caches. 
In our example, $va_c$ and $va_{nc}$ are  virtual addresses pointing to the same memory location $pa$; $va_c$ is the cacheable alias
while $va_{nc}$ is non-cacheable. Initially, the memory location $pa$ contains the value $0$ and the corresponding cache line is either invalid or the line has valid data but it is clean.
In a sequential model where reads and writes are guaranteed to take place in program order and their effects are instantly visible to all system components,
the program of Figure~\ref{fig:int}  has the following effects: V1) a victim accesses address $va_c$, reading $0$; A1) the
attacker writes $1$ into $pa$ using the virtual alias
$va_{nc}$; V2) the victim accesses again $va_c$, this time
reading $1$; V3) if $1$ does not respect a security policy, then the
victim rejects it;
otherwise V4) the victim uses $1$ as the input for a security-relevant functionality.

On a real processor with a relaxed memory model the same system can behave differently, in
particular:
V1) using $va_c$, the victim reads initial value $0$
from the memory at the location $pa$
and fills the corresponding line in the cache;
A1) the attacker use $va_{nc}$ to write $1$ directly into the memory,
bypassing the cache; 
V2) the victim accesses again $va_c$, reading $0$ from the
cache; V3) the policy is evaluated based on $0$;
possibly, the cache line is evicted and, since it is not dirty, the memory
is not affected; V4) the next time that the victim accesses $pa$ it
will read $1$ and will use this value as input of the functionality,
but $1$ has not been checked against the policy.
This enables an attacker to bypass a reference monitor, here
represented by the check of the security policy, and to inject unchecked
input as parameter of security-critical functions.

\subsection{Attacking Confidentiality Using Instruction Caches}\label{sec:attack:conf:icache}
\begin{figure}[t]
\centering
\begin{subfigure}[t]{0.2\textwidth}
\begin{Verbatim}[fontsize=\footnotesize]
A1) jmp A8
A2) write(&A8, {R0=1})
A3) call victim
A4) jmp A8
A5) D = R0
...
A8) R0=0
A9) return

V1) if secret
      jmp f1
    else
      jmp f2
\end{Verbatim}
\end{subfigure}
\caption{Confidentiality threat due to instruction-cache}
\label{fig:conf2}
\end{figure}
Similar to data caches, instruction caches can be used to
mount access-driven cache attacks; in this case the attacker probes the
instruction cache to extract information about the victim execution path.

Our attack vector uses self-modifying code.
The program in Figure~\ref{fig:conf2} demonstrates the principles of
the attack. We assume that the instruction cache is physically indexed
and that it has only one way. We also assume that the attacker's
executable address space is cacheable and that the processor
uses separate instruction and data caches.

Initially, the attacker's program contains a function at the address
$A8$ that writes $0$ into the register $R0$ and immediately returns.
The attacker starts in $A1$, by invoking the function at $A8$.
Let $idx$ be the line index corresponding to the address of $A8$:
Since the executable address space is cacheable, the
execution of the function has the side effect of temporarily storing the
instructions of the function into the instruction cache.
Then ($A2$), the attacker modifies its own code, overwriting the
instruction at $A8$ with an instruction that
updates register $R0$ with the value $1$.
Since the processor uses separate instruction and data
caches the new instruction is not written into the
instruction cache.
After that the victim completes the execution of its own code,
the attacker ($A4$) re-executes the function at $A8$.
The instruction executed by the second invocation of the function depends on
the behaviour of the victim: 
if the execution path of the victim contains at least one address
whose line index is $idx$ then the attacker code is evicted,
the second execution of the function fetches the new instruction
from the memory and the register is updated with the value $1$;
otherwise, the attacker code is not evicted and
the second execution of the function uses the old code
updating the register with $0$.

In practice, the attacker can probe if the line index $idx$ is
evicted by the victim. By repeating the probing phase for every
cache line,  the attacker can mount access-driven instruction cache
attacks.
The program V1 in Figure~\ref{fig:conf2} exemplifies a victim of such
attack, where the control flow of the victim (and thus the lines
evicted by the program) depends on a confidential variable $\mathit{secret}$.

\subsection{Scenarios}\label{sec:attack:scenario}

In this section we investigate practical applications and limits
of the attack vectors.
To simplify the presentation, we assumed one-way physically indexed caches. However, 
all attacks above can be straightforwardly applied to virtually indexed caches.
Also, the examples can be extended to support
multi-way caches if the way-allocation strategy of the cache does not
depend on the addresses that are accessed: the attacker repeats the cache filling phase using several
addresses that are all mapped to the same line index.
The attack presented in Section~\ref{sec:attack:conf:icache}
also assumes that the processor uses separate instruction and data caches.
This is the case in most modern processors, since they usually use the ``modified
Harvard architecture''. Modern x64 processors, however, implement a snooping mechanism that invalidates corresponding instruction cache lines automatically 
in case of self-modifying code (\cite{Intel}, Vol.~3, Sect.~11.6); in such a scenario the attack cannot succeed.

The critical assumptions of the attacks 
are the ability of building virtual aliases
with mismatched cacheability attributes
(for the attacks in Sections~\ref{sec:attack:conf:dcache} and~\ref{sec:attack:int:dcache})
and the ability of self-modifying code (for the attack in Section~\ref{sec:attack:conf:icache}).
These assumptions can be easily met if the attacker
is a (possibly compromised) operating system and the victim is a
colocated guest in a virtualized environment.
In this case, the attacker is usually free to create several virtual aliases and
to self-modify its own code.
A similar scenario consists of systems that use specialised hardware to
isolate security-critical components (like SGX and TrustZone), where a malicious operating system shares the caches
with trusted components.
Notice also that in case of TrustZone and hardware assisted virtualization,
the security software (e.g. the hypervisor) is not informed 
about the creation of setups that enable the attack vectors, since it usually does not interfere with the
manipulation of the guest page tables. 

In some cases it is possible to enable the attack vectors even if the
attacker is executed in non-privileged mode.
Some operating systems can allow user processes to reconfigure cacheability of their own virtual
memory. The main reason of this functionality is to speed up
some specialised computations that need to avoid polluting the cache
with data that is accesses infrequently~\cite{qu2005using}.
In this case two malicious programs can collude to build the
aliases having mismatched attributes.

Since buffer overflows can be used to inject
malicious code, modern operating systems enforce the executable space
protection policy: a memory page can be either writable or executable,
but it can not be both at the same time. However, to support just in
time compilers, the operating systems allow user processes to change
at run-time the permission of virtual memory pages, allowing to switch
a writable page into an executable and vice versa (e.g. Linux provides
the syscall ``mprotect'', which changes protection for a
memory page of the calling process). Thus, the attack of
Section~\ref{sec:attack:conf:icache} can still succeed if:
(i) initially the page containing the function $A8$ is executable,
(ii) the malicious process requests the operating system to switch the page as
writable (i.e. between step $A1$ and $A2$)
and (iii) the process requests the operating system to change back the page as
executable before re-executing the function (i.e. between step $A2$
and $A4$).
If the operating system does not invalidate the instruction cache whenever the
permissions of memory pages are changed, the confidentiality threat
can easily be exploited by a malicious process.

In Sections~\ref{sec:background} and~\ref{sec:related}  we provide a summary of existing
literature on side channel attacks that use caches. In general, every
attack (e.g.~\cite{Aciicmez:2006:TCA:2092880.2092891,zhang2012cross,Neve:2006:AAC:1756516.1756531}) that is access-driven and that has been
implemented by probing access times can be reimplemented
using the new vectors. However, we stress that the new vectors have
two distinguishing characteristics with respect to the time based ones:
{(i) the probing phase does not need the support of an external
measurement,
(ii) the vectors build a cache-based storage channel that has relatively low noise compared channels based on execution time which depend on many other factors than cache misses, e.g., TLB
misses and branch mispredictions.

In fact, probing the cache state by measuring execution time
requires the attacker to access the system time. If this resource is not
directly accessible in the execution level of the attacker, the attacker needs
to invoke a privileged function that can introduce delays and noise in the cache state
(e.g. by causing the eviction from the data cache when accessing internal
data-structures). For this reason, the authors
of~\cite{zhang2012cross} disabled the timing virtualization of XEN (thus
giving the attacker direct access to the system timer) to demonstrate
a side channel attack.}
Finally, one of the storage channels presented here poses integrity threats clearly
outside the scope of timing based attacks.




\section{Case Studies}\label{sec:experiments}
To substantiate the importance of the new attack vectors, and the need to augment the verification
methodology to properly take caches and cache attributes into account, we examine the
attack vectors in practice. Three cases are presented: A malicious OS that extracts a secret
AES key from a cryptoservice hosted in TrustZone, a malicious paravirtualized OS that subverts the memory protection of
a hypervisor, and a user process
that extracts the exponent of a modular exponentiation
procedure executed by another process.

\subsection{Extraction of AES Keys}
AES~\cite{daemen2013design} is a widely used symmetric encryption scheme that
uses a succession of rounds, where four operations (SubBytes,
ShiftRows, MixColumn and AddRoundKey) are iteratively applied to
temporary results. For every round $i$, the algorithm derives the sub key $K_i$ from an
initial key $k$. For AES-128 it is possible to derive $k$ from any sub
key $K_i$.

Traditionally, efficient AES software takes advantage of precomputed
SBox tables to reach a high performance and compensate the lack of native
support to low-level finite field operations. The fact that disclosing access
patterns to the SBoxes can make AES software insecure is well known in literature
(e.g.~\cite{AESattack,Tromer:2010:ECA:1713125.1713127,Aciicmez:2006:TCA:2092880.2092891}).
The existing implementations of these attacks
probe the data cache using time channels, here we demonstrate that
such attacks can be replicated using the storage channel described in
Section~\ref{sec:attack:conf:dcache}. With this aim, we implement the
attack described in~\cite{Neve:2006:AAC:1756516.1756531}.

The attack exploits a common implementation pattern.
The last round of AES is slightly different from the others since the
MixColumn operation is skipped. For this reason, implementations often
use four SBox tables $T_0, T_1, T_2, T_3$ of 1KB
for all the rounds except the last one, whereas a dedicated $T_4$
is used.
Let $c$ be the resulting cipher-text, $n$ be the
total number of rounds and $x_i$ be the intermediate output of the
round $i$. The last AES round computes the cipher-text as follows:
\[
 c = K_n \oplus \textit{ShiftRows}(\textit{SubBytes}(x_{n-1}))
\]
Instead of computing $\textit{ShiftRows}(\textit{SubBytes}(x_{n-1}))$, the implementation
accesses the precomputed table $T_4$ according to an index that depends on
$x_{n-1}$. Let $b[j]$ denote the $j$-th byte of $b$ and $[T_4\ \textit{output}]$ be one of the actual accesses to $T_4$, then 
\[
 c[j] = K_n[j] \oplus [T_4\ \textit{output}]\,.
\]
Therefore, it is straightforward to compute $K_n$ knowing the
cipher-text and the entry yielded by the access to $T_4$:
\[
 K_n[j] = c[j] \oplus [T_4\ \textit{output}]
\]
Thus the challenge is to identify the exact $[T_4\ output]$ for a
given byte $j$. We use the ``non-elimination'' method described in~\cite{Neve:2006:AAC:1756516.1756531}.
Let $L$ be a log of encryptions, consisting of a set of pairs
$(c_l, e_l)$. Here, $c_l$ is the resulting cipher-text and $e_l$ is
the set of cache lines accessed by the AES implementation.
We define $L_{j,v}$ to be the subset of $L$ such that the byte $j$ of the
cipher-text is $v$:
\[
  L_{j,v} = \{(c_l, e_l) \in L\ \mbox{such that} \ c_l[j]=v\}
\]
Since $c[j] = K_n[j] \oplus [T_4\ output]$ and the key is constant, 
if the $j$-th byte of two cipher-texts have the same value then the
accesses to $T_4$ for such cipher-text must contain at least one
common entry.
Namely, the cache line accessed by the implementation while computing
$c[j] = K_n[j] \oplus [T_4\ output]$ is (together with some false
positives) in the non-empty set
\[
E_{j,v} = \bigcap_{(c_l, e_l) \in L_{j,v}} e_l
\]
Let $T_4^{j,v}$ be the set of distinct bytes of $T_4$ that can be
allocated in the cache lines $E_{j,v}$. Let $v,v'$ be two different values recorded
in the log for the byte $j$. We know that exist $t_4^{i,v} \in T_4^{j,v}$ and
$t_4^{i,v'} \in T_4^{j,v'}$ such that $v = K_n[j] \oplus t_4^{i,v}$
and $v' = K_n[j] \oplus t_4^{i,v'}$. Thus
\[
v \oplus v' = t_4^{j,v} \oplus t_4^{j,v'}
\]
This is used to iteratively shrink the sets $T_4^{j,v}$
and $T_4^{j,v'}$ by removing the pairs that do not satisfy the
equation. The attacker repeats this process until for a byte value $v$
the set $T_4^{j,v}$ contains a single value; then the byte $j$ of key
is recovered using $K_n[j] = v \oplus t_4^{j,v}$.
Notice that the complete process can be repeated for every byte
without gathering further logs and that the attacker does not need to
know the plain-texts used to produce the cipher-texts.

We implemented the attack on a Raspberry Pi 2~\cite{rpi2}, because this platform
is equipped with a widely used CPU (ARM Cortex A7) and allows to use
the TrustZone extensions.
The system starts in TrustZone and executes the bootloader of our minimal TrustZone operating system. This installs a
secure service that allows an untrusted kernel to encrypt blocks
(e.g. to deliver packets over a VPN) using a secret key. This key is intended to
be confidential and should not be
leaked to the untrusted software. 
The trusted service is implemented using an existing AES library for
embedded devices~\cite{wolfssl},
that is relatively easy to deploy in the resource constrained environment
of TrustZone. However, several other implementations (including 
OpenSSL~\cite{openssl}) expose the same weakness due to the use of precomputed SBoxes.
The boot code terminates by exiting TrustZone and activating the
untrusted kernel. This operating system is not able to directly access the TrustZone memory but can invoke the secure service by executing Secure Monitor
Calls (SMC).

In this setting, the attacker (the untrusted kernel), which is executed
as privileged software outside TrustZone, is free to
manipulate its own page tables (which are different from the ones used
by the TrustZone service). Moreover, the attacker can invalidate
and clean cache lines, but may not use debugging
instructions to directly inspect the state of the caches.

The attacker uses the algorithm presented in
Figure~\ref{fig:conf1b}, however several considerations
must be taken into account to make the attack practical.
The attacker repeats the filling and probing phases
for each possible line index (128) and way (4) of the data-cache.
In practice, since the cache eviction strategy is pseudo random, the
filling phase is also repeated several times, until the L1 cache is
completely filled with the probing data (i.e. for every pair of
virtual addresses used, accessing to the two addresses yield different
values).

On Raspberry Pi 2, the presence of a unified L2 cache can obstruct the
probing phase: even if a cache line is evicted from the L1 cache by
the victim, the system can temporarily store the line into the L2 cache,
thus making the probing phase yield false negatives.
It is in general possible to extend the attack to deal with L2 caches
(by repeating the filling and probing phases for every line index and
way of the L2 cache subsystem), however, in Raspberry Pi 2 the L2 cache
is shared between the CPU and the GPU, introducing a considerable amount of noise in
the measurements. For this reason we always flushes the L2 cache between the step
$A5$ and $A6$ of the attack. We stress that this operation can be done by the
privileged software outside TrustZone without requiring any support by
TrustZone itself.

To make the demonstrator realistic, we allow the TrustZone service to
cache its own stack, heap, and static data. This pollutes the data
extracted by the probing phase of the attack: it can now yield false positives due to access of
the victim to such memory areas.
The key extraction algorithm can handle such false positives, but we decide to
filter them out to speed up the analysis phase.
For this reason, the attacker first identifies the
cache lines that are frequently evicted independently of the resulting
cipher-text (e.g. lines where the victim stack is probably allocated)
and removes them from the sets $E_{j,v}$.
As common, the AES implementation defines the SBox tables as
consecutive arrays. Since they all consists of 1 KB of data, the
cache lines where different SBoxes are allocated are non-overlapping, helping
the attacker in the task of reducing the sets $E_{j,v}$ to contain a
single line belonging to the table $T_4$ and of filtering out
all evictions that are due the previous rounds of AES.

For practical reasons we implemented the filling and probing phase
online, while we implemented the key extraction algorithm as a offline Python
program that analyses the logs saved by the online phase.
The complete online phase (including the set-up of the page tables)
consists of 552 lines of C, while the Python programs consists of
152 lines of code. The online attacker generates a stream of random
128 bits plain-texts and requests to the TrustZone service their encryption.
Thus, the frequency of the attacker's  measurements isolates one
AES encryption of one block per measurement. Moreover, even if the
attacker knows the input
plain-texts, they are not used in the offline phase.
We repeated the attack for several randomly generated keys and in the
worst case, the offline phase recovered the complete 128-bit key after
850 encryption in less than one second.

\subsection{Violating Spatial Isolation in a Hypervisor}\label{sec:impl:hyper}
A hypervisor is a low-level execution platform controlling accesses to system resources
and is used to provide isolated partitions on a shared hardware. The partitions are used to
execute software with unknown degree of trustworthiness. Each partition has 
access to its own resources and cannot encroach on protected parts
of the system, like the memory used by the hypervisor or the other partitions.
Here we demonstrate that a malicious operating system (guest) running on a hypervisor
can gain illicit access to protected resources using the mechanism described in
Section~\ref{sec:attack:int:dcache}.

As basis for our study  we use a hypervisor~\cite{SOFSEM}
that has been formally verified previously with respect to a cache-less model.
The hypervisor runs on an \armv~Cortex-A8 processor~\cite{Cortexa8},
where both  L1 and L2 caches are enabled.
On \armv~the address translation depends on the page
tables stored in the memory. Entries of the page tables encode a
virtual-to-physical mapping for a memory page as well as access
permissions and cacheability setting. On Cortex-A8 the MMU
consults the data cache before accessing the main memory whenever a
page table descriptor must be fetched.

The architecture is paravirtualized by the hypervisor for
several guests. Only the hypervisor is executing in privileged mode,
while the guests are executed in non-privileged mode and need to
invoke hypervisor functionality to alter the critical resources of the
system, like page tables.

A peculiarity of the hypervisor (and others~\cite{xen}) that makes it 
particularly relevant for our purpose is the use of so-called direct paging~\cite{SOFSEM}.
%
%
Direct paging enables a guest to manage its own memory space with assistance of the hypervisor. 
Direct paging allows the guest to allocate the page tables inside its own memory and to
directly manipulate them while the tables are not in active use by the
MMU. Then, the guest uses dedicated hypervisor calls to effectuate
and monitor the transition of page tables between passive and active state.
The hypervisor provides a number of system calls that support
the allocation, deallocation, linking, and activation of guest page
tables. These calls need to read the content of page tables that are
located in guest memory and ensure that the proposed MMU setup
does not introduce any illicit access grant. Thus
the hypervisor acts as a reference monitor of the page tables.

As described in Section~\ref{sec:attack:int:dcache}, on a Cortex-A8 processor sequential consistency
is not guaranteed if the same memory location is accessed by virtual aliases with mismatched cacheability
attributes. This opens up for vulnerabilities. The hypervisor may check
a page table by fetching its content from the cache. However, if the content of the page table in the cache
is clean and different from what has been placed by the attacker in the main memory and the page table is later evicted
from the cache, the MMU will use a configuration that is different from
what has been validated by the hypervisor.

Figure~\ref{cache:attacks} illustrates how a guest can use the aliasing of the physical memory to bypass
the validation needed to create a new page table.
Hereafter we assume that the guest and the hypervisor use two different virtual addresses to point to the same memory location.
Initially, the hypervisor (1) is induced to load a valid page table in the cache.
This can be done by writing a valid page table, requesting the hypervisor to verify and allocate it and
then requesting the hypervisor to deallocate the table.
Then, the guest (2) stores an invalid page table in the same memory location.
If the guest uses a non-cacheable virtual alias, the guest write (3) is directly applied to the memory bypassing the cache.
The guest (4) requests the hypervisor to validate and allocate this memory area, so that it can later be used
as page table for the MMU.
At this point, the hypervisor is in charge of verifying that the memory area contains a valid page table and of revoking
any direct access of the guest to this memory. In this way, a validated page table can be later used securely
by the MMU.
Since the hypervisor (4) accesses the same physical location through the cache, it
can potentially validate stale data, for example the ones fetched during the step (1).
At a later point in time, the validated data is evicted from the cache. This data is not written back to the memory since the hypervisor has only checked
the page table content and thus the corresponding cache lines are clean. 
Finally, the MMU (5) uses the invalid page table and its settings become
untrusted.

{
  Note that this attack is different from existing ``double mapping'' attacks.
  In double-mapping attacks the same physical memory is
  mapped ``simultaneously'' to multiple virtual memory addresses used by
  different agents; the attack occurs when the untrusted agent owns the
  writable alias, thus being able to directly modify the memory
  accessed by the trusted one.
  Here, the attacker  exploits the fact that the
  same physical memory is first allocated to the untrusted agent and then
  re-allocated to the trusted one. After that the ownership is
  transferred (after step A1), the untrusted agent has no mapping to this memory
  area. However, if the cache contains stale data the trusted agent may be compromised.
  Moreover, the attack does not depend on misconfiguration of the TLBs; the hypervisor is programmed to completely clean the TLBs whenever the MMU is reconfigured.
}

We implemented a malicious guest that managed to bypass the hypervisor validation using the above mechanism.
The untrusted data, that is used as configuration of the MMU, is used to obtain writable access to the master
page table of the hypervisor. This enables the attacker to reconfigure
its own access rights to all memory pages and
thus to completely take over the system.

\begin{figure}
{\small
  
  \begin{center}
\includegraphics[width=0.6\linewidth]{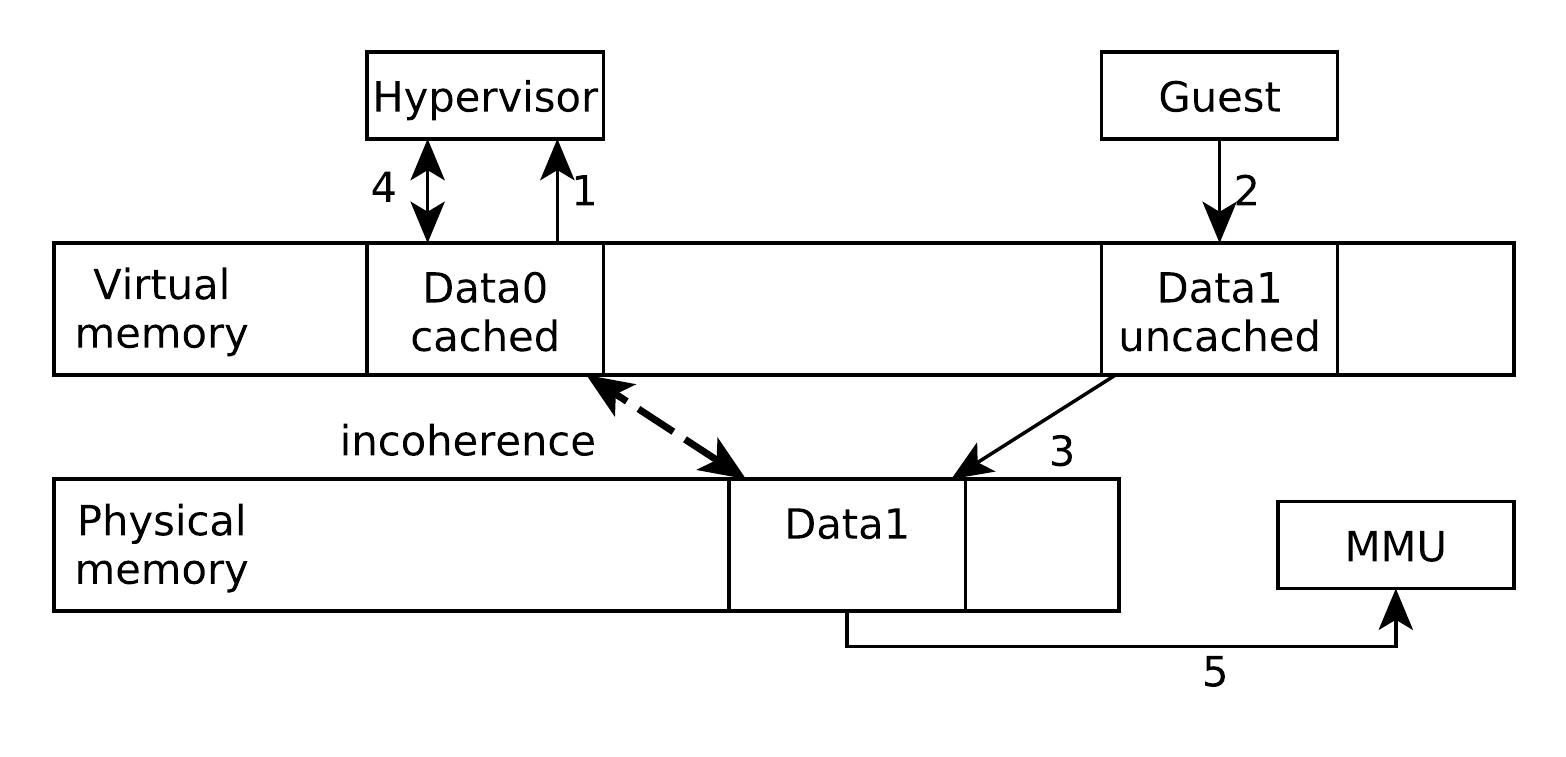}
  \end{center}
  \caption{Compromising integrity of a direct paging mechanism using incoherent memory. The MMU is configured to use a page table that was not validated by the hypervisor.}
  \label{cache:attacks}
}
\end{figure}

Not all hypervisors are subject to this kind of vulnerability. For example, if a hypervisor uses shadow paging, then 
guest pages are copied into the hypervisor's own memory where they are transformed into so-called shadow page tables. The guest has no access to this memory area 
and the hypervisor always copies cached data (if present), so the attack described above cannot be replicated.
On the other hand, the adversary can still attack secure services hosted by the hypervisor, for example a virtual machine introspector.
In~\cite{DBLP:conf/esorics/ChfoukaNGDE15} the hypervisor is used to implement a run-time monitor to protect an untrusted guest from its internal threats.
The monitor is deployed in a separate partition to isolate it from the untrusted guest.
The policy enforced by the monitor is executable space protection:
each page in the memory can be either writable or executable but not both at the same time. 
The monitor, via the hypervisor, intercepts all changes to the executable codes. This allows to use standard signature checking to prevent code injection.
Each time the guest operating system tries to execute an application, the monitor checks if the binary of the application has a valid signature.
In case the signature is valid, the monitor requests the hypervisor to
make
executable the physical pages that contain the binary code.
The security of this system depends on the fact that the adversary cannot directly modify a validated executable due to executable space protection.
However, if a memory block of the application code is accessed using virtual aliases with mismatched cacheability attributes, 
the untrusted guest can easily mislead the monitor to validate wrong data and execute unsigned code.

\subsection{Extraction of Exponent From a Modular Exponentiation Procedure}\label{sec:impl:sqmult}
\begin{figure}
\centering
\begin{subfigure}[t]{0.2\textwidth}
 \begin{Verbatim}[fontsize=\footnotesize]
y := 1
for i = m down to 1
    y = Square(y)
    y = ModReduce(y, N)
    if e_i == 1
       y = Mult(y,x)
       y = ModReduce(y, N)
\end{Verbatim}
\end{subfigure}
\caption{Square and multiply algorithm}\label{fig:sq:mult}
\end{figure}
The square and multiply algorithm of Figure~\ref{fig:sq:mult} is
often used to compute the modular exponentiation $x^e \mathop{\mathrm{mod}} N$,
where $e_m$ \dots $e_1$ are the bits of the binary representation of $e$.
This algorithm has been exploited in access-driven attacks,
since the sequence of function calls directly leaks
$e$, which corresponds to the private key in several decryption
algorithms.
Here we demonstrate that an attacker that is interleaved with a victim
can infer $e$ using the storage channel described in
Section~\ref{sec:attack:conf:icache}.

The attack was implemented on Raspberry Pi 2. We build a 
setting where a malicious process (e.g. a just in time compiler)
can self-modify its own code. Moreover, we implement a scheduler
that allows the attacker to be scheduled after every loop of the
victim.\footnote{Forcing the scheduler of a general purpose OS to grant such
high frequency of measurements  is out of the scope of this paper. The interested reader can refer to~\cite{Neve:2006:AAC:1756516.1756531,zhang2012cross}.}

The attacker uses the vector presented in
Figure~\ref{fig:conf2}, repeating the filling and probing phases
for every way of the instruction cache
and for every line index where the code of
the functions \verb|Mult| and \verb|ModReduce| can be mapped.
Due to the separate instruction and data L1 caches, the presence of the
L2 cache does not interfere with the probing phase. However, we must
ensure that the instruction overwritten in the step ($A2$) does not
sit in the L1 data-cache when the step ($A4$) is executed. Since user
processes cannot directly invalidate or clean cache lines, we
satisfy this requirement by adding a further step ($A3.b$). This step
writes several addresses whose line indices in the L1 data-cache are
the same of the address $\&A8$, thus forcing the eviction from the L1
data-cache of the line that has contains the instruction stored at $\&A8$.

We repeated the attack for several randomly generated values of $e$
and in each case the attacker correctly identified the execution path
of the victim. This accuracy is due to the simple environment (no
other process is scheduled except the victim and the attacker) and the
lack of noise that is typical in attacks that use time channels.



\section{Countermeasures}

Literature on access-based timing channel attacks suggests a number of well-known countermeasures. Specifically, for attacks on the confidentiality of AES encryption, a rather comprehensive 
list of protective means is provided in \cite{AESattack}. Some of the approaches are specific to AES, e.g., using registers instead of memory or dedicated hardware instructions for the SBox 
table look-up. Others are specific to the timing attack vector, e.g., reducing the accuracy of timing information available to the attacker. Still, there are well-known solutions addressing 
the presence of caches in general, thus they are suitable to defend against attacks built on the cache storage channel described in this paper.

In what follows we identify such known general countermeasures
(Sections~\ref{sec:count:disabled} and~\ref{sec:normalize}.1-5)
and propose new ones that are specific to the attack vector using uncacheable
aliases (Sections~\ref{sec:count:seqconst},~\ref{sec:normalize}.6, and~\ref{sec:countermeasure:hw}). In addition it is examined which countermeasures are suitable to protect against
the integrity threat posed by incoherent aliases in the memory system and propose a fix for the hypervisor example.

Different countermeasures are evaluated by implementing them for the AES and hypervisor scenarios introduced in the previous section and analysing their performance. The corresponding benchmark results are shown in Tables \ref{table benchmark hypervisor} and \ref{table benchmark AES}. Since our main focus is on verifying systems in the presence of caches, for each group of countermeasures we also sketch how a correctness proof would be conducted. Naturally, such proofs require a suitable model of the memory system including instruction and data caches.

It should be emphasised that the verification of the countermeasures is meant to be performed separately from the verification of the overall system which is usually assuming a much simpler
memory model for feasibility. The goal is to show that the countermeasures neutralise the cache storage channels and re-establish a coherent memory model. The necessary program verification
conditions from such a proof can then be incorporated into the overall verification methodology, supporting its soundness.

\subsection{Disabling Cacheability}\label{sec:count:disabled}

The simplest way to eliminate the cache side channel is to block an attacker from using the caches altogether. In a virtualization platform, like an operating system or a hypervisor,
this can be achieved by enforcing the memory allocated to untrusted guest partitions to be uncacheable. Consequently, cache-driven attacks on confidentiality and integrity of a system are
no longer possible. Unfortunately, this countermeasure comes at great performance costs, potentially slowing down a system by several orders of magnitude. On the other hand, a proof of the 
correctness of the approach is straight-forward. Since the attacker cannot access the caches, they are effectively invisible to him. The threat model can then be specified using 
a coherent memory semantics that is a sound abstraction of a system model where caches are only used by trusted code.

\begin{table}[t]
\scriptsize
\centering
\begin{tabular}{l|r|r|r|r|r}
LMbench micro benchmark &	Native  &	Hyp            &	ACPT            &	SelFl           &	Flush\\
\hline
null syscall            &       0.41    &       1.75           &      	1.76            &       1.77            &	1.76\\
read	                &       0.84    &	2.19           &	2.20            &	2.20	        &	2.38\\
write	                &       0.74	&       2.09           &	2.10            &	2.15		&       2.22\\
stat	                &       3.22    &	5.61           &	5.50            &       5.89            &	5.92\\
fstat                   &	1.19    &	2.53           &	2.55            &	2.56            &	2.65\\
open/close	        &       6.73	&       14.50          &	14.42	        &       14.86		&       14.71\\
select(10)	        &       1.86	&       3.29           &	3.30            &	3.33            &	3.42\\
sig handler install     &	0.85	&       2.87           &	2.89            &	2.92            &	2.95\\
sig handler overhead	&       4.43	&       14.45          &	14.48           &	15.11           &	14.91\\
protection fault	&       2.66	&       3.73	       &        3.83            &	3.91            &	3.70\\
pipe                    &	21.83	&       48.78	       &        47.79	        &       47.62           &	692.91\\
fork+exit               &	1978	&       5106           &        5126            &	6148            &	38787\\
fork+execve             &  	2068	&       5249	       &        5248            &	6285            &	39029\\
pagefaults              &	3.76	&       11.21          &	11.12	        &       21.55	        &	332.82\vspace{3mm}\\

Application benchmark   &	Native  &	Hyp             &	ACPT            &	SelFl   	&	Flush\\
\hline
tar (500K)              &	70	&       70	        &       70	        &       70	        &	190\\
tar (1M)                &	120     &	120             &	120             &	120             &	250\\
tar (2M)	        &       230     &	210             &	200             &	210             &	370\\
dd (10M)                &	90	&       140	        &       140             &	160             &	990\\
dd (20M)                &	190     &	260             &	260             &	570             &	1960\\
dd (40M)                &	330     &	500             &	450             &	600             &	3830\\
jpg2gif(5KB)            &	60      &	60              &	60              &	60              &	130\\
jpg2gif(250KB)          &	920     &	810             &	820             &	830             &	1230\\
jpg2gif(750KB)	        &       930	&       870             &	870             &	880             &	1270\\
jpg2bmp(5KB)	        &       40      &	40              &	40              &	40              &	110\\
jpg2bmp(250KB)	        &       1350    &	1340            &	1340            &	1350            &       1720\\
jpg2bmp(750KB)	        &       1440    &	1420	        &       1420            &	1430            &	1790\\
jpegtrans(270', 5KB)	&       10      &       10              &	10              &      	10              &	80\\
jpegtrans(270', 250KB)	&       220     &	240             &	240             &	250             &	880\\
jpegtrans(270', 750KB)	&       380     &	400             &	400             &	420		&       1050\\
bmp2tiff(90 KB)	        &       10      &	10              &	10              &	10		&       60\\
bmp2tiff(800 KB)	&       20	&       20              &	20              &	20              &	80\\
ppm2tiff(100 KB)	&       10      &	10              &	10              &	10              &	70\\
ppm2tiff(250 KB)	&       10      &	10              &	10              &	20		&       80\\
ppm2tiff(1.3 MB)	&       20      &	30              &	30              &	30              &	90\\
tif2rgb(200 KB)	        &       10      &	20              &	20              &	20              &       120\\
tif2rgb(800 KB)         &	40      &	40              &	40              &	50              &	270\\
tif2rgb(1.200 MB)	&       130     &	160             &	160             &	180             &	730\\
sox(aif2wav 100KB)      &	20      &	20              &	20              &	30              &	140\\
sox(aif2wav 500KB)	&       40	&       60	        &       60	        &       60		&       180\\
sox(aif2wav 800KB)	&       60	&       100	        &       100	        &       110		&       220
\end{tabular}\vspace{3mm}\\
\caption{Hypervisor Micro and Application Benchmarks.
LMbench micro benchmarks [$\mathit{\mu s}$] and application benchmarks [$\mathit{ms}$] for the Linux kernel v2.6.34 running natively on BeagleBone Black, paravirtualized 
on the hypervisor without protection against the integrity threat (Hyp), with always cacheable page tables (ACPT), with selective flushing (SelFl), and with full cache flushes on entry (Flush).}
\label{table benchmark hypervisor}
\end{table}
\subsection{Enforcing Memory Coherency}\label{sec:count:seqconst}
Given the dramatic slowdown expected for a virtualization platform, it seems out
of the question to completely deny the use of caches to untrusted guests.
Nevertheless, the idea of enforcing that guest processes cannot break 
{memory coherency} through uncacheable aliases still seems appealing.

\subsubsection{Always Cacheable Guest Memory} When making all guest memory uncacheable is prohibitively expensive, an intuitive alternative could be to just make all guest memory cacheable.
Indeed, if guests are user processes in an operating system this can be easily implemented by adapting the page table setup for user processes accordingly, i.e., enforcing cacheability for
all user pages. Then user processes cannot create uncacheable aliases to measure cache contents and start cache-based time-of-check-to-time-of-use attacks on 
their host operating system.

However, for hypervisors, where guests are whole operating systems, the approach has several drawbacks. First of all, operating systems are usually controlling memory mapped I/O devices which should be operated through uncacheable memory accesses. If a hypervisor would make all memory accesses of a guest OS cacheable, the OS will not be able to properly control I/O devices and probably not work correctly. Thus, making all untrusted guest memory cacheable only works for (rather useless) operating systems that do not control I/O devices. Furthermore, there are cases when a guest can optimise its performance by making seldomly used pages uncacheable \cite{qu2005using}.  

\subsubsection{$C\oplus U$ Policy}

Instead of making all guest pages cacheable, a hypervisor could make sure that
at all times a given physical page can either be accessed in cacheable or
uncacheable mode ($C\oplus U$ policy). To this end it would need to monitor the
page table setup of the guests and forbid them to define both cacheable and
uncacheable aliases of the same physical address. Then guests may set up
uncacheable virtual pages only if no cacheable alias exists for the targeted
physical page. Moreover, the hypervisor has to flush a cacheable page from the
caches when it becomes uncacheable, in order to remove stale copies of the page
that might be abused to set up an alias-driven cache attack. In this way, the
hypervisor would enforce {memory coherency} for the guest memory by making sure that no content from
uncacheable guest pages is ever cached and for cacheable pages cache entries may
only differ from main memory if they are dirty.

A Trust-zone cryptoservice that intends to prevent a malicious OS to use
memory incoherency to measure the Trust-zone accesses to the cache can use
TZ-RKP~\cite{azab2014hypervision} 
and extend its run-time checks to force the OS to respect the $C\oplus U$
policy.

\subsubsection{Second-Stage MMU}Still, for both the static and the dynamic case, the $C\oplus U$ policy may be expensive to implement for fully virtualizing hypervisors that rely on a second stage of address translation. For example, the ARMv8 architecture provides a second stage MMU that is controlled by the hypervisor, while the first stage MMU is controlled by the guests. Intermediate physical addresses provided by the guests are then remapped through the second stage to the actual physical address space. The mechanism allows also to control the cacheability of the intermediate addresses, but it can only enforce non-cacheability. In order to enforce cacheability, the hypervisor would need to enforce it on the first stage of translation by intercepting the page table setup of its guests, which creates an undesirable performance overhead and undermines the idea of having two independently operated stages of address translation.

\subsubsection{$W\oplus X$ Policy}Unfortunately, enforcing cacheability of
memory accesses does not protect against the instruction-cache-based
confidentiality threat described earlier. In order to prevent an attacker from
storing incoherent copies for the same instruction address in the memory system, the hypervisor would also need to prohibit self-modifying code for the guests, i.e., ensure that all
guest pages are either writable or executable ($W\oplus X$ policy). Since operating systems regularly use self-modification, e.g., when installing kernel updates or swapping in pages,
the association of pages to the executable or writable attribute is dynamic as well and must be monitored by the hypervisor. It also needs to flush instruction caches when an executable
page becomes writable.

Overall, the solutions presented above seem to be more suitable for paravirtualizing hypervisors, that are invoked by the guests explicitly to configure their virtual memory. Adding the required changes to the corresponding MMU virtualization functionality seems straightforward. In fact, for the paravirtualizing hypervisor presented in this paper a tamper-proof security monitor has been implemented and formally verified, which enforces executable space protection on guest memory and checks code signatures in order to protect the guests from malicious code injection \cite{DBLP:conf/esorics/ChfoukaNGDE15}.

\subsubsection{Always Cacheable Page Tables}To protect the hypervisor against
the integrity threat a lightweight specialization of the $C\oplus U$ policy
introduced above was implemented. It is based on the observation that
uncacheable aliases can only subvert the integrity of the hypervisor if they are
constructed for the inputs of its MMU virtualization functions. Thus the
hypervisor needs only to enforce the $C \oplus U$ policy, and consequently 
memory coherency, on its inputs. While this can be achieved by flushing the caches appropriately (see Section~\ref{sec:normalize}), a more efficient approach is to allocate the page tables of the guests in regions that are always cacheable. These regions of physical memory are fixed for each guest and the hypervisor only validates a page table for the guest if it is allocated in this area. In all virtual addresses mapping to the area are forced to be cacheable. Obviously, also the guest system needs to be adapted to support the new requirement on the allocation of page tables. However, given a guest system that was already prepared to run on the original hypervisor, the remaining additional changes should be straight-forward. For instance, the adaptation of the hypervisor required changes to roughly 35 LoC in the paravirtualized Linux kernel and an addition of 45 LoC to the hypervisor for the necessary checks.

The performance of the hypervisor with always cacheable page tables (ACPT) can be observed in Table~\ref{table benchmark hypervisor}. Compared to the original hypervisor there are basically
no performance penalties. In some cases the new version even outperforms the original hypervisor, due to the ensured cacheability of page tables. It turns out that in the evaluated 
Linux kernel, page tables are not always allocated in cacheable memory areas. The correctness of the approach is discussed in detail in Section~\ref{sec:verification}. The main verification
condition to be discharged in a formal proof of integrity is that the hypervisor always works on coherent memory, hence any correctness proof based on a 
coherent model also holds in a more detailed model with caches.

\begin{table}[t]
\scriptsize
\centering
\begin{tabular}{l|r|r|r|r}
\multirow{2}{*}{AES encryption}         & \multicolumn{2}{c|}{5 000 000 $\times$ 16B}    & \multicolumn{2}{c}{10 000 $\times$ 8KB} \\
                                        & Time          & Throughput    & Time          & Throughput\\
\hline
Original SBoxes                         & 23s           & 3.317 MB/s    & 13s           & 6.010 MB/s\\
Compact Last SBox                       & 24s	        & 3.179 MB/s    & 16s           & 4.883 MB/s\\
Scrambled Last SBox	                & 30s	        & 2.543 MB/s    & 20s           & 3.901 MB/s\\
Uncached Last SBox                      & 36s	        & 2.119 MB/s    & 26s           & 3.005 MB/s\\
Scrambled All SBoxes                    & 132s          & 0.578 MB/s    & 125s          & 0.625 MB/s\\
Uncached All SBoxes                     & 152s      	& 0.502 MB/s    & 145s          & 0.539 MB/s
\end{tabular}\vspace{3mm}\\
\caption{AES Encryption Benchmarks. AES encryption on Raspberry Pi 2 of one block (128 bits = 16 Bytes) and 512 blocks for different SBox layouts.} 
\label{table benchmark AES}
\end{table}

\subsection{Repelling Alias-Driven Attacks}\label{sec:normalize}
The countermeasures treated so far were aimed at restricting the behaviour of the attacker to prevent him from harvesting information from the cache channel or break
memory coherency in an attack on integrity. A different angle to the problem lies in focusing on the trusted victim process and ways it can protect itself against an unrestricted attacker
that is allowed to break memory coherency of its memory and run alias-driven cache attacks. The main idea to protect integrity against such attacks is to (re)establish 
coherency for all memory touched by the trusted process. For confidentiality, the idea is to adapt the code of the victim in a way that
its execution leaks no additional information to the attacker through the cache channel. Interestingly, many of the techniques described below are suitable for both purposes, neutralizing 
undesirable side effects of using the caches.

\subsubsection{Complete Cache Flush} One of the traditional means to tackle cache side channels is to flush all instruction and data caches before executing trusted code. In this way, 
all aliases in the cache are either written back to memory (in case they are dirty) or simply removed from the cache (in case they are clean). Any kind of priming of the caches by the 
attacker becomes ineffective since all his cache entries are evicted by the trusted process, foiling any subsequent probing attempts using addresses with mismatched cacheability. Similarly, 
all input data the victim reads from the attacker's memory are obtained from coherent main memory due to the flush, thus thwarting 
alias-driven attacks on integrity.

A possible correctness proof that flushing all caches eliminates the information side channel would rely on the assertion that, after the execution of the trusted service, an attacker
will always make the same observation using mismatched aliases, i.e., that all incoherent lines were evicted from the cache. Thus he cannot infer any additional knowledge from the cache
storage channel. Note, that here it suffices to flush the caches before returning to the attacker, but to protect against the integrity threat, data caches need to be flushed before any
input data from the attacker is read.

For performance evaluation the flushing approach was implemented in the AES and hypervisor examples. At each call of an AES encryption or hypervisor function, all data and instruction caches are flushed completely. Naturally this introduces an overhead for the execution of legitimate guest code due to an increased cache miss rate after calls to trusted processes. At the same time the trusted process gets slowed down for the same reason, if normally some of its data and instructions were still allocated in the caches from a previous call. Additionally the flushing itself is often expensive, e.g., for ARM processors the corresponding code has to traverse all cache lines in all ways and levels of cache to flush them individually. That all these overheads can add up to a sizeable delay of even one order of magnitude is clearly demonstrated by the benchmarks given in Tables \ref{table benchmark AES} and \ref{table benchmark hypervisor}.

\subsubsection{Cache Normalization}Instead of flushing, the victim can eliminate the cache information side channel by reading a sequence of memory cells so that the cache is brought into a known state. For instruction caches the same can be achieved by executing a sequence of jumps that are allocated at a set of memory locations mapping to the cache lines to be evicted. In the context of timing channels this process is called normalization. If subsequent memory accesses only hit the normalized cache lines, the attacker cannot observe the memory access pattern of the victim, because the victim always evicts the same lines. However the correctness of this approach strongly depends on the hardware platform used and the replacement policy of its caches. In case several memory accesses map to the same cache line the normalization process may in theory evict lines that were loaded previously. Therefore, in the verification a detailed cache model is needed to show that all memory accesses of the trusted service hit the cache ways touched during normalization.  

\subsubsection{Selective Eviction}The normalization method shows that cache side effects can be neutralized without evicting the whole
cache. In fact, it is enough to focus on selected cache lines that are critical for integrity or confidentiality. For example, the integrity threat on the hypervisor can be eliminated
by evicting the cache lines corresponding to the page table provided by the attacker. The flushing or normalization establishes  
memory coherency for the hypervisor's inputs, thus making sure it validates the right data. The method of selective flushing was implemented for the hypervisor scenario and benchmark
results in Table \ref{table benchmark hypervisor} show, as one would expect, that it is more efficient than flushing the whole cache, but still slower than our specialized ACPT solution. 

To ensure confidentiality in the AES example it suffices to evict the cache
lines occupied by the SBoxes. Since the 
incoherent entries placed in the same cache lines are removed by the victim using flushing or normalization, the attacker subsequently cannot measure key-dependent data accesses to
these cache lines. For the modular exponentiation example the same technique can be used, evicting only the lines in the instruction cache where the code of the functions
\verb|Mult| and \verb|ModReduce| is mapped. 

The correctness of selective eviction of lines for confidentiality depends on the fact that accesses to other lines do not leak secret information through the cache side channel, e.g., for the AES encryption algorithm lines that are not mapped to an SBox are accessed in every computation, independent of the value of the secret key. Clearly, this kind of trace property needs to be added as a verification condition on the code of the trusted service. Then the classic confidentiality property can be established, that observations of the attacker are the same in two computations where only the initial values of the secret are different (non-infiltration~\cite{Heitmeyer:2008:AFM:1340674.1340715}).

\subsubsection{Secret-Independent Memory Accesses}The last method of eliminating the cache information side channel is a special case of this approach. It aims to transform the victim's code such that it produces a memory access trace that is completely independent of the secret, both for data accesses and instruction fetches. Consequently, there is no need to modify the cache state set up by the attacker, it will be transformed in the same way even for different secret values, given the trusted service receives the same input parameters and all hidden states in the service or the cache model are part of the secret information.

As an example we have implemented a modification of AES suggested in \cite{AESattack}, where a 1KB SBox look-up table is scrambled in such a way that a look-up needs to touch all cache lines occupied by the SBox. In our implementation on Raspberry Pi 2 each L1 cache line consists of 64 Bytes, hence a 32bit entry is spread over 16 lines where each line contains two bits of the entry. While the decision which 2 bits from every line are used is depending on the secret AES key, the attacker only observes that the encryption touches the 16 cache lines occupied by the SBox, hence the key is not leaked.

Naturally the look-up becomes more expensive now
because a high number of bitfield and shift operations is required to reconstruct the original table entry.
For a single look-up, a single memory access is substituted by 16 memory accesses, 32 shifts,
16 additions and 32 bitfield operations.
The resulting overhead is roughly 50\% if only the last box is scrambled (see Table~\ref{table benchmark AES}). This is sufficient if all SBoxes are mapped to the same cache lines and the attacker cannot interrupt the trusted service, probing the intermediate cache state. Scrambling all SBoxes seems prohibitively expensive though, slowing the encryption down by an order of magnitude. However, since the number of operations depends on the number of lines used to store the SBox, if the system has bigger cache lines the countermeasure becomes cheaper.

\subsubsection{Reducing the Channel Bandwidth}

Finally for the AES example there is a countermeasure that does not completely eliminate the cache side channel, but makes it harder for the attacker to derive the secret key. The idea described in \cite{AESattack} is to use a more compact SBox that can be allocated on less lines, undoing an optimization in wolfSSL for the last round of AES. There the look-up only needs to retrieve one byte instead four, still the implementation word-aligns these bytes to avoid bit masking and shifting. By byte-aligning the entries again, the table shrinks by a factor of four, taking up four lines instead of 16 on Raspberry Pi 2. Since the attacker can distinguish less entries by the cache line they are allocated on, the channel leaks less information. This theory is confirmed in practice where retrieving the AES key required about eight times as many encryptions compared to the original one. At the same time, the added complexity resulted in a performance delay of roughly 23\% (see Table~\ref{table benchmark AES}).

\subsubsection{Detecting Memory Incoherency}
A reference monitor (e.g. the hypervisor) can counter 
the integrity threat by preventing the invocation of the 
critical functions (e.g. the MMU virtualization functions) if 
memory incoherency is detected.
The monitor can itself use mismatched cache attributes to
detect incoherency as follows.
For every address that is used as the input of a critical function, the
monitor checks if reading the location using the cacheable and
non-cacheable aliases yield the same result. If the two reads differs, then 
memory incoherency is detected and the monitor rejects the request, otherwise
then request is processed.

\subsection{Hardware Based Countermeasures}\label{sec:countermeasure:hw}
The cache-based storage channels rely on misbehaviour of the
system due to misconfigurations. For this reason, the hardware
could directly take care of them. The vector based on mismatched
cacheability attributes can be easily made ineffective if the processor does
not ignore unexpected cache hits. For example, if a physical address is written
using a non-cacheable alias, the processor can invalidate every line
having the corresponding tag. 
Virtually indexed caches are usually equipped with similar mechanisms
to guarantee that there can not be aliases inside the cache itself.

Hardware inhibition of the vector that uses the instruction cache 
can be achieved using a snooping mechanism that invalidates instruction cache lines whenever self-modification is detected, similar to what happens in x64 processors.
In architectures that perform weakly ordered memory accesses and aggressive speculative execution, implementing such a mechanism can become quite complex and make the out-of-order execution
logic more expensive. There is also a potential slow-down due to misspeculation when instructions are fetched before they are overwritten.

\vskip 10pt
Overall, the presented countermeasures show that a trusted service can be
efficiently secured against alias-driven cache attacks if two properties are
ensured: (1) for integrity, the trusted service may only
  accesses  coherent  memory (2) for confidentiality, the cache must be transformed in a way such
  that the attacker cannot observe memory accesses depending on
  secrets. In next section, a verification methodology presented that aims to prove these properties for the code of the trusted service.

\section{Verification Methodology}\label{sec:verification}
The attacks presented in Section~\ref{sec:experiments} demonstrate that the presence of caches can make a trustworthy, i.e. formally verified,
program vulnerable to both confidentiality and security threats. These vulnerabilities depend on the fact that for some resources (i.e.
some physical addresses of the memory) the actual system behaves
differently from what is predicted by the formal model:
we refer to this misbehaviour as ``loss of sequential consistency''.

As basis for the study we assume a low level program (e.g.~a hypervisor, a separation kernel, a security monitor, or a TrustZone crypto-service)
running on a commodity CPU such as the ARMv7 Cortex A7 of Raspberry Pi 2. We refer to the trusted program as ``the kernel''. 
The kernel shares the system with an untrusted application, henceforth ``the application''. We assume that the kernel has been
subject to a pervasive formal verification that established its functional correctness and isolation properties using a model that reflects the ARMv7 ISA
specification to some level of granularity. For instance for both seL4 and the Prosper kernel the
processor model is based on Anthony Fox's cacheless L3 model of ARMv7~\footnote{In case of Prosper, augmented with a detailed model of the MMU \cite{SOFSEM}.}.

We identify two special classes of system resources (read: Memory locations):
\begin{itemize}
\item Critical resources: These are the resources whose integrity must be protected, but which the 
application needs access to for its correct operation.
\item Confidential resources: These are the resources that should be read protected against the
application.
\end{itemize}
There may in addition be resources that are both critical and confidential. We call those \textit{internal} resources.
Examples of critical resources are the page tables of a hypervisor,
the executable code of the untrusted software in a run-time monitor,
and in general the resources used by the invariants needed for the verification of
functional correctness.
Confidential (internal) resources can be cryptographic keys, internal kernel data structures, or the memory of a guest colocated with the application.

The goal is to repair the formal analysis of the kernel, reusing as much as possible of the prior analysis.
In particular, our goals are:
\begin{enumerate}
\item To demonstrate that critical and internal resources
cannot be directly affected by the application and that for these resources the actual
system behaves according to the formal specification (i.e.~that 
sequential consistency is preserved and the integrity attacks
described in Section~\ref{sec:attack:int:dcache} cannot succeed).
\item To guarantee that no side channel is present due to caches,
i.e.~that the real system exposes all and only the channels that are present
in the formal functional specification that have been used to verify
the kernel using the formal model.
\end{enumerate}

\subsection{Repairing the Integrity Verification}\label{sec:ver:integrity}
For simplicity, we assume that the kernel accesses all resources
using cacheable virtual addresses.
To preserve integrity we must ensure two properties:
\begin{itemize}
\item That an address belonging to a critical resource cannot be directly or indirectly modified by the
application. 
\item Sequential consistency of the kernel.
\end{itemize}
 The latter property is equivalent to guaranteeing that what
is observed in presence of caches is exactly what is predicted by the
ISA specification.

The verification depends on a system invariant that must be preserved
by all executions:
For every address that belongs to the critical and internal resources,
if there is a cache hit and the corresponding cache line
differs from the main memory then the cache line must be dirty.
The mechanism used to establish this invariant depends on the
specific countermeasure used.
It is obvious that if the caches are disabled
(Section~\ref{sec:count:disabled}) the invariant holds, since the
caches are always empty.
In the case of ``Always Cacheable Memory''
(Section~\ref{sec:count:seqconst}) the invariant is preserved because
no non-cacheable alias is used to access these resources: the
content of the cache can differ from the content of the memory only
due to a memory update that changed the cache, thus the corresponding
cache line is dirty.
Similar arguments apply to the $C\oplus U$ Policy, taking into account
that the cache is cleaned whenever a resource type switch from cacheable ($C$) to
uncacheable ($U$) and vice versa.

  More complex reasoning is necessary for other countermeasures, where the
  attacker can build uncacheable aliases in its own memory. In this case we know that the
system is configured so that the application cannot write the critical resources, since otherwise the integrity
property cannot be established for the formal model in the first place. Thus, 
if the cache contains critical or internal data different from main memory it must have been written there by the kernel that only uses cacheable memory only, hence the line is dirty as well.

To show that a physical address $pa$ belonging to a critical resource cannot
not be directly or indirectly modified by the application we proceed
as follows.
By the assumed formal verification, the application has
no direct writable access to $pa$, otherwise the integrity
property would not have been established at the ISA level.
Then, the untrusted application can not directly update $pa$ neither in
the cache nor in the memory. The mechanism that can be used to indirectly update the
view of the kernel of the address $pa$ consists in evicting a cache line that has
a value for $pa$ different from the one stored in the memory and that
is not dirty. However, this case is prevented by the new invariant.

Proving that sequential consistency of the kernel is preserved is
trivial: The kernel always uses cacheable addresses so it is unable to break
the new invariant: a memory write always updates the cache line if
there is a cache hit.

\subsection{Repairing the Confidentiality Verification}
Section~\ref{sec:attacks} demonstrates the capabilities of the
attacker: Additionally to the resources that can be accessed in the
formal model (registers, memory locations access to which is granted
by the MMU configuration, etc) the attacker is able to measure which cache
lines are evicted.
Then the attacker can (indirectly) observe all the resources
that can affect the eviction. Identifying this set of resources is
critical to identify the constraints that must be satisfied by the
trusted kernel. For this reason, approximating this set
(e.g. by making the entire cache observable) can strongly reduce the freedom of the trusted code. 
A more refined (still conservative) analysis considers observable by
the attacker the cache line tag\footnote{On direct mapped caches, we can disregard the line tag, because they contain only one way for each line. In order to observe the tags of addresses accessed by the kernel, the attacker requires at least two ways per cache line: one that contains an address accessible by the kernel and one that the attacker can prime in order to measure whether the first line has been accessed.} 
  and whether a cache line is empty (cache line emptiness). Then to guarantee confidentiality it is necessary to ensure that, while the application is executing, the cache line tag and emptiness never depend on the confidential resources.
We stress that this is a sufficient condition to guarantee that
no additional information is leaked due to presence of caches with respect to
the formal model

Showing that the condition is met by execution of the application
is trivial. By the assumed formal verification we already know that the
application has no direct read access (e.g.~through a virtual memory mapping)
to confidential resources. On the other hand, the kernel is able to access these
resources, for example to perform encryption. The goal is to show that the
caches do not introduce any channel that has not been taken into
account at the level of the formal model. Due to the overapproximation described above,
this task is reduced to a ``cache-state non-interference property'', i.e.
showing that if an arbitrary functionality of the kernel is executed then the cache line
emptiness and the line tags in the final state do not depend on confidential data.

The analysis of this last verification condition depends on the countermeasure used by the kernel. If the kernel always terminates
with caches empty, then
the non-interference property trivially holds, since a constant value
can not carry any sensible information.
This is the case if the kernel always flushes the caches before
exiting, never use cacheable aliases (for both program counter and
memory accesses) or the caches are completely disabled.

In other cases (e.g. ``Secret-Independent Memory Accesses'' and ``Selective
Eviction'') the verification condition is further decomposed to two
tasks:
\begin{enumerate}
 \item Showing that starting from two states that
   have the same cache states, if
   two programs access at the same time the same memory locations then the final states have the same
   cache states.
 \item Showing that the sequence of memory accesses performed by the
   kernel only depends on values that are not confidential.
\end{enumerate}
The first property is purely architectural and thus independent of the
kernel. Hereafter we summarise the reasoning for a system with a
single level of caches, with separated instruction and data caches and
whose caches are physically indexed and physically tagged (e.g. the L1
memory subsystem of ARMv7 CPUs).
We use
$\ArmMachineVar_1,\ArmMachineVar'_1,\ArmMachineVar_2,\ArmMachineVar'_2$
to range over machine 
states and $\ArmMachineVar_1 \to \ArmMachineVar'_1$ to represent the
execution of a single instruction.
From an execution $\ArmMachineVar_1 \to \ArmMachineVar_2 \dots \to
\ArmMachineVar_n$ we define two projections:
$\pi_I(\ArmMachineVar_1 \to \ArmMachineVar_2 \dots \to
\ArmMachineVar_n)$ is the list of encountered program counters and
$\pi_D(\ArmMachineVar_1 \to \ArmMachineVar_2 \dots \to
\ArmMachineVar_n)$ is the list of executed memory operations (type of
operation and physical address).
We define $\RelP$ as the biggest relation such that if $\ArmMachineVar_1 \RelP \ArmMachineVar_2$ then
for both data and instruction cache
\begin{itemize}
  \item a line in the cache of $\ArmMachineVar_1$ is empty if and only
    if the same line is empty in $\ArmMachineVar_2$, and
  \item the caches of $\ArmMachineVar_1$ and $\ArmMachineVar_2$ have the
    same tags for every line.
\end{itemize}
The predicate $\RelP$ is preserved by executions
$\ArmMachineVar_1 \to \dots$ and $\ArmMachineVar_2 \to \dots$
if the corresponding projections are \emph{cache safe}:
\begin{inparaenum}[(i)]
\item the instruction tag and index of $\pi_I(\ArmMachineVar_1
  \to \dots)[i]$ is equal to 
  the instruction tag and index of $\pi_I(\ArmMachineVar_2
  \to \dots)[i]$
\item if  $\pi_D(\ArmMachineVar_1
  \to \dots)[i]$ is a read (write) then 
  $\pi_D(\ArmMachineVar_2  \to \dots)[i]$ is a read (write)
\item the cache line tag and index of the address in $\pi_D(\ArmMachineVar_1
  \to \dots)[i]$ is equal to 
  the cache line tag and index of the address in $\pi_D(\ArmMachineVar_2
  \to \dots)[i]$
\end{inparaenum}

Consider the example in Figure~\ref{fig:conf1a}, where 
 $va_3$ and $va_4$ are different addresses. In our
current setting this is secure only if $va_3$ and $va_4$ share the same
data cache index and tag (but they could point to different positions within a
line). Similarly, 
the example in Figure~\ref{fig:conf2}
is secure only if the addresses of both targets of the conditional
branch have the same instruction cache index and tag.
Notice that these conditions are less restrictive than the ones
imposed by the program counter security model.
Moreover, these restrictions dot not forbid completely data-dependent
look-up tables.
For example, the scrambled implementation of AES presented In
Section~\ref{sec:normalize} satisfies the rules that we identified
even if it uses data-dependent look-up tables.

In practice, to show that the trusted code satisfies the cache safety policy,
we rely on a relational observation equivalence and we use
existing tools for relational verification that
support trace based observations. In our experiments we adapted the
tool presented in~\cite{DBLP:conf/ccs/BalliuDG14}.
The tool executes two analyses of the code.
The first analysis handles the instruction cache:
we make every instruction observable and we
require that the matched instructions have the same
set index and tag for the program counter. The second 
analysis handles the data cache:
we make every memory access an
observation and we require that the matched memory accesses use the
same set index and tag (originally the tool considered observable only memory
writes and required that the matched memory writes access the same
address and store the same value).
Note that the computation of set index and tag are platform-dependent, thus when porting the same verified code to a processor, whose caches use a different method for indexing lines, the code might not be cache safe anymore.
To demonstrate the feasibility of our approach we applied the tool to
one functionality of the hypervisor described in
Section~\ref{sec:impl:hyper}, which is implemented by 60 lines of
assembly and whose analysis required $183$ seconds.


\section{Related Work}\label{sec:related}
Kocher~\cite{Kocher:1996:TAI:646761.706156} and Kelsey et
al.~\cite{Kelsey:2000:SCC:1297828.1297833} were the  first  to  demonstrate
cache-based side-channels. 
They showed that these channels contain enough information to
enable an attacker to extract the secret key of cryptographic
algorithms.
Later, Page formally studied cache
side-channels and showed how one can use them to attack 
na\"{\i}ve implementations of the DES cryptosystem~\cite{Page02theoreticaluse}.
Among the existing cache attacks, the trace-driven and  access-driven
attacks are the most closely related to this paper since
they can be reproduced using the vectors presented in Section~\ref{sec:attacks}.

In trace-driven attacks~\cite{Page02theoreticaluse} an adversary
profiles the cache activities while the victim is executed.
Ac{\i}i\c{c}mez showed
a trace-driven cache attack on the first two rounds of AES~\cite{Aciicmez:2006:TCA:2092880.2092891}, which
has been later improved and extended by X. Zhao~\cite{eprint-2010-22957}
to compromise a CLEFIA block cipher. 
A similar result is reported in~\cite{Bertoni:2005:APA:1058430.1059131}. 
In an access-driven, or Prime+Probe, attack the adversary can determine the cache sets
modified by the victim. 
In several papers this technique is used to compromise real
cryptographic algorithms like RSA~\cite{percival2005cache,InciGAES15} and
AES~\cite{Gullasch:2011:CGB:2006077.2006784,
  Neve:2006:AAC:1756516.1756531,
  Tromer:2010:ECA:1713125.1713127}.

Due to the security concerns related to cache channels,
research on the security implications of shared caches has
so far been focusing on padding~\cite{ZhangAM12} and
mitigation~\cite{Agat00} techniques to address timing channels.
Notably, Godfrey and Zulkernine have proposed efficient host-based
solutions to close timing channels through selective flushing and
cache partitioning~\cite{GodfreyZ14}.
In the \textsc{StealthMem}
approach~\cite{Kim:2012:SSP:2362793.2362804}
each guest is given exclusive access to a small portion
of the shared cache for its security critical computations. By
ensuring that this stealth memory is always allocated in the cache, no
timing differences are observable to an attacker.

In literature, few works investigated cache based storage channels.
In fact, all implementations of the above attacks use timing channels as 
the attack vector.
Brumley~\cite{Brumley15} recently conjectured 
 the existence of a storage channel that can be implemented using cache debug
functionality on some ARM embedded microprocessors. However, the ARM
technical specification~\cite{Cortexa7}  explicitly states that
such debug instructions can be executed only by privileged software
in TrustZone, making practically 
impossible for an attacker to access them with the exception of
a faulty hardware implementation.

The attack based on mismatched cacheability attributes opens up for
TOCTTOU like vulnerabilities.
Watson \cite{Watson07} demonstrated this vulnerability for Linux system call wrappers.
A similar approach is used in~\cite{Bratus:2008:TTT:1422929.1422932} to invalidate security
guarantees, attestation of a platform's software, provided by a Trusted Platform Module (TPM).
TPM takes  integrity measurements only before software is loaded into the memory, and it assumes that once the software is loaded 
it remains unchanged. However, this assumption is not met if the
attacker can indirectly change the software before is used. 

Cache-related architectural problems have been exploited before to 
bypass memory protection.
In~\cite{wojtczuk2009attacking,duflot2009getting} the authors use a
weakness of some Intel x86 implementations to 
bypass SMRAM protection and execute 
malicious code in 
System Management Mode (SMM).
The attack relies on the fact that the SMRAM protection is implemented by the
memory controller, which is external to the CPU cache. 
A malicious operating system first marks the SMRAM memory region as cacheable and write-back,
then it writes to the physical addresses of the SMRAM.
Since the cache is unaware of the SMRAM
configuration, the writes
are cached and do not raise exceptions.
When the execution is transferred to SMM, the
CPU fetches the instructions from the
poisoned cache.
While this work shows similarities to the integrity threat posed by cache
storage channels, the above attack is specific to certain Intel implementations
and targets only the highest security level of x86.
On ARM, the cache keeps track which lines have been filled 
due to accesses performed by TrustZone SW.
The TrustZone SW can configure via its page tables the memory regions that
are considered ``secure'' (e.g. where its code and internal data structure are
stored). A TrustZone access to a secure memory location can hit a cache line
only if it belongs to TrustZone.

The attack vectors for data caches presented in this paper abuse undefined behaviour in the ISA specification (i.e., accessing the same memory address with different cacheability types)
and deterministic behaviour of the underlying hardware (i.e., that non-cacheable accesses completely bypass the data caches and unexpected cache hits are ignored). While we focused on an 
ARMv7 processor here, there is a strong suspicion that other architectures exhibit similar behaviour. In fact, in experiments we succeeded to replicate the behaviour of the memory subsystem 
on an ARMv8 processor (Cortex-A53), i.e., uncacheable accesses do not hit valid entries in the data cache. For Intel x64, the reference manual states that memory 
type aliases using the page tables and page attribute table (PAT) ``may lead to undefined operations that can result in a system failure'' (\cite{Intel}, Vol.~3, 11.12.4). It is also explicitly 
stated that the accesses using the (non-cacheable) WC memory type may not check the caches. Hence, a similar behaviour as on ARM processors should be expected. On the other hand, some Intel
processors provide a self-snooping mechanism to support changing the cacheability type of pages without requiring cache flushes. It seems to be similar in effect as the hardware countermeasure
suggested in Section~\ref{sec:countermeasure:hw}. In the Power ISA manual (\cite{Power}, 5.8.2), memory types are assumed to be unique for all aliases of a given address. Nevertheless this is 
a software condition that is not enforced by the architecture. When changing the storage control bits in page table entries the programmer is required to flush the caches. This also hints 
to the point that no hardware mechanisms are mandated to handle unexpected cache hits. 

Recently, several works successfully verified low level execution platforms
that provide trustworthy mechanisms to isolate commodity software.
In this context caches are mostly excluded from the analysis. An exception is
the work by Barthe et al. \cite{DBLP:conf/types/BartheBCCL13} that provide an abstract
model of cache behaviour sufficient to replicate various timing-based exploits
and countermeasures from the literature such as \textsc{StealthMEM}.

The verification of seL4 assumes that caches are correctly
handled~\cite{KleinAEMSKH14} and ignores timing channels.  The
bandwidth of timing channels in seL4 and possible countermeasures were
examined by Cock et al~\cite{CockGMH14}. While
storage based channels have not been addressed, integrity of the
kernel seems in practice to be preserved by the fact that system call
arguments are passed through registers only.

The VerisoftXT project targeted the verification of
Microsoft Hyper~V and a semantic stack was devised to underpin the
code verification with the VCC tool \cite{CohenPS13}. Guests are
modelled as full x64 machines where caches cannot be made transparent
if the same address is accessed in cacheable and uncacheable mode,
however no implications on security have been discussed. Since the
hypervisor uses a shadow page algorithm, where guest translations are
concatenated with a secure host translation, the integrity properties do
not seem to be jeopardised by any actions of the guest.

Similarly the Nova \cite{SteinbergK10,TewsVW09} and CertiKOS
\cite{GuVFSC11} microvisors do not consider caches in their formal
analysis, but they use hardware which supports a second level address
translation which is controlled by the host and cannot be affected by
the guest. Nevertheless the CertiKOS system keeps a partition
management software in a separate partition that can be contacted by
other guests via IPC to request access to resources. This IPC
interface is clearly a possible target for attacks using uncacheable
aliases.

In any case all of the aforementioned systems seem to be
vulnerable to cache storage channel information leakage, assuming they
allow the guest systems to set up uncacheable memory mappings. In
order to be sound, any proof of information flow properties then needs
to take the caches into account. In this paper we show for the first
time how to conduct such a non-interference proof that treats also possible
data cache storage channels.



\section{Concluding Remarks}
We presented novel cache based attack vectors that use storage
channels and we demonstrated their usage to threaten integrity and
confidentiality of real software.
To the best of our knowledge, it is the first time that cache-based
storage channels are demonstrated on commodity hardware. 

The new attack vectors partially invalidate the results of formal
verification performed at the 
ISA level. In fact, using storage-channels, the adversary can extract
information without accessing variables that are external to the ISA
specification. This is not the case for timing attacks and power
consumption attacks. 
For this reason it is important to provide methodologies to fix the
existing verification efforts.
We show that for some of the existing countermeasures this task can be
reduced to checking relational observation equivalence. To make this
analysis practical, we adapted an existing tool~\cite{DBLP:conf/ccs/BalliuDG14}
to check the conditions that are sufficient to prevent information
leakage due to the new cache-channels. In general, the additional checks in the code verification need to be complemented by a correctness proof of a given countermeasure on a 
suitable cache model. In particular it has to be shown that memory coherency for the verified code is preserved by the countermeasure and that an attacker cannot observe sensitive 
information even if it can create non-cacheable aliases.

The attack presented in Section~\ref{sec:attack:int:dcache} 
raises particular concerns, since it poses integrity threats that
cannot be carried out using timing channels.
The possible victims of such an attack are systems where
the ownership of memory is transferred from the untrusted agent 
to the trusted one and where the trusted agent checks the content of this memory
before using it as parameter of a critical function.
After that the ownership is
transferred, if the cache is not clean, the trusted agent may validate stale
input while the critical function uses different data.
The practice of transferring ownership between security domains is usually
employed to reduce memory copies and is used for example by
hypervisors that use direct paging, run-time monitors
that inspect executable code to prevent execution of malware, as well as reference monitors
that inspect the content of IP packets or validate requests for device drivers.

There are several issues we leave out as future work.
We did not provide a mechanism to check the security of some
of the countermeasures like Cache Normalisation and we did not
apply the methodology that we described to a complete software.
Moreover, the channels that we identified probably do not cover all the existing
storage side channels. Branch prediction, TLBs, sharebility attributes
are all architectural details that, if misconfigured,
can lead to behaviours that are inconsistent with the ISA
specification. If the adversary is capable of configuring these
resources, like in virtualized environments, it is important to
identify under which conditions the trusted software preserves its security
properties. 

From a practical point of view, we focused our experiments on
exploiting the L1 cache. For example, to 
extract the secret key of the AES service on Raspberry Pi 2
we have been forced to flush and clean the L2 cache. The reason is
that on this platform the L2 cache is shared with the GPU and we have little to no knowledge
about the memory accesses it performs. On the other hand, shared L2
caches open to the experimentation with concurrent attacks, where the attacker 
can use a shader executed on the GPU.
Similarly, here we only treated cache channels on a
  single processor core.
Nevertheless the same channels can be built in a multi-core settings
using the shared caches (e.g. L2 on Raspberry Pi 2). The new vectors can then be
used to
replicate known timing attacks on shared caches
(e.g.~\cite{InciGAES15}).


\chapter{Formal Analysis of Countermeasures against Cache Storage Side Channels}\label{paper:csf}
\chaptermark{Formal Analysis of Countermeasures}
\backgroundsetup{position={current page.north east},vshift=1cm,hshift=-9cm,contents={\VerBar{OrangeRed}{2cm}}}
\BgThispage

\begin{center}
Hamed Nemati, Roberto Guanciale, Christoph Baumann, Mads Dam
\end{center}

\begin{abstract}
Formal verification of systems-level software such as 
hypervisors and operating systems can enhance system trustworthiness. However,
without taking  low level features like caches into account
the verification may become unsound. While this is a well-known fact w.r.t.~timing leaks, few works have addressed latent cache storage side-channels.
We present a verification methodology to analyse 
soundness of countermeasures used to neutralise cache storage channels.
We apply the proposed methodology to existing countermeasures,
showing that they allow to restore integrity and
prove confidentiality of the system.
We decompose the proof effort into verification conditions that allow for an
easy adaption of our strategy to various software and hardware platforms.
As case study, we extend the verification
of an existing hypervisor whose integrity can be tampered
using cache storage channels.
We used the HOL4 theorem prover
to validate our security analysis,
applying the verification methodology to formal models of ARMv7
and ARMv8.
\end{abstract}

\newenvironment{prove}[1]{\noindent { {\textit{Proof}}\textit{#1.}~}}{\hfill $\square$} \let\proof\relax\let\endproof\relax

\newtheorem*{definition*}{Definition}
\newtheorem{prop}[theorem]{Proof Obligation}
\newtheorem{assumption}[theorem]{Assumption}

\usetikzlibrary{calc,decorations.pathmorphing,shapes,arrows}
\newcommand\xrsquigarrow[1]{%
\mathrel{%
\begin{tikzpicture}[baseline= {( $ (current bounding box.south) + (0,-0.5ex) $ )}]
  \node[inner sep=.5ex] (a) {$\scriptstyle #1$};
  \path[draw,implies-,double distance between line centers=1.5pt,decorate,
    decoration={zigzag,amplitude=0.7pt,segment length=1.2mm,pre=lineto,
    pre   length=4pt}] 
    (a.south east) -- (a.south west);
\end{tikzpicture}}%
}

\tikzset{decorate sep/.style 2 args=
{decorate,decoration={shape backgrounds,shape=circle,shape size=#1,shape sep=#2}}}
\newcommand*{\DashedArrow}{\mathbin{\tikz 
\draw [->,
line join=round,
decorate, decoration={
    zigzag,
    segment length=3,
    amplitude=1.2,post=lineto,
    post length=2pt
}]  (0,0.8ex) -- (1.0em,0.8ex);
}}
\newcommand{\WkTrs}[1]{\DashedArrow_{#1}}
\newcommand{\WTrs}[1]{\rightarrow^{*}_{#1}}
\newcommand{\WTrsN}[2]{\rightarrow^{#2}_{#1}}
\newcommand{\implc}{\Rightarrow}

\newcommand{\emptymsg}{\perp}
\newcommand{\idx}{id}
\newcommand*\midpoint[1]{\widebar{#1}}
\newcommand{\Exe}{\mathit{Cv}}
\newcommand{\altExe}{\mathit{Mv}}
\newcommand{\Trans}[1]{\xrightarrow{}}

\newcommand{\Act}{\mathbb{A}}
\newcommand{\Csh}{\mathbb{C}}
\renewcommand{\L}{\mathbb{L}}
\newcommand{\D}{\mathbb{D}}
\newcommand{\His}{\Act^\ast}
\newcommand{\his}{\mathit{H}}
\newcommand{\fhis}[2]{\his(#1,#2)}
\newcommand{\Chis}{\mathit{hist}}
\newcommand{\SL}{\mathit{slice}}
\newcommand{\M}{\mathbb{M}}
\newcommand{\N}{\mathbb{N}}
\newcommand{\B}{\mathbb{B}}
\newcommand{\T}{\B^{T}} 
\newcommand{\I}{\mathbb{B}^N}
\renewcommand{\S}{\mathbb{SL}}
\newcommand{\VA}{\mathbb{VA}}
\newcommand{\PA}{\mathbb{PA}}
\newcommand{\regId}{\mathbb{RG}}

\newcommand{\RegionUser}{\mathbf{G_m}}
\newcommand{\RegionKernel}{\mathbf{K_m}}
\newcommand{\RegionKernelV}{\mathbf{K_{vm}}}

\newcommand{\SimR}{\mathrel{\mathcal{R}_{\mathit{sim}}}}
\newcommand{\BsimR}{\mathrel{\mathcal{R}_{\tiny {sim+\oracle}}}}
\newcommand{\exmode}{\mathit{m}}
\newcommand{\UMode}{\mathit{U}}
\newcommand{\TMode}{\mathit{P}}
\newcommand{\Modecnt}{\mathit{Mode}}
\newcommand{\Mode}[1]{\Modecnt(#1)}
\newcommand{\callConv}[1]{\textit{ex-entry}(#1)}
\newcommand{\Critical}{\mathit{Ext}}
\newcommand{\SeqObs}{\mathcal{O}}
\newcommand{\RealObs}{\mathcal{O}}
\newcommand{\Obs}{\mathcal{O}}
\newcommand{\ObsEq}[1]{\sim_{#1}}
\newcommand{\cmeq}{\approx}
\newcommand{\cmstep}{\mathrel{\hat{\rightarrow}}}
\newcommand{\cmsteps}{\mathrel{\hat{\DashedArrow}}}
\newcommand{\Obscoh}{\mathit{coh}}

\newcommand{\resource}{r}
\newcommand{\resSpc}{\mathbb{R}}
\newcommand{\seclevel}{\ell}
\newcommand{\acc}{\mathit{acc}}
\newcommand{\cl}{\mathit{cl}}
\newcommand{\RealMonitor}{\mathit{Mon}}
\newcommand{\RealMMU}{\mathit{Mmu}}
\newcommand{\RealMonInt}{\mathit{\bar I}_{\prod}}
\newcommand{\RealMonIntSW}{\mathit{\bar I}}
\newcommand{\RealMonIntSWPriv}{\overline{\mathit{II}}}
\newcommand{\RealState}{\bar \sigma}
\newcommand{\RealSpace}{\bar \Sigma}
\newcommand{\RealTrs}[1]{\rightarrow_{#1}}
\newcommand{\touched}[1]{\ [{#1}]}
\newcommand{\TlsTrs}[1]{\rightarrow_{#1}}

\newcommand{\SeqMonInt}{\mathcal{I}_{\prod}}
\newcommand{\SeqMonIntSW}{\mathit{I}}
\newcommand{\SeqMonIntSWPriv}{\mathit{II}}
\newcommand{\SeqSpace}{\Sigma}
\newcommand{\SeqTrs}[1]{\rightarrow_{#1}}
\newcommand{\histev}{\mathit{h}}
\newcommand{\inst}{dop}

\newcommand\defequiv{\mathrel{\overset{\makebox[0pt]{\mbox{\normalfont\tiny\sffamily def}}}{\equiv}}}
\newcommand{\Preg}{\mathit{M_{ac}}}

\newcommand{\cachetrans}[1]{\mathrel{\rightarrowtail_{#1}}}
\newcommand{\oracle}{\tiny \varOmega}
\newcommand{\oracleTrs}{\rightarrowtail}
\newcommand{\orclsp} {\mathbb{O}}
\newcommand{\drvabl}{\rhd}
\newcommand{\coh}{\mathit{Coh}}
\newcommand{\class}{\mathbb{T}}
\newcommand{\TypeDom}{\Delta}
\newcommand{\clevel}{\mathit{Type}}

\newcommand{\id}{\mathit {id}}
\newcommand{\cache}{\mathit {cache}}
\newcommand{\ureg}{\mathit {reg}}
\newcommand{\breg}{\mathit {breg}}
\newcommand{\cop}{\mathit {coreg}}
\newcommand{\psrsold}{\mathit {psrs}}
\newcommand{\pc}{\mathit {pc}}
\newcommand{\stkp}{\mathit {sp}}
\newcommand{\lr}{\mathit {lr}}
\newcommand{\cpsr}{\mathit {cpsr}}
\newcommand{\psr}{\mathit {psr}}
\newcommand{\domCritc}{\mathit{CR}}
\newcommand{\domConf}{\mathit{CO}}

\newcommand{\critc}  {\mathrm{cr}}
\newcommand{\pscritc}{\mathrm{sr}}
\newcommand{\user}   {\mathrm{ur}}
\newcommand{\swcritc}{\mathrm{scr}}
\newcommand{\smem}   {\mathrm{mem}}
\newcommand{\scache} {\mathrm{cache}}

\newcommand{\slc}{\mathit{SL}}
\newcommand{\fslc}[2]{\slc(#1,#2)}
\mathchardef\mhyphen="2D 
\newcommand{\fdirty}[2]{\mathit{c}\mhyphen\mathit{dirty}(#1,#2)}
\newcommand{\fhit}[2]{\mathit{c}\mhyphen\mathit{hit}(#1,#2)}
\newcommand{\cdSl}[2]{\mathit{c}\mhyphen\mathit{cnt}({#1},{#2})}
\newcommand{\way}[2]{\mathit{W}(#1,#2)}
\newcommand{\fway}[3]{\mathit{W}(#1,#2,#3)}
\newcommand{\si}    {\mathit {si}}
\newcommand{\tags}  {\mathit {tag}}
\newcommand{\widx}  {\mathit {widx}}
\newcommand{\ftouch}{\mathit {touch}}
\newcommand{\flfill}{\mathit {lfill}}
\newcommand{\fevict}{\mathit {evict}}
\newcommand{\fwriba}{\mathit {wriba}}
\newcommand{\falias}{\mathit {alias}}
\newcommand{\invba} {\mathit {invba}}
\newcommand{\clnba} {\mathit {clnba}}
\newcommand{\ffill} {\mathit {fill}}
\newcommand{\alloc} {\mathit{alloc}}
\newcommand{\fread} {\mathit {read}}
\newcommand{\fwrite}{\mathit {write}}
\newcommand{\touch} {\mathrm{touch}}
\newcommand{\evict} {\mathrm{evict}}
\newcommand{\lfill} {\mathrm{lfill}}
\newcommand{\cread} {\mathrm{read}}
\newcommand{\cwrite}{\mathrm{write}}
\newcommand{\domres}[2]{\mathit{Dom}(#1, {#2})}
\newcommand{\tagstate}{\mathit{ts}}
\newcommand{\action}{\mathit{act}}
\newcommand{\filter}{\mathcal{\varphi}}
\newcommand{\push}{\mathit{push}}
\newcommand{\pop}{\mathit{pop}}
\newcommand{\tail}{\mathit{back}}
\newcommand{\qu}{\mathbb{Q}}
\newcommand{\po}{Obligation}
\newcommand{\wt}{\mathit{wt}}
\newcommand{\rdd}{\mathit{rd}}

\newcommand{\Conf}{\domConf}

\newcommand*{\Scale}[2][4]{\scalebox{#1}{$#2$}}%
\newcommand{\fsize}[1]{\Scale[0.84]{#1}}
\newcommand{\ssize}[1]{\Scale[1]{#1}}
\newcommand{\monDom}{\mathit{MD}}
\newcommand{\NC}{\mathit{NC}}
\newcommand{\CR}{\domCritc}
\newcommand{\cohR}{\mathit{cohR}}
\newcommand{\CM}{\mathit{CM}}
\newcommand{\rCM}{\overline{\mathit{CM}}}
\newcommand{\CRx}{\CR_\ex}
\newcommand{\col}{\kappa}
\newcommand{\AL}{L}


\section{Introduction}\label{sec:intor}
Formal verification of low-level software such as microkernels, hypervisors, and drivers has made big strides in recent
years~\cite{sel4,AlkassarHLSST09,WildingGRH10,dam2013formal,Heitmeyer:2006:FSV:1180405.1180448,zhao2011armor,SteinbergK10,GuVFSC11}.
We appear to be approaching the point where the promise of provably secure, practical system software is becoming a reality. 
However, existing verification uses models that are far simpler than contemporary state-of-the-art hardware. Many features pose significant challenges: Memory models, pipelines, speculation, out-of-order execution, peripherals, and various coprocessors, for instance for system management. In a security context, caches are notorious. They have been known for years to give rise to timing side channels that are difficult to fully counteract \cite{2016arXiv161204474G}.
Also, cache management is closely tied to memory management, which---since it governs memory mapping, access control, and cache configuration through page-tables residing in memory---is one of the most complex and 
 security-critical components in the computer architecture flora.  

Computer architects strive to hide this complexity from application programmers, but system software, device drivers, and high-performance software, for which tuning of cache usage is critical, need explicit control over 
features like cacheability attributes. In virtualization scenarios, for instance, it is critical for performance to be able to delegate cache management authority for pages belonging to a guest OS to the guest itself.
With such a delegated authority a guest is free to configure
its share of the memory system as it wishes, including configurations that may break conventions normally expected for a well-behaved OS. For instance, a guest OS will usually be able to create memory aliases and to
set cacheability attributes as it wishes. Put together, these capabilities can, however, give rise to memory incoherence, since the same physical address can now be pointed to by two virtual 
addresses, one to cache and one to memory. This opens up for cache storage
attacks on both confidentiality and integrity, as was shown in~\cite{DBLP:conf/sp/GuancialeNBD16}.  Similarly to cache timing channels 
that use variations in execution time to discover hardware hidden state, storage channels use aliasing to profile cache activities and to attack system confidentiality. However, while the timing channels are external to models
used for formal analysis and do not invalidate verification of integrity
properties, storage channels simply make the models unsound: Using them for security analysis can lead to conclusions that are false.

This shows the need to develop verification frameworks for low-level system
software that are able to adequately reflect the presence of caches. 
It is particularly desirable if this can be done in a manner that allows
to reuse existing  verification tools on simpler models that do not consider
caches. This is the goal we set ourselves in this paper. 
We augment an existing cacheless model by adding a cache and accompanying cache
management functionality in MMU and page-tables.
We use this augmented model to derive proof obligations that can be imposed to ensure absence of both integrity and confidentiality
attacks. 
This provides a verification framework that we use to analyse soundness of
countermeasures.
The countermeasures are formally modelled as new
proof obligations that can be analysed on the cacheless model to
ensure absence of vulnerabilities due to cache storage channels. 
Since these obligations can be verified using the cacheless model, existing
tools~\cite{DBLP:conf/ccs/BalliuDG14,Brumley:2011:BBA:2032305.2032342,Song:2008:BNA:1496255.1496257}
(mostly not available on a cache enabled model) can automate this task  to a large extent.
We then complete the paper by repairing the verification of 
an existing and vulnerable
hypervisor~\cite{DBLP:journals/jcs/GuancialeNDB16}, sketching
how the derived proof obligations are discharged.


\section{Related Work}\label{sec:related}
\textbf{\itshape Formal Verification}
Existing work on formal verification do not 
takes into account cache storage channels.
The verification of seL4 assumes a sequential memory model and leaves cache
issues to be managed by means external to model~\cite{DBLP:conf/sosp/KleinEHACDEEKNSTW09, KleinAEMSKH14}.
Cock et al~\cite{CockGMH14} examined the bandwidth of timing channels in seL4 and possible countermeasures including cache coloring. 
The verification of the Prosper kernel ~\cite{dam2013formal, DBLP:journals/jcs/GuancialeNDB16} assumes that caches are invisible and correctly handled. 
Similarly Barthe et al.~\cite{DBLP:conf/fm/BartheBCL11} ignores caches for the verification of an isolation property for an idealised hypervisor.
Later, in~\cite{Barthe:2012:CRO:2354412.2355248} the authors extended the model
to include an abstract account of caches and verified that timing channels
are neutralised by cache flushing.
The CVM framework ~\cite{AlkassarHLSST09} treats caches only in the context of device management~\cite{HillebrandRP05}. Similarly, the Nova
\cite{SteinbergK10,TewsVW09, Becker2016} and CertiKOS \cite{Gu:2016:CEA:3026877.3026928, gu2015deep, GuVFSC11} microvisors do not consider caches in their formal analysis. 
In a follow-up paper~\cite{Costanzo:2016:EVI:2908080.2908100} the verification
was extended to machine code level, using a sequential memory model
and relying on the absence of address aliasing.

In scenarios such as OS virtualization, where untrusted software is allowed 
to configure cacheability of its own memory, all of the above systems can be
vulnerable to cache storage channel attacks. For instance, these channels can be
used to 
create illicit information flows among threads of seL4.

\textbf{\itshape Timing Channels}
Timing attacks and countermeasures have been formally verified to varying degrees of detail in the literature.
Almeida et al.~\cite{DBLP:conf/fse/AlmeidaBBD16} prove functional correctness and leakage security of MEE-CBC in presence of  timing attackers.
Barthe et al. \cite{DBLP:conf/types/BartheBCCL13} provide an abstract model of cache behaviour sufficient to replicate various timing-based exploits
and countermeasures from the literature such as \textsc{StealthMEM}~\cite{Kim:2012:SSP:2362793.2362804}. In a follow-up work Barthe et al.~\cite{Barthe:2014:SNC:2660267.2660283} formally showed that cryptographic algorithms that are implemented based-on
\textsc{StealthMEM} approach and thus are S-constant-time\footnote{An implementation is called S-constant-time, if it does not branch on secrets and only memory accesses to stealth addresses can be secret-dependent.}
are protected against cache-timing attacks. 
FlowTracker~\cite{silvatecnica} detects leaks using an information flow analysis at
compile time. Similarly, Ford et al.~\cite{GuVFSC11} uses information flow analysis based on explicit labeling to detect side-channels, and Vieira uses a deductive formal approach to
guarantee that side-channel countermeasures are correctly deployed~\cite{vieira2012formal}.
Other related work includes those adopting formal analysis to either check the rigour of countermeasures~\cite{DBLP:journals/corr/DoychevK16, DBLP:conf/sp/GuancialeNBD16, Tiwari:2011:CUM:2000064.2000087} 
or to examine bandwidth of side-channels~\cite{Kopf:2012:AQC:2362216.2362268, Doychev:2013:CTS:2534766.2534804}. Zhou~\cite{Zhou:2016:SAD:2976749.2978324} proposed  page access based solutions to mitigate the access-driven cache attacks and used model checking to show these countermeasure restore security.
Illicit information flows due to caches can also be countered by masking timing fluctuations by noise injection~\cite{Wang:2007:NCD:1250662.1250723} or by clock manipulation~\cite{Hu:1992:RTC:2699806.2699810, Vattikonda:2011:EFG:2046660.2046671,resolutionTiming}. 
A extensive list of protective means for timing attacks is given in~\cite{AESattack, ge2016survey}.

By contrast, we tackle storage channels. These channels carry information
through  memory and, additionally to permit  illicit information flows, can be
used to compromise integrity.
Storage channels have been used
in~\cite{DBLP:conf/sp/GuancialeNBD16} to show how cache management features
could be used to attack both integrity and confidentiality of
several types of application.


\section{Threats and Countermeasures }\label{sec:attack}

\begin{figure} \footnotesize
\centering
\begin{subfigure}[t]{0.4\textwidth}
\begin{lstlisting}[]
    V1) D = access($va1$)
    A1) write($va2$,1);
        free($va2$)
    V2) D = access($va1$)
    V3) if not policy(D)
           reject
        [evict $va1$]
    V4) use($va1$)
\end{lstlisting}
\caption{Integrity}
\label{fig:int}
\end{subfigure} \hspace{2cm}
\begin{subfigure}[t]{0.4\textwidth}
\begin{lstlisting}[]
    A1) write($va1$, 1)
    A2) write($va2$, 0)
    V1) if secret
          access($va3$)
        else
          access($va4$)
        
    A3) D = access($va2$) 
\end{lstlisting}
\caption{Confidentiality}
\label{fig:conf}
\end{subfigure}
\caption{Mismatched memory attribute threats}
\end{figure}
The presence of caches and ability to configure cacheability of virtual alias
enable the class of  attacks called ``alias-driven
attacks''~\cite{DBLP:conf/sp/GuancialeNBD16}. These attacks are based on building virtual aliases with mismatched cacheability attributes to break memory coherence; i.e, causing inconsistency between the values
stored in a memory location and the corresponding cache line, without making the cache line dirty.
We present here two examples to demonstrate how integrity and confidentiality
can be attacked using this vector.

\subsection{Integrity}
\label{subsec:IntCm}
Figure~\ref{fig:int} demonstrates an integrity threat. 
Here, we assume the cache is direct-mapped, physically indexed and write-back.
Also, both the attacker and victim are executed interleaved on a single core.
Virtual addresses $va1$ and $va2$ are aliasing the same memory  $pa$, $va1$ is cacheable and
 $va2$ is uncacheable. Initially, the memory $pa$ contains the value $0$ and the corresponding cache line is empty.
In a sequential model reads and writes are executed in order and their effects are instantly visible: V1) a victim accesses $va1$, reading $0$; A1) the
attacker writes $1$ into $pa$ using $va2$ and
releases the alias $va2$; V2) the victim accesses again $va1$, this time reading $1$; V3) if $1$ does not respect a security policy, then the victim rejects it;
otherwise V4) the victim passes $1$ to a security-critical functionality.

On a CPU with a weaker memory model the same code behaves differently:
V1) using $va1$, the victim reads $0$ from the memory and fills the cache;
A1) the attacker uses $va2$ to directly write $1$ in memory, bypasses the cache, and then frees the mapping; 
V2) the victim accesses again $va1$, reading $0$ from the cache; V3) the security policy is evaluated based on $0$;
possibly, the cache line is evicted and, since it is not dirty, the memory is not affected; V4) next time the victim accesses $pa$ it
reads $1$, but $1$ is not the value that
has been checked against the security policy.
This permits the attacker to bypass the policy.\footnote{Note that the attacker release its alias $va2$ before returning the control to the victim, making this attack different from the standard double mapping attacks.}

Intuitive countermeasures against alias-driven attacks are to forbid the
attacker from allocating cacheable aliases at all or to make cacheable its
entire memory.
A lightweight specialization of these approaches is ``always cacheability'':
A fixed region of memory is made always cacheable and
the victim rejects any input pointing  
outside this region.
Coherency can also be achieved by flushing the entire cache before the victim
accesses the attacker memory. Unfortunately, this countermeasure comes with
severe performance penalties~\cite{DBLP:conf/sp/GuancialeNBD16}. 
``Selective eviction'' is a more efficient solution and 
consists in removing from the cache every location 
that is accessed by the victim and that has been previously accessed
by the attacker.
Alternatively, the victim can use mismatched cache attributes itself
to detect memory incoherence and abort dangerous requests.

\subsection{Confidentiality}
Figure~\ref{fig:conf} shows a
confidentiality threat. Both $va_1$ and
$va_2$ point to the location $pa$ and 
say $idx$ is the cache line index of $pa$.
All virtual addresses except $va_2$ are
cacheable, and we assume that both $pa$
and the physical address pointed by $va_3$ are allocated in the same cache line.
The attacker writes A1) $1$ in the cache, making the line dirty, and
A2) $0$ in the memory. 
From this point, the value stored in the memory after the execution of the victim depends on the victim behaviour; if the victim accesses at
least one address (e.g. $va_3$) whose line index is $idx$, then the dirty line is evicted and  $1$ is written back to the memory;
otherwise the line is not evicted and $pa$ still contains  $0$.
This allows the attacker to measure evicted lines and thus to launch an access-driven attack.
%
In the following we summarise some of the countermeasures against cache storage channels presented in~\cite{DBLP:conf/sp/GuancialeNBD16} and relevant to our work.

Similarly, to ensure that no secret can leak through the cache storage channel, one can forbid allocating cacheable aliases or always flush the cache after executing victim functionalities.
An alternative is cache partitioning, where each process gets a dedicated part
of the cache and there is no intersection between any two partitions. This
makes it impossible for the victim activities to affect the
 attacker behaviour, thus preventing the attacker to infer information about
 victim's internal variables.
A further countermeasure is secret independent memory accesses, which aims at transforming the victim's
code so that the victim  accesses do not depend on secret.

Cache normalisation can also be used to close storage channels. In this approach, the cache is brought to a known state by reading a sequence of memory cells. This guarantees the subsequent secret
dependent accesses only hit the normalised cache lines, preventing the attacker from observing access patterns of the victim.


\newcommand{\safetyp}{\psi}
\section{High-Level Security Properties}\label{sec:modeling}
In this work we consider a trusted system software (the ``kernel'')
that shares the system with an untrusted software (the ``application'').
Possible instances for the kernel include
hypervisors, runtime monitors, 
low-level operating system routines, and
cryptographic services. The application is a software that requests services from the kernel and can be a user process or even a complete operating system. 
The hardware execution mode used by the application is less
privileged than the mode use by the kernel. The application is
potentially malicious and takes the role of the attacker here.

Some system resources
are owned by the kernel and are called ``critical'', some other resources 
should not be disclosed to the application and are called ``confidential''.
The kernel dynamically tracks memory ownership and provides mechanisms
for secure ownership transfer. 
This enables the application to pass data to the kernel services, while avoiding expensive copy operations:
the application prepares the input inside its own memory, the ownership of this memory is transferred to the kernel, and the corresponding kernel routine operates on the input in-place.
Two instances of this are the direct-paging memory virtualization mechanism 
introduced by Xen~\cite{xen} and runtime monitors
that forbid self-modifying code and prevent execution of unsigned code~\cite{DBLP:conf/esorics/ChfoukaNGDE15}. 
In these cases, predictability of the kernel behaviour must be ensured, regardless of any incoherent memory configuration created by the application. 

In such a setting, a common approach to formalise security is
via an integrity and a confidentiality property.
We use $\SeqState \in \SeqSpace$ to represent a state of the system and
$\WkTrs{}$ to denote
a transition relation. The transition relation models
the execution of one instruction by the application or the execution of a
complete handler of the kernel. The integrity property
ensures functional correctness (by showing that a state
invariant $\SeqMonIntSW$ is preserved by all transitions) and that the critical resources can
not be modified by the application 
(by showing a relation $\safetyp$):
\begin{property}[Correctness]
\label{prop:highLevelIntegrity}
  For all $\SeqState$ if $\SeqMonIntSW(\SeqState)$ and $\SeqState \WkTrs{}
  \SeqState'$ then $\SeqMonIntSW(\SeqState')$ and $\safetyp(\SeqState, \SeqState')$ 
\end{property}
The confidentiality property ensures that 
confidential resources are not leaked to the application
and is expressed using standard non-interference. 
Let $\SeqObs$ be the application's observations (i.e., the resources that are not
confidential) and let $\ObsEq{\SeqObs}$ be observational equivalence (which
requires states to have the same observations),
then
confidentiality is expressed by the following property
\begin{property}[Confidentiality]
\label{thm:highLevelConfidentiality}
  Let $\SeqState_1, \SeqState_2$ are initial states of the system
  such that $\SeqState_1 \ObsEq{\SeqObs} \SeqState_2$. If 
  $\SeqState_1 \WkTrs{}^{*} \SeqState'_1$  then
  $\exists \SeqState_2'.\ \SeqState_2 \WkTrs{}^{*} \SeqState'_2$  and $\SeqState'_1 \ObsEq{\SeqObs} \SeqState'_2$
\end{property}

The ability of the application to configure cacheability of its resources can
lead to incoherency, making formal program analysis on a cacheless model
unsound.
Nevertheless, directly verifying properties~\ref{prop:highLevelIntegrity}
and~\ref{thm:highLevelConfidentiality} using a complete cache-aware model
is unfeasible for any software of meaningful size.
Our goal is to show that the countermeasures can be used to restore coherency.
We demonstrate that if the countermeasures are correctly implemented by the
kernel then verification of the security properties on the cache-aware model
can be soundly reduced to proof obligations on the cacheless model.

\section{Formalisation}
As basis for our study we define two models, a cacheless and a cache-aware model. The cacheless model represents a memory coherent single-core system where all caches are disabled. 
The cache-aware model is the same system augmented by a single level data cache.



 \newcommand{\wds}{\alpha}
\newcommand{\wdl}{2^{(3 + \wds)}}
\newcommand{\bitsize}{\mathit{w}}

\subsection{Cacheless Model}
\label{subsec:seqmodel}
The cacheless model is ARM-flavoured but general enough to apply to other architectures. A state
$\SeqState = \tuple{\ureg, \psrs ,\cop, \mem} \in \SeqSpace$ is a tuple of general-purpose registers $\ureg$ (including program counter $\pc$), control registers $\psrs$, coprocessor state $\cop$, and memory $\mem$. 
The core executes either in non-privileged mode $\UMode$
or privileged mode
$\TMode$, ${\Modecnt(\SeqState) \in \{\UMode, \TMode\}}$. The control registers
$\psrs$ encode the execution mode and other execution parameters
such as the arithmetic flags. The coprocessor
state $\cop$ determines a range of system configuration parameters.
The word addressable memory is represented by $mem:\PA \to \B^{\bitsize}$,
where  $\B = \{0,1\}$, $\PA$ be the sets of physical addresses, and $\bitsize$
is
the word size.

The set $\resSpc$ identifies all resources in the system,
including registers, control registers, coprocessor states and
physical memory locations (i.e. $\PA \subseteq \resSpc$).
We use $\Exe: \SeqSpace \times \resSpc \to \B^{*}$ 
to represent the \emph{core-view} of a resource, which looks up the resource
and yields the corresponding value; e.g., for a physical address $pa \in \PA$, $\Exe$ returns the memory content in $pa$, $\Exe(\SeqState, pa) = \SeqState.mem(pa)$.

All software activities are restricted by a hardware monitor. The monitor
configuration can depend on  coprocessor states
(e.g., control registers for a TrustZone memory controller) and
 on regions of memory (e.g., page-tables for a Memory Management Unit (MMU)). We use the predicate
$\RealMonitor(\SeqState,\resource,\exmode,\acc) \in \B$ to represent the
hardware monitor, which holds if in the state $\SeqState$ the access  $\acc \in \{\wt, \rdd\}$
(for write and read, respectively) to the resource $\resource \in \resSpc$  is granted for the execution mode $\exmode \in \{\UMode, \TMode\}$. In ARM the hardware monitor consists of a static and a dynamic part. The static access control 
is defined by the processor architecture, which prevents non-privileged modifications to coprocessor states and control registers and whose model is trivial. The dynamic part is defined by the MMU and controls memory accesses. We use 
${\RealMMU(\SeqState,va,\exmode,\acc) \in (\PA\times\mathbb{B}) \cup \{\perp\}}$ to model the memory management unit. This function yields  for a virtual address $va$ the translation and the cacheability attribute
if the access permission is granted and $\perp$ otherwise. Therefore,
$\RealMonitor(\SeqState, pa,\exmode,\acc) $ is defined as $  \exists va.\
\RealMMU(\SeqState,va,\exmode,\acc) = (pa, -)$.
Further, $\monDom: \SeqSpace \to \resSpc$ is the function determining resources (i.e., the coprocessor registers,  the current master page-table, the linked page-tables)
which affect the monitor's behaviour.

The behaviour of the system is defined by an LTS ${\RealTrs{\exmode} \subseteq
  \SeqSpace \times \SeqSpace}$, where ${\exmode \in \{\UMode, \TMode\}}$
and if $\SeqState \SeqTrs{\exmode} \SeqState'$ then $\Mode{\SeqState} =
\exmode$.
 Each transition represents the execution of a single instruction.
When needed, we let $\SeqState \RealTrs{\exmode} \SeqState' \touched{\inst}$ 
denote that the instruction executed has $\inst$ effects on the data memory
subsystem,
where $\inst$ can be $\wt(R)$, $\rdd(R)$, or $\cl(R)$ to represent update, read and
cache cleaning of resources $R \subseteq \resSpc$. 
Finally, we use $\SeqState_0 \WkTrs{} \SeqState_n$ to represent the weak  transition relation that holds if there is a finite execution 
${\SeqState_0\Trans{}\cdots\Trans{}\SeqState_n}$ such that $\Mode{\SeqState_n} =
\UMode$ and $\Mode{\SeqState_j}\neq \UMode$ for $0 < j < n$
(i.e., the weak transition hides internal states of the kernel). 
%
%


 \newcommand{\drt}[2]{\fdirty{#1}{#2}}

\subsection{Cache-Aware Model}\label{subsec:realmodle} 

We model a single-core processor with single level unified cache.
A state  $\RealState
\in \RealSpace$ has all the components of the cacheless model together with the
cache, $\RealState = \tuple{\ureg, \psrs ,\cop, \mem, \cache}$.
The function $\RealMMU$ and transition relation   ${\RealTrs{\exmode} \subseteq  \RealSpace \times \RealSpace}$
are extended to take into account caches. Other definitions of the previous subsection are extended trivially.
In the following we use
$\fhit{\RealState}{pa}$ to denote a cache hit for address $pa$,
$\drt{\RealState}{pa}$ to identify dirtiness of the corresponding cache-line  (i.e., if the value of $pa$ has been modified in cache but not written back in the memory),
and $\cdSl{\RealState}{pa}$ to obtain value for $pa$ stored in the cache.

Both the kernel and the hardware monitor  (i.e., MMU) have the same view of the memory.
In fact, the kernel always uses cacheable virtual addresses and the MMU always consults first
the cache when it fetches a page table descriptor.\footnote{This is for instance
  the case of ARM Cortex-A53 and ARM Cortex-A8.}
This allows us to define the core-view for memory resources $pa \in \PA$ 
in presence of caches:
\[
 \Exe(\RealState, pa) = 
 \begin{cases}
  \cdSl{\RealState}{pa} & :\ \fhit{\RealState}{pa}\\
  \RealState.mem(pa)   & :\ \ow \\
 \end{cases}
\]


 \subsection{Security Properties}
\label{sec:topgoal}

Since the application is untrusted, we assume its code is unknown and that 
it can
break its memory coherence at will.
Therefore, the behaviour of the application in the cacheless and the cache-aware models can differ significantly. In particular,
memory incoherence caused by mismatched cacheability may lead a control variable of the application to have different values in these two models, 
causing the application to have different control flows.
This makes the long-established method of verification by
refinement~\cite{Roever:2008:DRM:1525579} not feasible for analysing
behaviour of the attacker.

To accomplish our security analysis we identify a subset of resources that
are \emph{critical}: resources for which integrity must be preserved and
on which critical functionality depends.
The security type of registers, control and coprocessor registers is statically
assigned. The security type of memory locations, however, 
can dynamically change due to transfer of memory ownership;
i.e., the type of these resources depends on the state of the system.
The function 
$\domCritc :(\RealSpace \cup \SeqSpace) \rightarrow 2^\resSpc$
retrieves the subset of resources that are critical.
Function $\domCritc$ usually depends on a subset of resources, for instance the
internal kernel data-structures used to store the security type of memory resources.
The definition of function $\domCritc$ must be based on the \emph{core-view} to make it usable for both the
cacheless and cache-aware models. 
We also define
(for both cacheless and cache-aware state)
$\Critical_{\critc}(\SeqState) = \{(a, \Exe(\SeqState, a))\ |\  a \in
\domCritc(\SeqState) \}$, this is the function that extracts the value of all
critical resources.
Finally, the set of resources that are confidential is statically defined
and identified by $\domConf \subseteq \resSpc$.
This region of memory includes all 
internal kernel data-structures whose value can depend on
secret information.

Property~\ref{prop:highLevelIntegrity} requires to introduce a system invariant
$\RealMonIntSW$ that is software-dependent and defined per kernel.
The invariant specifies: (i) the shape of a sound page-table
(e.g., to prohibit the non-privileged writable accesses 
to the kernel memory and the page-tables themselves),
(ii) properties that permits the kernel to work properly (e.g., the kernel stack pointer and its data structures are correctly
configured), (iii) functional properties specific for the selected countermeasure, and (iv) cache-related properties  that allow to restore coherency. 
A corresponding  invariant $\SeqMonIntSW$ for the cacheless model is derived
from $\RealMonIntSW$ by excluding properties that constrain caches.
Property~\ref{prop:highLevelIntegrity} is formally demonstrated by two theorems:
one constraining the behaviour of the application and one showing functional
correctness of the kernel.
Let $\callConv{\RealState}$ be a predicate identifying the state of the system
immediately after switching to the kernel (i.e., when an exception handler is
executed the mode is privileged and the program counter points to the
exception table).
Theorem~\ref{thm:topIntegrityUser} enforces that the execution of the
application in the cache enabled setting cannot affect the critical resources.
\begin{theorem}[Application-integrity]
\label{thm:topIntegrityUser}
  For all $\RealState$, 
  if $\RealMonIntSW(\RealState)$ and $\RealState \TlsTrs{\UMode} \RealState'$ then 
  $\RealMonIntSW(\RealState')$,
  $\Critical_{\critc}(\RealState)=\Critical_{\critc}(\RealState')$, and
  if  $\Mode{\RealState'} \neq \UMode$ then  $\callConv{\RealState'}$
\end{theorem}

While characteristics of the application prevents establishing refinement for
non-privileged transitions, for the kernel we show that the use of proper countermeasures enables 
transferring the security properties from the cacheless to the cache-aware model. This demands proving that the two models behave equivalently for kernel transitions.
We prove the behavioural equivalence by showing
refinement between two models using forward
simulation. We define the simulation relation $\SimR$  
(guaranteeing equality of critical resources) and show that both the invariant and
the relation are preserved
by privileged transitions:
\begin{theorem}[Kernel-integrity]
\label{thm:topIntegrityPriv}
  For all $\RealState_1$ and $\SeqState_1$
  such that $\RealMonIntSW(\RealState_1)$, $\RealState_1 \SimR \SeqState_1$,
  and $\callConv{\RealState_1}$
  if $\RealState_1 \WkTrs{} \RealState_2$ then $\exists\SeqState_2 .\ \SeqState_1 \WkTrs{} \SeqState_2$,
  $\RealState_2 \SimR \SeqState_2$ and  $\RealMonIntSW(\RealState_2)$
\end{theorem}

Applying a similar methodology, in Section~\ref{sec:confidentiality} we prove
the confidentiality property (i.e., Theorem~\ref{thm:topConfidentiality}) in
presence of caches. 
Here, we use bisimulation (equality of the application's observations) 
as unwinding condition: 
\begin{theorem}[Confidentiality]
\label{thm:topConfidentiality}
  For all $\RealState_1$ and $\RealState_2$ such that
  $\RealMonIntSW(\RealState_1)$, $\RealMonIntSW(\RealState_2)$, and
  $\RealState_1 \ObsEq{\RealObs} \RealState_2$, if
  $\RealState_1 \WkTrs{} \RealState'_1$ then
  $\exists \RealState_2'.\ \RealState_2 \WkTrs{} \RealState'_2$
and $\RealState'_1 \ObsEq{\RealObs} \RealState'_2$ as well as $\RealMonIntSW(\RealState_1')$, $\RealMonIntSW(\RealState_2')$.
\end{theorem}


\newcommand{\safe}{\mathit{safe}}

\section{Integrity}\label{sec:integrity}
Our strategy to demonstrate correctness of integrity countermeasures consists of
two steps. We first decompose the proof of
Theorems~\ref{thm:topIntegrityUser} and~\ref{thm:topIntegrityPriv} and show
that the integrity properties are met if a set of proof conditions are
satisfied. The goal of this step is to provide a theoretical framework that
permits to analyse soundness of a countermeasure without the need of dealing
with the complex transition relation of the cache-aware model.
Then, we demonstrate correctness of two countermeasures of Section~\ref{subsec:IntCm}, namely always cacheability and selective eviction, by showing that
if they are correctly implemented by the kernel then verification of the
integrity properties can be soundly reduced to analysing properties of the
kernel in the cacheless model.

\subsection{Integrity: Application Level (Theorem~\ref{thm:topIntegrityUser})}
\label{subsec:userIntegrity}
To formalise the proof
we introduce the auxiliary definitions of \textit{coherency}, \textit{derivability} and \textit{safety}.

\begin{definition*}[Coherency]
\label{def:coherency}
 In $\RealState \in \RealSpace$ a set of memory resources $R \subseteq \PA$ is coherent $(\coh(\RealState, R))$, 
 if for all $pa \in R$,
 such that $pa$ hits the cache and its value differs from the memory $($i.e.
 $\RealState.\mem(pa) \neq \cdSl{\RealState}{pa})$,
 the corresponding cache line is dirty $(\drt{\RealState}{pa})$.
\end{definition*}
Coherency of the critical resources  is essential to prove integrity. In fact, for an incoherent resource the \emph{core-view} can be changed indirectly without an explicit memory write, i.e., through evicting the clean
cache-line corresponding to the resource which has different values in
the cache and memory. Moreover, in some cores, (e.g., ARMv7/v8) the MMU looks first into the caches when it fetches a descriptor. 
Then if the page-tables are coherent, a cache eviction cannot indirectly affect the behaviour of the MMU. 

To allow the analysis of specific countermeasures to abstract from the
complex cache-aware transition system, we introduce the notion of
\emph{derivability}.
This is a relation overapproximating the effects over the memory and cache for instructions executed in non-privileged mode.
Derivability is an architectural property and it is independent of the software executing.
\begin{definition*}[Derivability]
\label{def:derivability}
We say $\RealState'$ is derivable from  $\RealState$ 
in non-privileged mode (denoted as $\RealState \drvabl_{\UMode} \RealState'$) if 
$\RealState.\cop = \RealState'.\cop$ and
for every $pa \in \PA$ one of the following properties holds:
\begin{itemize}
  \item[$D_{\emptyset}(\RealState, \RealState', pa)$:]
Independently of the access rights for the address $pa$, the corresponding memory can be changed due to an eviction
of a dirty line and subsequent write-back of the cached value into the memory.
Moreover, the cache can always change due to an eviction.
  \item[$D_{\rdd}(\RealState, \RealState', pa)$:]
If non-privileged mode can read the address $pa$, the cache state can change through a fill operation which loads the cache with the value of $pa$ in the memory.
\item[$D_{\wt}(\RealState, \RealState', pa)$:]
If non-privileged mode can write the address $pa$, it can either write directly into the cache, making it dirty, or bypass it, by using an uncacheable alias.
\end{itemize}
Figure~\ref{fig:derivability} reports the formal definition of these predicates.
\end{definition*}

\begin{figure}
\[\begin{array}{l}
\begin{array}{ll}
& D_{\emptyset}(\RealState, \RealState', pa) \defequiv 
 M'(pa) \neq M(pa)  \Rightarrow (
\drt{\RealState}{pa} \land
M'(pa) = \cdSl{\RealState}{pa})  \\
&
\,\,\ \land \ 
 \way{\RealState'}{pa} \neq \way{\RealState}{pa}  \Rightarrow \\
& \qquad\qquad \left (\neg \fhit{\RealState'}{pa}
 \land (\drt{\RealState}{pa} \implc M'(pa) = \cdSl{\RealState}{pa} )
 \right ) \\
\end{array} 
\\ \\
\begin{array}{ll}
&D_{\rdd}(\RealState, \RealState', pa)  \defequiv \RealMonitor(\RealState,pa, \UMode, \rdd) \\
&
\,\,\ \land \ 
 M'(pa) = M(pa) \land
 \way{\RealState'}{pa} \neq \way{\RealState}{pa}  \Rightarrow \\
& \qquad\qquad
  (\cdSl{\RealState'}{pa} = M(pa)  \land \neg \fhit{\RealState}{pa})
\end{array}
\\ \\
\begin{array}{rll}
&D_{\wt}(\RealState, \RealState', pa)  \defequiv 
 \RealMonitor(\RealState,pa, \UMode, \wt) \\
&
\,\,\ \land \ 
(\way{\RealState'}{pa} \neq \way{\RealState}{pa}  \Rightarrow 
\drt{\RealState'}{pa} )\\
&
\,\,\ \land \ 
(M'(pa) \neq M(pa) \implc \\ &  \qquad\qquad (\neg \drt{\RealState'}{pa} \implc  \exists va . \RealMMU(\RealState, va, \UMode, \wt) = (pa, 0)))
\end{array}
\end{array}\]
\caption{Derivability. Here $M \!=\!\RealState.\mem$,
  $M'\!=\!\RealState'.\mem$, and $\way{\RealState}{pa} = \tuple{\fhit{\RealState}{pa},
  \drt{\RealState}{pa}, \cdSl{\RealState}{pa}}$ 
denotes the cache-line corresponding to $pa$ in $\RealState.\cache$.}
\label{fig:derivability}
\end{figure}

\begin{definition*}[Safety]
\label{def:monitorsafe}
 A state $\ \RealState\ $ is safe, $\safe(\RealState)$, if for every state $\RealState'$, resource $\resource$, mode $\exmode$ and access request $\acc$ if $\RealState \drvabl_{\UMode} \RealState'$ then
 $\RealMonitor(\RealState,\resource, \exmode,\acc) = \RealMonitor(\RealState',\resource, \exmode,\acc)$.
\end{definition*}
A state is \emph{safe} if non-privileged executions cannot affect the hardware
monitor, i.e., only the kernel can change page-tables.

\newcommand{\RealInvF}{\RealMonIntSW_{\mathit{fun}}}
\newcommand{\RealInvCC}{\RealMonIntSW_{\mathit{coh}}}
\newcommand{\RealInvCMC}{\RealMonIntSW_{\mathit{cm}}}
To decompose proof of Theorem~\ref{thm:topIntegrityUser}, 
the invariant $\RealMonIntSW$ must be split in three parts:
a functional part $\RealInvF$
which only depends on the \emph{core-view} for the critical resources, 
an invariant $\RealInvCC$ 
which only depends on coherency of the critical resources, 
and an optional countermeasure-specific invariant $\RealInvCMC$ 
which depends on coherency of non-critical memory
resources such as resources in an always-cacheable region. 
\begin{prop}\label{po:monintor:inv:threeparts}
For all $\RealState$, $\RealMonIntSW(\RealState)\!=\!\RealInvF(\RealState)\!\wedge\!\RealInvCC(\RealState)\!\wedge\!\RealInvCMC(\RealState)$ 
and: 
\begin{enumerate}
\item \label{thm:invUserDrvbl}
  for all $\RealState'$ if $\Critical_{\critc}(\RealState)\!=\!\Critical_{\critc}(\RealState')$ then  $\RealInvF(\RealState)\!=\!\RealInvF(\RealState')$;
\item \label{thm:invUserCCChoerence}
  for all $\RealState'$ if $\coh(\RealState, \domCritc(\RealState))$, $\coh(\RealState', \domCritc(\RealState'))$, and $\Critical_{\critc}(\RealState)\!=\!\Critical_{\critc}(\RealState')$
  then $\RealInvCC(\RealState)\!=\!\RealInvCC(\RealState')$;
\item \label{thm:invUserCMChoerence}
 for all $\RealState'$ if $\RealMonIntSW(\RealState)$ and $\RealState \drvabl_{\UMode} \RealState' $ then $\RealInvCMC(\RealState')$.
\end{enumerate}
\end{prop}

Also, the function   $\domCritc$ 
must be correctly defined: i.e. resources
affecting the set of critical resources are critical themselves.
\begin{prop}\label{po:type:delat:soundness}
  For all $\RealState, \RealState'$ if  
$\Critical_{\critc}(\RealState) = \Critical_{\critc}(\RealState')$
then $\domCritc(\RealState) = \domCritc(\RealState')$
\end{prop}

Safety is essential to prove integrity:
if a state is not safe, the application can potentially elevate its permissions by changing configurations of the hardware monitor and get access to resources beyond 
its rights. 
\begin{lemma}
\label{lem:monitorsafe}
 If $\RealMonIntSW(\RealState)$ then $\safe(\RealState)$
\end{lemma}
Proof of Lemma~\ref{lem:monitorsafe} depends on the formal model of the hardware monitor and guarantees provided by the invariant. 
Using the invariant $\RealMonIntSW$, we identify three main proof obligations
that are needed to prove this lemma.
\begin{enumerate}
 \item The functional part of the invariant must guarantee that the resources that control the hardware monitor are considered critical.
\begin{prop}\label{po:monintor:dom:critic}
If $\RealInvF(\RealState)$ then $\monDom(\RealState ) \subseteq  \domCritc(\RealState)$
\end{prop}

\item The application should not be allowed to directly affect the critical
 resources. This means there is no address writable in non-privileged mode that points to a critical resource.
\begin{prop}\label{vc:critical:no:write}
 If $\ \RealInvF(\RealState)\ $ and $\resource \in \domCritc(\RealState )$ then $\neg \RealMonitor(\RealState, \resource, \UMode, \wt)$
\end{prop}

\item Finally, to prevent  the application from indirectly affecting the
  hardware monitor, e.g., by line eviction, the invariant must ensure coherency of critical resources.
\begin{prop}\label{vc:critical:coherency}
 If $\RealInvCC(\RealState)$  then $\coh(\RealState, \domCritc(\RealState ))$
\end{prop}
\end{enumerate}

We overapproximate the reachable states in non-privileged mode to prove
properties that do not depend on the software platform or countermeasure. This eliminates the need for revalidating properties of the instruction set 
(i.e., properties that must be verified for every possible instruction) in every new software scenario.
\begin{lemma} \label{lem:userswitchlem}
 For all $\RealState$ such that
 $\safe(\RealState)$
and $\coh(\RealState, \monDom(\RealState))$,\linebreak
 if  $\RealState  \RealTrs{\UMode} \RealState'$
 then 
 \footnote{In~\cite{khakpour2013machine} the authors proved a similar theorem for the HOL4 ARMv7 model provided by Anthony Fox et. al.~\cite{DBLP:conf/itp/FoxM10}.}
\begin{enumerate}
  \item\label{thm:trans2drvbl} $\RealState \drvabl_{\UMode} \RealState'$, i.e. non-privileged transitions from safe states can
    only lead into derivable states
  \item \label{thm:cnxtSwUsr}  if  $\Modecnt(\RealState')\!\neq\!\UMode$   then
    $\callConv{\RealState'}$, i.e., the mode can only change by entering an
    exception handler
\end{enumerate}
\end{lemma}

Corollary~\ref{thm:trns2drvbl} shows that derivability can be used
as sound overapproximation of the behaviour of non-privileged transitions
if the invariant holds.
\begin{corollary}
\label{thm:trns2drvbl}
 For all $\RealState$ if $\RealMonIntSW(\RealState)$ and $\RealState
 \TlsTrs{\UMode} \RealState'$ then $\RealState \drvabl_{\UMode} \RealState'$
\end{corollary}
\begin{prove}{}
The statement directly follows by
Lemma~\ref{lem:userswitchlem}.\ref{thm:trans2drvbl} which is enabled by
Lemma~\ref{lem:monitorsafe}, 
Obligation~\ref{po:monintor:dom:critic}, and 
Obligation~\ref{vc:critical:coherency}.
\end{prove}

We now proceed to show that the hardware monitor enforces access policy
correctly; i.e., the application transitions cannot modify critical resources.
\begin{lemma}
\label{thm:User-No-Exfiltration}
  For all $\RealState, \RealState'$ such that
  $\RealMonIntSW(\RealState)$  if 
  $\RealState \drvabl_{\UMode} \RealState'$
  then $\Critical_{\critc}(\RealState) = \Critical_{\critc}(\RealState')$
\end{lemma}
\begin{prove} {}
Since $\RealMonIntSW(\RealState)$ holds, the hardware monitor prohibits
writable accesses of the application to critical resources (Obligation~\ref{vc:critical:no:write}) and $\safe(\RealState)$ holds (Lemma~\ref{lem:monitorsafe}).
Also,  derivability  shows that the application can directly change only
resources that are writable according to the monitor. Thus, the application
 cannot directly update $\domCritc(\RealState)$.
Beside, the invariant guarantees coherency of critical resources in
$\RealState$ (Obligation~\ref{vc:critical:coherency}). This prevents indirect modification of these resources.
\end{prove}

%
To complete the proof of Theorem~\ref{thm:topIntegrityUser} we additionally
need to show that coherency of critical resources
(Lemma~\ref{thm:coherencyUser}) 
 and invariant (Lemma~\ref{thm:invUser})
are preserved by non-privileged transitions. 
\begin{lemma}
\label{thm:coherencyUser}
 For all $\RealState$ if $\RealMonIntSW(\RealState)$ and $\RealState \drvabl_{\UMode} \RealState'\ $ then $\coh(\RealState', \domCritc(\RealState'))$
\end{lemma}
\begin{prove}{}
The proof depends on \po~\ref{vc:critical:coherency} and
\po~\ref{vc:critical:no:write}: coherency can be invalidated only through
non-cacheable writes, which are not possible since  
aliases to critical resources that  are writable by non-privileged mode are forbidden.
\end{prove}

\begin{lemma}
\label{thm:invUser}
 For all $\RealState$ and $\RealState'$ if $\RealMonIntSW(\RealState)$ and $\RealState \TlsTrs{\UMode} \RealState'$ then $\RealInvF(\RealState')$
\end{lemma}
\begin{prove}{}
To show that non-privileged transitions preserve the invariant we use
Corollary~\ref{thm:trns2drvbl},
Lemma~\ref{thm:User-No-Exfiltration}, and Obligation~\ref{po:monintor:inv:threeparts}.\ref{thm:invUserDrvbl}.
\end{prove}

Finally, Lemma~\ref{lem:userswitchlem}.\ref{thm:cnxtSwUsr},
Corollary~\ref{thm:trns2drvbl},
Lemma~\ref{thm:User-No-Exfiltration}, Lemma~\ref{thm:coherencyUser}, Lemma~\ref{thm:invUser}, and \po~\ref{po:monintor:inv:threeparts}
imply Theorem~\ref{thm:topIntegrityUser}, completing the proof of integrity for non-privileged transitions.
%
%
%
%


 \subsection{Integrity: Kernel Level (Theorem~\ref{thm:topIntegrityPriv})}
\label{subsec:trustedIntegrity}
The proof of kernel integrity (Theorem~\ref{thm:topIntegrityPriv}) is done by:
(i) showing that the kernel preserves the invariant $\SeqMonIntSW$ in the cacheless model (\po~\ref{lem:invPrevsISA});
(ii) requiring a set of general proof obligations (e.g., the kernel does not jump
outside its address space)
that must be verified at the cacheless level (\po~\ref{lem:kernelCFGandVM}),
(iii) demonstrating that the simulation relation permits transferring the
invariant from the cache-aware model to the cacheless one
(Obligation~\ref{vc:isaInitStsfyInv}),
(iv) verifying a refinement between the cacheless model and the cache-aware
model (Lemma~\ref{lem:privIntegByMiniInv}) assuming correctness of the
countermeasure (\po~\ref{vc:kernel:resource:value:inmodels}), and
finally 
(v) proving that the refinement theorem allow to transfer the invariant from the
cacheless model to the cache-aware model (Lemma~\ref{lem:privInvEnd}). Figure~\ref{fig:simR}
indicates our approach to prove kernel integrity.

\begin{figure}
\centering
 \includegraphics[width=0.7\linewidth]{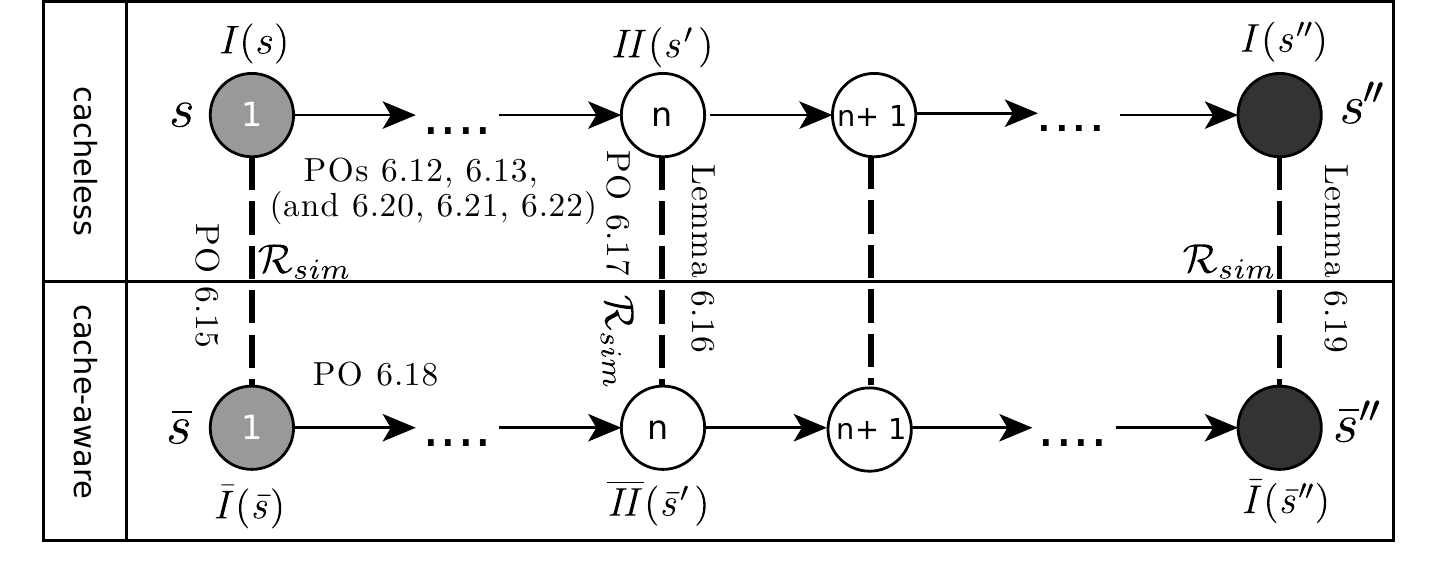}
 \caption{Verification of kernel integrity}  
 \label{fig:simR}
\end{figure}

%
Our goal is to lift the code verification to the cacheless model, so that existing tools can be used to discharge proof obligations.
The first proof obligation requires to show that the kernel is functionally
 correct when there is no cache:
\begin{prop}
\label{lem:invPrevsISA}
For all $\SeqState$ such that $\SeqMonIntSW(\SeqState)$ and
$\callConv{\SeqState}$ if $\SeqState \WkTrs{} \SeqState'$
then $\SeqMonIntSW(\SeqState')$.
\end{prop}

To enable the verification of simulation, it must be guaranteed that
the kernel's behaviour is predictable:
(\ref{lem:kernelCFGandVM}.1) the kernel cannot execute instructions outside
the critical region (this property is usually verified by extracting the control
flow graph of the kernel and demonstrating that it is contained in a static
region of the memory), and
(\ref{lem:kernelCFGandVM}.2) the kernel does not change page-tables entries that
maps the kernel virtual memory (here this region is assumed to be static and identified by $\RegionKernelV$) which must be accessed only using cacheable aliases.
\begin{prop}
\label{lem:kernelCFGandVM}
For all $\SeqState$ 
such that $\SeqMonIntSW(\SeqState)$ and
$\callConv{\SeqState}$ if $\SeqState \WTrs{\TMode} \SeqState'$ then
\begin{enumerate}
\item
\label{lem:kernelPcReange}
if 
$\RealMMU(\SeqState', \SeqState'\!.\ureg.\pc, \TMode, \acc) = (pa, -)$ then $pa \in \domCritc(\SeqState)$
\item
\label{lem:ptFixedPriv}
for every $va \in \RegionKernelV$ and $\acc$ if $\RealMMU(\SeqState, va, \TMode, \acc) = (pa,c)$
then \linebreak $\RealMMU(\SeqState', va, \TMode, \acc) = (pa,c)$ and $c=1$
\end{enumerate}
\end{prop}

In order to establish a refinement, we introduce the  \emph{memory-view} $\altExe$ of the cache-aware model.
This function models what can be observed in the memory after a cache line
eviction and is used to handle dynamic growth of the critical resources, which are potentially not coherent when a kernel handler starts.
\[{
  \altExe(\RealState,r)\ \defequiv\ 
\begin{cases}
   \cdSl{\RealState}{pa} &:\ r\in\PA \land \fdirty{\RealState}{pa} \\
   \RealState.\mem(pa) &:\  r\in\PA \land \neg \fdirty{\RealState}{pa}\\
  \Exe(\RealState,r) &:\ \text{otherwise}
\end{cases}
}\]

Note that $\altExe$ only differs from $\Exe$ in memory resources that are cached, clean, and have different values stored in the cache and memory, 
i.e., incoherent memory resources. In particular, we can prove the following lemma about coherency:
\begin{lemma}
\label{lem:coherency1}
Given a memory resource $pa\in\PA$ and cache-aware state $\RealState$ then $\coh(\RealState,\{pa\}) \Leftrightarrow (\Exe(\RealState,pa)=\altExe(\RealState,pa))$, i.e., memory resources are coherent, iff both views agree on them.
\end{lemma} 

We define the simulation relation between the cacheless and cache-aware models 
using the memory-view:
$\RealState \SimR \SeqState \equiv \forall \resource\in\resSpc.\ \Exe(\SeqState,r) = \altExe(\RealState,r)$.
The functional invariant of the cache-aware model and the invariant of the
cacheless model must be defined analogously, that is the two invariants are 
equivalent if the critical resources are coherent,
thus ensuring that functional properties can be transfered between the two models.
\begin{prop}
\label{vc:isaInitStsfyInv}
 For all $\RealState$ and $\SeqState$, such that $\RealInvCC(\RealState)$
 and $\RealState \SimR \SeqState$,
 holds
 $\RealInvF(\RealState) \Leftrightarrow \SeqMonIntSW(\SeqState)$
\end{prop}

A common problem of verifying low-level software is coupling the invariant with every possible internal states of the kernel.
This is a major concern here, since the set of critical resources
changes dynamically and can be stale while the kernel is executing.
We solve this problem by defining a \emph{internal invariant}
$\RealMonIntSWPriv(\RealState, \RealState')$ 
which allows us to define properties of the state $\RealState'$ in relation with the initial state of
the kernel handler $\RealState$.
This invariant $\RealMonIntSWPriv$ (similarly $\SeqMonIntSWPriv$ for the cacheless model) includes:
(i) $\RealMonIntSW(\RealState)$ holds,
(ii) the program counter in $\RealState'$ points to the kernel memory,
(iii) coherency of resources critical in the initial state (i.e. $\coh(\RealState', \domCritc(\RealState))$)
(iv) all virtual addresses in $\RegionKernelV$ are cacheable and their mapping in $\RealState$ and $\RealState'$ is unchanged,
and (v) additional requirements that are countermeasure specific and will be described later.

\begin{lemma}
\label{lem:privIntegByMiniInv}
  For all $\RealState$ and $\SeqState$
  such that $\RealMonIntSW(\RealState)$
  ,
 $\callConv{\RealState}$, and $\RealState \SimR \SeqState$,
  if $\RealState \WTrsN{\TMode}{n} \RealState'$ and $\SeqState \WTrsN{\TMode}{n} \SeqState'$
  then $\RealState' \SimR \SeqState'$,  $\RealMonIntSWPriv(\RealState, \RealState')$, and $\SeqMonIntSWPriv(\SeqState, \SeqState')$  
\end{lemma}
\begin{prove}{}
By induction on the execution length. The base case is trivial, since no step is taken.
For the inductive case we first show that the instruction executed is the same in both the models:
$\SimR$ guarantees that the two states have the same program counter,
\po~\ref{lem:kernelCFGandVM}.\ref{lem:kernelPcReange} ensures that the program counter points to the kernel memory contained in the critical resources,
property (iii) of the invariant guarantees that this memory region is coherent, and
Lemma~\ref{lem:coherency1} shows that the fetched instruction  is the same.

Proving that the executed instruction preserves kernel virtual memory mapping is trivial, since \po~\ref{lem:kernelCFGandVM}.\ref{lem:ptFixedPriv}
ensures that the kernel does not change its own memory layout and that it only uses cacheable aliases.
Showing that the resources accessed by the instruction have the same value (thus guaranteeing $\SimR$ is preserved) in the cache-aware
and cacheless states depends on demonstrating their coherency. This is
countermeasure specific and is guaranteed by proof
obligation~\ref{vc:kernel:resource:value:internal:invariant:inmodels}.\ref{vc:kernel:resource:value:inmodels}.
Similarly, showing the internal invariant is maintained by privileged transitions depends on the specific countermeasure that is in place 
(\po~\ref{vc:kernel:resource:value:internal:invariant:inmodels}.\ref{vc:kernel:internal:invariant:inmodels}).
\end{prove}


\begin{prop}\label{vc:kernel:resource:value:internal:invariant:inmodels}
  For all 
  states $\RealState$ and $\RealState'$
  that satisfy the refinement (i.e. 
  $\RealMonIntSW(\RealState)$,
  $\callConv{\RealState}$, and $\RealState \SimR \SeqState$),
  after any number $n$ of instructions of the kernel that preserve the
  refinement and the internal invariants (
  $\RealState \WTrsN{\TMode}{n} \RealState'$,
  $\SeqState \WTrsN{\TMode}{n} \SeqState'$,
  $\RealState' \SimR \SeqState'$, 
  $\RealMonIntSWPriv(\RealState, \RealState')$, and
  $\SeqMonIntSWPriv(\SeqState, \SeqState')$
  ) 
  \begin{enumerate}
   \item \label{vc:kernel:resource:value:inmodels}
   if the execution of the $n+1$-th instruction in the cacheless model
   accesses resources $R$
  ($\SeqState' \SeqTrs{\TMode} \SeqState'' \touched{\inst}$ and either $\inst = rd(R)$
  or $\inst = wt(R)$) then
  $\coh(\RealState', R)$
   \item \label{vc:kernel:internal:invariant:inmodels}
   the execution of the $n+1$-th instruction in the cacheless and cache-aware models 
   ($\SeqState' \SeqTrs{\TMode} \SeqState''$ and $\RealState' \RealTrs{\TMode} \RealState''$) preserves 
   the internal invariants ($\SeqMonIntSWPriv(\SeqState, \SeqState'')$ and $\RealMonIntSWPriv(\RealState, \RealState'')$)
    \end{enumerate}
\end{prop}

Additionally, the internal invariant must ensure the countermeasure specific
requirements of coherency for all internal states of the kernel.
\begin{prop}
\label{vc:inv:cacheless:to:cache}
 For all $\RealState$ and $\RealState'$ if $\RealMonIntSWPriv(\RealState, \RealState')$ then $\RealInvCMC(\RealState')$
\end{prop}

Finally, we show that the invariant $\RealMonIntSW$ holds in a cache-aware state
when the control is returned  to non-privileged mode, i.e. when the invariant is
re-established in the cacheless model.
\begin{lemma}
\label{lem:privInvEnd}
 For all $\RealState$, $\RealState'$, and $\SeqState'$  if $\SeqMonIntSW(\SeqState')$, $\RealState' \SimR \SeqState'$, and
 $\RealMonIntSWPriv(\RealState, \RealState')$ then $\RealMonIntSW(\RealState')$ holds.
\end{lemma}
\begin{prove}{}
The three parts of invariant $\RealMonIntSW(\RealState')$ are demonstrated
independently. 
\po~\ref{vc:inv:cacheless:to:cache} establishes $\RealInvCMC(\RealState')$.
Property (iii) of  $\RealMonIntSWPriv(\RealState, \RealState')$ guarantees
$\RealInvCC(\RealState')$. Finally,
\po~\ref{vc:isaInitStsfyInv} demonstrates 
 $\RealInvF(\RealState)$.
\end{prove}

\newcommand{\alwayscachebility}{{AC}}
\subsection{Correctness of countermeasures}\label{subsec:countermeasures}
Next, we turn to show that selected countermeasures for the integrity attacks
prevent usage of cache  to violate the integrity property. Thus,
we show that the countermeasures help to discharge the coherency related proof
obligations reducing verification of integrity to analysing properties 
of the kernel code using the cacheless model.

\paragraph{Always Cacheability} 
We use $\Preg \subseteq \PA$ to statically identify the region of physical
memory that should be always accessed using cacheable aliases. 
The verification that the always cacheability  countermeasure is in place
can be done by discharging the following proof obligations at the cacheless level:
(\ref{prop:cm:always-cacheable}.1) the hardware monitor does not permit uncacheable accesses to $\Preg$,
(\ref{prop:cm:always-cacheable}.2.a) the kernel never allocates critical resources outside $\Preg$,
thus restricting  the application to use $\Preg$ to communicate with the kernel,
and
(\ref{prop:cm:always-cacheable}.2.b) the kernel accesses only resources in $\Preg$:
\begin{prop}
\label{prop:cm:always-cacheable}
 For all $\SeqState$ such that  $\SeqMonIntSW(\SeqState)$
\begin{enumerate}
\item
\label{prop:cm:always-cacheable:mmu}
 for every $va$, $m$ and $\acc$ if  $\RealMMU(\SeqState,va,m,\acc)\!=\!(pa,c)$
 and $pa \in \Preg$ then  $c = 1$,
\item
 if $\SeqState \WTrs{\TMode}  \SeqState'$ then 
\begin{enumerate}
  \item
\label{lem:Gcrproperty}
 $\domCritc(\SeqState') \subseteq \Preg$ and
  \item
\label{lem:alwayscacheable_policy}
if $\SeqState' \SeqTrs{\TMode} \SeqState'' \touched{\inst}$ and $R$ are the
resources in $\inst$ then $R \subseteq \Preg$
\end{enumerate}
\end{enumerate}
\end{prop}
These three properties, together with \po~\ref{vc:isaInitStsfyInv}, enable us to prove that the resources accessed by the instructions executed in the privileged mode have the same value in the cache-aware and
cacheless states (\po~\ref{vc:kernel:resource:value:internal:invariant:inmodels}.\ref{vc:kernel:resource:value:inmodels}).
In a similar vein, we instantiate part (v) of the internal invariant as $\coh(\RealState', \Preg)$, and then we use \po{s}~\ref{lem:kernelCFGandVM} and \ref{prop:cm:always-cacheable}.2.a to show that the internal invariant
is preserved by kernel steps (\po~\ref{vc:kernel:resource:value:internal:invariant:inmodels}.\ref{vc:kernel:internal:invariant:inmodels}).
To discharge coherency related proof obligations for non-privileged mode,
we set $\RealInvCC$ and $\RealInvCMC$ to be
$\coh(\RealState, \Preg \cap \domCritc(\RealState))$ and
$\coh(\RealState, \Preg\!\setminus\!\domCritc(\RealState))$
respectively. This makes the proof of 
Obligations~\ref{po:monintor:inv:threeparts}.\ref{thm:invUserCCChoerence},
~\ref{vc:critical:coherency}, 
and~\ref{vc:inv:cacheless:to:cache} trivial.

We use Lemma~\ref{lem:monitorsafe} to guarantee 
that derivability preserves 
page-tables and, thus, cacheability of the resources in $\Preg$,
and we use Lemma~\ref{thm:User-No-Exfiltration} and
\po~\ref{po:type:delat:soundness} to demonstrate that 
derivability preserves $\domCritc$. 
Finally, \po~\ref{prop:cm:always-cacheable}.\ref{prop:cm:always-cacheable:mmu}
 enforces that all aliases to $\Preg$ are cacheable,
demonstrating \po
s~\ref{po:monintor:inv:threeparts}.\ref{thm:invUserCMChoerence}.

\newcommand{\clean}{\mathit{cl}}
\paragraph{Selective Eviction}
This approach requires to selectively flush the lines that correspond to the
memory locations that become critical when the kernel acquire ownership of a
region of memory. 
To verify that the kernel correctly implements this countermeasure we need to
track evicted lines, by adding to the cacheless model a history variable
$\histev$.  
\[
\begin{array}{cc}
  \frac
      {\SeqState \SeqTrs{\TMode} \SeqState' \touched{\clean(R)}
      }
      {(\SeqState, \histev) \SeqTrs{\TMode} (\SeqState', \histev\ \cup\ R)\touched{\clean(R)}} 
  \quad \quad       
  \frac
      {\SeqState \SeqTrs{\TMode} \SeqState'  \touched{\inst}\ \wedge\
        \inst \neq \clean(R)
      }
      {(\SeqState, \histev) \SeqTrs{\TMode} (\SeqState', \histev)  \touched{\inst}}
\end{array}
\]
All the kernel accesses must be restricted to resources that are either critical or have been previously cleaned (i.e.  are in $\histev$).
\begin{prop} \label{vc:selective:eviction:accesses:resources:evicted}
 For all $\SeqState$ such that  $\SeqMonIntSW(\SeqState)$
 if $(\SeqState, \emptyset) \WTrs{\TMode}  (\SeqState', \histev')$,
 $(\SeqState', \histev') \SeqTrs{\TMode} (\SeqState'', \histev'')
 \touched{\inst}$, and either $\inst = rd(R)$ or $\inst = wt(R)$ 
 then $R \subseteq \domCritc(\SeqState) \cup \histev'$
\end{prop} 

Moreover, it must be guaranteed that the set of critical resources always remains coherent.
\begin{prop}
\label{vc:kernel:resource:history}
 For all $\SeqState$ such that  $\SeqMonIntSW(\SeqState)$ and $\callConv{\SeqState}$
 if $(\SeqState, \emptyset) \WkTrs{}  (\SeqState', \histev')$
 then $\domCritc(\SeqState') \subseteq \domCritc(\SeqState) \cup \histev'$
\end{prop} 

This ensures that the kernel accesses only coherent resources and
allows to establish \po~\ref{vc:kernel:resource:value:internal:invariant:inmodels}.\ref{vc:kernel:resource:value:inmodels}.
To discharge \po~\ref{vc:kernel:resource:value:internal:invariant:inmodels}.\ref{vc:kernel:internal:invariant:inmodels}, we first define 
part (v) of the internal invariant as $\coh(\RealState', \domCritc(\RealState) \cup \histev')$ and then use
\po{s}~\ref{lem:kernelCFGandVM} and \ref{vc:selective:eviction:accesses:resources:evicted} to discharge it.
To discharge proof obligations of non-privileged mode in the cache-aware model, for state $\RealState$ the invariant must ensure
$\RealInvCC(\RealState) = \coh(\RealState, \domCritc(\RealState))$ and
$\RealInvCMC(\RealState) = true$. This makes the proof of
Obligations~\ref{po:monintor:inv:threeparts}.\ref{thm:invUserCCChoerence},~\ref{vc:critical:coherency}
and~\ref{vc:inv:cacheless:to:cache} trivial.
Part (v) of the internal invariant and \po~\ref{vc:kernel:resource:history} ensure that 
$\RealInvCC$ is established when the kernel returns to the application.


 \section{Confidentiality}
\label{sec:confidentiality}
This section presents the proof of the confidentiality property. The proof relies on the cache behaviour, hence we briefly present the cache model.
\subsection{Generic Data Cache Model}\label{sec:generic:cache}

We have formally defined a generic model which fits a number of common processor data-cache implementations. The intuition behind is that most data-caches 
are direct mapped or set-associative caches, sharing a similar structure: \begin{paraenum} \item Memory is partitioned into sets of lines which are congruent w.r.t. to a set index, \item data-caches
contain a number of ways which can hold one corresponding line for every set index, being uniquely identified by a tag, 
\item writes can make lines dirty, i.e., potentially 
different from the associated data in memory, \item there is a small set of common cache operations, e.g., filling the cache with a line from memory, \item an eviction policy controls the allocation of ways for new lines, and the eviction of old lines if the cache is full\end{paraenum}.


In addition to the cache contents, partitioned into line sets, we keep history $\his \in \His$ of internal cache actions performed for each line set. An action $a\in \Act$ can be \begin{paraenum}\item a read or write access to a present line, \item a line eviction, or \item a line fill\end{paraenum}. All actions also specify the tag of the corresponding line.

As each line set can only store limited amounts of entries, eviction policy $\fevict?(H,t)$ returns the tag of the line to be replaced at a line fill for a given tag $t$, or $\bot$ if eviction is not required. Evicted dirty lines are then written back into the memory. 

We assume here that the eviction policy is only depending on the action history of a line set and the tag of the line to be filled in. This is trivially the case for direct-mapped caches, where each line set at most contains one entry and no choice of eviction targets is possible, and for purely random replacements strategies, where the cache state and history is completely ignored. Also more sophisticated eviction policies, like LRU or pseudo-LRU, usually only depend on preceding operations on the same line set.

Another observation for these replacement strategies is that they only depend on finite subsequences of the history. For instance, LRU only depends on the past actions on the entries currently present in a given set and after a cache flush subsequent line fills and evictions are not influenced by operations from before the flush. We formalize this notion as the \emph{filtered history} on which a given eviction policy depends, computed from the full history by policy-specific filter $\filter:\His \to \His$. The idea is that the eviction policy makes the same decisions for filtered and unfiltered histories.

\begin{assumption}
  \label{lem:lruEvctSameTag}
  If $\filter$ is a policy-specific filter then
  for all $\his\in\His$ and tag $t$ holds, $\fevict?(\his, t) = \fevict?(\filter(\his), t)$.
\end{assumption}
We use this property in combination with our countermeasures in the confidentiality proof, in order to make cache states indistinguishable for the attacker.
A complete formalization of the cache state and semantics is omitted here, as these details are irrelevant for the general proof strategy outlined below (see Appendix for the detailed model).


%
\subsection{Observations in the Cache-Aware Model}
\label{sec:conf_proof}
The kernel and untrusted application share and control parts of the memory and the caches.
Our goal is to ensure that through these channels the application cannot infer anything about the confidential resources of the kernel.
%
We fix the confidential resources of the kernel as a static set $\Conf \subset
\resSpc$ and demand that they cannot be directly accessed by the application
(\po~\ref{prop:monitor}). 

The observations of the application are
over-approximated by the set $\Obs = \{r\mid r \notin \Conf\}$. Note, that some
of the observable resources may be kernel resources that are not directly
accessible by the application, but affect its behaviour (e.g., page tables).
Two cacheless states $\SeqState_1$ and $\SeqState_2$ are considered
observationally equivalent, if all resources in $\Obs$ have the same core-view.
Formally we define $\SeqState_1 \ObsEq{\Obs} \SeqState_2 \equiv \forall r \in
\Obs.\ \Exe(\SeqState_1,r) = \Exe(\SeqState_2,r)$. 

In the cache-aware model, naturally, also the cache has to be taken into
account. Specifically, for addresses readable by the attacker, both the
corresponding cache line and  underlying main memory content can be extracted using uncacheable aliases. For all other cache lines, we overapproximate the
observable state, assuming that the attacker can infer whether a tag is present in the cache (\emph{tag state}\footnote{Tags of kernel accesses can be identified when a line set contains both a 
line of the attacker and the kernel, and subsequent secret-dependent accesses may either cause a hit on that line or a miss which evicts the attacker's line.}), measure whether an entry is dirty, and derive the 
filtered history\footnote{The history of kernel cache operations may be leaked in a similar way as the tag state, affecting subsequent evictions in the application, however only as far as the filter for the 
eviction policy allows it.} of cache actions on all line sets. For caches and memories $C_1$, $C_2$, $M_1$, and $M_2$ we denote observational equivalence w.r.t.~a set $A\subseteq \PA$ by $(C_1,M_1) \cmeq_A (C_2,M_2)$. The
relation holds if:
\begin{enumerate}
  \item the memories agree for all $pa\in A$, i.e., $M_1(pa) = M_2(pa)$,
  \item the line sets with any index $i$ in $C_1$, $C_2$:
    \begin{enumerate}
      \item agree on the tags of their entries (same tag state),
      \item agree on the dirtiness of their entries
      \item agree on the contents of those entries that have tags pointing to addresses in $A$, and
      \item agree on their filtered history (``$\filter(\his_1(i)) = \filter(\his_2(i))$'').
\end{enumerate}
\end{enumerate}
Notice that $\cmeq_A$ implies core-view equivalence for any address in $A$.

%
%
Now we distinguish observable resources which are always coherent, from potentially incoherent non-critical memory resources $\NC\subseteq \PA \cap \Obs$. Intuitively, this set contains all observable addresses that the application may access through uncacheable aliases. For coherent observable resources we introduce relation
\[\RealState_1 \ObsEq{\Obscoh} \RealState_2\ \defequiv\ \forall r\in \Obs\setminus\NC.\ \Exe(\RealState_1,r) = \Exe(\RealState_2,r)\,, \] 
and define observational equivalence for the cache-aware model:
\[\begin{array}{lcl}
   \RealState_1 \ObsEq{\Obs} \RealState_2\ &\!\!\defequiv\!\!&  \RealState_1 \ObsEq{\Obscoh} \RealState_2  \\
   &\!\! \land \!\!&(\RealState_1.\cache, \RealState_1.\mem) \cmeq_\NC (\RealState_2.\cache, \RealState_2.\mem)\,.
  \end{array}
\]
Note that we are overloading notation here and that $\ObsEq{\Obscoh}$ and $\ObsEq{\Obs}$ are equivalence relations. Allowing to apply relation $\ObsEq{\Obscoh}$ also to cacheless states, we get the following corollary. 

\begin{corollary}\label{cor:seqobscoh}
For all $\SeqState_1,\SeqState_2$, if $\SeqState_1\ObsEq{\Obs}\SeqState_2$ then $\SeqState_1\ObsEq{\Obscoh}\SeqState_2$.
\end{corollary}
%

Confidentiality (Theorem~\ref{thm:topConfidentiality}) 
is demonstrated by showing that relation
$\ObsEq{\Obs}$ is a bisimulation (i.e., it is preserved for pairs of computations
starting in observationally equivalent states). 
Below, we prove this property separately for of application and kernel steps.

\subsection{Confidentiality: Application Level}
\label{subsec:userConf}
As relation $\ObsEq{\RealObs}$ is countermeasure-independent, 
verification for application's transition can be shown once and for all for a given
hardware platform, assuming the kernel invariant guarantees  several properties
for  the hardware configuration. First, the hardware monitor must ensure that confidential resources are only readable in privileged mode.
\begin{prop}\label{prop:monitor}
For all $\RealState$ such that $\RealMonIntSW(\RealState)$ if $\resource \in
\Conf$ then $\neg \RealMonitor(\RealState, \resource, \UMode, \rdd)$.
\end{prop}

Secondly, the invariant needs to guarantee that the hardware monitor data is always observable. This implies that in observationally equivalent states the same access permissions are in place.
\begin{prop}\label{prop:samemmu}
 For all $\RealState$, if $\RealMonIntSW(\RealState)$ then $\monDom(\RealState) \subset \Obs$.
\end{prop}


In addition, the monitor never allows the application to access coherent resources through uncacheable aliases.
\begin{prop}\label{prop:nc}
For all $\RealState$ such that $\RealMonIntSW(\RealState)$ for every
$\mathit{va}$, access right $\acc$ if \linebreak
$\RealMMU(\RealState,\mathit{va},\UMode,\acc)=(\mathit{pa},c)$ 
and $\mathit{pa} \in \PA\setminus\NC$ then $c=1$.
\end{prop}

By this property it is then easy to derive the coherency of resources outside of $\NC$, assuming that the hardware starts in a coherent memory configuration and that the kernel never makes memory incoherent itself (\po~\ref{lem:kernelCFGandVM}.\ref{lem:ptFixedPriv}).
\begin{prop}\label{prop:cohnotnc}
For all $\RealState$, if $\RealMonIntSW(\RealState)$ then $\coh(\RealState,\PA\setminus\NC)$.
\end{prop}



Finally, it has to be shown that non-privileged cache operations preserve the equivalence of coherent and incoherent resources.

\begin{lemma}\label{prop:cacheeq}
 Given a set of potentially incoherent addresses $A\subset \Obs\cap \PA$ and cache-aware states $\RealState_1$ and $\RealState_2$ such that \begin{enumerate}
   \item $\coh(\RealState_1,(\Obs\cap \PA)\setminus A)$ and $\coh(\RealState_2,(\Obs\cap \PA)\setminus A)$, 
   \item $\forall pa \in (\Obs\cap \PA)\setminus A.\ \Exe(\RealState_1,pa)= \Exe(\RealState_2,pa)$, and 
   \item $(\RealState_1.\cache, \RealState_1.\mem) \cmeq_{A} (\RealState_2.\cache, \RealState_2.\mem)$, 
\end{enumerate} if $\RealState_1 \RealTrs{\UMode} \RealState_1'\touched \inst$ and
$\RealState_2 \RealTrs{\UMode}\RealState_2'\touched \inst$, and 
$\inst$ is
cacheable, 
then \begin{enumerate}
\item $\coh(\RealState_1',(\Obs\cap \PA)\setminus A)$ and $\coh(\RealState_2',(\Obs\cap \PA)\setminus A)$,
\item $\forall pa \in (\Obs\cap \PA)\setminus A.\ \Exe(\RealState_1',pa)= \Exe(\RealState_2',pa)$, and
\item $(\RealState_1'.\cache, \RealState_1'.\mem) \cmeq_A (\RealState_2'.\cache, \RealState_2'.\mem)$.
\end{enumerate}
\end{lemma}

This lemma captures three essential arguments about the underlying hardware: \begin{paraenum} \item on coherent memory, observational equivalence is preserved w.r.t.~the core-view and thus independent of whether the data is currently cached, \item cacheable accesses cannot break the coherency of these resources, and \item on potentially incoherent memory (addresses $A$), observational equivalence is preserved if cache and memory are equivalent for the corresponding lines, i.e., they have the same tag states, contents, and filtered history\end{paraenum}. 

These properties inherently depend on the specific cache architecture, in particular on the eviction policy and its history filter. For any two equal filtered histories the eviction policy selects the same entry to evict (Assumption~\ref{lem:lruEvctSameTag}), therefore corresponding cache states stay observationally equivalent. Moreover, they rely on the verification conditions that evictions do not change the core-view of a coherent memory resource and that cache line fills for coherent addresses read the same values from main memory. 

We have formally verified Lemma~\ref{prop:cacheeq} for an instantiation of our cache model that uses an LRU replacement policy (see Appendix).

Based on the proof obligations lined out above we can now prove confidentiality for steps of the application
\begin{lemma}
\label{thm:User-No-Infiltration}
  For all $\RealState_1$, $\RealState_2$, $\RealState_1'$ where $\RealMonIntSW(\RealState_1)$ and $\RealMonIntSW(\RealState_2)$ hold,
  if $\RealState_1 \ObsEq{\RealObs} \RealState_2$ and $\RealState_1 \RealTrs{\UMode}{} \RealState'_1$, then $\exists\RealState_2'.\ \RealState_2 \RealTrs{\UMode}{} \RealState'_2$ and $\RealState'_1 \ObsEq{\Obs} \RealState'_2$.
\end{lemma}

\begin{prove}{~Lemma~\ref{thm:User-No-Infiltration}} 
We perform a case split over all possible hardware steps in non-privileged mode. Observational equivalence on cache, memory, and program counter
provide that the same instruction is fetched and executed in both steps.\footnote{We do not model instruction caches in the scope of this work, but assume a unified data and instruction cache. In fact, instruction caches allow further storage-channel-based attacks~\cite{DBLP:conf/sp/GuancialeNBD16} that would need to be addressed at this point in the proof.}
For hardware transitions that do not access memory, we conclude as in the proof on the cacheless model (\po~\ref{thm:seqConfCoh}, similar theorems were proved in \cite{Schwarz2016}). 

In case of memory instructions, since  general purpose registers are equivalent, the same addresses are computed for the operation. Obligations~\ref{prop:monitor} and~\ref{prop:samemmu} yield that the same access permissions are in place in $\RealState_1$ and $\RealState_2$ such that the application can only directly read observationally equivalent memory resources. Then we distinguish cacheable and uncacheable accesses to an address $\mathit{pa}$. 

In the latter case, by Obligation~\ref{prop:nc} and we know that $\mathit{pa}\in\NC$, and we obtain $\RealState_1.\mem(\mathit{pa}) =\RealState_2.\mem(\mathit{pa})$ from $\RealState_1 \ObsEq{\Obs} \RealState_2$. Now, since the access bypasses the caches, they are unchanged. Moreover, the memory accesses in both states yield the same result and the equivalence of cache and memory follows trivially. 

In case of cacheable accesses we know by Obligation~\ref{prop:cohnotnc} that resources $(\Obs\cap \PA)\setminus \NC$ are coherent and we apply Lemma~\ref{prop:cacheeq} to deduce the observational equivalence of cache and memory for addresses in $(\Obs\cap \PA)\setminus \NC$ and $\NC$. As also register resources are updated with the same values, we conclude $\RealState_1' \ObsEq{\Obs} \RealState_2'$.
\end{prove}


 \subsection{Confidentiality: Kernel Level}
\label{subsec:kernelConf}
Kernel level confidentiality ensures that the execution of kernel steps does not leak confidential information to the application. Specifically, this entails showing that at the end of the kernel execution observational equivalence of resources in $\Obs$ is (re)established.
First, the kernel should not leak confidential resources in absence of caches:
\begin{prop}
\label{thm:seqConfCoh}
For all $\SeqState_1$, $\SeqState_2$, $\SeqState_1'$ such that
$\SeqMonIntSW(\SeqState_1)$, $\SeqMonIntSW(\SeqState_2)$
$\callConv{\SeqState_1}$, 
$\callConv{\SeqState_2}$,
and $\SeqState_1 \ObsEq{\Obs} \SeqState_2$
if $\SeqState_1 \WkTrs{} \SeqState_1'$  then $\exists \SeqState_2'\ .\ \SeqState_2 \WkTrs{} \SeqState'_2$ and $\SeqState'_1 \ObsEq{\Obs} \SeqState'_2$.
\end{prop}

The goal is now to show that the chosen countermeasure against information leakage through the caches allows transferring the confidentiality property to the cache-aware model. Formally,
this countermeasure is represented by a two-state property $\CM(\SeqState,\SeqState')$ on the cacheless and $\rCM(\RealState,\RealState')$ on the cache-aware model. Here the first argument 
is the starting state of the kernel execution, while the second argument is some arbitrary state that is reached from there by a privileged computation. Property $\CM(\SeqState,\SeqState')$
should only cover functional properties of the countermeasure that can be verified on the cacheless model as part of the kernel's internal invariant.
\begin{prop}\label{prop:cminv}
For all $\SeqState,\SeqState'$ with $\SeqState \SeqTrs{\TMode}^\ast \SeqState'$ and $\callConv{\SeqState}$, if $\SeqMonIntSWPriv(\SeqState,\SeqState')$ then $\CM(\SeqState,\SeqState')$
\end{prop}
The countermeasure property on the cache-aware model, on the other hand, extends $\CM$ with conditions on the cache state that prevent information leakage to the application. We demand that it can be established through the bisimulation between cacheless and cache-aware model for a given countermeasure. 
\begin{prop}\label{prop:rcminv}
For all $\SeqState,\SeqState',\RealState,\RealState'$ where $\SeqState \SeqTrs{\TMode}^\ast \SeqState'$, $\RealState \RealTrs{\TMode}^\ast \RealState'$, $\callConv{\SeqState}$, $\RealState \SimR \SeqState$, and $\RealState' \SimR \SeqState'$, if $\SeqMonIntSWPriv(\SeqState,\SeqState')$ then $\rCM(\RealState,\RealState')$ 
\end{prop}
%


Hereafter, to enable transferring non-interference properties from the
cacheless model to the cache-aware model we assume that the transition relations
$\SeqTrs{\TMode}$ are total functions for both models.
As we want to reuse Lemma~\ref{lem:privIntegByMiniInv} in the following proofs, we require a number of properties of the simulation relation and the invariants.
\begin{lemma}\label{lem:obsrsim}
For all $\RealState_1,\RealState_2, \SeqState_1,\SeqState_2$, such that $\RealMonIntSW(\RealState_1)$, $\RealMonIntSW(\RealState_2)$, $\RealMonIntSW(\SeqState_1)$, and $\RealMonIntSW(\SeqState_2)$ as well as $\RealState_1 \SimR \SeqState_1$ and $\RealState_2 \SimR \SeqState_2$:
\begin{center} (1)\ \  $\RealState_1 \ObsEq{\Obs} \RealState_2 \Rightarrow \SeqState_1 \ObsEq{\Obs} \SeqState_2$\qquad (2)\ \ $\RealState_1 \ObsEq{\Obscoh} \RealState_2 \Leftrightarrow \SeqState_1 \ObsEq{\Obscoh} \SeqState_2$\end{center}
\end{lemma}
The properties follow directly from the definition of $\SimR$, the coherency of resources in $\Obs\setminus\NC$ (Obligation~\ref{prop:cohnotnc}), and Lemma~\ref{lem:coherency1}.


\begin{figure}
 \centering
 \includegraphics[width=0.7\linewidth]{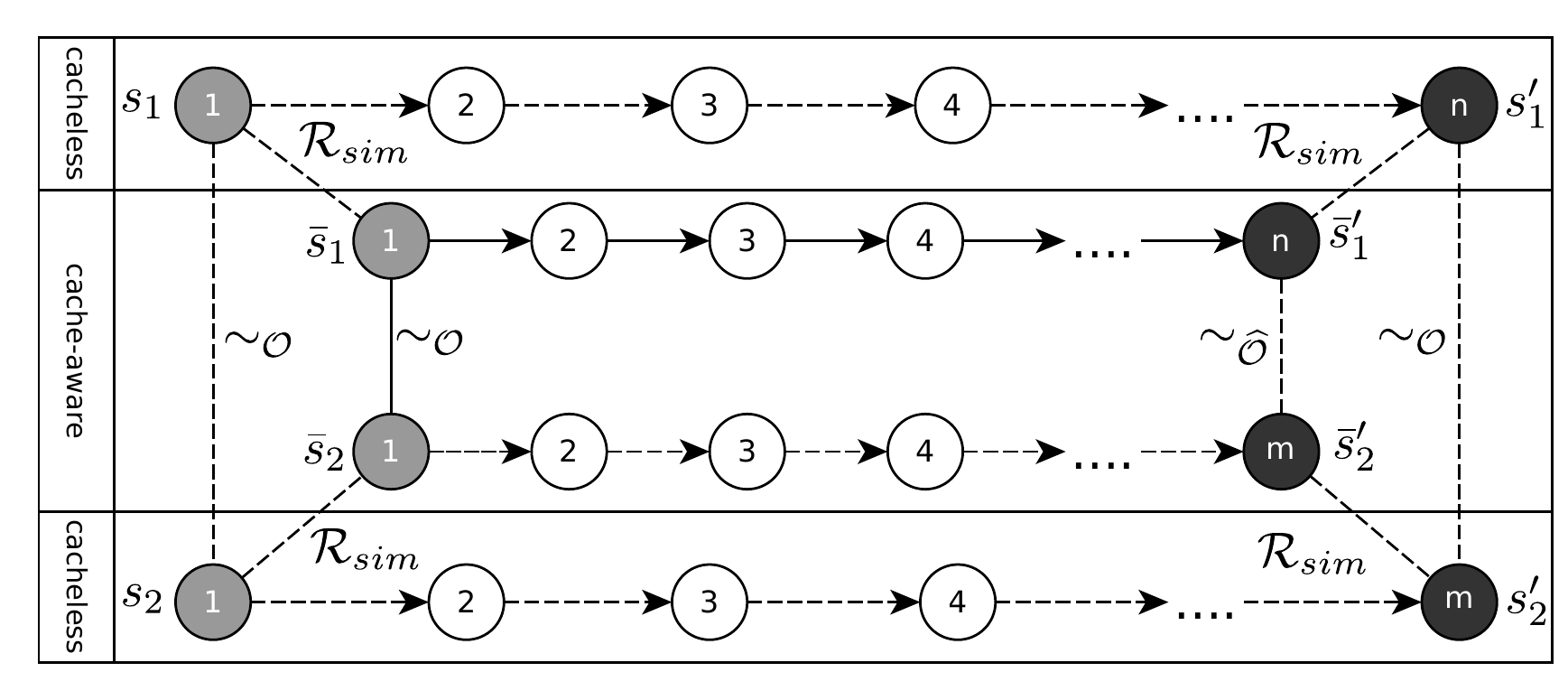}
 \caption{Transferring confidentiality property. Dashed lines indicate proof obligations.}
 \label{fig:confR}
\end{figure}

In addition we require the following technical condition on the relation of the cacheless and cache-aware invariant.
\begin{prop}\label{prop:rsiminv}
For each cache-aware state $\RealState$ with $\RealMonIntSW(\RealState)$, there exists a cacheless state $\SeqState$ such that $\RealState \SimR \SeqState$ and $\SeqMonIntSW(\SeqState)$.
\end{prop}

Now we can use Lemma~\ref{lem:privIntegByMiniInv} to transfer the equivalence of coherent resources after complete kernel executions.

\begin{lemma}
\label{thm:conftrans}
For all $\RealState_1$, $\RealState_1'$, $\RealState_2$ such that 
$\RealMonIntSW(\RealState_1)$, $\RealMonIntSW(\RealState_2)$, $\callConv{\RealState_1}$, $\callConv{\RealState_2}$, and $\RealState_1 \WkTrs{} \RealState_1'$, if $\RealState_1 \ObsEq{\Obs} \RealState_2$ then $\exists\RealState_2'.\ \RealState_2 \WkTrs{} \RealState'_2$ and $\RealState_1' \ObsEq{\Obscoh} \RealState_2'$.
\end{lemma}

\begin{prove}{}We follow the classical approach of mapping both initial cache-aware states to simulating ideal states by \po~\ref{prop:rsiminv}, transferring observational equivalence by 
Lemma~\ref{lem:obsrsim}.1, and using Lemma~\ref{lem:privIntegByMiniInv} to obtain a simulating computation $\SeqState_1 \WkTrs{} \SeqState_1'$ for $\RealState_1 \WkTrs{} \RealState_1'$ on 
the cacheless level (cf. Fig.~\ref{fig:confR}). Obligation~\ref{thm:seqConfCoh} yields an observationally equivalent cacheless computation $\SeqState_2 \WkTrs{} \SeqState_2'$ starting in the
state simulating $\RealState_2$. Applying Lemma~\ref{lem:privIntegByMiniInv} again results in a corresponding cache-aware computation $\RealState_2 \WkTrs{} \RealState_2'$. We transfer
equivalence of the coherent resources in the final states using Corollary~\ref{cor:seqobscoh} and Lemma~\ref{lem:obsrsim}.2.\end{prove}

It remains to be shown that the kernel execution also maintains observational equivalence for the potentially incoherent memory resources. This property depends on the specific countermeasure used, but can be proven once and for all for a given hardware platform and cache architecture.
\begin{prop}
\label{lem:eqcache}
For all $\RealState_1$, $\RealState_1'$, $\RealState_2$, and $\RealState_2'$ with $\RealState_1 \WkTrs{} \RealState_1'$ and $\RealState_2 \WkTrs{} \RealState'_2$ as well as $\RealMonIntSW(\RealState_1)$, $\RealMonIntSW(\RealState_2)$, $\callConv{\RealState_1}$, $\callConv{\RealState_2}$, $\rCM(\RealState_1, \RealState_1')$ and $\rCM(\RealState_2,\RealState_2')$, if $\RealState_1 \ObsEq{\Obs} \RealState_2$ then 
\begin{center}$(\RealState_1'.\cache, \RealState'_1.\mem) \cmeq_\NC (\RealState'_2.\cache, \RealState'_2.\mem)$.\end{center}
\end{prop}


Finally, we combine above results to prove that the confidentiality of the system is preserved on the cache-aware level. 

\begin{prove}{~Theorem~\ref{thm:topConfidentiality}}
Performing an induction on the computation length, we distinguish application and complete kernel steps. In the first case we apply Lemma~\ref{thm:User-No-Infiltration}. The invariants on both states are preserved by Lemmas~\ref{lem:monitorsafe}, \ref{lem:userswitchlem}.\ref{thm:trans2drvbl}, and~\ref{thm:invUser}. For kernel computations, we first step into non-privileged states $\RealState_1''$, $\RealState_2''$ deducing $\RealMonIntSW(\cdot)$ and $\callConv{\cdot}$ for both states by Theorem~\ref{thm:topIntegrityUser}, as well as $\RealState_1''\ObsEq{\Obs}\RealState_2''$ by Lemma~\ref{thm:User-No-Infiltration}. By Lemma~\ref{thm:conftrans} we show the existence of a cache-aware computation $\RealState_2'' \WkTrs{} \RealState_2'$ that preserves the equivalence of coherent resources. Using Lemma~\ref{lem:privIntegByMiniInv} we obtain a corresponding cacheless computation. By Obligation~\ref{prop:rcminv} we get $\rCM(\RealState_1'', \RealState_1')$ and $\rCM(\RealState_2'', \RealState_2')$. Invariants $\RealMonIntSW(\RealState_1')$ and $\RealMonIntSW(\RealState_2')$ hold due to Obligation~\ref{lem:invPrevsISA} and Lemma~\ref{lem:privInvEnd}. We conclude by Obligation~\ref{lem:eqcache}.
\end{prove}

The core requirements that enable this result are Proof Obligations~\ref{prop:cminv},~\ref{prop:rcminv}, and~\ref{lem:eqcache}. Below we discuss how to discharge them for two common countermeasures.

\subsection{Correctness of Countermeasures}\label{confidentiality:correctnessOfCountermeasure}

For a given cache leakage countermeasure, one needs to define predicates $\CM$, $\rCM$ and discharge the related proof obligations. Here we describe the countermeasures of secret-independent code and cache flushing that are well-known to work against cache timing channels. However, in the presence of uncacheable aliases they are not sufficient to prevent information leakage. In particular, if a secret-dependent value in memory is overshadowed by secret-independent dirty cache entry for the same address, the application can still extract the memory value by an uncacheable read. 

To tackle this issue we require that all countermeasures sanitise writes of the kernel to addresses in $\NC$ by a cache cleaning operation. Then if the core-views of $\NC$ are equivalent in two related states, so are the memories. A sufficient verification condition on the cacheless model is imposed as follows.
\begin{prop}\label{prop:satanize}
For all $\SeqState$, $\SeqState'$ such that $\CM(\SeqState, \SeqState')$ and $\SeqState \WkTrs{} \SeqState'$ performs $\inst_1\cdots\inst_n$ with $\inst_i=\wt(R_i)$ and $\mathit{pa}\in R_i$, if $\mathit{pa}\in\NC$ then there is a $j>i$ such that $\inst_j=\cl(R_j)$ and $\mathit{pa}\in R_j$. 
\end{prop}
This condition is augmented with countermeasure-specific proof obligations that guarantee the observational equivalence of caches.

\subsubsection{Secret-Independent Code}

The idea behind this countermeasure is that no information can leak through the data caches, if the kernel accesses the same memory addresses during all computations from observationally equivalent entry states, i.e., the memory accesses do not depend on the confidential resources. We approximate this functional requirement by predicate $\CM(\SeqState,\SeqState')$:

\begin{prop}\label{prop:seqsecind}
If a complete kernel execution $\SeqState \WkTrs{} \SeqState'$ performs $\inst_1~ \cdots~ \linebreak \inst_n$ on resources $R_1\cdots R_n$, then for all $\SeqState_2,\SeqState_2'$ with $\SeqState \ObsEq{\Obs} \SeqState_2$ and $\SeqState_2 \WkTrs{} \SeqState_2'$ operations $\inst_1'\cdots\inst_n'$ on $R_1'\cdots R_n'$ are performed such that addresses in $R_i$ and $R_i'$ map to the same set of tags.
\end{prop}

Note that this allows to access different addresses depending on a secret, as long as both addresses have the same tag. Such trace properties can be discharged by a relational analysis at the binary level~\cite{DBLP:conf/ccs/BalliuDG14}. On the cache-aware level, the property is refined by:
\[\rCM(\RealState, \RealState')\ \defequiv\ \forall \RealState_2, \RealState_2'.\ \RealState \WkTrs{} \RealState' \land  \RealState \cmeq_\NC \RealState_2 \land \RealState_2\WkTrs{} \RealState_2' \Rightarrow  \RealState' \cmeq_\NC \RealState_2'\]
The refinement (\po~\ref{prop:rcminv}) is proven similar to Lemma~\ref{thm:conftrans} using the Lemma \ref{lem:privIntegByMiniInv} and the determinism of the hardware when applying the confidentiality of the system in the cacheless model (Lemma~\ref{thm:seqConfCoh}) as well as \po~\ref{prop:satanize}. Then, proving \po~\ref{lem:eqcache} is straightforward.

\subsubsection{Cache Flushing}

A simple way for the kernel to conceal its secret-dependent operations is to flush the cache before returning to the application. The functional requirement $\CM(\SeqState,\SeqState')$ implies:

\begin{prop}\label{prop:flush}
For any kernel computation $\SeqState \WkTrs{} \SeqState'$ performing data operations $\inst_1\cdots\inst_n$:
\begin{enumerate}
\item there exists a contiguous subsequence $\inst_i\cdots\inst_j$ of clean operations on all cache lines,
\item operations $\inst_1\cdots\inst_j$ do not write resources in $\NC$,
\item operations $\inst_{j+1}\cdots\inst_n$ accessed address tags and written values for $\NC$ do not depend on confidential resources,
\end{enumerate} 
\end{prop}

Condition (3) is formalised and proven like the secret-independent code countermeasure discussed above. Condition (1) can be verified by binary code analysis, checking that the expected sequence of clean operations is eventually executed. We identify the resulting state by $\mathit{fl}(\SeqState)$. Condition (2) is not strictly necessary, but it reduces the overall verification effort. Then we demand by $\rCM(\RealState, \RealState')$:
\begin{definition*}
For all $\RealState''$ such that $\RealState \RealTrs{\TMode}^\ast \RealState''\RealTrs{\TMode}^\ast \RealState'$:
\begin{enumerate}
\item if $\mathit{fl}(\RealState'')$ or $\RealState''$ is a preceding state in the computation, then for all $\mathit{pa}\in \NC$ it holds that $\altExe(\RealState'',\mathit{pa})=\altExe(\RealState,\mathit{pa})$,
\item if $\mathit{fl}(\RealState'')$, all cache lines and their filtered histories are empty,
\item if $\mathit{fl}(\RealState'')$ and $\RealState''\WkTrs{} \RealState'$, then for all computations $\RealState_2\WkTrs{}\RealState_2'$ with $\RealState''\cmeq_\NC\RealState_2$ we have $\RealState'\cmeq_\NC\RealState_2'$.
\end{enumerate}
\end{definition*}

Here, Condition (1) holds as resources $\NC$ are not written and are affected only by cache evictions which preserve the memory-view. Condition (2) follow directly from the cache-flush semantics. Condition (3) is discharged using Lemma \ref{lem:privIntegByMiniInv} and \po~\ref{prop:flush}.3. In the proof of \po~\ref{lem:eqcache} we establish memory equivalence between intermediate states $\RealState_1''$ and $\RealState_2''$ where $\mathit{fl}(\RealState_1'')$ and $\RealState_2''$ using Condition (1). By Condition (2) we obtain $\RealState_1''\cmeq_\NC\RealState_2''$ and conclude by Condition (3).




 \section{Case Study}\label{sec:application}
As a case study  we use a real hypervisor capable of hosting a Linux guest along with security services that has been formally verified previously on
a cacheless model~\cite{dam2013formal,DBLP:journals/jcs/GuancialeNDB16} and
vulnerable  to attacks based on cache storage channel~\cite{DBLP:conf/sp/GuancialeNBD16}.

A hypervisor is a system software which controls access to resources and can be
used to create isolated partitions on a shared hardware. 
The hypervisor paravirtualizes the platform for several guests. Only the hypervisor executes in privileged mode, while guests entirely run in non-privileged mode and
need to invoke hypervisor functionalities to modify critical resources, such as
page-tables. The hypervisor uses direct paging~\cite{SOFSEM} to virtualize the
memory subsystem. Using this approach  a guest 
prepares a page-table in its own memory, which after validation, is used by the hypervisor to configure the MMU, without requiring
memory copy operations.
For validated page-tables the hypervisor ensures that the guest has no writable
accesses to the page-tables, thus ensuring that the guest cannot change the MMU configuration. 
This mechanism makes the hypervisor  particularly
relevant for our purpose since the critical resources change dynamically
and ownership transfer is used for communication.

To efficiently implement direct paging, the hypervisor keeps a type (either
\textit{page-table} or \textit{data}) and a reference counter for each
physical memory page. The counter
tracks (i) for a data-page the number of virtual aliases that enable
non-privileged writable
accesses to this page,  and (ii) for a page-table the number of times the page is
used as page-table. The intuition is that the
hypervisor can change type of a page (e.g., when it allocates or frees a page-table) only if the corresponding
reference counter is zero.

The security of this system can be subverted if page-tables are accessed using virtual aliases with mismatched cacheability attributes: The application can potentially mislead the hypervisor
to validate stale data and to make a non-validated page a page-table. The
hypervisor must counter this threat by preventing incoherent memory from being a
page-table. Here, we fix the vulnerability forcing the guest to create
page-tables only inside an always cacheable region of the memory.

\newcommand{\hyperMem}{\mathit{HM}}
\newcommand{\hyperType}{T}
We instantiate the general concepts of Section~\ref{subsec:seqmodel} for the
case study hypervisor. 
The hypervisor uses a static region of physical memory $\hyperMem$
to store its stack, data structures and code.
This region includes the data structure used to keep the reference counter and type 
for memory pages. Let $\hyperType(\SeqState, pa) = pt$ represent that
the hypervisor data-structure types the page containing the address $pa$
as \textit{page-table},
then the critical resources are $\domCritc(\SeqState) = \hyperMem \cup \{pa.\ \hyperType(\SeqState, pa) =
pt\}$. 

The state invariant $\RealMonIntSW$ guarantees:
(i) soundness of the reference counter, 
(ii) that the state of the system is well typed; i.e., the MMU uses only memory pages that are typed \textit{page-table},  and page-tables
forbid  non-privileged accesses to pages outside $\hyperMem$ or not typed data.
Since the hypervisor uses always cacheability, the invariant also requires (iii)
that
$\hyperMem \subseteq \Preg$, if $\hyperType(\SeqState, pa) = pt$ then $pa \in
\Preg$, coherency of $\Preg$,
and that all aliases to $\Preg$ are cacheable.

\po~\ref{po:monintor:inv:threeparts}.\ref{thm:invUserDrvbl} is 
demonstrated by showing that the functional part of the invariant only depends
on page-tables and the internal data-structures, which are critical and contained in $\hyperMem$.
\po s~\ref{po:monintor:inv:threeparts}.\ref{thm:invUserCCChoerence}
and~\ref{po:monintor:inv:threeparts}.\ref{thm:invUserCMChoerence} are
demonstrated by correctness of always cacheability.
\po~\ref{po:type:delat:soundness} trivially holds, since type of memory pages
only depends on the internal data structure of the hypervisor, which is always
considered critical.
Property (ii) of the invariant guarantees
\po~\ref{po:monintor:dom:critic}, since the MMU can use only memory blocks that are typed
page-table, which are considered critical.
The same property ensures that all critical resources are not writable by the
application, thus guaranteeing \po~\ref{vc:critical:no:write}.
Moreover, \po s~\ref{vc:critical:coherency} and~\ref{vc:kernel:resource:value:inmodels}
are  guaranteed by soundness of the
countermeasure.
\po~\ref{vc:isaInitStsfyInv} is proved by showing that the functional invariants
for the two models are defined analogously using the core-view. 
Finally, property (iii) guarantees countermeasure specific \po
s~\ref{prop:cm:always-cacheable}.1 and~\ref{prop:cm:always-cacheable}.2.a.

There remains to verify obligations on the hypervisor code. This can be done using a
binary analysis tool, since the obligations are defined solely on the cacheless model:
(\po~\ref{lem:invPrevsISA}) the hypervisor handlers are correct, preserving the
functional invariant,
(\po~\ref{lem:kernelCFGandVM}) the control flow graph of the hypervisor is correct
and the hypervisor never changes its own memory mapping,
(\po~\ref{prop:cm:always-cacheable}.2.b) all memory accesses of the hypervisor
are in the always cacheable region.

\section{Implementation}
\label{sec:impl}
We used the  HOL4 interactive theorem prover~\cite{hol4}
to validate our security analysis. We focused the validation on
Theorems~\ref{thm:topIntegrityUser} and~\ref{thm:topIntegrityPriv},
since the integrity threats posed by storage channels
cannot be countered by means external to model
(e.g., information leakage can be neutralised by introducing noise in the cache state).

Following the proof strategy of Section~\ref{sec:integrity},
the proof has been divided in three layers:
an abstract verification establishing lemmas that are platform-
and countermeasure-independent, the verification of
soundness of countermeasures, and a part that is platform-specific.
For the latter, we instantiated the models 
for both ARMv7 and ARMv8, 
using extensions of the models of Anthony Fox~\cite{DBLP:conf/itp/FoxM10, ARMV8} that include the cache model of Section~\ref{sec:generic:cache}. 
The formal specification used in our analysis consists of roughly 2500 LOC of
HOL4, and the complete proof consists of 10000 LOC.

For the case study we augmented the existing hypervisor~\cite{dam2013formal,DBLP:journals/jcs/GuancialeNDB16} with the always cacheability
countermeasure. This entailed some engineering effort to adapt the memory allocator of the Linux kernel to allocate page-tables inside $\Preg$.
The adaptation required changes to 45 LoC in the hypervisor and an addition of  35 LoC in the paravirtualized Linux kernel and imposes a negligible performance overhead. 
The formal model of the hypervisor has been modified to
include the additional checks performed by the hypervisor to prevent
allocation of page-tables outside $\Preg$ and
to forbid uncacheable aliases to $\Preg$.
Similarly, we extended the functional invariant with the
new properties guaranteed by the adopted countermeasure.
The model of the hypervisor has been used to show that the new design
 preserves the functional invariant. The invariant and the formal model of 
ARMv7 processor have been used to validate the proof obligations that do not require
binary code analysis.

We did not analyse the binary code of the hypervisor.
However, we believe that the proof of the corresponding proof obligations can be
automated to a large extent using binary analysis tools (e.g.~\cite{DBLP:journals/jcs/GuancialeNDB16}) 
or using refinement techniques (e.g.~\cite{Sewell:2013:TVV:2491956.2462183}). 

Finally, we validated the analysis of Section~\ref{subsec:userConf},
instantiating the verification strategy for the formal model of ARMv7 processor
and the flushing countermeasure. 


\section{Conclusion}
We presented an approach to verify countermeasures
for cache storage channels. We identified the conditions
that must be met by a security mechanism to neutralise the
attack vector and we verified correctness of some of the
existing techniques to counter  both integrity and confidentiality
attacks.

The countermeasures are formally modelled as new
proof obligations that can be imposed on the cacheless model to
ensure the absence of vulnerability due to cache storage channels. 
The result of this analysis are theorems in Section~\ref{sec:topgoal}.
They demonstrate that a software satisfying a set of 
proof obligations (i.e., correctly implementing the countermeasure)
is not vulnerable  because of cache storage channels. 
Since these proof obligations can be 
verified using a memory coherent setting, existing
verification tools can be used to analyse the target software. For example,
 the proof obligations required to demonstrate that a countermeasure is in place (e.g. \ref{prop:cm:always-cacheable}, \ref{vc:selective:eviction:accesses:resources:evicted}, \ref{prop:satanize}, \ref{prop:seqsecind}) can
 be verified using existing binary analysis tools.

While this paper exemplifies the approach for unified single-level data-caches, our methodology can be extended to counter leakage through timing channels and accommodate more complex scenarios and other hardware features 
too. For instance our approach can be used to counter storage channels due to enabling multi-core processing, multi-level caches, instruction caches, and TLB.

In a multi-core setting the integrity proof can be straightforwardly adopted. However, for confidentiality new countermeasures such as stealth memory or cache partitioning must be used to ensure that secret values cannot be
leaked. This entails defining new proof obligations to make sure that the countermeasures are correctly implemented and protect secret values. In the \textsc{StealthMem} approach~\cite{Kim:2012:SSP:2362793.2362804} each core is given exclusive access to a small portion
of the shared cache for its security critical computations. By ensuring that this stealth memory is always allocated in the cache, and thus no eviction can happen, one can guarantee that no storage channel can be built.

Multi-level caches can be handled iteratively in a straightforward fashion, starting from the cacheless model and adding CPU-closer levels of cache at each iteration.
Iterative refinement has three benefits: Enabling the use of existing (cache unaware) analysis tools for verification, enabling transfer of results of sections \ref{subsec:trustedIntegrity}, \ref{subsec:countermeasures},
\ref{subsec:kernelConf} and \ref{confidentiality:correctnessOfCountermeasure} to the more complex models, and allowing 
to focus on each hardware feature independently, so at least partially counteracting the pressure towards ever larger and more complex global models. For non-privileged transitions the key tools are derivability and
Lemmas \ref{lem:userswitchlem} and \ref{prop:cacheeq}. These fit new hardware features well. For instance, for separate instruction- and data-caches, coherency should be extended to require that instruction-cache hits are equal to core-view (and when this is not the case the 
address should then become attacker observable). Since derivability is not instruction dependent, Lemma \ref{lem:userswitchlem} (and consequently Theorem \ref{thm:topIntegrityUser}) can be easily re-verified.  

It worth noting that sometimes it is possible to transfer by refinement also properties for non-privileged transitions. This is the case for TLBs without virtualisation extensions: As non-privileged instructions are unable to directly
modify the TLB, incoherent behaviours can arise only by the assistance of kernel transitions. If a refinement has been proved for kernel transitions (i.e., that the TLB is properly handled by the kernel), then a refinement
can be established for the non-privileged transitions too.

In general, the verification of refinement for privileged transitions requires to show (i) that the refinement is established if a countermeasure is in place, (ii) that the countermeasure ensures the kernel to not 
leave secret dependent footprints in incoherent resources, and (iii) that the kernel code implements the countermeasures. For instance, instruction-caches require that instructions fetched by the kernel is the same in both models,
e.g., because the kernel code is not self-modifying or because the caches are flushed before executing modified code. For confidentiality, $\pc$-security can be shown to guarantee that instruction-cache is not affected by secrets.

Note also that the security analysis requires trustworthy models of hardware,
which are needed to verify platform-dependent proof obligations. Some of these
properties (e.g., Assumption~\ref{lem:lruEvctSameTag}) require extensive tests to demonstrate that corner cases are
correctly handled by models.
For example, while the conventional wisdom is that flushing caches can close side-channels, a new study~\cite{2016arXiv161204474G} showed flushing does not sanitise caches thoroughly and leaves some channels active, e.g.
instruction cache attack vectors. 


There are several open questions concerning side-channels due to similar shared
low-level hardware features such as branch prediction units,
which undermine the soundness of formal verification. This is an unsatisfactory situation since formal proofs are costly and should pay off by giving reliable guarantees.
Moreover, the complexity of contemporary hardware is such that a verification approach allowing reuse of models and proofs as new hardware features are added is essential for 
formal verification in this space to be economically sustainable. 
Our results represent a first step towards giving these guarantees in the presence of low level storage channels.


\newpage

\section{Appendix}\label{sec:appndx}

\subsection{Generic Data Cache Model}\label{sec:generic:cache}

Below with give the definition of our generic cache model which is the basis for discharging Assumptions~\ref{lem:lruEvctSameTag},~\ref{prop:cacheeq} and Proof Obligations~\ref{prop:rcminv},~\ref{lem:eqcache}.
 
Our cache model does not fix the cache size, the number of ways, the format of set indices, lines, tags, or the eviction policy. Moreover, cache lines can be indexed according to physical or virtual addresses, but are physically tagged, and our model allows both write-back and write-through caching. 

Let  $n$ and $m$ are the bitsize of virtual and physical addresses, then $\VA = \B^{n-\wds}$ and $\PA=\B^{m-\wds}$ are 
the sets of word-aligned addresses where $2^\wds$ is the word size in number of bytes (e.g. in a 64-bit architecture $\wds = 3$).
Our cache has $2^N$ sets, its lines are of size  $2^{L}$ words 
and we use functions $\si : \VA \times \PA \to \B^{N}$, $\tags : \PA \to \T$, and $\widx : \VA \times \PA \to \B^{L}$
to compute the set indices, tags and word indices of the cache. We have $T = m-(N+L)$ for physically-indexed and $T = m-L$ for virtually-indexed caches. 

We define a \emph{cache slice} as a mapping from a tag to a cache line,  $\ssize{\S= \T \to \L \cup \{\bot\}}$; $\bot$ if no line is cached for the given tag.
A line $\ssize{ln\in \mathbb{L}}$ is a pair $(ln.D,ln.d) \in \D \times \B$, $\D$ is the mapping from word-indices to data, and $d$ indicates the line dirtiness. 

Then, a cache $C\in\Csh$ maps a set index $i$ to a slice $C(i).\SL \in  \S$ and a history of actions performed $C(i).\Chis \in \Act^\ast$ on that slice, i.e., 
$\ssize{\Csh = \B^N \to \S \times \Act^\ast}$. The history records internal cache actions of type $\Act = \{\touch_{\mathrm{r}\mid\mathrm{w}}\,t, \evict\,t, \lfill\,t\}$, where $t\in\T$ denotes the tag associated with the action. 

Here, $\touch_{\mathrm{r}\mid\mathrm{w}}\,t$ denotes a read or write access to a line tagged with $t$, $\lfill\,t$ occurs when a line for tag $t$ is
loaded from memory and placed in the cache. Similarly, $\evict\,t$ represents the eviction of a line with tag $t$. To simplify formalization of properties, we define a number of abbreviations
in Figure~\ref{fig:abber}. Note that in case of virtual indexing, we assume an unaliased cache, i.e., each physical address is cached in at most one cache slice.
\begin{figure}[]
  \centering
\[
 \begin{array}{c}
 \fslc{C}{i}  \equiv C(i).\SL,           \quad 
 \fway{C}{i}{t}  \equiv (\fslc{C}{i})(t),     \quad
 \fhis{C}{i}  \equiv C(i).\Chis, \\
 \cdSl{C}{pa} \equiv \fway{C}{i}{t}.D(wi),  \quad  
 \fhit{C}{pa} \equiv  \fway{C}{i}{t} \neq \bot, \\
 \fdirty{C}{pa} \equiv \fhit{C}{pa} \land \fway{C}{i}{t}.d
\end{array}
\]
\caption{Abbreviations; for $pa\!\in\!\PA$ mapped by $va\!\in\!\VA$ and where $i$, $t$ and $wi$ are the corresponding set index, tag and word index, and $C\!\in\!\Csh$.}
\label{fig:abber}
\end{figure}
\begin{figure}[]
\[
\begin{array}{lcl}
\ftouch &:& \S \times \His \times \T \times \B^L \times (\B^{\bitsize}\cup \{\bot\}) \to \S \times \His  \\
\flfill &:& \S \times \His \times \M \times  \PA \to \S \times \His \\
\fevict &:& \S \times \His \times \T  \to \S \times \His  \vspace{1mm}\\
\fwriba &:& \L \times \M \to \M \qquad \qquad \fevict?\ \ \, :\ \ \, \His \times \T \to \T  \cup \bot
\end{array}
\]
\caption{Internal operations of the cache. The fifth input of $\ftouch$ is either $\bot$ for read accesses or the value $v'\in\B^w$ being written to the cache line.}
\label{fig:icfun}
\end{figure}

The semantics of internal cache actions on a cache slice and history are given by the corresponding functions in Figure~\ref{fig:icfun}. If tag $t$ hits the cache in the line set $i$, then $\ftouch$ and $\fevict$ update the cache $C \in \Csh$ as follows, where $D_i'=\fway{C}{i}{t}.D[ wi \mapsto v']$ is the resulting line of a write operation at word index
$wi$ with value $v'$:
\begin{eqnarray*} \arraycolsep=2pt
{
  (\fslc{C'}{i},\,\fhis{C'}{i})\!:= \left\{
     \begin{array}{ll}
       (\fslc{C}{i}\hspace{56pt},\fhis{C}{i}@(\touch_{\mathtt{r}}\,t))   \!&:\!\touch_{\mathrm{r}}\,t    \\
       (\fslc{C}{i}[t \mapsto (D_i',1)]\hspace{2.5pt},\fhis{C}{i}@(\touch_{\mathrm{w}}\,t))\!&:\!\touch_{\mathrm{w}}\,t    \\
       (\fslc{C}{i}[t \mapsto \bot]\hspace{23pt},\fhis{C}{i}@(\evict\,t))          \!&:\!\evict\,t                 \\                                
     \end{array} \right.
}     
\end{eqnarray*}

For cache misses on a physical address $pa$ with tag $t$, function $\flfill$ loads memory content $v=\mem[pa + L\cdot 2^{\alpha}-1 : pa]$ and places it as a clean line into the cache:
$(\fslc{C'}{i},\fhis{C'}{i}) := (\fslc{C}{i}[t \mapsto (\lambda i. v[i],0)],\ \fhis{C}{i}@(\lfill\,t))$.

As the cache can only store limited amounts of lines, eviction policy $evict?(\his, t)$ returns the tag of the line to be replaced at a line fill for tag $t$, or $\bot$ if eviction is not required. Evicted dirty lines are then written back into the memory using function $wriba$. As explained in the main text, the eviction policy is only depending on a finite subset of the history represented by filter function $\filter:\His\to\His$ (Assumption~\ref{lem:lruEvctSameTag}). 

The definitions above provide the minimal building blocks to define a detailed cached memory system that responds to reads, writes, and eviction requests from the core. Below we give an example.

\subsection{Operational Write-back Cache Semantics}
With the help of functions in Fig.~\ref{fig:icfun} and abbreviations of Fig.~\ref{fig:abber}.
 we give semantics to a write-back cache with LRU replacement strategy. Figure~\ref{fig:mops} lists the interface available to the core to control the cache. When one of these functionalities
 is called, the cache uses the internal actions to update its state according to the requested operation. In what follows we set $t=\tags(va,pa)$ and $i=\si(va,pa)$.

\begin{figure}
\small  \medmuskip=\muexpr\medmuskip*5/8\relax \centering
$\begin{array}{ll}
\ffill  &: \Csh \times \M \to \VA \times \PA \to \Csh \times \M\\
\fread  &: \Csh \times \M \times \VA \times \PA \to \Csh \times \M \times \B^{\bitsize}\\
\fwrite &: \Csh \times \M \times \VA \times \PA \times \B^{\bitsize} \to \Csh \times \M \\
\invba &: \Csh \times \M \times \VA \times \PA \to \Csh \times \M \\
\clnba &: \Csh \times \M \times \VA \times \PA \to \Csh \times \M \\
\end{array}$
\caption{Core accessible interface. Functions $clnba$ and $invba$ clean and invalidate cache lines for given virtual and physical addresses; $\clnba$ only resets the dirty bit and writes
back dirty lines, while $\invba$ also evicts the line. Function $\ffill$ is used to pre-load lines into the cache.}
\label{fig:mops}
\end{figure}

Function $\ffill(C,M,va,pa)$ loads the cache $C\in\Csh$ by invoking $\flfill$. However, if the cache is full the eviction policy determines the line to evict to make space. 
Using $\fwriba$, the evicted line is then written back in the memory $M \in \M$. We denote this conditional eviction  by $\alloc(C,M,va,pa)$, which is defined as:
\[\arraycolsep=0.5pt\def\arraystretch{1.3} {
\left \{ \begin{array}{ll}
(C[i \mapsto \fevict(\fslc{C}{i},\fhis{C}{i},t')], \fwriba(\fway{C}{i}{t'},M)) &:\!\fevict?(\fhis{C}{i},t)=t' \\
(C,M)                                                                          &:\!\fevict?(\fhis{C}{i},t)=\bot
\end{array}
\right. }
\]
We save the result of this function as the pair $(\bar C, \bar M)$. Moreover, if an alias for the filled line is present in another cache slice, i.e., a line with the same tag, that line has to be evicted as well. We define this condition as follows:
\[{ \falias?(C,t,i,i') \equiv \exists va',pa'.\ i' = si(va', pa') \land i \neq i' \land \fway{C}{si(va', pa')}{t}\ \neq \bot }\]
Then alias detection and eviction $\falias(C,M,va,pa)$ is defined as:
\[\arraycolsep=2pt {
\left \{ \begin{array}{ll}
(C[i' \mapsto \fevict(\fslc{C}{i'},\fhis{C}{i'},t)], \fwriba(\fway{C}{i'}{t},M)) &: \falias?(C,t,i,i') \\
(C,M)                                                                         &: \ow
\end{array}
\right. }
\]
The result of this function applied to $(\bar C, \bar M)$ is saved as $(\hat C, \hat M)$. The combination of these actions with a line fill is denoted by $\mathit{fillwb}(C,M,va,pa)$ and defined below.
\[{
\begin{cases}
(\hat{C}[i \mapsto \flfill(\fslc{\hat{C}}{i},\fhis{\hat{C}}{i}, M, pa)], \hat M) &:\ \neg\fhit{C}{pa} \\
(C,M) &: \ \text{otherwise}
\end{cases}
}\]
Thus, $(\tilde C , \tilde M) = \mathit{fillwb}(C,M,pa,va)$. Now the definition of reading, writing, flushing, and cleaning the cache is straightforward, for $x = C,M,va,pa$ we have:
\[\arraycolsep=3pt {
\begin{array}{lcl}
\fread(x) &=& (\tilde{C}[i \mapsto \ftouch(\fslc{\tilde C}{i},\fhis{\tilde C}{i},t,i,\bot), \tilde M, \cdSl{\tilde C}{pa})\\
\fwrite(x,v) &=& (\tilde{C}[i \mapsto \ftouch(\fslc{\tilde C}{i},\fhis{\tilde C}{i},t,i,v), \tilde M) \\
\invba(x) &=& (C[i \mapsto \fevict(\fslc{C}{i},\fhis{C}{i},t)], \fwriba(\fway{C}{i}{t},M)) \\
\clnba(x) &=& (C[i \mapsto \fslc{C}{i}[t \mapsto \fway{C}{i}{t}[d\mapsto 0]]], \fwriba(\fway{C}{i}{t},M))
\end{array}
}\]

\newcommand{\cons}{\mathit{Cons}}
\newcommand{\flru}{\filter_{\mathit{lru}}}
\newcommand{\last}{\mathit{last}}
\newcommand{\pres}[1]{T_{#1}}
Other cache functionalities can be defined similarly. It remains to instantiate the eviction policy and its filter $\filter$. We choose the Least Recently Used (LRU) policy, which always replaces the element that was not used for the longest time if there is no space. 
In a $k$ set associative cache we model LRU as a decision queue $q \in \qu$ of size $k$. This queue maintains an order on how the tags in a cache set are accessed. In this ordering the queue's front is the one that has been touched most recently and its back points to the tag to be replaced next upon a miss (or an empty way).

We assume two functions to manipulate queue content: \begin{inparaenum} \item $\push: \qu \times \T \to \qu$ adjusts the queue to make room for the coming tag, inserting the tag at the front of the queue, and 
\item $\pop: \qu \times \T \to \qu$ removes the input tag and shifts all elements to the front. \end{inparaenum} Additionally $\tail: \qu \to \T \cup \bot$ returns the back ($k-1th$) element in the queue or $\bot$ if there is still space. We construct the queue recursively from the history $h$ with the function $\cons : h \to \qu$:
\[{
\begin{array}{lcl}
\cons(\varepsilon) &=& \emptyset \\
\cons(h@a) &=& \begin{cases}
\pop_t(\cons(h)) &:\ a=\evict\,t \\
\push_t(\cons(h)) &:\ a=\lfill\,t \\
\push_t(\pop_t(\cons(h))) &:\ a=\touch_{\mathrm{r\mid w}}\,t
\end{cases}
\end{array}
}
\]
Then the eviction policy $\fevict?$ is defined by:
\[\fevict?_{\mathit{lru}}(h,t) := \tail(\cons(h)) \]
\begin{proposition}
 On each cache slice, the LRU replacement strategy depends only on the action history for the (at most $k$) present tags, since their lines were last filled into the cache.
\end{proposition}
We capture this part of the action history through the filter function $\flru : \His \to \His$. Its definition depends on a number of helper functions. First we introduce the last action $\last(h,t)$ on a tag $t$ in history $h$ and the set of tags $\pres{h}\subset \T$ that are currently present in a slice according to history information. For action $a\in\Act$, $\tags(a)$ returns the associated tag.
\[{
\begin{array}{rcl}
\last(\varepsilon, t) &=& \bot \qquad\quad
\last(h@a,t) \ \ =\ \ \begin{cases}
a &:\ \tags(a) = t \\
\last(h,t) &:\ \text{otherwise}
\end{cases}\vspace{2mm}\\
T_h &=& \{t \mid \last(h,t) \notin \{\bot, \evict\,t\}\}
\end{array}
}\]
For a set of tags $T\subset \T$ we define the LRU filter includes the least recent fills and subsequent touches to tags in the set, but leaves out evictions and actions on irrevelant tags.
\[{
\begin{array}{rcl}
\flru(\varepsilon, T) &=& \varepsilon \qquad\qquad
\flru(h, \emptyset) \ \ =\ \ \varepsilon \vspace{2mm}\\
\flru(h@a,T) &=& \begin{cases}
\flru(h,T)@a &:\ \exists t.\ a = \touch_{\mathrm{r|w}}\, t \land t\in T\\
\flru(h,T\setminus \{t\})@a &:\ \exists t.\ a = \lfill\, t \land t\in T\\
\flru(h,T) &:\ \text{otherwise}
\end{cases}
\end{array}
}\]
Then we set $\flru(h):= \flru(h,\pres{h})$. 

\subsection{Proof of Assumption~\ref{lem:lruEvctSameTag}}

Proving the claim \[\fevict?(h, t) = \fevict?(\flru(h), t)\] boils down to the following property.
\begin{lemma}\label{thm:lruQueueEq}
 For all $h\in\His$, $\cons(h) = \cons(\flru(h))$.
\end{lemma}

It is proven by induction on the length of $h$, with one additional induction hypothesis:
\[\forall t\in\pres{h}.\ \pop_t(\cons(h),t) = \cons(\flru(h,\pres{h}\setminus\{ t\}))\] 
We need two invariants of the cache semantics for histories $h'=h@a$: \begin{inparaenum}\item if $a$ is an eviction or a touch on tag $t$, then $t\in\pres{h}$, and \item if $a$ is a line fill of tag $t$ then $t\notin \pres{h}$. \end{inparaenum} Moreover, we use idempotence and commutativity properties of applying $\flru(h,T)$ several times.
\[
\begin{array}{rcl}
\flru(h,T) &=& \flru(\flru(h,T),T) \\
\flru(\flru(h,T),T') &=& \flru(\flru(h,T'),T)
\end{array}
\]
With these arguments the detailed proof becomes a straightforward exercise. We also introduce the following lemmas for later use.
\begin{lemma}\label{lem:flrumonoton}
For two histories $h_1$, $h_2$ and two tag sets $T$ and $T'\subseteq T$, if $\flru(h_1,T) = \flru(h_2,T)$ then $\flru(h_1,T') = \flru(h_2,T')$.
\end{lemma}
We prove this property by induction on the length of the filtered histories using the definition of $\flru(h,T)$.
\begin{lemma}\label{lem:taghistinv}
All cache actions preserve the following invariant on slices $i$ of cache $C$:
\[ T_{\his(C,i)} = \{t\mid  \fway{C}{i}{t} \neq \bot \} \]
\end{lemma}
This follows directly from the semantics of the cache actions.

\subsection{Proof of Assumption~\ref{prop:cacheeq}}

Before we conduct the proof we first give a formal definition of observational equivalence of cache and memory wrt.~a set of addresses $A$. Let 
$\mathit{TG_A} = \{\tags(a) \mid a \in A\}$ and $\mathit{SI_A} = \{\si(a) \mid a \in A\}$ denote the sets of tags and set indices corresponding to addresses in $A$, then:
\[
\begin{array}{l}
  (C_1,M_1) \cmeq_A (C_2,M_2)\ \  \defequiv \\
  \qquad\begin{array}{ll}
& \forall pa \in A.\ M_1(pa) = M_2(pa)     \\ 
&\land\ \forall i, t.\ \fway{C_1}{i}{t} = \bot \Leftrightarrow \fway{C_2}{i}{t} = \bot \\
&\land\ \forall i, t.\ \fway{C_1}{i}{t}.d = \fway{C_2}{i}{t}.d \\
&\land\ \forall i \in \mathit{SI_A}, t \in \mathit{TG_A}.\\
&\qquad\qquad \fway{C_1}{i}{t} \neq \bot \implc \fway{C_1}{i}{t} = \fway{C_2}{i}{t}  \\
&\land\ \forall i.\ \filter(\fhis{C_1}{i}) = \filter(\fhis{C_2}{i})
  \end{array}
  \end{array}
\]
Now let $C_1=\RealState_1.\cache$, $C_2=\RealState_2.\cache$, $M_1=\RealState_1.\mem$, $M_2=\RealState_2.\mem$, and similar for the primed states, then it needs to be shown:

\[
\begin{array}{l}
\forall \RealState_1,\RealState_2,\RealState_1',\RealState_2',\inst. \ \\
  \qquad\begin{array}{ll}
&\coh(\RealState_1,(\Obs\cap\PA)\setminus A) \\
\land& \coh(\RealState_2,(\Obs\cap\PA)\setminus A) \\ 
\land& \forall pa \in (\Obs\cap\PA)\setminus A.\ \Exe(\RealState_1,pa)= \Exe(\RealState_2,pa) \\
\land& (C_1, M_1) \cmeq_{A} (C_2, M_2) \\
\land& \RealState_1 \RealTrs{\UMode}\RealState_1'\touched{\inst} \\
\land& \RealState_2 \RealTrs{\UMode}\RealState_2'\touched{\inst}
  \end{array}\\
\ \ \Longrightarrow\\
  \qquad\begin{array}{ll}
&\coh(\RealState_1',(\Obs\cap\PA)\setminus A) \\
\land& \coh(\RealState_2',(\Obs\cap\PA)\setminus A) \\ 
\land& \forall pa \in (\Obs\cap\PA)\setminus A.\ \Exe(\RealState_1',pa)= \Exe(\RealState_1',pa) \\
\land& (C_1', M_1') \cmeq_{A} (C_2', M_2') 
  \end{array}
  \end{array}
\]

In general, cache operations cannot make coherent resources incoherent, thus we can focus on the last two claims. All memory instructions are broken down into a sequence of internal cache actions so it suffices to make a case split on the possible cache actions $a\in\Act$. We outline the general proof strategy for each case below. 
\begin{description}
 \item{$a=\touch_{\mathrm{r}}\,t$} --- A read hit does not change contents of cache and memory at all. We only need to consider the changes to the action history of the affected slice $i$. By definition of $\flru$ we have: 

\[
\begin{array}{lcl}
\flru(\his(C_1',i)) &=& \flru(\his(C_1,i)@a) \\
                    &=& \flru(\his(C_1,i))@a \\
                    &=& \flru(\his(C_2,i))@a \\
                    &=& \flru(\his(C_2,i)@a) \\
                    &=& \flru(\his(C_2',i)) 
\end{array}
\]

 \item{$a=\touch_{\mathrm{w}}\,t$} --- The case of write hits is analoguous to the read case, with the exception that the data content and dirty bit may change. Nevertheless the written line is present in both caches with the same contents and the dirty bit becomes $1$ in both states after the write operation. Since the same value is written, we can also show the claim that the data content for tags $t\in\mathit{TG}_A$ are equal.
 \item{$a=\lfill\,t$} --- A line fill leaves the memory and dirty bits unchanged and since we have the same tag states, the line fill occurs in both caches. 

For tags $t\notin\mathit{TG}_A$ that belong to coherent addresses we know that the core-view stays unchanged because a line is only fetched if it was not present in the cache before and the memory content that was visible in the core-view of the pre-state is loaded into the cache to be visible in the core-view of the post-state.

For $t\in \mathit{TG}_A$, relation $\cmeq_{A}$ guarantees the equivalence of the memory contents directly for addresses $A$ and again the same line is filled into the cache. 

In both cases, the tag states stay equivalent because the same tag is added into the cache slice. Concerning the history of the cache slice, we get from the definition of $\flru$ with $h_1=\his(C_1,i)$ and $h_2=\his(C_2,i)$:

\[
\begin{array}{lcl}
\flru(\his(C_1',i)) &=& \flru(h_1@a)\\
                    &=& \flru(h_1@a,T_{h_1@a}) \\
                    &=& \flru(h_1@a,T_{h_1} \cup \{a\})\\
                    &=& \flru(h_1,T_{h_1})@a \\
                    &=& \flru(h_1)@a \\ 
                    &=& \flru(h_2)@a \\
                    &=& \flru(h_2,T_{h_2})@a \\
                    &=& \flru(h_2@a,T_{h_2}\cup \{a\}) \\
                    &=& \flru(h_2@a,T_{h_2@a}) \\
                    &=& \flru(h_2@a) \\
                    &=& \flru(\his(C_2',i))
\end{array}
\]

\item{$a=\evict\,t$} --- For coherent resources evictions do not change the core-view, as any line that is evicted was either dirty before and thus written back to memory, maintaining its addresses' core-view, or it was clean but coherent with the corresponding memory content that becomes visible in the core-view after the eviction. If a confidential line is evicted there is nothing more to show. 

For tags in a line $i$ tag states and filtered histories are equal. By Assumption~\ref{lem:lruEvctSameTag} the eviction policy yields the same result in both states, thus if a line is evicted it is done so in both caches and these lines have the same tag. 

For evicted coherent lines or confidential lines we argue as above. For lines belonging to the set $A$ we know that they have the same contents, so if they are dirty, memory changes in the same way. In case they are clean, memories stay unchanged and are still equivalent. 

In all cases the tag state is manipulated in the same way, as the same tags are evicted, thus they stay equal. The filtered histories for line $i$ are still the same by definition of $\flru$ and the equality of tag states.

\[
\begin{array}{lcll}
\flru(\his(C_1',i) &=& \flru(h_1@a) \\
                   &=& \flru(h_1@a,T_{h_1@a}) \\
                   &=& \flru(h_1@a,T_{h_1} \setminus \{a\})\\
                   &=& \flru(h_1,T_{h_1} \setminus \{a\}) \\ 
                   &=& \flru(h_1,T_{h_2} \setminus \{a\}) & (Lemma~\ref{lem:taghistinv})\\
                   &=& \flru(h_2,T_{h_2}\setminus \{a\})  & (Lemma~\ref{lem:flrumonoton}) \\
                   &=& \flru(h_2@a,T_{h_2}\setminus \{a\}) \\
                   &=& \flru(h_2@a,T_{h_2@a}) \\
                   &=& \flru(h_2@a) \\
                   &=& \flru(\his(C_2',i))
\end{array}
\]

\end{description}

This concludes the proof of Assumption~\ref{prop:cacheeq}.



\bibliographystyle{kthplain}
\bibliography{bib/biblio}

\end{document}